\documentclass[envcountsect,envcountsame,runningheads]{llncs}

\usepackage{xspace}
\usepackage[all]{xy}
\usepackage{wrapfig}

\usepackage{amsmath}
\usepackage{amssymb}
\usepackage{mathrsfs}
\usepackage{xcolor}

\newdimen\proofrulebreadth \proofrulebreadth=.05em
\newdimen\proofdotseparation \proofdotseparation=1.25ex
\newdimen\proofrulebaseline \proofrulebaseline=2ex
\newcount\proofdotnumber \proofdotnumber=3
\let\then\relax
\def\hfi{\hskip0pt plus.0001fil}
\mathchardef\squigto="3A3B
\newif\ifinsideprooftree\insideprooftreefalse
\newif\ifonleftofproofrule\onleftofproofrulefalse
\newif\ifproofdots\proofdotsfalse
\newif\ifdoubleproof\doubleprooffalse
\let\wereinproofbit\relax
\newdimen\shortenproofleft
\newdimen\shortenproofright
\newdimen\proofbelowshift
\newbox\proofabove
\newbox\proofbelow
\newbox\proofrulename
\def\shiftproofbelow{\let\next\relax\afterassignment\setshiftproofbelow\dimen0 }
\def\shiftproofbelowneg{\def\next{\multiply\dimen0 by-1 }%
\afterassignment\setshiftproofbelow\dimen0 }
\def\setshiftproofbelow{\next\proofbelowshift=\dimen0 }
\def\setproofrulebreadth{\proofrulebreadth}

\def\prooftree{% NESTED ZERO (\ifonleftofproofrule)
\ifnum  \lastpenalty=1
\then   \unpenalty
\else   \onleftofproofrulefalse
\fi
\ifonleftofproofrule
\else   \ifinsideprooftree
        \then   \hskip.5em plus1fil
        \fi
\fi
\bgroup% NESTED ONE (\proofbelow, \proofrulename, \proofabove,
\setbox\proofbelow=\hbox{}\setbox\proofrulename=\hbox{}%
\let\justifies\proofover\let\leadsto\proofoverdots\let\Justifies\proofoverdbl
\let\using\proofusing\let\[\prooftree
\ifinsideprooftree\let\]\endprooftree\fi
\proofdotsfalse\doubleprooffalse
\let\thickness\setproofrulebreadth
\let\shiftright\shiftproofbelow \let\shift\shiftproofbelow
\let\shiftleft\shiftproofbelowneg
\let\ifwasinsideprooftree\ifinsideprooftree
\insideprooftreetrue
\setbox\proofabove=\hbox\bgroup$\displaystyle % NESTED TWO
\let\wereinproofbit\prooftree
\shortenproofleft=0pt \shortenproofright=0pt \proofbelowshift=0pt
\onleftofproofruletrue\penalty1
}

\def\eproofbit{% NESTED TWO
\ifx    \wereinproofbit\prooftree
\then   \ifcase \lastpenalty
        \then   \shortenproofright=0pt  % 0: some other object, no indentation
        \or     \unpenalty\hfil         % 1: empty hypotheses, just glue
        \or     \unpenalty\unskip       % 2: just had a tree, remove glue
        \else   \shortenproofright=0pt  % eh?
        \fi
\fi
\global\dimen0=\shortenproofleft
\global\dimen1=\shortenproofright
\global\dimen2=\proofrulebreadth
\global\dimen3=\proofbelowshift
\global\dimen4=\proofdotseparation
\global\count255=\proofdotnumber
$\egroup  % NESTED ONE
\shortenproofleft=\dimen0
\shortenproofright=\dimen1
\proofrulebreadth=\dimen2
\proofbelowshift=\dimen3
\proofdotseparation=\dimen4
\proofdotnumber=\count255
}

\def\proofover{% NESTED TWO
\eproofbit % NESTED ONE
\setbox\proofbelow=\hbox\bgroup % NESTED TWO
\let\wereinproofbit\proofover
$\displaystyle
}%
\def\proofoverdbl{% NESTED TWO
\eproofbit % NESTED ONE
\doubleprooftrue
\setbox\proofbelow=\hbox\bgroup % NESTED TWO
\let\wereinproofbit\proofoverdbl
$\displaystyle
}%
\def\proofoverdots{% NESTED TWO
\eproofbit % NESTED ONE
\proofdotstrue
\setbox\proofbelow=\hbox\bgroup % NESTED TWO
\let\wereinproofbit\proofoverdots
$\displaystyle
}%
\def\proofusing{% NESTED TWO
\eproofbit % NESTED ONE
\setbox\proofrulename=\hbox\bgroup % NESTED TWO
\let\wereinproofbit\proofusing
\kern0.3em$
}

\def\endprooftree{% NESTED TWO
\eproofbit % NESTED ONE
  \dimen5 =0pt% spread of hypotheses
\dimen0=\wd\proofabove \advance\dimen0-\shortenproofleft
\advance\dimen0-\shortenproofright
\dimen1=.5\dimen0 \advance\dimen1-.5\wd\proofbelow
\dimen4=\dimen1
\advance\dimen1\proofbelowshift \advance\dimen4-\proofbelowshift
\ifdim  \dimen1<0pt
\then   \advance\shortenproofleft\dimen1
        \advance\dimen0-\dimen1
        \dimen1=0pt
        \ifdim  \shortenproofleft<0pt
        \then   \setbox\proofabove=\hbox{%
                        \kern-\shortenproofleft\unhbox\proofabove}%
                \shortenproofleft=0pt
        \fi
\fi
\ifdim  \dimen4<0pt
\then   \advance\shortenproofright\dimen4
        \advance\dimen0-\dimen4
        \dimen4=0pt
\fi
\ifdim  \shortenproofright<\wd\proofrulename
\then   \shortenproofright=\wd\proofrulename
\fi
\dimen2=\shortenproofleft \advance\dimen2 by\dimen1
\dimen3=\shortenproofright\advance\dimen3 by\dimen4
\ifproofdots
\then
        \dimen6=\shortenproofleft \advance\dimen6 .5\dimen0
        \setbox1=\vbox to\proofdotseparation{\vss\hbox{$\cdot$}\vss}%
        \setbox0=\hbox{%
                \advance\dimen6-.5\wd1
                \kern\dimen6
                $\vcenter to\proofdotnumber\proofdotseparation
                        {\leaders\box1\vfill}$%
                \unhbox\proofrulename}%
\else   \dimen6=\fontdimen22\the\textfont2 % height of maths axis
        \dimen7=\dimen6
        \advance\dimen6by.5\proofrulebreadth
        \advance\dimen7by-.5\proofrulebreadth
        \setbox0=\hbox{%
                \kern\shortenproofleft
                \ifdoubleproof
                \then   \hbox to\dimen0{%
                        $\mathsurround0pt\mathord=\mkern-6mu%
                        \cleaders\hbox{$\mkern-2mu=\mkern-2mu$}\hfill
                        \mkern-6mu\mathord=$}%
                \else   \vrule height\dimen6 depth-\dimen7 width\dimen0
                \fi
                \unhbox\proofrulename}%
        \ht0=\dimen6 \dp0=-\dimen7
\fi
\let\doll\relax
\ifwasinsideprooftree
\then   \let\VBOX\vbox
\else   \ifmmode\else$\let\doll=$\fi
        \let\VBOX\vcenter
\fi
\VBOX   {\baselineskip\proofrulebaseline \lineskip.2ex
        \expandafter\lineskiplimit\ifproofdots0ex\else-0.6ex\fi
        \hbox   spread\dimen5   {\hfi\unhbox\proofabove\hfi}%
        \hbox{\box0}%
        \hbox   {\kern\dimen2 \box\proofbelow}}\doll%
\global\dimen2=\dimen2
\global\dimen3=\dimen3
\egroup % NESTED ZERO
\ifonleftofproofrule
\then   \shortenproofleft=\dimen2
\fi
\shortenproofright=\dimen3
\onleftofproofrulefalse
\ifinsideprooftree
\then   \hskip.5em plus 1fil \penalty2
\fi
}

\newcommand{\indrulename}[1]{\texttt{\textup{#1}}}
\newcommand{\indrule}[3]{\ensuremath{\begin{array}{c}\prooftree#2\justifies#3\thickness=0.05em\using\indrulename{#1}\endprooftree\end{array}}}

\renewcommand{\theenumi}{\arabic{enumi}}
\renewcommand{\theenumii}{\arabic{enumii}}
\renewcommand{\theenumiii}{\arabic{enumiii}}

\makeatletter
\renewcommand\p@enumii{\theenumi.}
\renewcommand\p@enumiii{\theenumi.\theenumii.}
\renewcommand\p@enumiv{\theenumi.\theenumii.\theenumiii.}
\makeatother

\newcommand{\llem}[1]{\label{lemma:#1}}
\newcommand{\rlem}[1]{Lem.~\ref{lemma:#1}}
\newcommand{\ldef}[1]{\label{def:#1}}
\newcommand{\rdef}[1]{Def.~\ref{def:#1}}
\newcommand{\lprop}[1]{\label{prop:#1}}
\newcommand{\rprop}[1]{Prop.~\ref{prop:#1}}
\newcommand{\lthm}[1]{\label{thm:#1}}
\newcommand{\rthm}[1]{Thm.~\ref{thm:#1}}
\newcommand{\lremark}[1]{\label{remark:#1}}
\newcommand{\rremark}[1]{Rem.~\ref{remark:#1}}
\newcommand{\lcoro}[1]{\label{coro:#1}}
\newcommand{\rcoro}[1]{Coro.~\ref{coro:#1}}
\newcommand{\lsec}[1]{\label{section:#1}}
\newcommand{\rsec}[1]{Sec.~\ref{section:#1}}
\newcommand{\lexample}[1]{\label{example:#1}}
\newcommand{\rexample}[1]{Ex.~\ref{example:#1}}
\newcommand{\leqn}[1]{\label{eqn:#1}}
\newcommand{\reqn}[1]{(\ref{eqn:#1})}
\newcommand{\resultname}[1]{{\bfseries{#1}}}
\newcommand{\itemname}[1]{#1}
\newcommand{\refcase}[1]{\itemname{\ref{#1}}}
\newcommand{\defn}[1]{{\em #1}}
\newcommand{\eg}{{\em e.g.}\xspace}
\newcommand{\ie}{{\em i.e.}\xspace}
\newcommand{\cf}{{\em cf.}\xspace}
\newcommand{\ih}{{\em i.h.}\xspace}
\newcommand{\ST}{\ |\ }
\newcommand{\HS}{\hspace{.5cm}}

\renewcommand{\emptyset}{\varnothing}
\newcommand{\set}[1]{\{#1\}}

\newcommand{\eqdef}{\overset{\mathrm{def}}{=}}
\newcommand{\iffdef}{\overset{\mathrm{def}}{\iff}}
\newcommand{\eqih}{\overset{\mathrm{h.i.}}{=}}

\newcommand{\OR}{\ \mid\ }

\newcommand{\tm}{t}
\newcommand{\tmtwo}{s}
\newcommand{\tmthree}{u}
\newcommand{\tmthreevariant}{v}
\newcommand{\tmfour}{r}
\newcommand{\tmfive}{p}
\newcommand{\tmsix}{q}

\newcommand{\var}{x}
\newcommand{\vartwo}{y}
\newcommand{\varthree}{z}

\newcommand{\basetyp}{\alpha}

\newcommand{\labelset}{\mathscr{L}}
\newcommand{\lab}{\ell}
\newcommand{\labtwo}{\lab'}
\newcommand{\labthree}{\lab''}

\newcommand{\conbase}{\Box}
\newcommand{\of}[1]{\langle#1\rangle}
\newcommand{\con}{\mathtt{C}}
\newcommand{\contwo}{\con'}

\newcommand{\conof}[1]{\con\of{#1}}
\newcommand{\contwoof}[1]{\contwo\of{#1}}
\newcommand{\conhat}{\hat{\con}}

\newcommand{\typ}{\mathcal{A}}
\newcommand{\typtwo}{\mathcal{B}}
\newcommand{\typthree}{\mathcal{C}}
\newcommand{\typfour}{\mathcal{D}}

\newcommand{\lset}[1]{[#1]}

\newcommand{\mtyp}{\mathcal{M}}
\newcommand{\mtyptwo}{\mathcal{N}}
\newcommand{\mtypthree}{\mathcal{P}}

\newcommand{\tctx}{\Gamma}
\newcommand{\tctxtwo}{\Delta}
\newcommand{\tctxthree}{\Theta}

\DeclareMathOperator{\dom}{dom}

\newcommand{\lam}[2]{\lambda #1. #2}
\newcommand{\lamp}[3]{\lambda^{#1} #2. #3}
\newcommand{\subs}[3]{#1\{#2 := #3\}}
\newcommand{\sub}[2]{\{\!\!\{#1 := #2\}\!\!\}}

\newcommand{\tolab}[1]{\overset{#1}{\to}}

\newcommand{\tobeta}{\mathrel{\to_\beta}}
\newcommand{\rtobeta}{\mathrel{\twoheadrightarrow_\beta}}
\newcommand{\toabeta}[1]{\mathrel{\xrightarrow{#1}_\beta}}
\newcommand{\todist}{\mathrel{\to_\dist}}
\newcommand{\lefttodist}{\mathrel{\leftarrow_\dist}}
\newcommand{\rtodist}{\mathrel{\twoheadrightarrow_\dist}}
\newcommand{\todistl}[1]{\mathrel{\xrightarrow{#1}_\dist}}

\newcommand{\Obj}{\mathsf{Obj}}
\newcommand{\Stp}{\mathsf{Stp}}
\newcommand{\src}{\mathsf{src}}
\newcommand{\tgt}{\mathsf{tgt}}

\newcommand{\emptyDerivation}{\epsilon}
\newcommand{\redex}{R}
\newcommand{\redextwo}{S}
\newcommand{\redexthree}{T}
\newcommand{\name}{\mathsf{lab}}
\newcommand{\names}{\mathsf{labs}}
\newcommand{\redexset}{\mathcal{M}}

\newcommand{\redseq}{\rho}
\newcommand{\redseqtwo}{\sigma}
\newcommand{\redseqthree}{\tau}

\newcommand{\permeq}{\equiv}
\newcommand{\permle}{\sqsubseteq}

\newcommand{\dist}{\texttt{\textup{\#}}}
\newcommand{\lambdadist}{\lambda^\dist}
\newcommand{\terms}{\mathcal{T}}
\newcommand{\termsdist}{\mathcal{T}^{\dist}}

\newcommand{\lengthof}[1]{|#1|}
\newcommand{\fv}[1]{\mathsf{fv}({#1})}
\newcommand{\ls}[1]{\bar{#1}}

\newcommand{\refines}{\mathrel{\ltimes}}

\newcommand{\tmlabel}[1]{\mathtt{T}(#1)}
\newcommand{\varlabel}[2]{\mathtt{T}_{#1}(#2)}

\newcommand{\lst}{\ls{x}}
\newcommand{\lsttwo}{\ls{y}}

\newcommand{\anon}{\mathtt{X}}

\newcommand{\classof}[1]{{[#1]}}

\newcommand{\sieve}{\Downarrow}

\newcommand{\id}{\mathsf{id}}

\newcommand{\Poset}{\mathsf{Poset}}
\newcommand{\cls}[1]{{[#1]}}

\newcommand{\grothy}[2]{\int_{#1}{#2}}
\newcommand{\ulbDeriv}[1]{\mathbb{D}(#1)}
\newcommand{\ulbDerivLam}[1]{\mathbb{D}^{\lambda}(#1)}
\newcommand{\ulbDerivDist}[1]{\mathbb{D}^{\dist}(#1)}
\newcommand{\ulbFree}[2]{\mathbb{F}(#1, #2)}
\newcommand{\ulbGarbage}[2]{\mathbb{G}(#1, #2)}
\newcommand{\ulbF}{\mathcal{F}}
\newcommand{\ulbG}{\mathcal{G}}
\newcommand{\leqF}{\unlhd}
\newcommand{\lorF}{\triangledown}
\newcommand{\landF}{\!\vartriangle\!}

\newcommand{\identity}{\mathsf{id}}

\newcommand{\toI}{lam}
\newcommand{\toE}{app}

\newcommand{\condp}[1]{\textrm{\texttt{[c#1]}}}

\newcommand{\SeeAppendix}{\textcolor{red}{$\clubsuit$}\xspace}
\newcommand{\SeeAppendixRef}[1]{\hyperref[#1]{\SeeAppendix}}
\newcommand{\occursin}{\preceq}

\newcommand{\CreationCase}[1]{\textup{\texttt{[cr#1]}}}

\newcommand{\subterms}[1]{\mathsf{sub}(#1)}
\newcommand{\fsubterms}[1]{\mathsf{sub}^{\circ}(#1)}
\newcommand{\pt}[2]{#2 \hookrightarrow #1}
\newcommand{\ptF}[2]{#2 \hookrightarrow #1}
\newcommand{\totwocell}{\leq}

\usepackage{hyperref}
\usepackage{chngcntr}
\usepackage{apptools}
\AtAppendix{\counterwithin{theorem}{section}}

\renewcommand\itshape{\slshape}

\begin{document}
\title{Factoring Derivation Spaces\\ via Intersection Types \\ (Extended Version)\thanks{Work partially supported by CONICET.}}
\author{
Pablo Barenbaum\inst{1}%\orcidID{0000-1111-2222-3333}
  \and
Gonzalo Ciruelos\inst{2} %\orcidID{1111-2222-3333-4444}
}
\authorrunning{P.~Barenbaum and G.~Ciruelos}

\institute{
Departamento de Computaci\'on, FCEyN, UBA, Argentina \\
IRIF, Universit\'e Paris 7, France \\
\email{pbarenbaum@dc.uba.ar}
\and
Departamento de Computaci\'on, FCEyN, UBA, Argentina \\
\email{gonzalo.ciruelos@gmail.com}
}
\maketitle              % typeset the header of the contribution
\begin{abstract}
% The abstract should briefly summarize the contents of the paper in
% 15--250 words.
In typical non-idempotent intersection type systems, proof normalization is not confluent.
In this paper we introduce a confluent non-idempotent intersection type system for the $\lambda$-calculus.
Typing derivations are presented using proof term syntax.
The system enjoys good properties: subject reduction, strong normalization,
and a very regular theory of residuals.
A correspondence with the $\lambda$-calculus is established by simulation theorems.
The machinery of non-idempotent intersection types allows us to track the usage of
resources required to obtain an answer.
In particular, it induces a notion of {\em garbage}: a computation is garbage
if it does not contribute to obtaining an answer.
Using these notions,
we show that the derivation space of a $\lambda$-term may be factorized using a
variant of the Grothendieck construction for semilattices.
This means, in particular, that any derivation in the $\lambda$-calculus can be
uniquely written as a garbage-free prefix followed by garbage.

\end{abstract}
\section{Introduction}
  \lsec{introduction}
  
Our goal in this paper is attempting to
understand the spaces of computations of programs.
Consider a hypothetical functional programming language
with arithmetic expressions and tuples.
All the possible computations starting from the tuple
$(1 + 1,\ 2 * 3 + 1)$
can be arranged to form its ``space of computations'':
\[
  {\small
  \xymatrix@R=.4cm{
    (1 + 1,\ 2 * 3 + 1) \ar[d] \ar[r] & (1 + 1,\ 6 + 1) \ar[d] \ar[r] & (1 + 1,\ 7) \ar[d] \\
    (2,\ 2 * 3 + 1)     \ar[r]        & (2,\ 6 + 1)      \ar[r]       & (2,\ 7) \\
  }}
\]
In this case, the space of computations is quite easy to understand,
because the subexpressions $(1 + 1)$ and $(2 * 3 + 1)$ cannot interact with each other.
Indeed, the space of computations of a tuple $(A,B)$ can always be understood as
the {\em product} of the spaces of $A$ and $B$.
In the general case, however, the space of computations of a program may have a
much more complex structure.
For example, it is not easy to characterize the space of computations of a function application $f(A)$.
The difficulty is that $f$ may use the value of $A$ zero, one, or
possibly many times.

\medskip

The quintessential functional programming language
is the pure $\lambda$-calculus.
Computations in the $\lambda$-calculus have been thoroughly
studied since its conception in the 1930s.
The well-known theorem by Church and Rosser~\cite{church1936some}
states that $\beta$-reduction in the $\lambda$-calculus is \emph{confluent},
which means, in particular, that terminating programs have unique normal forms.
Another result by Curry and Feys~\cite{curry1958combinatory}
states that computations in the
$\lambda$-calculus may be \emph{standardized},
meaning that they may be converted into a computation in canonical form.
A refinement of this theorem by L\'evy~\cite{Tesis:Levy:1978}
asserts that the canonical computation thus obtained is equivalent to the
original one in a strong sense, namely that they are \emph{permutation equivalent}.
In a series of papers~\cite{DBLP:conf/ctcs/Mellies97,DBLP:conf/rta/Mellies02,mellies2002axiomatic,DBLP:conf/birthday/Mellies05},
Melli\`es generalized many of these results to the abstract setting of {\em axiomatic rewrite systems}.

Let us discuss ``spaces of computations'' more precisely.
The \emph{derivation space} of an object~$x$ in some rewriting system
is the set of all {\em derivations}, \ie sequences of rewrite steps, starting from $x$.
In this paper, we will be interested in the pure $\lambda$-calculus,
and we will study \emph{finite} derivations only.
In the $\lambda$-calculus,
a transitive relation between derivations may be defined, the {\em prefix order}.
A derivation $\redseq$ is a prefix of a derivation $\redseqtwo$,
written $\redseq \permle \redseqtwo$,
whenever $\redseq$ performs less computational work than $\redseqtwo$.
Formally, $\redseq \permle \redseqtwo$ is defined to hold whenever the {\em projection} $\redseq/\redseqtwo$ 
is empty\footnote{The notion of projection defined by means of residuals
is the standard one, see \eg \cite[Chapter~12]{Barendregt:1984}
or \cite[Section~8.7]{Terese}.}.
For example, if $K = \lam{\var}{\lam{\vartwo}{\var}}$,
the derivation space of the term $(\lam{\var}{\var\var})(K\varthree)$
can be depicted with the \emph{reduction graph} below.
Derivations are directed paths in the reduction graph,
and $\redseq$ is a prefix of $\redseqtwo$ if there is a directed path from the
target of $\redseq$ to the target of $\redseqtwo$.
For instance, $\redextwo\redex_2$ is a prefix of $\redex\redextwo'\redexthree'$:
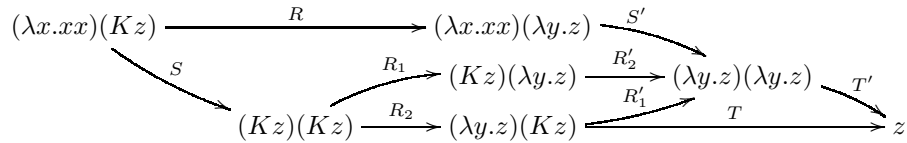
\[
  \xymatrix@R=.1cm{
    (\lam{\var}{\var\var})(K\varthree)
      \ar^-{\redex}[rr]
      \ar@/_.2cm/^-{\redextwo}[rdd]
  &
  &
    (\lam{\var}{\var\var})(\lam{\vartwo}{\varthree})
    \ar@/^.3cm/[rd]^{\redextwo'}
  &
  \\
  &
  &
    (K\varthree)(\lam{\vartwo}{\varthree}) \ar[r]^-{\redex'_2}
  &
    (\lam{\vartwo}{\varthree})(\lam{\vartwo}{\varthree})
    \ar@/^.2cm/[rd]^-{\redexthree'}
  \\
  &
    (K\varthree)(K\varthree) \ar@/^.3cm/[ru]^{\redex_1} \ar[r]^-{\redex_2}
  &
    (\lam{\vartwo}{\varthree})(K\varthree)
      \ar@/_.2cm/[ru]^-{\redex'_1}
      \ar[rr]^-{\redexthree}
  &
  &
    \varthree
  }
\]
Remark that $\permle$ is reflexive and transitive but not antisymmetric,
\ie it is a quasi-order but not an order.
For example $\redex\redextwo' \permle \redextwo\redex_1\redex_2' \permle \redex\redextwo'$
but $\redex\redextwo' \neq  \redextwo\redex_1\redex_2'$.
Antisymmetry may be recovered as usual when in presence of a quasi-order,
by working modulo {\em permutation equivalence}:
two derivations $\redseq$ and $\redseqtwo$
are said to be permutation equivalent,
written $\redseq \permeq \redseqtwo$, if $\redseq \permle \redseqtwo$
and $\redseqtwo \permle \redseq$.
Working modulo permutation equivalence is reasonable because L\'evy's
formulation of the
standardization theorem ensures that permutation equivalence is decidable.

Derivation spaces are known to exhibit various
regularities~\cite{Tesis:Levy:1978,DBLP:journals/tcs/Zilli84,laneve1994distributive,thesismellies,levy_redex_stability,DBLP:conf/lics/AspertiL13}.
In his PhD thesis, L\'evy~\cite{Tesis:Levy:1978} showed that the derivation space of a term is
an upper semilattice: any two derivations $\redseq,\redseqtwo$ from a term $\tm$
have a {\em least upper bound} $\redseq \sqcup \redseqtwo$, defined as $\redseq(\redseqtwo/\redseq)$,
unique up to permutation equivalence.
On the other hand, the derivation space of a term $\tm$ is not an easy structure to understand in general\footnote{Problem~2 in
the RTA List of Open Problems~\cite{dershowitz1991open} poses
the open-ended question of investigating the properties of ``spectra'', \ie derivation spaces.}.
For example, relating the derivation space of an application $\tm\tmtwo$
with the derivation spaces of $\tm$ and $\tmtwo$ appears to be a hard problem.
L\'evy also noted that the {\em greatest lower bound} of two derivations does
not necessarily exist, meaning that the derivation space of a term does not form a lattice in general.
Even when it forms a lattice, it may not necessarily be a {\em distributive} lattice, as observed for example by
Laneve~\cite{laneve1994distributive}.
In \cite{DBLP:conf/ctcs/Mellies97}, Melli\`es showed that
derivation spaces in any rewriting system satisfying certain axioms
may be factorized using two spaces,
one of {\em external} and one of {\em internal} derivations.

The difficulty to understand derivation spaces is due to three pervasive phenomena
of \emph{interaction} between computations.
The first phenomenon is \emph{duplication}:
in the reduction graph of above, the step $\redextwo$ duplicates the step $\redex$,
resulting in two copies of $\redex$: the steps $\redex_1$ and $\redex_2$.
In such situation, one says that $\redex_1$ and $\redex_2$ are \emph{residuals} of $\redex$,
and, conversely, $\redex$ is an \emph{ancestor} of $\redex_1$ and $\redex_2$.
The second phenomenon is \emph{erasure}:
in the diagram above, the step $\redexthree$ erases the step $\redex'_1$,
resulting in no copies of $\redex'_1$.
The third phenomenon is \emph{creation}:
in the diagram above, the step $\redex_2$ creates the step $\redexthree$,
meaning that $\redexthree$ is not a residual of a step that existed prior
to executing $\redex_2$; that is, $\redexthree$ has no ancestor.

These three interaction phenomena, especially duplication and erasure,
are intimately related with the management of \emph{resources}.
In this work, we aim to explore the hypothesis that {\bf having an explicit
representation of resource management may provide insight on
the structure of derivation spaces}.

There are many existing $\lambda$-calculi that deal with resource management explicitly~\cite{boudol1993lambda,ehrhard2003differential,kesner2007resource,kesner2009prismoid},
most of which draw inspiration from Girard's Linear Logic~\cite{girard1987linear}.
Recently, calculi endowed with
{\em non-idempotent intersection type systems},
have received some attention~\cite{ehrhard2012collapsing,bernadet2013non,bucciarelli2014inhabitation,bucciarelli2017non,kesner2016reasoning,thesisvial,KRV18}.
These type systems are able to statically capture non-trivial
dynamic properties of terms, particularly {\em normalization},
while at the same time being amenable to elementary proof techniques by induction.
Intersection types were originally proposed by
Coppo and Dezani-Ciancaglini~\cite{DBLP:journals/aml/CoppoD78}
to study termination in the $\lambda$-calculus.
They are characterized by the presence of an {\em intersection} type constructor $\typ \cap \typtwo$.
{\em Non-idempotent} intersection type systems are distinguished from their usual idempotent
counterparts by the fact that intersection
is not declared to be idempotent, \ie $\typ$ and $\typ \cap \typ$ are not equivalent types.
Rather, intersection behaves like a multiplicative connective in linear logic.
Arguments to functions are typed many times, typically once
per each time that the argument will be used.
Non-idempotent intersection types were originally formulated by
Gardner~\cite{gardner1994discovering},
and later reintroduced by de Carvalho~\cite{carvalho2007semantiques}.

In this paper, we will use a non-idempotent intersection type system
based on system~$\mathcal{W}$~of~\cite{bucciarelli2017non}
(called system~$\mathcal{H}$ in \cite{bucciarelli2014inhabitation}).
Let us recall its definition.
Terms are as usual in the $\lambda$-calculus ($\tm  ::= \var \mid \lam{\var}{\tm} \mid \tm\,\tm$).
Types $\typ,\typtwo,\typthree,\hdots$ are defined by the grammar:
\[
  \begin{array}{rcl@{\hspace{2cm}}rcl}
  \typ  & ::= & \basetyp \OR \mtyp \to \typ
  &
  \mtyp & ::= & [\typ_i]_{i=1}^{n} \HS\text{with $n \geq 0$} \\
  \end{array}
\]
where $\basetyp$ ranges over one of denumerably many {\em base types},
and $\mtyp$ represents a {\em multiset of types}.
Here $[\typ_i]_{i=1}^{n}$ denotes the multiset $\typ_1,\hdots,\typ_n$
with their respective multiplicities.
A multiset $[\typ_i]_{i=1}^{n}$ intuitively stands for the (non-idempotent) intersection $\typ_1 \cap \hdots \cap \typ_n$.
The {\em sum of multisets} $\mtyp + \mtyptwo$ is defined as their union (adding multiplicities).
A {\em typing context} $\tctx$ is a partial function mapping variables to multisets of types.
The domain of $\tctx$ is the set of variables $\var$ such that $\tctx(\var)$ is defined.
We assume that typing contexts always have {\em finite domain} and hence they may be written as
$\var_1 : \mtyp_1, \hdots, \var_n : \mtyp_n$.
The {\em sum of contexts} $\tctx + \tctxtwo$ is their pointwise sum,
\ie $(\tctx + \tctxtwo)(\var) := \tctx(\var) + \tctxtwo(\var)$ if $\tctx(\var)$ and $\tctxtwo(\var)$ are both defined,
$(\tctx + \tctxtwo)(\var) := \tctx(\var)$ if $\tctxtwo(\var)$ is undefined,
and $(\tctx + \tctxtwo)(\var) := \tctxtwo(\var)$ if $\tctx(\var)$ is undefined.
We write $\tctx +_{i=1}^{n} \tctxtwo_i$ to abbreviate $\tctx + \tctxtwo_1 + \hdots + \tctxtwo_n$.
The {\em disjoint sum of contexts} $\tctx \oplus \tctxtwo$ stands for $\tctx + \tctxtwo$, provided that
the domains of $\tctx$ and $\tctxtwo$ are disjoint.
A {\em typing judgment} is a triple $\tctx \vdash \tm : \typ$,
representing the knowledge that the term $\tm$ has type $\typ$ in the context $\tctx$.
Type assignment rules for system $\mathcal{W}$ are as follows.

\begin{definition}[System $\mathcal{W}$]
{\small
\[
  \indrule{var}{
  }{
    \var : [\typ] \vdash \typ
  }
  \hspace{.15cm}
  \indrule{\toI}{
    \tctx \oplus (\var:\mtyp) \vdash \tm : \typ
  }{
    \tctx \vdash \lam{\var}{\tm} : \mtyp \to \typ
  }
  \hspace{.15cm}
  \indrule{\toE}{
    \tctx \vdash \tm : [\typtwo_i]_{i=1}^{n} \to \typ
    \HS
    (\tctxtwo_i \vdash \tmtwo : \typtwo_i)_{i=1}^{n}
  }{
    \tctx +_{i=1}^{n} \tctxtwo_i \vdash \tm\,\tmtwo : \typ
  }
\]
}
\end{definition}
Observe that the \indrulename{\toE} rule has $n + 1$ premises, where $n \geq 0$.
System~$\mathcal{W}$ enjoys various properties, nicely summarized in~\cite{bucciarelli2017non}.

There are two obstacles to adopting system~$\mathcal{W}$ for studying derivation spaces.
The first obstacle is mostly a matter of presentation---typing derivations use a tree-like notation,
which is cumbersome. One would like to have an alternative notation based on proof terms.
For example, one may define proof terms for the typing rules above using the
syntax $\pi ::= \var^\typ \mid \lam{\var}{\pi} \mid \pi[\pi,\hdots,\pi]$,
in such a way that $\var^\typ$ encodes an application of the \texttt{var} axiom,
$\lam{\var}{\pi}$ encodes an application of the \indrulename{\toI} rule to the typing derivation encoded by $\pi$,
and $\pi_1[\pi_2,\hdots,\pi_n]$ encodes an application of the
\indrulename{\toE} rule to the typing derivations encoded by $\pi_1,\pi_2,\hdots,\pi_n$.
For example, using this notation
$
  \lam{\var}{\var^{[\alpha,\alpha] \to \beta}[\var^\alpha,\var^\alpha]}
$
would represent the following typing derivation:
\[
  {\small
  \indrule{\toI}{
    \indrule{\toE}{
      \indrule{var}{}{\var : [\alpha,\alpha] \to \beta \vdash \var : [\alpha,\alpha] \to \beta}
      \hspace{-.3cm}
      \indrule{var}{}{\var : [\alpha] \vdash \var : \alpha}
      \hspace{-.3cm}
      \indrule{var}{}{\var : [\alpha] \vdash \var : \alpha}
      \vspace{.1cm}
    }{
      \var : [[\alpha,\alpha] \to \beta, \alpha,\alpha] \vdash \var\var : \beta
    }
  }{
      \vdash \lam{\var}{\var\var} : [[\alpha,\alpha] \to \beta, \alpha,\alpha] \to \beta
  }}
\]
The second obstacle is a major one for our purposes: {\em proof normalization} in this system is not confluent.
The reason is that applications take multiple arguments, and a $\beta$-reduction step must choose a way to
distribute these arguments among the occurrences of the formal parameters.
For instance, the following critical pair cannot be closed:
\[
  \xymatrix@R=.1cm@C=.5cm{
  &
  \hspace{-3cm}
  \raisebox{.2cm}{
    $
    (\lam{\var}{\vartwo^{[\alpha] \to [\alpha] \to \beta}[\var^\alpha][\var^\alpha]})
    [\varthree^{[\gamma] \to \alpha}[\varthree^\gamma],\varthree^{[] \to \alpha}[]]
    $
  }
  \hspace{-3cm}
  \ar[dl] \ar[dr]
  \\
    \vartwo^{[\alpha] \to [\alpha] \to \beta}[\varthree^{[\gamma] \to \alpha}[\varthree^\gamma]][\varthree^{[] \to \alpha}[]]
  &
  &
    \vartwo^{[\alpha] \to [\alpha] \to \beta}[\varthree^{[] \to \alpha}[]][\varthree^{[\gamma] \to \alpha}[\varthree^\gamma]]
  }
\]
The remainder of this paper is organized as follows:~\begin{itemize}
\item 
  In \rsec{preliminaries}, we review some standard notions of order and rewriting theory.
\item 
  In \rsec{lambdadist}, we introduce a confluent calculus $\lambdadist$ based on system $\mathcal{W}$.
  The desirable properties of system~$\mathcal{W}$ of~\cite{bucciarelli2017non} still hold in $\lambdadist$.
  Moreover, $\lambdadist$ is confluent.
  We impose confluence forcibly, by decorating subtrees with distinct labels, so that
  a $\beta$-reduction step may distribute the arguments in a unique way.
  Derivation spaces in $\lambdadist$ have very regular structure,
  namely they are distributive lattices.
\item
  In \rsec{simulation}, we establish a correspondence between derivation spaces in the
  $\lambda$-calculus and the $\lambdadist$-calculus via simulation theorems,
  which defines a morphism of upper semilattices.
\item
  In \rsec{factorization}, we introduce the notion of a garbage derivation.
  Roughly, a derivation in the $\lambda$-calculus is {\em garbage} if it maps to
  an empty derivation in the $\lambdadist$-calculus.
  This gives rise to an orthogonal notion of {\em garbage-free} derivation.
  The notion of garbage-free derivation is closely related
  with the notions of {\em needed step}~\cite[Section~8.6]{Terese},
  {\em typed occurrence of a redex}~\cite{bucciarelli2017non},
  and {\em external} derivation~\cite{DBLP:conf/ctcs/Mellies97}.
  Using this notion of garbage we prove a {\em factorization theorem}
  reminiscent of Melli\`es'~\cite{DBLP:conf/ctcs/Mellies97}.
  The upper semilattice of derivations of a term in the $\lambda$-calculus
  is factorized using a variant of the Grothendieck construction.
  Every derivation is uniquely decomposed as a garbage-free prefix followed by a garbage suffix.
\item
  In \rsec{conclusions}, we conclude.
\end{itemize}

{\bf Note.} Proofs including a \SeeAppendix symbol are spelled out in detail in the appendix.

\section{Preliminaries}
  \lsec{preliminaries}
  
We recall some standard definitions.
An \defn{upper semilattice} is a poset $(A,\leq)$
with a least element or \defn{bottom} $\bot \in A$,
and such that for every two elements $a, b \in A$
there is a least upper bound or \defn{join} $(a \lor b) \in A$.
A \defn{lattice} is an upper semilattice
with a greatest element or \defn{top} $\top \in A$,
and such that for every two elements $a, b \in A$
there is a greatest lower bound or \defn{meet} $(a \land b) \in A$.
A lattice is \defn{distributive} if $\land$ distributes over $\lor$ and vice versa.
A \defn{morphism} of upper semilattices
is given by a monotonic function $f : A \to B$,
\ie $a \leq b$ implies $f(a) \leq f(b)$,
preserving the bottom element, \ie $f(\bot) = \bot$,
and joins, \ie $f(a \lor b) = f(a) \lor f(b)$ for all $a, b \in A$.
Similarly for morphisms of lattices.
Any poset $(A,\leq)$ forms a category whose objects are the elements of $A$
and morphisms are of the form $\pt{b}{a}$ for all $a \leq b$. 
The category of posets with monotonic functions is denoted by $\Poset$.
In fact, we regard it as a 2-category: given morphisms $f,g : A \to B$ of posets,
we have that $f \totwocell g$ whenever $f(a) \leq g(a)$ for all $a \in A$.

An \defn{axiomatic rewrite system} (\cf~\cite[Def.~2.1]{thesismellies})
is given by a set of objects $\Obj$, a set of steps $\Stp$,
two functions $\src,\tgt : \Stp \to \Obj$ indicating the source and target of each step,
and a \defn{residual function} $(/)$ such that given any two steps $\redex,\redextwo \in \Stp$
with the same source, yields a set of steps $\redex/\redextwo$ such that $\src(\redex') = \tgt(\redextwo)$
for all $\redex' \in \redex/\redextwo$.
Steps are ranged over by $\redex,\redextwo,\redexthree,\hdots$.
A step $\redex' \in \redex/\redextwo$ is called a \defn{residual} of $\redex$ after $\redextwo$,
and $\redex$ is called an \defn{ancestor} of $\redex'$.
Steps are \defn{coinitial} (resp. \defn{cofinal}) if they have the same source (resp. target).
A \defn{derivation} is a possibly empty sequence of composable steps $\redex_1\hdots\redex_n$.
Derivations are ranged over by $\redseq,\redseqtwo,\redseqthree,\hdots$.
The functions $\src$ and $\tgt$ are extended to derivations.
Composition of derivations is defined when $\tgt(\redseq) = \src(\redseqtwo)$ and written $\redseq\,\redseqtwo$.
Residuals after a derivation can be defined by
$\redex_n \in \redex_0/\redextwo_1\hdots\redextwo_n$
if and only if there exist $\redex_1,\hdots,\redex_{n-1}$
such that $\redex_{i+1} \in \redex_i/\redextwo_{i+1}$ for all $0 \leq i \leq n - 1$.
Let $\redexset$ be a set of coinitial steps.
A \defn{development} of $\redexset$ is a (possibly infinite) derivation $\redex_1\hdots\redex_n\hdots$
such that for every index $i$ there exists a step $\redextwo \in \redexset$
such that $\redex_i \in \redextwo/\redex_1\hdots\redex_{i-1}$.
A development is \defn{complete} if it is maximal.

An \defn{orthogonal} axiomatic rewrite system (\cf~\cite[Sec.~2.3]{thesismellies})
has four additional axioms\footnote{In \cite{thesismellies},
Autoerasure is called Axiom~A, Finite Residuals is called Axiom~B, and Semantic Orthogonality is called PERM.
We follow the nomenclature of~\cite{DBLP:conf/popl/AccattoliBKL14}.}:
\begin{enumerate}
\item {\em Autoerasure}.
  $\redex/\redex = \emptyset$ for all $\redex \in \Stp$.
\item {\em Finite Residuals}.
  The set $\redex/\redextwo$ is finite for all coinitial $\redex,\redextwo \in \Stp$.
\item {\em Finite Developments}.
  If $\redexset$ is a set of coinitial steps,
  all developments of $\redexset$ are finite.
\item {\em Semantic Orthogonality}.
  Let $\redex,\redextwo \in \Stp$ be coinitial steps.
  Then there exist a complete development $\redseq$ of $\redex/\redextwo$
  and a complete development $\redseqtwo$ of $\redextwo/\redex$
  such that $\redseq$ and $\redseqtwo$ are cofinal.
  Moreover,
  for every step $\redexthree \in \Stp$ such that $\redexthree$ is coinitial to $\redex$,
  the following equality between sets holds:
  $\redexthree/(\redex\redseqtwo) = \redexthree/(\redextwo\redseq)$.
\end{enumerate}
In~\cite{thesismellies}, Melli\`es develops
the theory of orthogonal axiomatic rewrite systems.
A notion of \defn{projection}
$\redseq/\redseqtwo$ may be defined between coinitial derivations,
essentially by setting $\emptyDerivation/\redseqtwo \eqdef \emptyDerivation$
and $\redex\redseq'/\redseqtwo \eqdef (\redex/\redseqtwo)(\redseq'/(\redseqtwo/\redex))$
where, by abuse of notation, $\redex/\redseqtwo$ stands for a (canonical) complete development
of the set $\redex/\redseqtwo$.
Using this notion, one may define
a transitive relation of \defn{prefix} ($\redseq \permle \redseqtwo$),
a \defn{permutation equivalence} relation ($\redseq \permeq \redseqtwo$),
and the \defn{join} of derivations ($\redseq \sqcup \redseqtwo$).
Some of their properties are summed up in the figure
below:
\begin{center}
{\bf Summary of properties of orthogonal axiomatic rewrite systems}
\end{center}
\[
{\small
\begin{array}{|c|c|c|}
\hline
\begin{array}{r@{\hspace{0.1cm}}c@{\hspace{0.1cm}}l}
  \emptyDerivation\,\redseq & = & \redseq
\\
  \redseq\,\emptyDerivation & = & \redseq
\\
  \emptyDerivation/\redseq & = & \emptyDerivation
\\
  \redseq/\emptyDerivation & = & \redseq
\\
  \redseq/\redseqtwo\redseqthree & = & (\redseq/\redseqtwo)/\redseqthree
\\
  \redseq\redseqtwo/\redseqthree & = & (\redseq/\redseqthree)(\redseqtwo/(\redseqthree/\redseq))
\\
  \redseq/\redseq & = & \emptyDerivation
\end{array}
&
\begin{array}{r@{\hspace{0.1cm}}c@{\hspace{0.1cm}}l}
  \redseq \permle \redseqtwo & \iffdef & \redseq/\redseqtwo = \emptyDerivation
\\
  \redseq \permeq \redseqtwo & \iffdef & \redseq \permle \redseqtwo \land \redseqtwo \permle \redseq
\\
  \redseq \sqcup \redseqtwo  & \eqdef & \redseq(\redseqtwo/\redseq) 
\\
  \redseq \permeq \redseqtwo & \implies & \redseqthree/\redseq = \redseqthree/\redseqtwo
\\
  \redseq \permle \redseqtwo & \iff & \exists \redseqthree.\ \redseq\redseqthree \permeq \redseqtwo
\\
  \redseq \permle \redseqtwo & \iff & \redseq \sqcup \redseqtwo \permeq \redseqtwo
\end{array}
&
\begin{array}{r@{\hspace{0.1cm}}c@{\hspace{0.1cm}}l}
  \redseq \permle \redseqtwo & \implies & \redseq/\redseqthree \permle \redseqtwo/\redseqthree
\\
  \redseq \permle \redseqtwo & \iff & \redseqthree\redseq \permle \redseqthree\redseqtwo
\\
  \redseq \sqcup \redseqtwo  & \permeq & \redseqtwo \sqcup \redseq
\\
  (\redseq \sqcup \redseqtwo) \sqcup \redseqthree & = & \redseq \sqcup (\redseqtwo \sqcup \redseqthree)
\\
  \redseq & \permle & \redseq \sqcup \redseqtwo
\\
  (\redseq \sqcup \redseqtwo)/\redseqthree & = & (\redseq/\redseqthree) \sqcup (\redseqtwo/\redseqthree) \\
\end{array}
\\
\hline
\end{array}
}
\]
Let $\cls{\redseq} = \set{\redseqtwo \ST \redseq \permeq \redseqtwo}$
denote the permutation equivalence class of $\redseq$.
In an orthogonal axiomatic rewrite system,
the set $\ulbDeriv{x} = \set{\cls{\redseq} \ST \src(\redseq) = x}$
forms an upper semilattice~\cite[Thm.~2.2,Thm.~2.3]{thesismellies}.
The order $\cls{\redseq} \permle \cls{\redseqtwo}$ is declared to hold if
$\redseq \permle \redseqtwo$,
the join is $\cls{\redseq} \sqcup \cls{\redseqtwo} = \cls{\redseq \sqcup \redseqtwo}$,
and the bottom is $\bot = \cls{\emptyDerivation}$.
The $\lambda$-calculus is an example of an orthogonal axiomatic rewrite system.
Our structures of interest are the semilattices of derivations of the $\lambda$-calculus,
written $\ulbDerivLam{\tm}$ for any given $\lambda$-term $\tm$.
As usual, $\beta$-reduction in the $\lambda$-calculus
is written $\tm \tobeta \tmtwo$
and defined by the contextual closure of the axiom $(\lam{\var}{\tm})\tmtwo \tobeta \subs{\tm}{\var}{\tmtwo}$.

\section{The Distributive $\lambda$-Calculus}
  \lsec{lambdadist}
  
In this section we introduce the
{\em distributive $\lambda$-calculus} ($\lambdadist$),
and we prove some basic results.
Terms of the $\lambdadist$-calculus are typing derivations of a non-idempotent intersection type
system, written using proof term syntax.
The underlying type system is a variant of system~$\mathcal{W}$ of \cite{bucciarelli2014inhabitation,bucciarelli2017non},
the main difference being that $\lambdadist$
uses {\em labels} and a suitable invariant on terms,
to ensure that the formal parameters of all functions
are in 1--1 correspondence with the actual arguments that they receive.

\begin{definition}[Syntax of the $\lambdadist$-calculus]
Let $\labelset = \set{\lab,\lab',\lab'',\hdots}$ be a denumerable set of labels.
The set of \defn{types} is ranged over by $\typ,\typtwo,\typthree,\hdots$,
and defined inductively as follows:
\[
  \begin{array}{rcl@{\hspace{2cm}}rcl}
  \typ  & ::= & \basetyp^{\lab} \mid \mtyp \tolab{\lab} \typ
  &
  \mtyp & ::= & [\typ_i]_{i=1}^{n} \HS\text{with $n \geq 0$} \\
  \end{array}
\]
where $\basetyp$ ranges over one of denumerably many {\em base types},
and $\mtyp$ represents a {\em multiset of types}.
In a type like $\basetyp^{\lab}$ and $\mtyp \tolab{\lab} \typ$,
the label $\lab$ is called the \defn{external label}.
The \defn{typing contexts} are defined as in \rsec{introduction} for system~$\mathcal{W}$.
We write $\dom\tctx$ for the domain of $\tctx$. 
A type $\typ$ is said to \defn{occur} inside another type $\typtwo$,
written $\typ \occursin \typtwo$, if $\typ$ is a subformula of $\typtwo$.
This is extended to say that a type $\typ$ occurs in a multiset $[\typtwo_1,\hdots,\typtwo_n]$,
declaring that $\typ \occursin [\typtwo_1,\hdots,\typtwo_n]$ if $\typ \occursin \typtwo_i$ for some $i=1..n$,
and that a type $\typ$ occurs in a typing context $\tctx$,
declaring that $\typ \occursin \tctx$ if $\typ \occursin \tctx(\var)$ for some $\var \in \dom\tctx$.
\medskip

The set of
\defn{terms},
ranged over by $\tm,\tmtwo,\tmthree,\hdots$, is given by the grammar
$\tm ::= \var^{\typ} \mid \lamp{\lab}{\var}{\tm} \mid \tm\,\ls{\tm}$,
where $\ls{\tm}$ represents a (possibly empty) {\em finite list} of terms.
The notations $[x_i]_{i=1}^{n}$, $[x_1,\hdots,x_n]$, and $\ls{x}$
all stand simultaneously for multisets and for lists of elements.
Note that there is no confusion since we only work with {\em multisets of types},
and with {\em lists of terms}.
The concatenation of the lists $\lst,\lsttwo$ is denoted by $\lst + \lsttwo$.
A sequence of $n$ lists $(\lst_1,\hdots,\lst_n)$ is a \defn{partition} of $\lst$
if $\lst_1 + \hdots + \lst_n$ is a permutation of $\lst$.
The set of \defn{free variables} of a term $\tm$
is written $\fv{\tm}$ and defined as expected.
We also write $\fv{[\tm_i]_{i=1}^n}$ for $\cup_{i=1}^n \fv{\tm_i}$.
A \defn{context} is a term $\con$ with an occurrence of a distinguished {\em hole} $\conbase$.
We write $\conof{\tm}$ for the capturing substitution of $\conbase$ by $\tm$.
\defn{Typing judgments} are triples $\tctx \vdash \tm : \typ$ representing the knowledge
that the term $\tm$ has type $\typ$ in the context $\tctx$.
Type assignment rules are:
\[
  \indrule{var}{
  }{
    \var : \lset{\typ} \vdash \var^\typ : \typ
  }
  \HS
  \indrule{\toI}{
    \tctx \oplus (\var : \mtyp) \vdash \tm : \typtwo
  }{
    \tctx \vdash \lamp{\lab}{\var}{\tm} : \mtyp \tolab{\lab} \typtwo
  }
\]
\[
  \indrule{\toE}{
    \tctx \vdash \tm : \lset{\typtwo_1,\hdots,\typtwo_n} \tolab{\lab} \typ
    \HS
    \left( \tctxtwo_i \vdash \tmtwo_i : \typtwo_i \right)_{i=1}^{n}
  }{
    \tctx +_{i=1}^{n} \tctxtwo_i \vdash \tm[\tmtwo_1,\hdots,\tmtwo_n] : \typ
  }
\]
\end{definition}
For example
$\vdash \lamp{1}{\var}{\var^{[\alpha^2,\alpha^3] \tolab{4} \beta^5}[\var^{\alpha^3},\var^{\alpha^2}]} : [[\alpha^2,\alpha^3] \tolab{4} \beta^5, \alpha^2, \alpha^3] \tolab{1} \beta^5$
is a derivable judgment (using integer labels).

\begin{remark}[Unique typing]
Let $\tctx \vdash \tm : \typ$ and $\tctxtwo \vdash \tm : \typtwo$ be derivable judgments.
Then $\tctx = \tctxtwo$ and $\typ = \typtwo$.
Moreover, the derivation trees coincide.
\end{remark}
This can be checked by induction on $\tm$.
It means that $\lambdadist$ is an {\em \`a la Church} type system,
that is, types are an {\em intrinsic} property of the syntax of terms,
as opposed to an {\em \`a la Curry} type system like~${\mathcal W}$, in which
types are {\em extrinsic} properties that a given term might or might not have.

To define a confluent rewriting rule,
we impose a further constraint on the syntax of terms, called \defn{correctness}.
The $\lambdadist$-calculus will be defined over the set of correct terms.

\begin{definition}[Correct terms]
\ldef{sequentiality_and_correctness}
A multiset of types $[\typ_1,\hdots,\typ_n]$ is \defn{sequential}
if the external labels of $\typ_i$ and $\typ_j$ are different for all $i \neq j$.
A typing context $\tctx$ is sequential if $\tctx(\var)$ is sequential for every $\var \in \dom\tctx$.
A term $\tm$ is correct if it is typable and it verifies the following three conditions:
\begin{enumerate}
\item {\em Uniquely labeled lambdas.}
  If $\lamp{\lab}{\var}{\tmtwo}$ and $\lamp{\labtwo}{\vartwo}{\tmthree}$ 
  are subterms of $\tm$ at different positions, then $\lab$ and $\labtwo$
  must be different labels.
\item {\em Sequential contexts.}
  If $\tmtwo$ is a subterm of $\tm$ and $\tctx \vdash \tmtwo : \typ$
  is derivable, then $\tctx$ must be sequential.
\item {\em Sequential types.}
  If $\tmtwo$ is a subterm of $\tm$,
  the judgment $\tctx \vdash \tmtwo : \typ$ is derivable,
  and there exists a type such that
  $(\mtyp \tolab{\lab} \typtwo \occursin \tctx) \lor (\mtyp \tolab{\lab} \typtwo \occursin \typ)$,
  then $\mtyp$ must be sequential.
\end{enumerate}
The set of \defn{correct} terms is denoted by $\termsdist$.
\end{definition}
For example,
$\var^{[\alpha^1] \tolab{2} \beta^3}[\var^{\alpha^1}]$ is a correct term,
$\lamp{1}{\var}{\lamp{1}{\vartwo}{\vartwo^{\alpha^2}}}$ 
is not a correct term since labels for lambdas are not unique,
and
$\lamp{1}{\var}{\var^{\alpha^2 \tolab{3} [\beta^4,\beta^4] \tolab{5} \gamma^6}}$ is not a correct term
since $[\beta^4,\beta^4]$ is not sequential.

Substitution is defined explicitly below.
If $\tm$ is typable,
$\varlabel{\var}{\tm}$
stands for the multiset
of types of the free occurrences of $\var$ in $\tm$.
If $\tm_1,\hdots,\tm_n$ are typable,
$\tmlabel{[\tm_1,\hdots,\tm_n]}$ stands for the multiset
of types of $\tm_1,\hdots,\tm_n$.
For example,
$\varlabel{\var}{\var^{[\alpha^1] \tolab{2} \beta^3}[\var^{\alpha^1}]} =
\tmlabel{[\vartwo^{\alpha^1},\varthree^{[\alpha^1] \tolab{2} \beta^3}]} = [[\alpha^1] \tolab{2} \beta^3,\alpha^1]$.
To perform a substitution $\subs{\tm}{\var}{[\tmtwo_1,\hdots,\tmtwo_n]}$
we will require that $\varlabel{\var}{\tm} = \tmlabel{[\tmtwo_1,\hdots,\tmtwo_n]}$.

\begin{definition}[Substitution]
\ldef{substitution}
Let $\tm$ and $\tmtwo_1,\hdots,\tmtwo_n$ be correct terms such that $\varlabel{\var}{\tm} = \tmlabel{[\tmtwo_1,\hdots,\tmtwo_n]}$.
The capture-avoiding substitution of $\var$ in $\tm$ by $\ls{\tmtwo} = [\tmtwo_1,\hdots,\tmtwo_n]$
is denoted by $\subs{\tm}{\var}{\ls{\tmtwo}}$ and defined as follows:
\[
  \begin{array}{rcll}
    \subs{\var^\typ}{\var}{[\tmtwo]} & \eqdef & \tmtwo
  \\
    \subs{\vartwo^\typ}{\var}{[]} & \eqdef & \vartwo^\typ
    & \text{ if $\var \neq \vartwo$}
  \\
    \subs{(\lamp{\lab}{\vartwo}{\tmthree})}{\var}{\ls{\tmtwo}} & \eqdef &  \lamp{\lab}{\vartwo}{ \subs{\tmthree}{\var}{\ls{\tmtwo}} }
    & \text{ if $\var \neq \vartwo$ and $\vartwo \not\in \fv{\ls{\tmtwo}}$}
  \\
    \subs{\tmthree_0[\tmthree_j]_{j=1}^{m}}{\var}{\ls{\tmtwo}} & \eqdef &
    \subs{\tmthree_0}{\var}{\ls{\tmtwo}_0}[\subs{\tmthree_j}{\var}{\ls{\tmtwo}_j}]_{j=1}^{m}
  \end{array}
\]
In the last case, $(\ls{\tmtwo}_0, \hdots, \ls{\tmtwo}_m)$
is a partition of $\ls{\tmtwo}$
such that $\varlabel{\var}{\tmthree_j} = \tmlabel{\ls{\tmtwo}_j}$ for all $j=0..m$.
\end{definition}

\begin{remark}
Substitution is {\em type-directed}: the arguments $[\tmtwo_1,\hdots,\tmtwo_n]$
are propagated throughout the term
so that $\tmtwo_i$ reaches the free occurrence of $\var$
that has the same type as $\tmtwo_i$.
Note that the definition of substitution requires that $\varlabel{\var}{\tm} = \tmlabel{[\tmtwo_1,\hdots,\tmtwo_n]}$,
which means that the types of the terms $\tmtwo_1,\hdots,\tmtwo_n$
are in 1--1 correspondence with the types of the free occurrences of $\var$.
Moreover, since $\tm$ is a correct term,
the multiset $\varlabel{\var}{\tm}$ is sequential,
which implies in particular that each free occurrence of $\var$ has a different type.
Hence there is a {\em unique} correspondence matching the free occurrences of $\var$
with the arguments $\tmtwo_1,\hdots,\tmtwo_n$ that respects their types.
As a consequence, in the definition of substitution for an application
$\subs{\tmthree_0[\tmthree_j]_{j=1}^{m}}{\var}{\ls{\tmtwo}}$
there is essentially a unique way to split $\ls{\tmtwo}$
into $n+1$ lists $(\ls{\tmtwo}_0, \ls{\tmtwo}_1, \hdots, \ls{\tmtwo}_n)$
in such a way that $\varlabel{\var}{\tmthree_i} = \tmlabel{\ls{\tmtwo}_i}$.
More precisely, if $(\ls{\tmtwo}_0, \ls{\tmtwo}_1, \hdots, \ls{\tmtwo}_n)$
and $(\ls{\tmthree}_0, \ls{\tmthree}_1, \hdots, \ls{\tmthree}_n)$
are two partitions of $\ls{\tmtwo}$ with the stated property,
then $\ls{\tmtwo}_i$ is a permutation of $\ls{\tmthree}_i$ for all $i=0..n$.
Using this argument, it is easy to check by induction on $\tm$
that the value of $\subs{\tm}{\var}{\ls{\tmtwo}}$ is uniquely determined and
does not depend on this choice.
\end{remark}

For example,
$\subs{(\var^{[\alpha^1] \tolab{2} \beta^3}[\var^{\alpha^1}])}{\var}{[\vartwo^{[\alpha^1] \tolab{2} \beta^3},\varthree^{\alpha^1}]}
  = \vartwo^{[\alpha^1] \tolab{2} \beta^3}\varthree^{\alpha^1}$
while, on the other hand,
$\subs{(\var^{[\alpha^1] \tolab{2} \beta^3}[\var^{\alpha^1}])}{\var}{[\vartwo^{\alpha^1},\varthree^{[\alpha^1] \tolab{2} \beta^3}]}
  = \varthree^{[\alpha^1] \tolab{2} \beta^3}\vartwo^{\alpha^1}$.

The operation of substitution preserves term correctness and typability:
\begin{lemma}[Subject Reduction]
\llem{subject_reduction}
If $\conof{(\lamp{\lab}{\var}{\tm})\ls{\tmtwo}}$ is a correct term
such that
the judgment $\tctx \vdash \conof{(\lamp{\lab}{\var}{\tm})\ls{\tmtwo}} : \typ$
is derivable, then $\conof{\subs{\tm}{\var}{\ls{\tmtwo}}}$ is correct
and $\tctx \vdash \conof{\subs{\tm}{\var}{\ls{\tmtwo}}} : \typ$
is derivable.
\end{lemma}
\begin{proof}
\SeeAppendixRef{appendix_subject_reduction}
By induction on $\con$.
\end{proof}

\begin{definition}[The $\lambdadist$-calculus]
The \defn{$\lambdadist$-calculus} is the rewriting system whose objects
are the set of correct terms $\termsdist$.
The rewrite relation $\todist$ is the closure under arbitrary contexts
of the rule $(\lamp{\lab}{\var}{\tm})\ls{\tmtwo} \todist \subs{\tm}{\var}{\ls{\tmtwo}}$.
\rlem{subject_reduction} justifies that $\todist$ is well-defined, \ie that
the right-hand side is a correct term.
The label of a step is the label $\lab$ decorating the contracted lambda.
We write $\tm \todistl{\lab} \tmtwo$ whenever $\tm \todist \tmtwo$ and the label of
the step is $\lab$.
\end{definition}

\begin{example}
\lexample{lambdadist_reduction_graph}
Let $I^3 \eqdef \lamp{3}{\var}{\var^{\alpha^2}}$
and $I^4 \eqdef \lamp{4}{\var}{\var^{\alpha^2}}$.
The reduction graph of
the term
$(\lamp{1}{\var}{\var^{[\alpha^2] \tolab{3} \alpha^2}[\var^{\alpha^2}]})[
  I^3,
  I^4[\varthree^{\alpha^2}]
]$ is:
\[
  {\small
    \xymatrix@R=.5cm{
      (\lamp{1}{\var}{\var^{[\alpha^2] \tolab{3} \alpha^2}[\var^{\alpha^2}]})[
        I^3,
        I^4[\varthree^{\alpha^2}]
      ]
      \ar[r]^-{1}_-{\redextwo}
      \ar[d]_-{4}^-{\redex}
    &
      I^3[
        I^4[\varthree^{\alpha^2}]
      ]
      \ar[r]^-{3}_-{\redexthree}
      \ar[d]_-{4}^-{\redex'}
    &
      I^4[\varthree^{\alpha^2}]
      \ar[d]_-{4}^-{\redex''}
    \\
      (\lamp{1}{\var}{\var^{[\alpha^2] \tolab{3} \alpha^2}[\var^{\alpha^2}]})[
        I^3,
        \varthree^{\alpha^2}
      ]
      \ar[r]^-{1}_-{\redextwo'}
    &
      I^3[ \varthree^{\alpha^2} ]
      \ar[r]^-{3}_-{\redexthree'}
    &
      \varthree^{\alpha^2}
    }
  }
\]
\end{example}
Note that numbers over arrows are the labels of the steps,
while $\redex,\redex',\redextwo,...$ are metalanguage names to refer to the steps.
Next, we state and prove some basic properties of $\lambdadist$.

\begin{proposition}[Strong Normalization]
\lprop{strong_normalization}
There is no infinite reduction $\tm_0 \todist \tm_1 \todist \hdots$. 
\end{proposition}
\begin{proof}
Observe that a reduction step
$\conof{(\lamp{\lab}{\var}{\tm})\ls{\tmtwo}} \todist \conof{\subs{\tm}{\var}{\ls{\tmtwo}}}$
decreases the number of lambdas in a term by exactly 1,
because substitution is {\em linear}, \ie
the term $\subs{\tm}{\var}{[\tmtwo_1,\hdots,\tmtwo_n]}$
uses $\tmtwo_i$ exactly once for all $i=1..n$.
{\em Note:}
this is an adaptation of~\cite[Theorem~4.1]{bucciarelli2017non}.
\end{proof}

The substitution operator may be extended to work on lists,
by defining
$\subs{[\tm_i]_{i=1}^{n}}{\var}{\ls{\tmtwo}} \eqdef [\subs{\tm_i}{\var}{\ls{\tmtwo}_i}]_{i=1}^{n}$
where $(\ls{\tmtwo}_1,\hdots,\ls{\tmtwo}_n)$ is a partition of $\ls{\tmtwo}$
such that $\varlabel{\var}{\tm_i} = \tmlabel{\ls{\tmtwo}_i}$ for all $i=1..n$.

\begin{lemma}[Substitution Lemma]
\llem{substitution_lemma}
Let $\var \neq \vartwo$ and $\var \not\in \fv{\ls{\tmthree}}$.
If $(\ls{\tmthree}_1,\ls{\tmthree}_2)$ is a partition of $\ls{\tmthree}$
then
$
  \subs{\subs{\tm}{\var}{\ls{\tmtwo}}}{\vartwo}{\ls{\tmthree}}
  =
  \subs{\subs{\tm}{\vartwo}{\ls{\tmthree}_1}}{\var}{\subs{\ls{\tmtwo}}{\vartwo}{\ls{\tmthree}_2}}
$,
provided that both sides of the equation are defined.
{\em Note:} there exists a list $\ls{\tmthree}$ that makes the left-hand side defined
if and only if there exist lists $\ls{\tmthree}_1,\ls{\tmthree}_2$ that make the
right-hand side defined.
\end{lemma}
\begin{proof}
\SeeAppendixRef{appendix_substitution_lemma}
By induction on $\tm$.
\end{proof}

\begin{proposition}[Permutation]
\lprop{strong_permutation}
If $\tm_0 \todistl{\lab_1} \tm_1$
and $\tm_0 \todistl{\lab_2} \tm_2$
are different steps, then there exists a term $\tm_3 \in \termsdist$ such that
$\tm_1 \todistl{\lab_2} \tm_3$ and $\tm_2 \todistl{\lab_1} \tm_3$.
\end{proposition}
\begin{proof}
\SeeAppendixRef{appendix_strong_permutation}
By exhaustive case analysis of permutation diagrams.
Two representative cases are depicted below.
The proof uses the \resultname{Substitution Lemma}~(\rlem{substitution_lemma}).
\end{proof}
\[
    \xymatrix@R=.5cm@C=.5cm{
     (\lamp{\lab}{\var}{(\lamp{\lab'}{\vartwo}{\tmthree}) \ls{\tmfour}}) \ls{\tmtwo}
                        \ar[d]_-{\lab'}
                        \ar[r]^-{\lab} &
     \subs{((\lamp{\lab'}{\vartwo}{\tmthree}) \ls{\tmfour})}{\var}{\ls{\tmtwo}}
                        \ar[d]_-{\lab'} \\
                        \ar[r]^-{\lab}
     (\lamp{\lab}{\var}{\subs{\tmthree}{\vartwo}{\ls{\tmfour}}}) \ls{\tmtwo}
                        &
     \subs{\subs{\tmthree}{\vartwo}{\ls{\tmfour}}}{\var}{\ls{\tmtwo}}
                        \\
    }
\]
\[
    \xymatrix@R=.5cm@C=.5cm{
     (\lamp{\lab}{\var}{\tm})[\ls{\tmtwo}_1,(\lamp{\lab'}{\vartwo}{\tmthree})\ls{\tmfour},\ls{\tmtwo}_2]
                        \ar[d]_-{\lab'}
                        \ar[r]^-{\lab} &
     \subs{\tm}{\var}{[\ls{\tmtwo}_1,(\lamp{\lab'}{\vartwo}{\tmthree})\ls{\tmfour},\ls{\tmtwo}_2]}
                        \ar[d]_-{\lab'} \\
     (\lamp{\lab}{\var}{\tm})[\ls{\tmtwo}_1,\subs{\tmthree}{\vartwo}{\ls{\tmfour}},\ls{\tmtwo}_2]
                        \ar[r]^-{\lab}
                        &
     \subs{\tm}{\var}{[\ls{\tmtwo}_1,\subs{\tmthree}{\vartwo}{\ls{\tmfour}},\ls{\tmtwo}_2]}
                        \\
    }
\]
As a consequence of
\rprop{strong_permutation},
reduction is subcommutative,
\ie $(\lefttodist \circ \todist)\ \subseteq (\todist^= \circ \lefttodist^=)$
where
$\lefttodist$ denotes $(\todist)^{-1}$
and $R^=$ denotes the reflexive closure of $R$.
Moreover, it is well-known that subcommutativity implies {\bf confluence},
\ie $(\lefttodist^* \circ \todist^*)\ \subseteq (\todist^* \circ \lefttodist^*)$;
see \cite[Prop.~1.1.10]{Terese} for a proof of this fact.

\begin{proposition}[Orthogonality]
\lprop{orthogonality}
$\lambdadist$ is an orthogonal axiomatic rewrite system.
\end{proposition}
\begin{proof}
Let $\redex : \tm \todist \tmtwo$ and $\redextwo : \tm \todist \tmthree$.
Define the set of residuals $\redex/\redextwo$ as the set of steps starting
on $\tmthree$ that have the same label as $\redex$.
Note that $\redex/\redextwo$ is empty if $\redex = \redextwo$,
and it is a singleton if $\redex \neq \redextwo$, since terms are correct so their lambdas
are uniquely labeled.
Then it is immediate to observe that axioms {\em Autoerasure} and {\em Finite Residuals} hold.
The {\em Finite Developments} axiom
is a consequence of \resultname{Strong Normalization}~(\rprop{strong_normalization}).
The {\em Semantic Orthogonality} axiom is a consequence of \resultname{Permutation}~(\rprop{strong_permutation}).
\end{proof}

For instance, in the reduction graph of~\rexample{lambdadist_reduction_graph},
$\redextwo\redexthree/\redex\redextwo' = \redexthree'$,
$\redextwo \sqcup \redex = \redextwo\redex'$,
and $\redextwo\redex'\redexthree' \permeq \redex\redextwo'\redexthree'$.
Observe that in~\rexample{lambdadist_reduction_graph} there is no
duplication or erasure of steps. This is a general phenomenon.
Indeed, \resultname{Permutation}~(\rprop{strong_permutation}) ensures that all
non-trivial permutation diagrams are closed with exactly one step on each side.

Let us write $\ulbDerivDist{\tm}$ for the set of derivations of $\tm$ in the $\lambdadist$-calculus,
modulo permutation equivalence.
As a consequence of
\resultname{Orthogonality}~(\rprop{orthogonality})
and axiomatic results~\cite{thesismellies},
the set $\ulbDerivDist{\tm}$ is an upper semilattice.
Actually, we show that moreover the space $\ulbDerivDist{\tm}$
is a {\em distributive lattice}.
To prove this, let us start by mentioning
the property that we call \resultname{Full Stability}.
This is a strong version of stability in the sense of L\'evy~\cite{levy_redex_stability}.
It means that steps are created in an essentially unique way.
In what follows, we write $\name(\redex)$ for the label of a step, 
and $\names(\redex_1\hdots\redex_n) = \set{\name(\redex_i) \ST 1 \leq i \leq n}$ for the set of
labels of a derivation.

\begin{lemma}[Full Stability]
\llem{full_stability}
Let $\redseq,\redseqtwo$ be coinitial derivations
with disjoint labels, \ie $\names(\redseq) \cap \names(\redseqtwo) = \emptyset$.
Let $\redexthree_1,\redexthree_2,\redexthree_3$ be steps such that
$\redexthree_3 = \redexthree_1/(\redseqtwo/\redseq) = \redexthree_2/(\redseq/\redseqtwo)$.
Then there is a step $\redexthree_0$ such that $\redexthree_1 = \redexthree_0/\redseq$
and $\redexthree_2 = \redexthree_0/\redseqtwo$.
\end{lemma}
\begin{proof}
\SeeAppendixRef{appendix_full_stability}
The proof is easily reduced to a \resultname{Basic Stability} result:
a particular case of \resultname{Full Stability} when $\redseq$ and $\redseqtwo$
consist of single steps.
\resultname{Basic Stability} is proved by exhaustive case analysis.
\end{proof}

\begin{proposition}
\lprop{lambdadist_lattice}
$\ulbDerivDist{\tm}$ is a lattice.
\end{proposition}
\begin{proof}
\SeeAppendixRef{appendix_lambdadist_lattice}
The missing components are the {\em top} and the {\em meet}.
The top element is given by $\top := \cls{\redseq}$ where $\redseq : \tm \todist^* \tmtwo$
is a derivation to normal form, which exists by {\bf Strong Normalization}~(\rprop{strong_normalization}).
The meet of $\set{\cls{\redseq},\cls{\redseqtwo}}$ is constructed using
{\bf Full Stability}~(\rlem{full_stability}).
If $\names(\redseq) \cap \names(\redseqtwo) = \emptyset$,
define $(\redseq \sqcap \redseqtwo) := \emptyDerivation$.
Otherwise, the stability result ensures that there is a step
$\redex$ coinitial to $\redseq$ and $\redseqtwo$
such that $\name(\redex) \in \names(\redseq) \cap \names(\redseqtwo)$.
Let $\redex$ be one such step,
and, recursively, define
$(\redseq \sqcap \redseqtwo) :=
\redex((\redseq/\redex) \sqcap (\redseqtwo/\redex))$.
It can be checked that recursion terminates,
because $\names(\redseq/\redex) \subset \names(\redseq)$ is a strict
inclusion.
Moreover, $\redseq \sqcap \redseqtwo$
is the greatest lower bound of $\set{\redseq,\redseqtwo}$,
up to permutation equivalence.
\end{proof}
For instance, in~\rexample{lambdadist_reduction_graph}
we have that
$\redextwo\redexthree \sqcap \redex = \emptyDerivation$,
$\redextwo\redexthree \sqcap \redex\redextwo' = \redextwo$,
and $\redextwo\redexthree \sqcap \redex\redextwo'\redexthree' = \redextwo\redexthree$.

\begin{proposition}
\lprop{labels_morphism}
There is a monomorphism of lattices $\ulbDerivDist{\tm} \to \mathcal{P}(X)$
for some set~$X$.
The lattice $(\mathcal{P}(X),\subseteq,\emptyset,\cup,X,\cap)$ consists of the subsets of $X$, ordered by inclusion.
\end{proposition}
\begin{proof}
\SeeAppendixRef{appendix_labels_morphism}
The morphism is the function $\names$, mapping each derivation to its set of labels.
\end{proof}
This means that a derivation in $\lambdadist$ is characterized, up to permutation equivalence,
by the set of labels of its steps.
Since $\mathcal{P}(X)$ is a distributive lattice, in particular we have:
\begin{corollary}
\lcoro{ulbderivdist_distributive_lattice}
$\ulbDerivDist{\tm}$ is a distributive lattice.
\end{corollary}

\section{Simulation of the $\lambda$-Calculus in the $\lambdadist$-Calculus}
  \lsec{simulation}
  
In this section we establish a precise relationship between derivations in
the $\lambda$-calculus and derivations in $\lambdadist$.
To begin, we need a way to relate $\lambda$-terms and correct terms ($\termsdist$):

\begin{definition}[Refinement]
A correct term $\tm' \in \termsdist$ \defn{refines} a $\lambda$-term $\tm$,
written $\tm' \refines \tm$, according to the following inductive definition:
\[
  \indrule{r-var}{
  }{
    \var^\typ \refines \var
  }
  \HS
  \indrule{r-lam}{
    \tm' \refines \tm
  }{
    \lamp{\lab}{\var}{\tm'} \refines \lam{\var}{\tm}
  }
  \HS
  \indrule{r-app}{
    \tm' \refines \tm
    \HS
    \tmtwo'_i \refines \tmtwo \text{ for all $i=1..n$}
  }{
    \tm'[\tmtwo'_i]_{i=1}^{n} \refines \tm\tmtwo
  }
\]
\end{definition}
A $\lambda$-term may have many refinements.
For example, the following terms refine $(\lam{\var}{\var\var})\vartwo$:
\[
  (\lamp{1}{\var}{\var^{[\,] \tolab{2} \alpha^3}[\,]})[\vartwo^{[\,] \tolab{2} \alpha^3}]
  \HS
  (\lamp{1}{\var}{\var^{[\alpha^2] \tolab{3} \beta^4}[\var^{\alpha^2}]})[\vartwo^{[\alpha^2] \tolab{3} \beta^4},\vartwo^{\alpha^2}]
\]
\[
  (\lamp{1}{\var}{\var^{[\alpha^2,\beta^3] \tolab{4} \gamma^5}[\var^{\alpha^2},\var^{\beta^3}]})[\vartwo^{[\alpha^2,\beta^3] \tolab{4} \gamma^5},\vartwo^{\alpha^2},\vartwo^{\beta^3}]
\]
The refinement relation establishes a relation of {\em simulation} between the $\lambda$-calculus and $\lambdadist$.
\begin{proposition}[Simulation]
\lprop{simulation}
\lprop{reverse_simulation}
Let $\tm' \refines \tm$. Then:
\begin{enumerate}
\item
  \label{reverse_simulation__item_fwd}
      If $\tm \tobeta \tmtwo$, there exists $\tmtwo'$ 
      such that $\tm' \todist^* \tmtwo'$
      and $\tmtwo' \refines \tmtwo$.
\item
      If $\tm' \todist \tmtwo'$, there exist $\tmtwo$ and $\tmtwo''$
      such that
      $\tm \tobeta \tmtwo$,
      $\tmtwo' \todist^* \tmtwo''$,
      and $\tmtwo'' \refines \tmtwo$.
\end{enumerate}
\end{proposition}
\begin{proof}
\SeeAppendixRef{appendix_simulation}
By case analysis. The proof is constructive.
Moreover, in item \refcase{reverse_simulation__item_fwd},
the derivation $\tm' \todist^* \tmtwo'$ is shown to be a {\em multistep}, \ie the complete
development of a set $\set{\redex_1,\hdots,\redex_n}$.
\end{proof}
The following example illustrates that a $\beta$-step in the $\lambda$-calculus may be simulated by zero,
one, or possibly many steps in $\lambdadist$, depending on the refinement chosen.

\begin{example}
The following are simulations of the step $\var\,((\lam{\var}{\var})\vartwo) \tobeta \var\,\vartwo$
using $\todist$-steps:
\[
  {\small
    \xymatrix@R=.25cm@C=.5cm{
      \var\,((\lam{\var}{\var})\vartwo)
      \ar[r]^-{\beta}
      \ar@{}[d]|*=0[@]{\rtimes}
      &
      \var\,\vartwo
      \ar@{}[d]|*=0[@]{\rtimes}
    \\
      \var^{[] \tolab{1} \alpha^2}\,[]
      \ar@{=}[r]
      &
      \var^{[] \tolab{1} \alpha^2}\,[]
    }
    \HS
    \xymatrix@R=.25cm@C=.5cm{
      \var\,((\lam{\var}{\var})\vartwo)
      \ar[r]^-{\beta}
      \ar@{}[d]|*=0[@]{\rtimes}
      &
      \var\,\vartwo
      \ar@{}[d]|*=0[@]{\rtimes}
    \\
      \var^{[\alpha^1] \tolab{2} \beta^3}\,[(\lamp{4}{\var}{\var^{\alpha^1}})[\vartwo^{\alpha^1}]]
      \ar[r]^-{\dist}
      &
      \var^{[\alpha^1] \tolab{2} \beta^3}\,[\vartwo^{\alpha^1}]
    }
  }
\]
\[
  {\small
    \xymatrix@R=.25cm{
      \var\,((\lam{\var}{\var})\vartwo)
      \ar[rr]^-{\beta}
      \ar@{}[d]|*=0[@]{\rtimes}
      &
      &
      \var\,\vartwo
      \ar@{}[d]|*=0[@]{\rtimes}
    \\
      \var^{[\alpha^1,\beta^2] \tolab{3} \gamma^4}\,[
        (\lamp{5}{\var}{\var^{\alpha^1}})[\vartwo^{\alpha^1}],
        (\lamp{6}{\var}{\var^{\beta^2}})[\vartwo^{\beta^2}]
      ]
      \ar@{->>}[rr]^-{\dist}
      &
      &
      \var^{[\alpha^1,\beta^2] \tolab{3} \gamma^4}\,[
        \vartwo^{\alpha^1},
        \vartwo^{\beta^2}
      ]
    }
  }
\]
\end{example}
The next result relates typability and normalization.
This is an adaptation of existing results from
non-idempotent intersection types, \eg~\cite[Lemma~5.1]{bucciarelli2017non}.
Recall that a
\defn{head normal form} is a term of the form
$\lam{\var_1}{\hdots\lam{\var_n}{\vartwo\,\tm_1\,\hdots\,\tm_m}}$.

\begin{proposition}[Typability characterizes head normalization]
\lprop{refinements_and_head_normal_forms}
The following are equivalent:
\begin{enumerate}
\item There exists $\tm' \in \termsdist$ such that $\tm' \refines \tm$.
%\item There exists $\tm' \in \termsdist$ such that $\tm' \refines \tm$
%      and $\tm' \todist^* \lamp{\lab_1}{\var_1}{\hdots\lamp{\lab_n}{\var_n}{\vartwo^\typ[]\hdots[]}}$.
\item There exists a head normal form $\tmtwo$ such that $\tm \tobeta^* \tmtwo$.
\end{enumerate}
\end{proposition}
\begin{proof}
\SeeAppendixRef{appendix_refinements_and_head_normal_forms}
The implication $(1 \implies 2)$
relies on \resultname{Simulation}~(\rprop{simulation}).
The implication $(2 \implies 1)$
relies on the fact that head normal forms are typable,
plus an auxiliary result of {\em Subject Expansion}.
\end{proof}

The first item of \resultname{Simulation}~(\rprop{simulation})
ensures that every step $\tm \tobeta \tmtwo$ can be simulated
in $\lambdadist$ starting from a term $\tm' \refines \tm$.
Actually, a finer relationship can be established between the
derivation spaces $\ulbDerivLam{\tm}$ and $\ulbDerivDist{\tm'}$.
For this, we introduce the notion of \defn{simulation residual}.

\begin{definition}[Simulation residuals]
Let $\tm' \refines \tm$ and let $\redex : \tm \tobeta \tmtwo$ be a step.
The constructive proof of \resultname{Simulation}~(\rprop{simulation})
associates the $\tobeta$-step $\redex$ to a possibly empty set of $\todist$-steps
$\{\redex_1,\hdots,\redex_n\}$ all of which start from $\tm'$.
We write $\redex/\tm' \eqdef \set{\redex_1,\hdots,\redex_n}$,
and we call $\redex_1,\hdots,\redex_n$ the \defn{simulation residuals of $\redex$ after $\tm'$}.
All the complete developments of $\redex/\tm'$ have a common target,
which we denote by $\tm'/\redex$,
called the {\em simulation residual of $\tm'$ after $\redex$}.
\end{definition}

Recall that, by abuse of notation, $\redex/\tm'$ stands
for some complete development of the set $\redex/\tm'$.
By \resultname{Simulation}~(\rprop{simulation}),
the following diagram always holds given $\tm' \refines \tm \tobeta \tmtwo$:
\[
  \xymatrix@R=.5cm@C=2cm{
    \tm
      \ar[r]^{\beta}_{\redex}
      \ar@{}[d]|*=0[@]{\rtimes}
  &
    \tmtwo
      \ar@{}[d]|*=0[@]{\rtimes}
  \\
    \tm'
      \ar@{->>}[r]^{\dist}_{\redex/\tm'}
  &
    \tm'/\redex
  }
\]

\begin{example}[Simulation residuals]
Let $\redex : \var\,((\lam{\var}{\var})\vartwo) \tobeta \var\,\vartwo$
and consider the terms:
\begin{itemize}
\item[]
$
\tm'_0 = (\var^{[\alpha^1,\beta^2] \tolab{3} \gamma^4}\,[
                (\lamp{5}{\var}{\var^{\alpha^1}})[\vartwo^{\alpha^1}],
                (\lamp{6}{\var}{\var^{\beta^2}})[\vartwo^{\beta^2}]
          ])
$
\item[]
$
\tm'_1 = \var^{[\alpha^1,\beta^2] \tolab{3} \gamma^4}\,[
           \vartwo^{\alpha^1},
           (\lamp{6}{\var}{\var^{\beta^2}})[\vartwo^{\beta^2}]
$
\item[]
$
\tm'_2 = \var^{[\alpha^1,\beta^2] \tolab{3} \gamma^4}\,[
           (\lamp{5}{\var}{\var^{\alpha^1}})[\vartwo^{\alpha^1}],
           \vartwo^{\beta^2}
         ]
$
\item[]
$
\tm'_3 = \var^{[\alpha^1,\beta^2] \tolab{3} \gamma^4}\,[
             \vartwo^{\alpha^1},
             \vartwo^{\beta^2}
           ]
$
\end{itemize}
Then $\tm'_0/\redex = \tm'_3$
and $\redex/\tm'_0 = \set{\redex_1,\redex_2}$, where
$\redex_1 : \tm'_0 \todist \tm'_1$
and
$\redex_2 : \tm'_0 \todist \tm'_2$.
\end{example}

The notion of simulation residual can be extended for many-step derivations.
\begin{definition}[Simulation residuals of/after derivations]
If $\tm' \refines \tm$ and $\redseq : \tm \tobeta^* \tmtwo$ is a derivation,
then $\redseq/\tm'$ and $\tm'/\redseq$ are defined as follows by induction on $\redseq$:
\[
  \begin{array}{c@{\hspace{.5cm}}c@{\hspace{.5cm}}c@{\hspace{.5cm}}c}
    \emptyDerivation/\tm' \eqdef \emptyDerivation
  &
    \redex\redseqtwo/\tm' \eqdef (\redex/\tm')(\redseqtwo/(\tm'/\redex))
  &
    \tm'/\emptyDerivation \eqdef \tm'
  &
    \tm'/\redex\redseqtwo \eqdef (\tm'/\redex)/\redseqtwo
  \\
  \end{array}
\]
\end{definition}
It is then easy to check that $\redseq/\tm' : \tm' \todist^* \tm'/\redseq$ and $\tm'/\redseq \refines \tmtwo$, by induction on $\redseq$.
Moreover, simulation residuals are well-defined modulo permutation equivalence:

\begin{proposition}[Compatibility]
\lprop{compatibility_of_simulation_residuals_and_permutation_equivalence}
If $\redseq \permeq \redseqtwo$
and $\tm \refines \src(\redseq)$ 
then
$\redseq/\tm \permeq \redseqtwo/\tm$
and $\tm/\redseq = \tm/\redseqtwo$.
\end{proposition}
\begin{proof}
\SeeAppendixRef{appendix_compatibility_of_simulation_residuals_and_permutation_equivalence}
By case analysis,
studying how permutation diagrams in the $\lambda$-calculus
are transported to permutation diagrams in $\lambdadist$
via simulation.
\end{proof}

The following result resembles the usual Cube Lemma~\cite[Lemma~12.2.6]{Barendregt:1984}:

\begin{lemma}[Cube]
\llem{generalized_cube_lemma}
If $\tm \refines \src(\redseq) = \src(\redseqtwo)$, then $(\redseq/\redseqtwo)/(\tm/\redseqtwo) \permeq (\redseq/\tm)/(\redseqtwo/\tm)$.
\end{lemma}
\begin{proof}
\SeeAppendixRef{appendix_generalized_cube_lemma}
By induction on $\redseq$ and $\redseqtwo$,
relying on an auxiliary result, the {\em Basic Cube Lemma},
when $\redseq$ and $\redseqtwo$ are single steps,
proved by exhaustive case analysis.
\end{proof}
As a result,
$
  (\redseq \sqcup \redseqtwo)/\tm =
  \redseq(\redseqtwo/\redseq)/\tm =
  (\redseq/\tm)((\redseqtwo/\redseq)/(\tm/\redseq)) \permeq
  (\redseq/\tm)((\redseqtwo/\tm)/(\redseqtwo/\redseq)) =
  (\redseq/\tm) \sqcup (\redseqtwo/\tm)
$.
Moreover, if $\redseq \permle \redseqtwo$ 
then $\redseq\redseqthree \permeq \redseqtwo$ for some $\redseqthree$.
So we have that $\redseq/\tm \permle (\redseq/\tm)(\redseqthree/(\tm/\redseq)) = \redseq\redseqthree/\tm \permeq \redseqtwo/\tm$
by \resultname{Compatibility}~(\rprop{compatibility_of_simulation_residuals_and_permutation_equivalence}).
Hence we may formulate a stronger simulation result:

\begin{corollary}[Algebraic Simulation]
\lcoro{algebraic_simulation}
\lcoro{simulation_residuals_and_prefixes}
Let $\tm' \refines \tm$.
Then the mapping $\ulbDerivLam{\tm} \to \ulbDerivDist{\tm'}$
given by $\cls{\redseq} \mapsto \cls{\redseq/\tm'}$ is a morphism of upper
semilattices.
\end{corollary}

\begin{example}
\lexample{algebraic_simulation_example}
Let $I = \lam{\var}{\var}$
and $\Delta = (\lamp{5}{\var}{\var^{\alpha^2}})[\varthree^{\alpha^2}]$
and let $\hat{\vartwo} = \vartwo^{[\alpha^2] \tolab{3} [\,] \tolab{4} \beta^5}$.
The refinement $\tm' := (\lamp{1}{\var}{\hat{\vartwo}[\var^{\alpha^2}][\,]})[\Delta] \refines (\lam{\var}{\vartwo\var\var})(I\varthree)$
induces a morphism between the upper semilattices represented by the following reduction graphs:
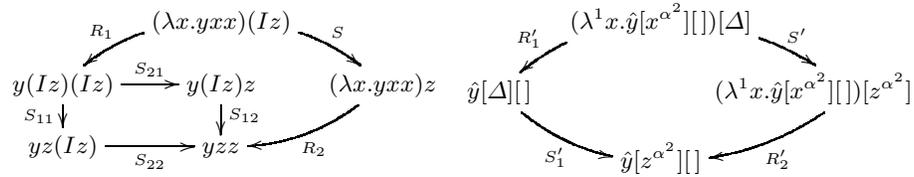
\[
{\small
  \xymatrix@R=.3cm@C=.3cm{
  &
    (\lam{\var}{\vartwo\var\var})(I\varthree)
    \ar@/_.25cm/[dl]_-{\redex_1}
    \ar@/^.25cm/[dr]^-{\redextwo}
  &
  &
  \\
    \vartwo(I\varthree)(I\varthree)
    \ar[d]_-{\redextwo_{11}}
    \ar[r]^-{\redextwo_{21}}
  &
    \vartwo(I\varthree)\varthree
    \ar[d]^{\redextwo_{12}}
  &
    (\lam{\var}{\vartwo\var\var})\varthree
    \ar@/^.25cm/[dl]^-{\redex_2}
  \\
    \vartwo \varthree (I\varthree)
    \ar[r]_-{\redextwo_{22}}
  &
    \vartwo \varthree \varthree
  }
  \hspace{-.3cm}
  \xymatrix@R=.3cm@C=.3cm{
  &
    (\lamp{1}{\var}{\hat{\vartwo}[\var^{\alpha^2}][\,]})[\Delta]
    \ar@/_.25cm/[dl]_-{\redex'_1}
    \ar@/^.25cm/[dr]^-{\redextwo'}
  \\
    \hat{\vartwo}[\Delta][\,]
    \ar@/_.25cm/[dr]_-{\redextwo'_1}
  &
  &
    \hspace{-1cm}(\lamp{1}{\var}{\hat{\vartwo}[\var^{\alpha^2}][\,]})[\varthree^{\alpha^2}]
    \ar@/^.25cm/[dl]^-{\redex'_2}
  \\
  &
    \hat{\vartwo}[\varthree^{\alpha^2}][\,]
  }
}
\]
For example
$(\redex_1 \sqcup \redextwo)/\tm' = (\redex_1\redextwo_{11}\redextwo_{22})/\tm' = \redex'_1\redextwo'_1 = \redex'_1 \sqcup \redextwo' = \redex_1/\tm' \sqcup \redextwo/\tm'$.
Note that the step $\redextwo_{22}$ is erased by the simulation: $\redextwo_{22}/(\hat{\vartwo}[\varthree^{\alpha^2}][\,]) = \emptyset$.
Intuitively, $\redextwo_{22}$ is ``garbage'' with respect to the refinement $\hat{\vartwo}[\varthree^{\alpha^2}][\,]$,
because it lies inside an {\em untyped} argument.
\end{example}

\section{Factoring Derivation Spaces}
  \lsec{factorization}
  
In this section we prove that the upper semilattice $\ulbDerivLam{\tm}$
may be factorized using a variant of the Grothendieck construction.
We start by formally defining the notion of {\em garbage}.

\begin{definition}[Garbage]
Let $\tm' \refines \tm$.
A derivation $\redseq : \tm \tobeta^* \tmtwo$ is \defn{$\tm'$-garbage}
if $\redseq/\tm' = \emptyDerivation$.
\end{definition}
The informal idea is that each refinement $\tm' \refines \tm$
specifies that some subterms of $\tm$ are ``useless''.
A subterm $\tmthree$ is useless if it lies inside the argument
of an application $\tmtwo(...\tmthree...)$
in such a way that the argument is not typed, \ie the refinement is
of the form $\tmtwo'[\,] \refines \tmtwo(...\tmthree...)$.
A single step $\redex$ is $\tm'$-garbage if the pattern of the contracted
redex lies inside a useless subterm.
A sequence of steps $\redex_1\redex_2\hdots\redex_n$
is $\tm'$-garbage if $\redex_1$ is $\tm'$-garbage,
$\redex_2$ is $(\tm'/\redex_1)$-garbage,
$\hdots$,
$\redex_i$ is $(\tm'/\redex_1\hdots\redex_{i-1})$-garbage,
$\hdots$,
and so on.

Usually we say that $\redseq$ is just {\em garbage}, when $\tm'$ is clear from the context.
For instance, in \rexample{algebraic_simulation_example},
$\redextwo_{21}$ is garbage, since $\redextwo_{21}/(\hat{\vartwo}[\Delta][\,]) = \emptyDerivation$.
Similarly, $\redextwo_{22}$ is garbage, since $\redextwo_{22}/(\hat{\vartwo}[\varthree^{\alpha^2}][\,]) = \emptyDerivation$.
On the other hand, $\redex_1\redextwo_{21}$ is not garbage, since
$\redex_1\redextwo_{21}/((\lamp{1}{\var}{\hat{\vartwo}[\var^{\alpha^2}][\,]})[\Delta]) = \redex'_1 \neq \emptyDerivation$.
For each $\tm' \refines \tm$, the set of $\tm'$-garbage derivations forms an {\em ideal} of the upper semilattice $\ulbDerivLam{\tm}$.
More precisely:
\begin{proposition}[Properties of garbage]
\lprop{properties_of_garbage}
Let $\tm' \refines \tm$. Then:
\begin{enumerate}
\item If $\redseq$ is $\tm'$-garbage and $\redseqtwo \permle \redseq$, then $\redseqtwo$ is $\tm'$-garbage.
  \label{properties_of_garbage__downwards_closed}
\item The composition $\redseq\redseqtwo$ is $\tm'$-garbage if and only if $\redseq$ is $\tm'$-garbage and $\redseqtwo$ is $(\tm'/\redseq)$-garbage.
\item If $\redseq$ is $\tm'$-garbage then $\redseq/\redseqtwo$ is $(\tm'/\redseqtwo)$-garbage.
\item The join $\redseq \sqcup \redseqtwo$ is $\tm'$-garbage if and only if $\redseq$ and $\redseqtwo$ are $\tm'$-garbage.
\end{enumerate}
\end{proposition}
\begin{proof}
\SeeAppendixRef{appendix_properties_of_garbage}
The proof is easy using \rprop{compatibility_of_simulation_residuals_and_permutation_equivalence} and \rlem{generalized_cube_lemma}.
\end{proof}
Our aim is to show that given $\redseq : \tm \tobeta^* \tmtwo$ and $\tm' \refines \tm$,
there is a unique way of decomposing $\redseq$ as $\redseqtwo\redseqthree$,
where $\redseqthree$ is $\tm'$-garbage and $\redseqtwo$ ``has no $\tm'$-garbage''.
Garbage is well-defined modulo permutation equivalence,
\ie
given $\redseq \permeq \redseqtwo$, we have that $\redseq$ is garbage if and only if $\redseqtwo$ is garbage.
In contrast, it is not immediate to give a well-defined notion of ``having no garbage''.
For example, in \rexample{algebraic_simulation_example},
$\redextwo\redex_2$ has no garbage steps, so it appears to have no garbage;
however, it is permutation equivalent to $\redex_1\redextwo_{11}\redextwo_{22}$, which
does contain a garbage step ($\redextwo_{22}$).
The following definition seems to capture the right notion of having no garbage:

\begin{definition}[Garbage-free derivation]
\ldef{garbage_free_derivation}
Let $\tm' \refines \tm$. A derivation $\redseq : \tm \tobeta^* \tmtwo$
is $\tm'$-garbage-free if
for any derivation $\redseqtwo$
such that $\redseqtwo \permle \redseq$
and $\redseq/\redseqtwo$ is $(\tm'/\redseqtwo)$-garbage,
then $\redseq/\redseqtwo = \emptyDerivation$.
\end{definition}
Again, we omit the $\tm'$ if clear from the context.
Going back to \rexample{algebraic_simulation_example},
the derivation $\redextwo\redex_2$ is not garbage-free,
because $\redex_1\redextwo_{11} \permle \redextwo\redex_2$
and $\redextwo\redex_2/\redex_1\redextwo_{11} = \redextwo_{22}$ is garbage but non-empty.
Note that \rdef{garbage_free_derivation}
is defined in terms of the prefix order ($\permle$), so:

\begin{remark}
If $\redseq \permeq \redseqtwo$,
then $\redseq$ is $\tm'$-garbage-free if and only if $\redseqtwo$ is $\tm'$-garbage-free.
\end{remark}

Next, we define an effective procedure ({\em sieving})
to erase all the garbage from a derivation.
The idea is that if $\redseq : \tm \tobeta^\star \tmtwo$ is a derivation
in the $\lambda$-calculus and $\tm' \refines \tm$ is any refinement,
we may constructively build a $\tm'$-garbage-free derivation $(\redseq \sieve \tm') : \tm \tobeta^\star \tmthree$
by erasing all the $\tm'$-garbage from $\redseq$.
Our goal will then be to show that $\redseq \permeq (\redseq \sieve \tm')\redseqtwo$
where $\redseqtwo$ is garbage.

\begin{definition}[Sieving]
Let $\tm' \refines \tm$ and $\redseq : \tm \tobeta^\star \tmtwo$.
A step $\redex$ is \defn{coarse for $(\redseq,\tm')$} if
$\redex \permle \redseq$ and $\redex/\tm' \neq \emptyset$.
The \defn{sieve of $\redseq$ with respect to $\tm'$},
written $\redseq \sieve \tm'$,
is defined as follows.
\begin{itemize}
\item If there are no coarse steps for $(\redseq,\tm')$,
      then $(\redseq \sieve \tm') \eqdef \emptyDerivation$.
\item If there is a coarse step for $(\redseq,\tm')$,
      then $(\redseq \sieve \tm') \eqdef \redex_0 ((\redseq/\redex_0) \sieve (\tm'/\redex_0))$
      where $\redex_0$ is the leftmost such step.
\end{itemize}
\end{definition}

\begin{lemma}
\llem{sieving_well_defined}
The sieving operation $\redseq \sieve \tm'$ is well-defined.
\end{lemma}
\begin{proof}
\SeeAppendixRef{appendix_sieving_well_defined}
To see that recursion terminates, consider the measure $M$ given by $M(\redseq,\tm') := \#\names(\redseq/\tm')$, and note that $M(\redseq,\tm') > M(\redseq/\redex_0,\tm'/\redex_0)$.
\end{proof}
For example, in \rexample{algebraic_simulation_example},
we have that
$\redextwo \sieve \tm' = \redextwo$
and
$\redextwo\redex_2 \sieve \tm' = \redex_1\redextwo_{11}$.

\begin{proposition}[Properties of sieving]
\lprop{properties_of_sieving}
Let $\tm' \refines \tm$ and $\redseq : \tm \tobeta^* \tmtwo$. Then:
\begin{enumerate}
\item $\redseq \sieve \tm'$ is $\tm'$-garbage-free and $\redseq \sieve \tm' \permle \redseq$.
\item $\redseq/(\redseq \sieve \tm')$ is $(\tm'/(\redseq \sieve \tm'))$-garbage.
\item $\redseq$ is $\tm'$-garbage if and only if $\redseq \sieve \tm' = \emptyDerivation$.
\item $\redseq$ is $\tm'$-garbage-free if and only if $\redseq \sieve \tm' \permeq \redseq$.
\end{enumerate}
\end{proposition}
\begin{proof}
\SeeAppendixRef{appendix_properties_of_sieving}
By induction on the length of $\redseq \sieve \tm'$, using various technical lemmas.
\end{proof}

As a consequence of the definition of the sieving construction and its properties,
given any derivation $\redseq : \tm \tobeta^* \tmtwo$
and any refinement $\tm' \refines \tm$,
we can always write $\redseq$, modulo permutation equivalence, as of the form
$\redseq \permeq \redseqtwo\redseqthree$
in such a way that $\redseqtwo$ is garbage-free and $\redseqthree$ is garbage.
To prove this take $\redseqtwo := \redseq \sieve \tm'$
and $\redseqthree := \redseq/(\redseq \sieve \tm')$,
and note that $\redseqtwo$ is garbage-free by item~1. of \rprop{properties_of_sieving},
$\redseqthree$ is garbage by item~2. of \rprop{properties_of_sieving},
and $\redseq \permeq \redseqtwo(\redseq/\redseqtwo) = \redseqtwo\redseqthree$
because $\redseqtwo \permle \redseq$ by item~1. of \rprop{properties_of_sieving}.

In the following we give a stronger version of this result.
The \resultname{Factorization} theorem below (\rthm{factorization_ulb_derivations})
states that this decomposition is actually an isomorphism of upper semilattices.
This means, on one hand, that given any derivation $\redseq : \tm \tobeta^* \tmtwo$
and any refinement $\tm' \refines \tm$
there is a {\em unique} way to factor $\redseq$ as of the form
$\redseq \equiv \redseqtwo\redseqthree$ where $\redseqtwo$ is garbage-free
and $\redseqthree$ is garbage.
On the other hand, it means that the decomposition
$\redseq \mapsto (\redseq\sieve\tm', \redseq/(\redseq\sieve\tm'))$
mapping each derivation to a of a garbage-free plus a garbage derivation
is {\em functorial}.
This means, essentially, that the set of pairs $(\redseqtwo,\redseqthree)$ such that $\redseqtwo$ is garbage-free
and $\redseqthree$ is garbage can be given the structure of an upper semilattice
in such a way that:
\begin{itemize}
\item If $\redseq \mapsto (\redseqtwo,\redseqthree)$
      and $\redseq' \mapsto (\redseqtwo',\redseqthree')$
      then $\redseq \permle \redseq' \iff (\redseqtwo,\redseqthree) \leq (\redseqtwo',\redseqthree')$.
\item If $\redseq \mapsto (\redseqtwo,\redseqthree)$
      and $\redseq' \mapsto (\redseqtwo',\redseqthree')$
      then $(\redseq \sqcup \redseq') \mapsto (\redseqtwo,\redseqthree) \lor (\redseqtwo',\redseqthree')$.
\end{itemize}
The upper semilattice structure of the set of pairs $(\redseqtwo,\redseqthree)$
is given using a variant of the Grothendieck construction:

\begin{definition}[Grothendieck construction for partially ordered sets]
Let $A$ be a poset, and let $B : A \to \Poset$ be a mapping associating
each object $a \in A$ to a poset $B(a)$.
Suppose moreover that $B$ is a {\em lax 2-functor}.
More precisely, for each $a \leq b$ in $A$,
the function $B(\pt{b}{a}) : B(a) \to B(b)$ 
is monotonic and such that:
\begin{enumerate}
\item $B(\pt{a}{a}) = \id_{B(a)}$ for all $a \in A$,
\item $B((\pt{c}{b}) \circ (\pt{b}{a})) \totwocell B(\pt{c}{b}) \circ B(\pt{b}{a})$ for all $a \leq b \leq c$ in $A$.
\end{enumerate}
The {\em Grothendieck construction} $\grothy{A}{B}$
is defined as the poset
given by the set of objects
$
  \set{(a,b) \ST a \in A,\ b \in B(a)}
$
and such that $(a, b) \leq (a', b')$
is declared to hold if and only if $A \leq a' \text{ and } B(\pt{a'}{a})(b) \leq b'$.
\end{definition}

The following proposition states that garbage-free derivations form a finite lattice,
while garbage derivations form an upper semilattice.

\begin{proposition}[Garbage-free and garbage semilattices]
\lprop{semilattices_of_garbage_free_and_garbage_derivations}
Let $\tm' \refines \tm$.
\begin{enumerate}
\item The set
      $F = \set{\cls{\redseq} \ST \src(\redseq) = \tm \text{ and } \redseq \text{ is $\tm'$-garbage-free}}$
      of $\tm'$-garbage-free derivations forms a finite lattice
      $\ulbFree{\tm'}{\tm} = (F,\leqF,\bot,\lorF,\top,\landF)$, with:
  \begin{itemize}
  \item {\em Partial order:} $\cls{\redseq} \leqF \cls{\redseqtwo} \iffdef \redseq/\redseqtwo \text{ is $(\tm'/\redseqtwo)$-garbage}$.
  \item {\em Bottom:} $\bot := \cls{\emptyDerivation}$.
  \item {\em Join:} $\cls{\redseq} \lorF \cls{\redseqtwo} \eqdef \cls{(\redseq \sqcup \redseqtwo) \sieve \tm'}$.
  \item {\em Top:} $\top$, defined as the join of all the $\cls{\redseqthree}$ such that $\redseqthree$ is $\tm'$-garbage-free.
  \item {\em Meet:} $\cls{\redseq} \landF \cls{\redseqtwo}$,
        defined as the join of all the $\cls{\redseqthree}$ such that $\cls{\redseqthree} \leqF \cls{\redseq}$ and $\cls{\redseqthree} \leqF \cls{\redseqtwo}$.
  \end{itemize}
\item The set $G = \set{\cls{\redseq} \ST \src(\redseq) = \tm \text{ and } \redseq \text{ is $\tm'$-garbage}}$
      of $\tm'$-garbage derivations forms an upper semilattice $\ulbGarbage{\tm'}{\tm} = (G,\permle,\bot,\sqcup)$,
      with the structure inherited from $\ulbDerivLam{\tm}$.
\end{enumerate}
\end{proposition}
\begin{proof}
\SeeAppendixRef{appendix_semilattices_of_garbage_free_and_garbage_derivations}
The proof relies on the properties of garbage and sieving~(\rprop{properties_of_garbage}, \rprop{properties_of_sieving}).
\end{proof}

Suppose that $\tm' \refines \tm$,
and let $\ulbF \eqdef \ulbFree{\tm'}{\tm}$ denote the lattice
of $\tm'$-garbage-free derivations.
Let $\ulbG : \ulbF \to \Poset$ be the lax 2-functor
$
  \ulbG(\cls{\redseq}) \eqdef \ulbGarbage{\tm'/\redseq}{\tgt(\redseq)}
$
with the following action on morphisms:
\[
  \begin{array}{rrlll}
  \ulbG(\ptF{\cls{\redseqtwo}}{\cls{\redseq}}) & : & \ulbG(\cls{\redseq}) & \to & \ulbG(\cls{\redseqtwo}) \\
  && \cls{\alpha} & \mapsto & \cls{\redseq\alpha/\redseqtwo}
  \end{array}
\]
Using the previous proposition (\rprop{semilattices_of_garbage_free_and_garbage_derivations})
it can be checked that $\ulbG$ is indeed a lax 2-functor, and that
the Grothendieck construction $\grothy{\ulbF}{\ulbG}$ forms an upper semilattice. The join is given by
$
  (a,b) \lor (a',b') = (a \lorF a', \ulbG(\ptF{a \lorF a'}{a})(b) \sqcup \ulbG(\ptF{a \lorF a'}{a'})(b'))
$.
Finally we can state the main theorem:

\begin{theorem}[Factorization]
\lthm{factorization_ulb_derivations}
The following maps form an isomorphism of upper semilattices:
\[
  \begin{array}{rcl}
    \ulbDerivLam{\tm}       & \to    & \grothy{\ulbF}{\ulbG} \\
    \cls{\redseq}           & \mapsto   & (\cls{\redseq \sieve \tm'}, \cls{\redseq / (\redseq \sieve \tm')} \\
  \end{array}
  \HS
  \begin{array}{rcl}
    \grothy{\ulbF}{\ulbG}             & \to    & \ulbDerivLam{\tm} \\
    (\cls{\redseq}, \cls{\redseqtwo}) & \mapsto & \cls{\redseq\redseqtwo} \\
  \end{array}
\]
\end{theorem}
\begin{proof}
\SeeAppendixRef{appendix_factorization_ulb_derivations}
The proof consists in checking that both maps are morphisms of upper semilattices and
that they are mutual inverses, resorting to \rprop{properties_of_garbage} and \rprop{properties_of_sieving}.
\end{proof}

\begin{example}
Let $\tm = (\lam{\var}{\vartwo\var\var})(I\varthree)$
and $\tm'$ be as in \rexample{algebraic_simulation_example}.
The upper semilattice $\ulbDerivLam{\tm}$ can be factorized as $\grothy{\ulbF}{\ulbG}$
as follows. Here posets are represented by their Hasse diagrams:
\[
\begin{array}{ccc}
  \xymatrix@R=.3cm@C=.3cm{
  &
    \cls{\emptyDerivation}
    \ar@/_.25cm/[dl]
    \ar@/^.25cm/[dr]
  &
  &
  \\
    \cls{\redex_1}
    \ar[d]
    \ar[r]
  &
    \cls{\redex_1\redextwo_{21}}
    \ar[d]
  &
    \cls{\redextwo}
    \ar@/^.25cm/[dl]
  \\
    \cls{\redex_1\redextwo_{11}}
    \ar[r]
  &
    \cls{\redex_1 \sqcup \redextwo}
  }
&
  \raisebox{-.85cm}{$\simeq$}
&
  \xymatrix@R=.3cm@C=.3cm{
  &
    (\cls{\emptyDerivation},\cls{\emptyDerivation})
    \ar@/_.25cm/[dl]
    \ar@/^.25cm/[dr]
  &
  &
  \\
    (\cls{\redex_1},\cls{\emptyDerivation})
    \ar[d]
    \ar[r]
  &
    (\cls{\redex_1},\cls{\redextwo_{21}})
    \ar[d]
  &
    (\cls{\redextwo},\cls{\emptyDerivation})
    \ar@/^.25cm/[dl]
  \\
    (\cls{\redex_1\redextwo_{11}},\cls{\emptyDerivation})
    \ar[r]
  &
    (\cls{\redex_1\redextwo_{11}},\cls{\redextwo_{22}})
  }
\end{array}
\]
For example $(\cls{\redextwo},\cls{\emptyDerivation}) \leq (\cls{\redex_1\redextwo_{11}},\cls{\redextwo_{22}})$
because $\cls{\redextwo} \leqF \cls{\redex_1\redextwo_{11}}$, that is,
$\redextwo/\redex_1\redextwo_{11} = \redextwo_{22}$ is garbage,
and
$\ulbG(\ptF{\cls{\redex_1\redextwo_{11}}}{\cls{\redextwo}})(\cls{\emptyDerivation})
= \cls{\redextwo/\redex_1\redextwo_{11}} = \cls{\redextwo_{22}} \permle \cls{\redextwo_{22}}$.
\end{example}

\section{Conclusions}
  \lsec{conclusions}
  
We have defined a calculus ($\lambdadist$) based on
non-idempotent intersection types.
Its syntax and semantics are complex due to the presence of an admittedly {\em ad hoc} correctness invariant for terms,
enforced so that reduction is confluent.
In contrast, derivation spaces in this calculus turn out to be very simple structures:
they are representable as {\em rings of sets}~(\rprop{labels_morphism})
and as a consequence they are distributive lattices~(\rcoro{ulbderivdist_distributive_lattice}).
Derivation spaces in the $\lambda$-calculus can be mapped to these
much simpler spaces using a strong notion of simulation~(\rcoro{algebraic_simulation})
inspired by residual theory. Building on this, we showed how the derivation space of any typable $\lambda$-term
may be factorized as a ``twisted product'' of garbage-free and garbage derivations~(\rthm{factorization_ulb_derivations}).

We believe that this validates the (soft) hypothesis
that explicitly representing resource management can provide insight on the structure
of derivation spaces.

\medskip

{\bf Related work.}
The \resultname{Factorization} theorem~(\rthm{factorization_ulb_derivations}) is reminiscent
of Melli\`es' abstract factorization result~\cite{DBLP:conf/ctcs/Mellies97}.
Given an axiomatic rewriting system fulfilling a number of axioms,
Melli\`es proves that every derivation can be uniquely
factorized as an {\em external} prefix followed by an {\em internal} suffix.
We conjecture that each refinement $\tm' \refines \tm$ should provide
an instantiation of Melli\`es' axioms, in such a way that our
$\tm'$-garbage-free/$\tm'$-garbage factorization coincides with his external/internal
factorization.
Melli\`es notes that any evaluation strategy that always selects external steps
is hypernormalizing. A similar result should hold for evaluation strategies that
always select {\em non-garbage} steps.

The notion of {\em garbage-free} derivation is closely related with the
notion of $X$-{\em neededness}~\cite{barendregt1987needed}.
A step $\redex$ is $X$-needed
if every reduction to a term $\tm \in X$ contracts a residual of $\redex$.
Recently, Kesner et al.~\cite{KRV18} have related
typability in a non-idempotent intersection type system $\mathcal{V}$
and weak-head neededness.
Using similar techniques, it should be possible to prove that
$\tm'$-garbage-free steps are $X$-needed,
where $X = \set{\tmtwo \ST \tmtwo' \refines \tmtwo}$
and $\tmtwo'$ is the $\todist$-normal form of $\tm'$.

There are several resource calculi in the literature
which perhaps could play a similar role as $\lambdadist$ to recover factorization results akin to \rthm{factorization_ulb_derivations}.
Kfoury~\cite{kfoury1996linearization} embeds the $\lambda$-calculus
in a {\em linear} $\lambda$-calculus that has no duplication nor erasure.
Ehrard and Regnier prove that the Taylor expansion of $\lambda$-terms~\cite{ehrhard2008uniformity}
commutes with normalization, similarly as in \resultname{Algebraic Simulation}~(\rcoro{algebraic_simulation}).
Mazza et al.~\cite{MazzaPellissierVial} study a general framework for {\em polyadic approximations},
corresponding roughly to the notion of {\em refinement} in this paper.

\medskip
\noindent{\bf Acknowledgements.}
To Eduardo Bonelli and Delia Kesner for introducing the first author to these topics.
To Luis Scoccola and the anonymous reviewers for helpful suggestions.

\bibliographystyle{splncs04}
\bibliography{biblio}

\newpage
\appendix
  \section{Appendix}
  
This technical appendix includes the detailed proofs of the results stated in the main body
of the paper that have been marked with the \SeeAppendix symbol.

\subsection{Proof of \rlem{subject_reduction} -- Subject Reduction}
\label{appendix_subject_reduction}

We are to prove that
if $\conof{(\lamp{\lab}{\var}{\tm})\ls{\tmtwo}}$ is correct
then $\conof{\subs{\tm}{\var}{\ls{\tmtwo}}}$ is correct,
and, moreover, that their unique typings are under the same typing context
and have the same type.
We need a few auxiliary results:

\begin{lemma}
\llem{substitution_permutation}
If $\ls{\tmtwo}$ is a permutation of $\ls{\tmthree}$, then $\subs{\tm}{\var}{\ls{\tmtwo}} = \subs{\tm}{\var}{\ls{\tmthree}}$.
\end{lemma}
\begin{proof}
By induction on $\tm$.
\end{proof}

\begin{lemma}
\llem{correct_subterms}
If $\tm$ is correct, then any subterm of $\tm$ is correct.
\end{lemma}
\begin{proof}
Note, by definition of correctness~(\rdef{sequentiality_and_correctness}),
that if $\tm$ is a correct abstraction $\tm = \lamp{\lab}{\var}{\tmtwo}$ then $\tmtwo$ is correct,
and if $\tm$ is a correct application $\tm = \tmtwo[\tmthree_1,\hdots,\tmthree_n]$
then $\tmtwo$ and $\tmthree_1,\hdots,\tmthree_n$ are correct.
This allows us to conclude by induction on $\tm$.
\end{proof}

\begin{lemma}[Relevance]
\llem{relevance}
If $\tctx \vdash \tm : \typ$
and $\var \in \dom\tctx$ then $\var \in \fv{\tm}$.
\end{lemma}
\begin{proof}
By induction on $\tm$.
\end{proof}

\begin{definition}
$\Lambda(\tm)$ stands for the multiset of labels decorating the lambdas of $\tm$:
\[
  \begin{array}{rcl}
    \Lambda(\var^\typ) & \eqdef & [\,] \\
    \Lambda(\lamp{\lab}{\var}{\tm}) & \eqdef & [\lab] + \Lambda(\tm)  \\
    \Lambda(\tm[\tmtwo_i]_{i=1}^{n}) & \eqdef & \Lambda(\tm) +_{i=1}^{n} \Lambda(\tmtwo_i)  \\
  \end{array}
\]
\end{definition}
For example, $\Lambda((\lamp{1}{\var}{\var^{[\alpha^2] \tolab{3} \alpha^2}})[\lamp{3}{\var}{\var^{\alpha^2}}]) = [1,3]$.

\begin{lemma}
\llem{labels_over_lambdas_substitution}
Let $\tm,\tmtwo_1,\hdots,\tmtwo_n$ be correct terms.
Then
$\Lambda(\subs{\tm}{\var}{[\tmtwo_i]_{i=1}^{n}}) = \Lambda(\tm) +_{i=1}^n \Lambda(\tmtwo_i)$.
\end{lemma}
\begin{proof}
By induction on $\tm$.
\end{proof}

To prove \rlem{subject_reduction}, we first show that substitution preserves typing,
and then that it preserves correctness.

\subsubsection{Substitution preserves typing}
\lsec{substitution_preserves_typing}

  More precisely, let us show that
  if $\tctx \vdash \conof{(\lamp{\lab}{\var}{\tm})\ls{\tmtwo}} : \typ$ is derivable,
  then $\tctx \vdash \conof{\subs{\tm}{\var}{\ls{\tmtwo}}} : \typ$ is derivable.
  By induction on the context $\con$.
  \begin{enumerate}
  \item {\bf Empty, $\con = \conbase$.}
    \label{subject_reduction__case_base_empty_context}
    By induction on $\tm$.
    \begin{enumerate}
    \item {\bf Variable (same), $\tm = \var^\typ$, $\ls{\tmtwo} = [\tmtwo]$.}
      We have that $\var : [\typ] \vdash \var^\typ : \typ$
      and $\tctxtwo \vdash \tmtwo : \typtwo$ are derivable,
      so we are done.
    \item {\bf Variable (different), $\tm = \vartwo^\typ$, $\vartwo \neq \var$, $\ls{\tmtwo} = []$.}
      We have that $\vartwo : [\typ] \vdash \vartwo^\typ : \typ$ is derivable, so we are done.
    \item {\bf Abstraction, $\tm = \lamp{\lab}{\vartwo}{\tmthree}$.}
      Let $\tctx \oplus \var : [\typtwo_i]_{i=1}^n \vdash \lamp{\lab}{\var}{\tmthree} : \mtyp \tolab{\lab} \typ$
      be derivable
      and $\tctxtwo_i \vdash \tmtwo_i : \typtwo_i$ be derivable for all $i=1..n$.
      By inversion of the \indrulename{\toI} rule,
      we have that $\tctx \oplus \vartwo : \mtyp \oplus \var : [\typtwo_i]_{i=1}^n \vdash \tmthree : \typ$
      is derivable, so by \ih
      $(\tctx \oplus \vartwo : \mtyp) +_{i=1}^{n} \tctxtwo_i \vdash \subs{\tmthree}{\var}{[\tmtwo_i]_{i=1}^n} : \typ$ is derivable.
      Observe that $\vartwo \not\in \fv{\tmtwo_i}$ so $\vartwo \not\in \dom\tctxtwo_i$ by \rlem{relevance}.
      Hence the previous judgment can be written as
      $(\tctx +_{i=1}^{n} \tctxtwo_i) \oplus \vartwo : \mtyp \vdash \subs{\tmthree}{\var}{[\tmtwo_i]_{i=1}^n} : \typ$.
      Applying the \indrulename{\toI} rule we obtain
      $\tctx +_{i=1}^{n} \tctxtwo_i \vdash \subs{\lamp{\lab}{\vartwo}{\tmthree}}{\var}{[\tmtwo_i]_{i=1}^n} : \mtyp \tolab{\lab} \typ$
      as required.
    \item {\bf Application, $\tm = \tmthree[\tmfour_j]_{j=1}^m$.}
      Let $\tctx \oplus \var : [\typtwo_i]_{i=1}^n \vdash \tmthree[\tmfour_j]_{j=1}^m : \typ$
      be derivable
      and $\tctxtwo_i \vdash \tmtwo_i : \typtwo_i$ be derivable for all $i=1..n$.
      By inversion of the \indrulename{\toE} rule,
      the multiset of types $[\typtwo_i]_{i=1}^n$
      may be partitioned as $[\typtwo_i]_{i=1}^n = \sum_{j=0}^{m} \mtyp_j$,
      and
      the typing context $\tctx$
      may be partitioned as $\tctx = \sum_{j=0}^{m} \tctx_j$
      in such a way that
      $\tctx_0 \oplus \var : \mtyp_0 \vdash \tmthree : [\typthree_j]_{j=1}^m \tolab{\lab} \typ$
      is derivable
      and
      $\tctx_j \oplus \var : \mtyp_j \vdash \tmfour_j : \typthree_j$ is derivable for all $j=1..m$.
      Consider a partition $(\ls{\tmtwo}_0,\hdots,\ls{\tmtwo}_j)$
      of the list $\ls{\tmtwo}$ 
      such that for every $j=0..m$ we have $\tmlabel{\ls{\tmtwo}_j} = \mtyp_j$.
      Observe that this partition exists since $\tmlabel{\ls{\tmtwo}_0 + \hdots + \ls{\tmtwo}_j} = \tmlabel{\ls{\tmtwo}} = \sum_{j=0}^m \mtyp_j = [\typtwo_i]_{i=1}^n = \varlabel{\var}{\tm}$.

      Moreover, let $\tctxthree_j = \sum_{i : \tmtwo_i \in \ls{\tmtwo}_j} \tctxtwo_i$ for all $j=0..m$.
      By \ih we have that
      $\tctx_0 + \tctxthree_0 \vdash \subs{\tmthree}{\var}{\ls{\tmtwo}_0} : [\typthree_j]_{j=1}^m \tolab{\lab} \typ$
      is derivable
      and
      $\tctx_j + \tctxthree_j \vdash \subs{\tmfour_j}{\var}{\ls{\tmtwo}_j} : \typthree_j$
      is derivable for all $j=1..m$.
      Applying the \indrulename{\toE} rule we obtain
      that
      $\sum_{j=0}^m \tctx_j + \sum_{j=0}^m \tctxthree_j \vdash \subs{\tmthree[\tmfour_j]}{\var}{\sum_{j=0}^m \ls{\tmtwo}_j} : \typ$
      is derivable.
      By definition of
      $\tctx_0,\hdots,\tctx_m$ and
      $\tctxthree_0,\hdots,\tctxthree_m$
      this judgment equals
      $\tctx +_{i=1}^n \tctxtwo_i \vdash \subs{\tmthree[\tmfour_j]}{\var}{\sum_{j=0}^m \ls{\tmtwo}_j} : \typ$.
      By definition of $\ls{\tmtwo}_0,\hdots,\ls{\tmtwo}_m$
      and \rlem{substitution_permutation}
      this in turn equals
      $\tctx +_{i=1}^n \tctxtwo_i \vdash \subs{\tmthree[\tmfour_j]}{\var}{\ls{\tmtwo}} : \typ$,
      as required.
    \end{enumerate}
  \item {\bf Under an abstraction, $\con = \lamp{\labtwo}{\vartwo}{\con'}$.}
    Straightforward by \ih.
  \item {\bf Left of an application, $\con = \con'\ls{\tmthree}$.}
    Straightforward by \ih.
  \item {\bf Right of an application, $\con = \tmthree[\ls{\tmfour}_1,\con',\ls{\tmfour}_2]$.}
    Straightforward by \ih.
  \end{enumerate}

\subsubsection{Substitution preserves correctness}

  More precisely, let us show that
  if $\conof{(\lamp{\lab}{\var}{\tm})\ls{\tmtwo}}$ is correct
  then $\conof{\subs{\tm}{\var}{\ls{\tmtwo}}}$ is correct.
  By induction on $\con$:
  \begin{enumerate}
  \item {\bf Empty, $\con = \conbase$.}
    \label{subject_reduction__case_base_empty_context}
    Let $\ls{\tmtwo} = [\tmtwo_1,\hdots,\tmtwo_n]$.
    Observe that if $(\lamp{\lab}{\var}{\tm})[\tmtwo_1,\hdots,\tmtwo_n]$ is correct then:
    \begin{itemize}
    \item \condp{1} $\tctx \oplus \var : [\typtwo_i]_{i=1}^{n} \vdash \tm : \typ$ and $\tctxtwo_i \vdash \tmtwo_i : \typtwo_i$ are derivable for all $i=1..n$,
    \item \condp{2} $\tm,\tmtwo_1,\hdots,\tmtwo_n$ are correct,
    \item \condp{3} there are no free occurrences of $\var$ among $\tmtwo_1,\hdots,\tmtwo_n$,
    \item \condp{4} all the lambdas occurring in $\tm,\tmtwo_1,\hdots,\tmtwo_n$ have pairwise distinct labels,
    \item \condp{5} $\tctx +_{i=1}^n \tctxtwo_i$ is a sequential context.
    \end{itemize}
    Condition \condp{1} holds by inversion of the typing rules,
    condition \condp{2} holds by~\rlem{correct_subterms},
    condition \condp{3} holds by Barendregt's convention,
    and conditions \condp{4} and \condp{5} hold because the source term is supposed to be correct.

    By induction on $\tm$, we check that
    if $(\lamp{\lab}{\var}{\tm})[\tmtwo_1,\hdots,\tmtwo_n]$
    is correct, then
    $\subs{\tm}{\var}{\ls{\tmtwo}}$ is correct.
    \begin{enumerate}
    \item {\bf Variable (same), $\tm = \var^{\typ}$, $\ls{\tmtwo} = [\tmtwo]$.}
      Note that $\subs{\tm}{\var}{\ls{\tmtwo}} = \tmtwo$. Conclude by \condp{2}.
    \item {\bf Variable (different), $\tm = \vartwo^{\typ}$, $\ls{\tmtwo} = []$ with $\var \neq \vartwo$.}
      Note that $\subs{\tm}{\var}{\ls{\tmtwo}} = \vartwo^\typ$.
      Conclude by \condp{2}.
    \item {\bf Abstraction, $\tm = \lamp{\labtwo}{\vartwo}{\tmthree}$.}
      \label{subject_reduction__case_abstraction}
      Then $\typ = \mtyp \tolab{\labtwo} \typthree$ and
      by inversion $\tctx \oplus \vartwo : \mtyp \vdash \tmthree : \typthree$ is derivable.
      Note that $(\lamp{\lab}{\var}{\tmthree})\ls{\tmtwo}$ is correct,
      so by \ih $\subs{\tmthree}{\var}{\ls{\tmtwo}}$ is correct.
      The variable $\vartwo$ does not occur free in $\ls{\tmtwo}$,
      so $(\tctx \oplus \vartwo : \mtyp) +_{i=1}^{n} \tctxtwo_i = (\tctx +_{i=1}^{n} \tctxtwo_i) \oplus (\vartwo : \mtyp)$.
      Let us check that $\lamp{\labtwo}{\vartwo}{\subs{\tmthree}{\var}{\ls{\tmtwo}}}$ is correct:
      \begin{enumerate}
      \item {\em Uniquely labeled lambdas.}
        Let $\lab_1$ and $\lab_2$ be two labels decorating different lambdas of 
        $\lamp{\labtwo}{\vartwo}{\subs{\tmthree}{\var}{\ls{\tmtwo}}}$,
        and let us show that $\lab_1 \neq \lab_2$.
        There are two subcases, depending on whether one of the labels
        decorates the outermost lambda:
        \begin{enumerate}
        \item
          {\bf If $\lab_1$ or $\lab_2$ decorates the outermost lambda.}
          Suppose without loss of generality that $\lab_1 = \labtwo$ is the label decorating the outermost lambda.
          Then by \rlem{labels_over_lambdas_substitution},
          there are two cases: either $\lab_2$ decorates a lambda of $\tmthree$,
          or $\lab_2$ decorates a lambda of some term in the list $\ls{\tmtwo}$.
          If $\lab_2$ decorates a lambda of $\tmthree$, then $\lab_1 \neq \lab_2$
          since we knew that $\lamp{\labtwo}{\vartwo}{\tmthree}$ was a correct term by \condp{2}.
          If $\lab_2$ decorates a lambda of some term in the list $\ls{\tmtwo}$, then $\lab_1 \neq \lab_2$
          by condition \condp{4}.
        \item
          {\bf If $\lab_1$ and $\lab_2$ do not decorate the outermost lambda.}
          Then $\lab_1$ and $\lab_2$ decorate different lambdas of the term $\subs{\tmthree}{\var}{\ls{\tmtwo}}$,
          and we conclude by \ih.
        \end{enumerate}
      \item {\em Sequential contexts.}
        \label{subject_reduction__case_abstraction_sequential_contexts}
        Let $\tm'$ be a subterm of $\lamp{\labtwo}{\var}{\subs{\tmthree}{\var}{\ls{\tmtwo}}}$.
        If $\tm'$ is a subterm of $\subs{\tmthree}{\var}{\ls{\tmtwo}}$, we conclude by \ih.
        Otherwise $\tm'$ is the whole term and the context is $\tctx +_{i=1}^{n} \tctxtwo_i$,
        which is sequential by hypothesis \condp{5}.
      \item {\em Sequential types.}
        \label{subject_reduction__case_abstraction_sequential_types}
        Let $\tm'$ be a subterm of $\lamp{\labtwo}{\var}{\subs{\tmthree}{\var}{\ls{\tmtwo}}}$.
        If $\tm'$ is a subterm of $\subs{\tmthree}{\var}{\ls{\tmtwo}}$, we conclude by \ih.
        Otherwise $\tm'$ is the whole term. Then $\tctx +_{i=1}^{n} \tctxtwo_i \vdash \tm' : \typ$ is derivable,
        since we have already shown that substitution preserves typing.
        Let $\mtyptwo \tolab{\labthree} \typfour$ be
        a type such that $\mtyptwo \tolab{\labthree} \typfour \occursin \tctx +_{i=1}^{n} \tctxtwo_i$ or $\mtyptwo \tolab{\labthree} \typfour \occursin \typ$,
        and let us show that $\mtyptwo$ is sequential.
        In the first case, \ie if $\mtyptwo \tolab{\labthree} \typfour \occursin \tctx +_{i=1}^{n} \tctxtwo_i$ holds,
        then either $\mtyptwo \tolab{\labthree} \typfour \occursin \tctx$ or $\mtyptwo \tolab{\labthree} \typfour \occursin \tctxtwo_i$ for some $i=1..n$,
        and we have that $\mtyptwo$ is sequential because all the terms $\tm,\tmtwo_1,\hdots,\tmtwo_n$ are correct by \condp{2}.
        In the second case, \ie if $\mtyptwo \tolab{\labthree} \typfour \occursin \typ$ holds,
        then we have that $\mtyptwo$ is sequential because $\tm$ is correct by \condp{2}.
      \end{enumerate}
    \item {\bf Application, $\tm = \tmthree\,\ls{\tmfour}$.}
      \label{subject_reduction__case_application} 
      Let $\ls{\tmfour} = [\tmfour_1,\hdots,\tmfour_m]$.
      Note that $(\lamp{\lab}{\var}{\tmthree})\ls{\tmtwo}_0$ is correct
      and $(\lamp{\lab}{\var}{\tmfour_j})\ls{\tmtwo}_j$ is correct for all $j=1..m$,
      which means that we may apply the \ih in all these cases.
      Let us show that $\subs{\tmthree[\tmfour_j]_{j=1}^m}{\var}{\ls{\tmtwo}}$ is correct:
      \begin{enumerate}
      \item {\em Uniquely labeled lambdas.}
        Let $\lab_1$ and $\lab_2$ be two labels decorating different lambdas of
        $\subs{\tmthree[\tmfour_j]_{j=1}^m}{\var}{\ls{\tmtwo}}$, and let us show that $\lab_1 \neq \lab_2$.
        Observe that the term
        $\subs{\tmthree[\tmfour_j]_{j=1}^m}{\var}{\ls{\tmtwo}} =
         \subs{\tmthree}{\var}{\ls{\tmtwo}_0}[\subs{\tmfour_j}{\var}{\ls{\tmtwo}_j}]_{j=1}^m$
        has $m + 1$ immediate subterms, namely
        $\subs{\tmthree}{\var}{\ls{\tmtwo}_0}$
        and $\subs{\tmfour_j}{\var}{\ls{\tmtwo}_j}$ for each $j=1..m$.
        We consider two subcases, depending on whether $\lab_1$ and $\lab_2$
        decorate two lambdas in the same immediate subterm or in different immediate subterms.

        \begin{enumerate}
        \item {\bf The labels $\lab_1$ and $\lab_2$ decorate the same immediate subterm.}
          That is, $\lab_1$ and $\lab_2$ both decorate lambdas in $\subs{\tmthree}{\var}{\ls{\tmtwo}_0}$
          or both decorate lambdas in some $\subs{\tmfour_j}{\var}{\ls{\tmtwo}_j}$ for some $j=1..m$.
          Then we conclude, since both $\subs{\tmthree}{\var}{\ls{\tmtwo}_0}$ and
          the $\subs{\tmfour_j}{\var}{\ls{\tmtwo}_j}$ are correct by \ih.
        \item {\bf The labels $\lab_1$ and $\lab_2$ decorate different subterms.}
          Let $\tmfour_0 := \tmthree$.
          Then we have that
          $\lab_1$ decorates a lambda in $\subs{\tmfour_j}{\var}{\ls{\tmtwo}_j}$ for some $j=0..m$
          and
          $\lab_2$ decorates a lambda in $\subs{\tmfour_k}{\var}{\ls{\tmtwo}_k}$ for some $k=0..m$, $j \neq k$.
          By \rlem{labels_over_lambdas_substitution},
          $\lab_1$ decorates a lambda in $\tmfour_j$ or a lambda in a term of the list $\ls{\tmtwo}_j$,
          and similarly
          $\lab_2$ decorates a lambda in $\tmfour_k$ or a lambda in a term of the list $\ls{\tmtwo}_k$.
          This leaves four possibilities, which are all consequence of \condp{4}.
        \end{enumerate}
      \item {\em Sequential contexts.}
        Similar to item~\refcase{subject_reduction__case_abstraction_sequential_contexts}.
      \item {\em Sequential types.}
        Similar to item~\refcase{subject_reduction__case_abstraction_sequential_types}.
      \end{enumerate}
    \end{enumerate}
  \item {\bf Under an abstraction, $\con = \lamp{\labtwo}{\vartwo}{\con'}$.}
    Note that $\tctx \vdash \lamp{\labtwo}{\vartwo}{\con'\of{\subs{\tm}{\var}{\ls{\tmtwo}}}} : \mtyp \tolab{\labtwo} \typ$ is derivable.
    Let us check the three conditions to see that $\lamp{\labtwo}{\vartwo}{\con'\of{\subs{\tm}{\var}{\ls{\tmtwo}}}}$ is correct:
    \begin{enumerate}
    \item {\em Uniquely labeled lambdas.}
      Any two lambdas in $\con'\of{\subs{\tm}{\var}{\ls{\tmtwo}}}$ have different labels by \ih.
      We are left to check that $\labtwo$ does not decorate any lambda in $\con'\of{\subs{\tm}{\var}{\ls{\tmtwo}}}$.
      Let $\lab_1$ be a label that decorates a lambda in $\con'\of{\subs{\tm}{\var}{\ls{\tmtwo}}}$.
      Then we have that $\lab_1$ decorates a lambda in $\con'$,
      or it decorates a lambda in $\subs{\tm}{\var}{\ls{\tmtwo}}$.
      By what we proved in item~\refcase{subject_reduction__case_base_empty_context}
      this in turn means that it decorates a lambda in $\tm$ or a lambda in some of the terms of the list $\ls{\tmtwo}$.
      In any of these cases we have that $\lab_1 \neq \labtwo$
      since $\lamp{\labtwo}{\vartwo}{\con'\of{(\lamp{\lab}{\var}{\tm})\ls{\tmtwo}}}$ is correct.
    \item {\em Sequential contexts.}
      Let $\tm'$ be a subterm of $\lamp{\labtwo}{\vartwo}{\con'\of{\subs{\tm}{\var}{\ls{\tmtwo}}}}$
      and let us check that its typing context is sequential.
      If $\tm'$ is a subterm of $\con'\of{\subs{\tm}{\var}{\ls{\tmtwo}}}$ we conclude by \ih.
      We are left to check the property for $\tm'$ being the whole term, \ie that $\tctx$ is sequential.
      By \ih, $\tctx \oplus \vartwo : \mtyp$ is sequential, which implies that $\tctx$ is sequential.
    \item {\em Sequential types.}
      Let $\tm'$ be a subterm of $\lamp{\labtwo}{\vartwo}{\con'\of{\subs{\tm}{\var}{\ls{\tmtwo}}}}$
      and let us check that, if $\mtyptwo \tolab{\labthree} \typthree$ is any type occurring in the typing
      context or in the type of $\tm'$, then $\mtyptwo$ is sequential.
      If $\tm'$ is a subterm of $\con'\of{\subs{\tm}{\var}{\ls{\tmtwo}}}$ we conclude by \ih.
      We are left to check the property for $\tm'$ being the whole term.

      If $\mtyptwo \tolab{\labthree} \typthree \occursin \tctx$,
      then $\mtyptwo \tolab{\labthree} \typthree \occursin \tctx \oplus \vartwo:\mtyptwo$
      which is the type of $\con'\of{\subs{\tm}{\var}{\ls{\tmtwo}}}$,
      so by \ih $\mtyptwo$ is sequential.

      If $\mtyptwo \tolab{\labthree} \typthree \occursin \mtyp \tolab{\lab''} \typ$,
      there are three subcases:
      \begin{enumerate}
      \item If $\mtyptwo = \mtyp$, then note that $\mtyp$ is sequential
            because $\tctx \oplus \vartwo : \mtyp$ is the typing context of $\con'\of{\subs{\tm}{\var}{\ls{\tmtwo}}}$,
           which is sequential by \ih.
      \item If $\mtyptwo \tolab{\labthree} \typthree \occursin \typtwo$ where $\typtwo$ is one of the types of $\mtyp$,
            then $\mtyptwo \tolab{\labthree} \typthree \occursin \mtyp \tolab{\lab''} \typ$ which is the typing context
            of $\con'\of{\subs{\tm}{\var}{\ls{\tmtwo}}}$,
            and we conclude for this term has sequential types by \ih.
      \item If $\mtyptwo \tolab{\labthree} \typthree \occursin \typ$,
            note that $\typ$ is the type of $\con'\of{\subs{\tm}{\var}{\ls{\tmtwo}}}$,
            and we conclude for this term has sequential types by \ih.
      \end{enumerate}
    \end{enumerate}
  \item {\bf Left of an application, $\con = \con'\ls{\tmthree}$.}
    \label{subject_reduction__case_left_application}
    Note that $\tctx \vdash \con'\of{\subs{\tm}{\var}{\ls{\tmtwo}}} : [\typtwo_j]_{j=1}^m \tolab{\labtwo} \typ$
    is derivable.
    Moreover the list of arguments is of the form
    $\ls{\tmthree} = [\tmthree_1,\hdots,\tmthree_m]$ where all the $\tmthree_j$ are correct
    and $\tctxtwo_j \vdash \tmthree_j : \typtwo_j$ is derivable for all $j=1..m$.
    Then $\tctx +_{j=1}^m \tctxtwo_j \vdash \con'\of{\subs{\tm}{\var}{\ls{\tmtwo}}}[\tmthree_j]_{j=1}^{m} : \typ$
    is derivable.
    Let us check the three conditions to see that $\con'\of{\subs{\tm}{\var}{\ls{\tmtwo}}}[\tmthree_j]_{j=1}^{m}$ is correct:
    \begin{enumerate}
    \item {\em Uniquely labeled lambdas.}
      Let $\lab_1$ and $\lab_2$ be two labels decorating different lambdas in $\con'\of{\subs{\tm}{\var}{\ls{\tmtwo}}}[\tmthree_j]_{j=1}^{m}$.
      There are three subcases.
      \begin{enumerate}
      \item
        If $\lab_1$ and $\lab_2$ both decorate lambdas in the subterm $\con'\of{\subs{\tm}{\var}{\ls{\tmtwo}}}$
        then $\lab_1 \neq \lab_2$ since $\con'\of{\subs{\tm}{\var}{\ls{\tmtwo}}}$ is correct by \ih.
      \item
        If $\lab_1$ and $\lab_2$ both decorate lambdas somewhere in $[\tmthree_1,\hdots,\tmthree_m]$
        then $\lab_1 \neq \lab_2$ since $\con'\of{\tm}[\tmthree_1,\hdots,\tmthree_m]$ is correct by hypothesis.
      \item
        If $\lab_1$ decorates a lambda in $\con'\of{\subs{\tm}{\var}{\ls{\tmtwo}}}$
        and $\lab_2$ decorates a lambda in one of the terms $\tmthree_j$ for some $j=1..m$,
        then note that $\lab_1$ must either decorate a lambda in $\con'$ or a lambda in $\subs{\tm}{\var}{\ls{\tmtwo}}$.
        By what we proved in item~\refcase{subject_reduction__case_base_empty_context}
        this in turn means that it decorates a lambda in $\tm$ or a lambda in some of the terms of the list $\ls{\tmtwo}$.
        In any of these cases we have that $\lab_1 \neq \lab_2$
        since $\con'\of{(\lamp{\lab}{\var}{\tm})\ls{\tmtwo}}[\tmthree_j]_{j=1}^m$ is correct by hypothesis.
      \end{enumerate}
    \item {\em Sequential contexts.}
      Let $\tm'$ be a subterm of $\con'\of{\subs{\tm}{\var}{\ls{\tmtwo}}}[\tmthree_j]_{j=1}^m$
      and let us show that its typing context is sequential.
      If $\tm'$ is a subterm of $\con'\of{\subs{\tm}{\var}{\ls{\tmtwo}}}$ we conclude by \ih.
      If $\tm'$ is a subterm of one of the $\tmthree_j$ for some $j=1..m$, we conclude using that $\tmthree_j$ is correct by hypothesis.
      It remains to check that the whole term is correct, \ie that $\tctx +_{j=1}^m \tctxtwo_j$ is sequential.
      Observe that $\tctx +_{j=1}^m \tctxtwo_j$ is also the typing context of
      $\con'\of{(\lamp{\lab}{\var}{\tm})\ls{\tmtwo}}[\tmthree_j]_{j=1}^m$, which is correct by hypothesis.
    \item {\em Sequential types.}
      Let $\tm'$ be a subterm of $\con'\of{\subs{\tm}{\var}{\ls{\tmtwo}}}[\tmthree_j]_{j=1}^m$
      and let us show if $\mtyptwo \tolab{\labtwo} \typthree$ is a type
      that occurs in the typing context or in the type of $\tm'$, then $\mtyptwo$ is sequential.
      If $\tm'$ is a subterm of $\con'\of{\subs{\tm}{\var}{\ls{\tmtwo}}}$ we conclude by \ih.
      If $\tm'$ is a subterm of one of the $\tmthree_j$ for some $j=1..m$, we conclude using that $\tmthree_j$ is correct by hypothesis.
      We are left to check the property for $\tm'$ being the whole term.

      If $\mtyptwo \tolab{\labtwo} \typthree \occursin \tctx +_{j=1}^m \tctxtwo_j$,
      we conclude by observing
      that $\tctx +_{j=1}^m \tctxtwo_j$ is also the typing context of $\con'\of{(\lamp{\lab}{\var}{\tm})\ls{\tmtwo}}[\tmthree_j]_{j=1}^m$,
      which is correct by hypothesis, so it has sequential types.
      
      Similarly, if $\mtyptwo \tolab{\labtwo} \typthree \occursin \typ$,
      we conclude by observing
      that $\typ$ is also the type of $\con'\of{(\lamp{\lab}{\var}{\tm})\ls{\tmtwo}}[\tmthree_j]_{j=1}^m$.
    \end{enumerate}
  \item {\bf Right of an application, $\con = \tmthree[\ls{\tmfour}_1,\con',\ls{\tmfour}_2]$.}
    Similar to item~\refcase{subject_reduction__case_left_application}.
  \end{enumerate}

\subsection{Proof of \rlem{substitution_lemma} -- Substitution Lemma}
\label{appendix_substitution_lemma}

Let us prove that
$ \subs{\subs{\tm}{\var}{\ls{\tmtwo}}}{\vartwo}{\ls{\tmthree}}
  =
  \subs{\subs{\tm}{\vartwo}{\ls{\tmthree}_1}}{\var}{\subs{\ls{\tmtwo}}{\vartwo}{\ls{\tmthree}_2}}$
by induction on $\tm$.
The interesting case is the application.
Let $\tm = \tmfour [\tmfour_1,\hdots,\tmfour_n]$. Then:
\[
\begin{array}{rcl}
  && \subs{\subs{(\tmfour [\tmfour_1,\hdots,\tmfour_n])}{\var}{\ls{\tmtwo}}}{\vartwo}{\ls{\tmthree}}
  \\
  & = & \subs{(\subs{\tmfour}{\var}{\ls{\tmtwo}_0} [\subs{\tmfour_i}{\var}{\ls{\tmtwo}_i}]_{i=0}^n)}
             {\vartwo}{\ls{\tmthree}} \hspace{1cm}\\
  & = & \subs{\subs{\tmfour}{\var}{\ls{\tmtwo}_0}}{\vartwo}{\ls{\tmthree}_0}
             [\subs{\subs{\tmfour_i}{\var}{\ls{\tmtwo}_i}}{\vartwo}{\ls{\tmthree}_i}]_{i=0}^n \\
  & \eqih & \subs{\subs{\tmfour}{\vartwo}{\ls{\tmthree}_{0,1}}}{\var}{\subs{\ls{\tmtwo}_0}
                  {\vartwo}{\ls{\tmthree}_{0,2}}}
            [\subs{\subs{\tmfour_i}{\vartwo}{\ls{\tmthree}_{i,1}}}{\var}{\subs{\ls{\tmtwo}_0}
                  {\vartwo}{\ls{\tmthree}_{i,2}}}]_{i=0}^n \\
  & = & \subs{(\subs{\tmfour}{\vartwo}{\ls{\tmthree}_{0,1}}
                     [\subs{\tmfour_i}{\vartwo}{\ls{\tmthree}_{i,1}}]_{i=0}^n)}
              {\var}{\sum_{i=1}^n \subs{\ls{\tmtwo}_i}{\vartwo}{\ls{\tmthree}_{i,2}}} \\
\end{array}
\]
Given that $\sum_{i=0}^n \ls{\tmtwo}_i$ is a permutation of $\ls{\tmtwo}$,
by defining $\ls{\tmthreevariant}_2 := \sum_{i=0}^n \ls{\tmthree}_{i,2}$,
the last term equals
$
  \subs{(\subs{\tmfour}{\vartwo}{\ls{\tmthree}_{0,1}}
           [\subs{\tmfour_i}{\vartwo}{\ls{\tmthree}_{i,1}}]_{i=0}^n)}
       {\var}{\subs{\ls{\tmtwo}}{\vartwo}{\ls{\tmthreevariant}_2}}
$.
Finally, by defining $\ls{\tmthreevariant}_1 := \sum_{i=0}^n \ls{\tmthree}_{i,1}$,
we obtain
$
  \subs{\subs{(\tmfour [\tmfour_1, \hdots, \tmfour_n])}
             {\vartwo}{\ls{\tmthreevariant}_1}}
       {\var}{\subs{\ls{\tmtwo}}{\vartwo}{\ls{\tmthreevariant}_2}}
$.
To conclude observe that $\ls{\tmthreevariant}_1 + \ls{\tmthreevariant}_2$ is indeed a permutation of $\ls{\tmthree}$:
\[
  \begin{array}{rcll}
  \ls{\tmthreevariant}_1 + \ls{\tmthreevariant}_2
  & = & (\sum_{i=0}^n \ls{\tmthree}_{i,1}) + (\sum_{i=0}^n \ls{\tmthree}_{i,2}) \\
  & \approx & \sum_{i=0}^n \ls{\tmthree}_{i,1} + \ls{\tmthree}_{i,2} \\
  & \approx & \sum_{i=0}^n \ls{\tmthree}_i & \text{by \ih on each index $i=0..n$} \\
  & = & \ls{\tmthree}
  \end{array}
\]
where $\lst \approx \lsttwo$ is the equivalence relation on lists that holds whenever $\lst$ is a permutation of $\lsttwo$.

\subsection{Proof of \rprop{strong_permutation} -- Permutation}
\label{appendix_strong_permutation}

Let us extend the operation of substitution to operate on contexts
by declaring that $\varlabel{\var}{\conbase} = [\,]$
and $\subs{\conbase}{\var}{[\,]} = \conbase$.
We need two auxiliary lemmas:

\begin{lemma}[Substitution lemma for contexts]
\llem{substitution_lemma_for_contexts}
\llem{substitution_lemma_alt}
If both sides of the equation are defined
and $(\ls{\tmtwo}_1,\ls{\tmtwo}_2)$ is a partition of $\ls{\tmtwo}$
then
$\subs{\conof{\tm}}{\var}{\ls{\tmtwo}} = \subs{\con}{\var}{\ls{\tmtwo_1}}\of{\subs{\tm}{\var}{\ls{\tmtwo}_2}}$.
\end{lemma}
\begin{proof}
The proof is similar to the Substitution Lemma, by induction on $\con$.
\end{proof}

\begin{lemma}[Reduction inside a substitution]
\llem{reduce_in_substitution}
Let $1 \leq i \leq n$ and $\tmtwo_i \tolab{\lab} \tmtwo'_i$.
Then:
\[
 \subs{\tm}{\var}{[\tmtwo_1,\hdots,\tmtwo_{i-1},\tmtwo_i,\tmtwo_{i+1},\hdots,\tmtwo_n]} \tolab{\lab}
 \subs{\tm}{\var}{[\tmtwo_1,\hdots,\tmtwo_{i-1},\tmtwo'_i,\tmtwo_{i+1},\hdots,\tmtwo_n]}
\].
\end{lemma}
\begin{proof}
Straightforward by induction on $\tmtwo$.
\end{proof}

The proof of \rprop{strong_permutation} proceeds as follows.
Let $\redex : \tm_0 \todistl{\lab} \tm_1$ and $\redextwo : \tm_0 \todistl{\lab'} \tm_2$
be coinitial steps, and let us show that the diagram may be closed.
The step $\redex$ is of the form
$
  \tm_0 = \conof{(\lamp{\lab}{\var}{\tm})\ls{\tmtwo}}
           \todistl{\lab} \conof{\subs{\tm}{\var}{\ls{\tmtwo}}} = \tm_1
$.
Recall that $\redex \neq \redextwo$ by hypothesis.
We proceed by induction on $\con$.
\begin{enumerate}
\item {\bf Empty context, $\con = \conbase$.}
  Then
  $
    \redex : \tm_0 = (\lamp{\lab}{\var}{\tm})\ls{\tmtwo}
             \todistl{\lab} \subs{\tm}{\var}{\ls{\tmtwo}} = \tm_1
  $.
  There are two subcases, depending on whether the pattern of $\redextwo$ is inside $\tm$
  or inside $\ls{\tmtwo}$:
  \begin{enumerate}
  \item {\bf The pattern of $\redextwo$ is in $\tm$.}
    \label{strong_permutation__case_redextwo_inside_tm}
    Using \rlem{substitution_lemma_for_contexts}, the situation is:
    \[
    \xymatrix@C=.25cm@R=.5cm{
     (\lamp{\lab}{\var}{\conof{(\lamp{\lab'}{\vartwo}{\tmfour}) \ls{\tmthree}}}) \ls{\tmtwo}
                        \ar[d]_{\lab}
                        \ar[r]^{\lab'} &
     (\lamp{\lab}{\var}{\conof{\subs{\tmfour}{\vartwo}{\ls{\tmthree}}}}) \ls{\tmtwo}
                        \ar@{.>}[d]_{\lab} \\
     \subs{\conof{(\lamp{\lab'}{\vartwo}{\tmfour}) \ls{\tmthree}}}{\var}{\ls{\tmtwo}}
                        \ar@{=}[d] &
     \subs{\conof{\subs{\tmfour}{\vartwo}{\ls{\tmthree}}}}{\var}{\ls{\tmtwo}}
                        \ar@{=}[d] \\
     \subs{\con}{\var}{\ls{\tmtwo}_1}
       \of{(\lamp{\lab'}{\vartwo}{\subs{\tmfour}{\var}{\ls{\tmtwo}_2}})
           \subs{\ls{\tmthree}}{\var}{\ls{\tmtwo}_3}}
                        \ar@{.>}[d]_-{\lab'} &
     \subs{\con}{\var}{\ls{\tmtwo}_1}
       \of{\subs{\subs{\tmfour}{\vartwo}{\ls{\tmthree}}}{\var}{\ls{\tmtwo}_4}} \\
     \subs{\con}{\var}{\ls{\tmtwo}_1}
       \of{\subs{\subs{\tmfour}{\var}{\ls{\tmtwo}_2}}{\vartwo}
                 {\subs{\ls{\tmthree}}{\var}{\ls{\tmtwo}_3}}}
                         &  \\
    }
    \]
    where $(\ls{\tmtwo}_1,\ls{\tmtwo}_2,\ls{\tmtwo}_3)$ and $(\ls{\tmtwo}_1,\ls{\tmtwo}_4)$ are partitions of $\ls{\tmtwo}$.
    Note that $(\ls{\tmtwo}_2, \ls{\tmtwo}_3)$ is a partition of $\ls{\tmtwo}_4$,
    so it suffices to show that
    $\subs{\subs{\tmfour}{\vartwo}{\ls{\tmthree}}}{\var}{\ls{\tmtwo}_4}
      =
      \subs{\subs{\tmfour}{\var}{\ls{\tmtwo}_2}}{\vartwo}
        {\subs{\ls{\tmthree}}{\var}{\ls{\tmtwo}_3}}
    $, which is an immediate consequence of the Substitution Lemma (\rlem{substitution_lemma}).
  \item {\bf The pattern of $\redextwo$ is inside $\ls{\tmtwo}$.}
    \label{strong_permutation__case_redextwo_inside_tmtwo}
    In this case, $\ls{\tmtwo} = [\ls{\tmtwo}_1, \conof{(\lamp{\lab'}{\vartwo}{\tmfour}) \ls{\tmthree}}, \ls{\tmtwo}_2]$.
    \[
    \xymatrix@C=.5cm@R=.5cm{
     (\lamp{\lab}{\var}{\tm}) [\ls{\tmtwo}_1,
                               \conof{(\lamp{\lab'}{\vartwo}{\tmfour}) \ls{\tmthree}},
                               \ls{\tmtwo}_2]
                        \ar[d]_-{\lab}
                        \ar[r]^{\lab'} &
     (\lamp{\lab}{\var}{\tm}) [\ls{\tmtwo}_1,
                               \subs{\tmfour}{\vartwo}{\ls{\tmthree}},
                               \ls{\tmtwo}_2]
                        \ar@{.>}[d]_-{\lab} \\
     \subs{\tm}{\var}{[\ls{\tmtwo}_1,
                       \conof{(\lamp{\lab'}{\vartwo}{\tmfour}) \ls{\tmthree}},
                       \ls{\tmtwo}_2]}
                       \ar@{.>}[r]^-{\lab'}
                       &
      \subs{\tm}{\var}{[\ls{\tmtwo}_1,
                        \subs{\tmfour}{\vartwo}{\ls{\tmthree}},\ls{\tmtwo}_2]}        \\
    }
    \]
    The arrow of the bottom of the diagram exists as a consequence of \rlem{reduce_in_substitution}.
  \end{enumerate}
\item {\bf Under an abstraction, $\con = \lamp{\lab''}{\vartwo}{\con'}$.}
  Straightforward by \ih.
\item {\bf Left of an application, $\con = \con'\ls{\tmthree}$.}
    There are three subcases, depending on whether the redex $\redextwo$ is
    at the root, to the left of the application, or to the right of the application.
    \begin{enumerate}
    \item {\bf The pattern of $\redextwo$ is at the root}.
      Then $\contwoof{(\lamp{\lab'}{\var}{\tm})\ls{\tmtwo}}$ must have a lambda at the root,
      so it is of the form $\lamp{\lab'}{\vartwo}{\con'''\of{(\lamp{\lab}{\var}{\tm}) \ls{\tmtwo}}}$.
      Hence, the starting term is
      $(\lamp{\lab'}{\vartwo}{\con'''\of{(\lamp{\lab}{\var}{\tm}) \ls{\tmtwo}}}) \ls{\tmthree}$.
      The symmetric case has already been studied in item~\refcase{strong_permutation__case_redextwo_inside_tm}.
    \item {\bf The pattern of $\redextwo$ is inside $\con'$.} Straightforward by \ih.
    \item {\bf The pattern of $\redextwo$ is inside $\ls{\tmthree}$.}
      It is immediate to close the diagram since the steps are at disjoint positions:
        \[
        \xymatrix@C=.5cm@R=.5cm{
         \contwoof{(\lamp{\lab}{\var}{\tm}) \ls{\tmtwo}}
           [\ls{\tmthree}_1, \con''\of{(\lamp{\lab'}{\vartwo}{\tmfour}) \ls{\tmfive}}, \ls{\tmthree}_2]
                            \ar[d]_-{\lab}
                            \ar[r]^{\lab'} &
         \contwoof{(\lamp{\lab}{\var}{\tm})\ls{\tmtwo}}
             [\ls{\tmthree}_1, \con''\of{\subs{\tmfour}{\vartwo}{\ls{\tmfive}}}, \ls{\tmthree}_2]
                           \ar@{.>}[d]_-{\lab} \\
         \contwoof{\subs{\tm}{\var}{\ls{\tmtwo}}}
           [\ls{\tmthree}_1, \con''\of{(\lamp{\lab'}{\vartwo}{\tmfour}) \ls{\tmfive}}, \ls{\tmthree}_2]
                            \ar@{.>}[r]^{\lab'}
                           &
         \contwoof{\subs{\tm}{\var}{\ls{\tmtwo}}}
             [\ls{\tmthree}_1, \con''\of{\subs{\tmfour}{\vartwo}{\ls{\tmfive}}}, \ls{\tmthree}_2]
        }
        \]
    \end{enumerate}
\item {\bf Right of an application,
  $\con = \tmfour[\tmthree_1, \hdots, \tmthree_{i-1}, \con',\tmthree_{i+1}, \hdots, \tmthree_n]$}.
  There are four subcases, depending on whether the redex $\redextwo$ is
  at the root, to the left of the application (that is, inside $\tmfour$),
  or to the right of the application (that is, either inside $\contwo$ or $\tmthree_j$ for some $j$).
  \begin{enumerate}
  \item {\bf The pattern of $\redextwo$ is at the root.}
    Then $\tmfour$ has a lambda at the root, \ie it is of the form $\lamp{\lab'}{\vartwo}{\tmthree}$.
    Hence the starting term is
    $
      (\lamp{\lab'}{\vartwo}{\tmthree}) [\tmthree_{1:i-1}, \contwoof{(\lamp{\lab}{\var}{\tm})\tmtwo},\tmthree_{i+1:n}]
    $.
    The symmetric case has already been studied in item~\refcase{strong_permutation__case_redextwo_inside_tmtwo}.
  \item {\bf The pattern of $\redextwo$ is inside $\tmfour$.}
    The steps are disjoint, so it is immediate.
  \item {\bf The pattern of $\redextwo$ is inside $\tmthree_j$ for some $j \neq i$.}
    The steps are disjoint, so it is immediate.
  \item {\bf The pattern of $\redextwo$ is inside $\contwo$.}
    Straightforward by \ih.
  \end{enumerate}
\end{enumerate}

\subsection{Full Stability}
\label{appendix_full_stability}

In this section, we state and prove a strong version of Stability
in the sense of L\'evy~\cite{levy_redex_stability},
called Full Stability (\rlem{full_stability}).
This lemma is crucial in establishing that 
$\ulbDerivDist{\tm}$ forms a distributive lattice (\rprop{lambdadist_lattice}).
First we need a few auxiliary results:

\begin{definition}[Alternative Substitution]
An alternative definition for capture-avoiding substitution,
written $\tm\sub{\var}{\ls{\tmtwo}}$,
may be defined as follows,
provided that
$\varlabel{\var}{\tm} \subseteq \tmlabel{\ls{\tmtwo}}$
and the multiset of types of $\ls{\tmtwo} = [\tmtwo_1,\hdots,\tmtwo_m]$
is sequential:
\[
  \begin{array}{rcll}
    \var^{\typ}\sub{\var}{\ls{\tmtwo}}                 & \eqdef & \tmtwo_i & \text{ where $i$ is the unique index such that} \\
                                                                         &&& \text{ $\tmlabel{\tmtwo_i}$ is the external label of $\typ$}\\
    \vartwo^{\typ}\sub{\var}{\ls{\tmtwo}}              & \eqdef & \vartwo^\typ & \text{ if $\var \neq \vartwo$} \\
    (\lamp{\lab}{\vartwo}{\tm})\sub{\var}{\ls{\tmtwo}} & \eqdef & \lamp{\lab}{\vartwo}{\tm\sub{\var}{\ls{\tmtwo}}} & \text{ if $\var \neq \vartwo$ and $\vartwo \not\in \fv{\ls{\tmtwo}}$}\\
    (\tm[\tmthree_i]_{i=1}^{k})\sub{\var}{\ls{\tmtwo}} & \eqdef & \tm\sub{\var}{\ls{\tmtwo}}[\tmthree_i\sub{\var}{\ls{\tmtwo}}]_{i=1}^{k}
  \end{array}
\]
Moreover,
$[\tm_1,\hdots,\tm_n]\sub{\var}{\ls{\tmtwo}}$
stands for $[\tm_1\sub{\var}{\ls{\tmtwo}},\hdots,\tm_n\sub{\var}{\ls{\tmtwo}}]$,
whenever each substitution $\tm_i\sub{\var}{\ls{\tmtwo}}$ is well-defined.
\end{definition}

\begin{lemma}
Let $\tm,\ls{\tmtwo},\ls{\tmthree}$ be correct terms such that
$\varlabel{\var}{\tm} = \tmlabel{\ls{\tmtwo}}$ and
suppose that $\ls{\tmtwo}$ is contained in $\ls{\tmthree}$,
\ie there exists a list $\ls{\tmfour}$ such that
$(\ls{\tmtwo},\ls{\tmfour})$ is a partition of $\ls{\tmthree}$.
Then $\subs{\tm}{\var}{\ls{\tmtwo}} = \tm\sub{\var}{\ls{\tmthree}}$.
\end{lemma}
\begin{proof}
Straightforward by induction on $\tm$.
\end{proof}

\begin{lemma}[Creation]
\llem{creation}
If $\redex$ creates $\redextwo$ then it is according to
one of the following cases:
\begin{enumerate}
\item \CreationCase{I}:
  $
    \con\of{ (\lamp{\lab}{\var}{\var^{\typ}})\,[\lamp{\labtwo}{\vartwo}{\tm}]\, \ls{\tmtwo}}
    \ \todist\ %
    \con\of{ (\lamp{\labtwo}{\vartwo}{\tm})\,\ls{\tmtwo} }
    \ \todist\ %
    \con\of{ \tm\sub{\vartwo}{\ls{\tmtwo}} }
  $.
\item \CreationCase{II}:
  $
     \con\of{ (\lamp{\lab}{\var}{ \lamp{\labtwo}{\vartwo}{\tm} })\,\ls{\tmtwo}\,\ls{\tmthree} }
     \ \todist\ %
     \con\of{ (\lamp{\labtwo}{\vartwo}{\tm'})\,\ls{\tmthree} }
     \ \todist\ %
     \con\of{ \tm'\sub{\vartwo}{\ls{\tmthree}} }
  $, where
  $
     \tm' = \tm\sub{\var}{\ls{\tmtwo}}
  $.
\item \CreationCase{III}:
  $
    \con_1\of{(\lamp{\lab}{\var}{\con_2\of{\var^{\typ}\,\ls{\tm}}})\ls{\tmtwo}}
    \ \todist\ %
    \con_1\of{\con'_2\of{(\lamp{\labtwo}{\vartwo}{\tmthree})\,\ls{\tm'}}}
    \ \todist\ %
    \con_1\of{\con'_2\of{\tmthree\sub{\vartwo}{\ls{\tm'}}}}
  $,
  where \[
  \begin{array}{rcl}
    \con_2\sub{\var}{\ls{\tmtwo}}      & = & \con'_2 \\
    \var^{\typ}\sub{\var}{\ls{\tmtwo}} & = & \lamp{\labtwo}{\vartwo}{\tmthree} \\
    \ls{\tm}\sub{\var}{\ls{\tmtwo}}    & = & \ls{\tm'} \\
    \ls{\tmtwo}                        & = & [\ls{\tmtwo_1},\lamp{\labtwo}{\vartwo}{\tmthree},\ls{\tmtwo_2}] \\
  \end{array}
  \]
\end{enumerate}
\end{lemma}
\begin{proof}
Let $\redex : \conof{(\lamp{\lab}{\var}{\tm})\ls{\tmtwo}} \todist \conof{\tm\sub{\var}{\ls{\tmtwo}}}$ be a step,
and let $\redextwo : \conof{\tm\sub{\var}{\ls{\tmtwo}}} \todist \tmfive$ another step
such that $\redex$ creates $\redextwo$.
The redex contracted by the step $\redextwo$ is below a context $\con_1$,
so let $\conof{\tm\sub{\var}{\ls{\tmtwo}}} = \con_1\of{(\lamp{\labtwo}{\vartwo}{\tmthree})\ls{\tmfour}}$,
where $(\lamp{\labtwo}{\vartwo}{\tmthree})\ls{\tmfour}$ is the redex contracted by $\redextwo$.
We consider three cases, depending on the relative positions of the holes of $\con$ and $\con_1$,
namely they may be disjoint,
$\con$ may be a prefix of $\con_1$,
or $\con_1$ may be a prefix of $\con$:
\begin{enumerate}
\item {\bf If the holes of $\con$ and $\con_1$ are disjoint.}
  Then there is a two-hole context $\conhat$ such that
  $
    \con = \conhat\of{\conbase,(\lamp{\labtwo}{\vartwo}{\tmthree})\ls{\tmfour}}
  $ and $
    \con_1 = \conhat\of{\tm\sub{\var}{\ls{\tmtwo}},\conbase}
  $.
  So
  $
    \conhat\of{(\lamp{\lab}{\var}{\tm})\ls{\tmtwo},(\lamp{\labtwo}{\vartwo}{\tmthree})\ls{\tmfour}}
    \todist \conhat\of{\tm\sub{\var}{\ls{\tmtwo}},(\lamp{\labtwo}{\vartwo}{\tmthree})\ls{\tmfour}}
    \todist \conhat\of{\tm\sub{\var}{\ls{\tmtwo}},\tmthree\sub{\vartwo}{\tmfour}}
  $.
  Observe that $\redextwo$ has an ancestor $\redextwo_0$, contradicting the hypothesis that $\redex$ creates $\redextwo$.
  So this case is impossible.
\item {\bf If $\con$ is a prefix of $\con_1$.}
  Then there exists a context $\con_2$ such that $\con_1 = \conof{\con_2}$,
  and we have that
  $\tm\sub{\var}{\ls{\tmtwo}} = \con_2\of{(\lamp{\labtwo}{\vartwo}{\tmthree})\ls{\tmfour}}$.
  We consider two subcases, depending on whether the position of the hole $\con_2$ lies
  inside the term $\tm$, or it reaches a free occurrence of $\var$ in $\tm$
  and ``jumps'' to one of the arguments in the list $\ls{\tmtwo}$.
  \begin{enumerate}
  \item {\bf If the position of the hole of $\con_2$ lies in $\tm$.}
    More precisely, there is a context $\con'_2$ and a term $\tm'$ such that:
    $\tm = \con'_2\of{\tm'}$,
    $\con_2 = \con'_2\sub{\var}{\ls{\tmtwo}}$,
    and
    $(\lamp{\labtwo}{\vartwo}{\tmthree})\ls{\tmfour} = \tm'\sub{\var}{\ls{\tmtwo}}$.
    We consider three further subcases for the term $\tm'$.
    It cannot be an abstraction, so it may be a variable $\var^{\typ}$,
    or an application.
    Besides, if it is an application, the head may be a variable $\var^{\typ}$ or an abstraction.
    \begin{enumerate}
    \item {\bf Variable, \ie $\tm' = \var^{\typ}$.}
      Then the list $\ls{\tmtwo}$ may be split as $\ls{\tmtwo} = [\ls{\tmtwo_1},(\lamp{\labtwo}{\vartwo}{\tmthree})\ls{\tmfour},\ls{\tmtwo_2}]$,
      and:
      $
        \conof{(\lamp{\lab}{\var}{\con'_2\of{\var^{\typ}}})[\ls{\tmtwo_1},(\lamp{\labtwo}{\vartwo}{\tmthree})\ls{\tmfour},\ls{\tmtwo_2}]}
        \ \todist\ %
        \conof{\con_2\of{ (\lamp{\labtwo}{\vartwo}{\tmthree})\ls{\tmfour} }}
        \ \todist\ %
        \conof{\con_2\of{ \tmthree\sub{\vartwo}{\ls{\tmfour}} }}
      $.
      Observe that $\redextwo$ has an ancestor $\redextwo_0$, contradicting the hypothesis that $\redex$ creates $\redextwo$.
      So this case is impossible.
    \item {\bf Application of a variable, \ie $\tm' = \var^{\typ}\,\ls{\tmfour'}$.}
      Given that
      $(\lamp{\labtwo}{\vartwo}{\tmthree})\ls{\tmfour} = \tm'\sub{\var}{\ls{\tmtwo}}$,
      we have $\ls{\tmfour} = \ls{\tmfour'}\sub{\var}{\ls{\tmtwo}}$,
      and the list $\ls{\tmtwo}$ may be split as $\ls{\tmtwo} = [\ls{\tmtwo_1},\lamp{\labtwo}{\vartwo}{\tmthree},\ls{\tmtwo_2}]$.
      Then
      $
        \conof{(\lamp{\lab}{\var}{\con_2'\of{\var^{\typ}\,\ls{\tmfour'}}})[\ls{\tmtwo_1},\lamp{\labtwo}{\vartwo}{\tmthree},\ls{\tmtwo_2}]}
        \ \todist\ %
        \conof{\con_2\of{(\lamp{\labtwo}{\vartwo}{\tmthree})\,\ls{\tmfour}}}
        \ \todist\ %
        \conof{\con_2\of{\tmthree\sub{\vartwo}{\ls{\tmfour}}}}
      $,
      and we are in the situation of \CreationCase{III}.
    \item {\bf Application of an abstraction, \ie $\tm' = (\lamp{\lab}{\vartwo}{\tmthree'})\,\ls{\tmfour'}$.}
      Given that
      $(\lamp{\labtwo}{\vartwo}{\tmthree})\ls{\tmfour} = \tm'\sub{\var}{\ls{\tmtwo}}$,
      we have $\ls{\tmthree} = \ls{\tmthree'}\sub{\var}{\ls{\tmtwo}}$
      and $\ls{\tmfour} = \ls{\tmfour'}\sub{\var}{\ls{\tmtwo}}$.
      Then:
      $
        \conof{(\lamp{\lab}{\var}{\con_2'\of{ (\lamp{\lab}{\vartwo}{\tmthree'})\,\ls{\tmfour'} }})\ls{\tmtwo}}
        \ \todist\ %
        \conof{\con_2\of{ (\lamp{\lab}{\vartwo}{\tmthree})\,\ls{\tmfour} }}
        \ \todist\ %
        \conof{\con_2\of{ \tmthree\sub{\varthree}{\ls{\tmfour}} }}
      $.
      Then $\redextwo$ has an ancestor $\redextwo_0$, contradicting the hypothesis that $\redex$ creates $\redextwo$.
      So this case is impossible.
    \end{enumerate}
  \item {\bf If the position of the hole of $\con_2$ ``jumps'' through a free occurrence of $\var$.}
    More precisely,
    there exist $\con_3$, $\typ$, $\ls{\tmtwo_1}$, $\ls{\tmtwo_2}$, and $\con_4$
    such that:
    $
      \tm = \con_3\of{\var^{\typ}}
      \HS
      \ls{\tmtwo} = [\ls{\tmtwo_1},\con_4\of{ (\lamp{\labtwo}{\vartwo}{\tmthree})\ls{\tmfour} },\ls{\tmtwo_2}]
    $
    in such a way that $\tmlabel{\con_4\of{ (\lamp{\labtwo}{\vartwo}{\tmthree})\ls{\tmfour} }} = \varlabel{\var}{\var^{\typ}}$,
    and $\con_2 = \con_3\sub{\var}{\ls{\tmtwo}}\of{\con_4}$.
    Then:
    \[
      \begin{array}{rl}
              & \conof{(\lamp{\lab}{\var}{\con_3\of{\var^{\typ}}})[\ls{\tmtwo_1},\con_4\of{ (\lamp{\labtwo}{\vartwo}{\tmthree})\ls{\tmfour} },\ls{\tmtwo_2}]} \\
      \todist & \conof{\con_3\sub{\var}{\ls{\tmtwo}}\of{\con_4\of{(\lamp{\labtwo}{\vartwo}{\tmthree})\ls{\tmfour}}}} \\
      \todist & \conof{\con_3\sub{\var}{\ls{\tmtwo}}\of{\con_4\of{ \tmthree\sub{\vartwo}{\ls{\tmfour}} }}}
      \end{array}
    \]
    Then $\redextwo$ has an ancestor $\redextwo_0$, contradicting the hypothesis that $\redex$ creates $\redextwo$.
    So this case is impossible.
  \end{enumerate}
\item {\bf If $\con_1$ is a prefix of $\con$.}
  Then there exists a context $\con_2$ such that $\con = \con_1\of{\con_2}$.
  This means that $\con_2\of{\tm\sub{\var}{\ls{\tmtwo}}} = (\lamp{\labtwo}{\vartwo}{\tmthree})\ls{\tmfour}$.
  We consider three subcases, depending on whether $\con_2$ is empty,
  or the hole of $\con_2$ lies to the left or to the right of the application.
  \begin{enumerate}
  \item {\bf Empty, $\con_2 = \conbase$.}
    Then $\con = \con_1$ so $\con$ is a prefix of $\con_1$. We have already considered this case.
  \item {\bf Left, $\con_2 = \con'_2\,\ls{\tmfour}$.}
    Then the step $\redex$ is of the form
    $
      \redex :
      \con_1\of{\con'_2\of{(\lamp{\lab}{\var}{\tm})\ls{\tmtwo}}\,\ls{\tmfour}}
      \todist
      \con_1\of{\con'_2\of{\tm\sub{\var}{\ls{\tmtwo}}}\,\ls{\tmfour}}
    $.
    In particular we know that $\con'_2\of{\tm\sub{\var}{\ls{\tmtwo}}} = \lamp{\labtwo}{\vartwo}{\tmthree}$.
    There are two cases, depending on whether $\con'_2$ is empty or non-empty.
    \begin{enumerate}
    \item {\bf Empty, $\con'_2 = \conbase$}
      Then we have that $\tm\sub{\var}{\ls{\tmtwo}} = \lamp{\labtwo}{\vartwo}{\tmthree}$.
      We consider two further subcases, depending on whether $\tm$ is an occurrence of the variable $\var$.
      \begin{enumerate}
      \item {\bf If $\tm = \var^{\typ}$ for some type $\typ$.}
        Then the list $\ls{\tmtwo}$ must be of the form $[\lamp{\labtwo}{\vartwo}{\tmthree}]$
        where the external label of $\typ$ is precisely $\labtwo$.
        Then
        $
          \con_1\of{(\lamp{\lab}{\var}{\var^{\typ}})[\lamp{\labtwo}{\vartwo}{\tmthree}]\,\ls{\tmfour}}
          \todist
          \con_1\of{(\lamp{\labtwo}{\vartwo}{\tmthree})\,\ls{\tmfour}}
          \todist
          \con_1\of{\tmthree\sub{\vartwo}{\ls{\tmfour}}}
        $,
        and we are in the situation of \CreationCase{I}.
      \item {\bf If $\tm \neq \var^{\typ}$ for any type $\typ$.}
        Then $\tm$ must be an abstraction, namely $\tm = \lamp{\labtwo}{\vartwo}{\tmthree'}$
        where $\tmthree'\sub{\var}{\ls{\tmtwo}} = \tmthree$. 
        Then
        $
          \con_1\of{(\lamp{\lab}{\var}{\lamp{\labtwo}{\vartwo}{\tmthree'}})\,\ls{\tmtwo}\,\ls{\tmfour}}
          \todist
          \con_1\of{(\lamp{\labtwo}{\vartwo}{\tmthree'\sub{\var}{\ls{\tmtwo}}})\,\ls{\tmfour}}
          \todist
          \con_1\of{\tmthree'\sub{\var}{\ls{\tmtwo}}\sub{\vartwo}{\ls{\tmfour}}}
        $,
        and we are in the situation of \CreationCase{II}.
      \end{enumerate}
    \item {\bf Non-empty, $\con'_2 \neq \conbase$}
      Then necessarily $\con'_2$ must be an abstraction, namely $\con'_2 = \lamp{\labtwo}{\vartwo}{\con''_2}$.
      Then:
      \[
        \begin{array}{rcl}
        \con_1\of{(\lamp{\labtwo}{\vartwo}{\con''_2\of{ (\lamp{\lab}{\var}{\tm})\,\ls{\tmtwo} }})\,\ls{\tmfour}}
        & \todist &
        \con_1\of{(\lamp{\labtwo}{\vartwo}{\con''_2\of{ \tm\sub{\var}{\ls{\tmtwo}} }})\,\ls{\tmfour}}
        \\
        & \todist &
        \con_1\of{\con''_2\of{ \tm\sub{\var}{\ls{\tmtwo}} }\sub{\vartwo}{\ls{\tmfour}}}
        \end{array}
      \]
      Then $\redextwo$ has an ancestor $\redextwo_0$, contradicting the hypothesis that $\redex$ creates $\redextwo$.
      So this case is impossible.
    \end{enumerate}
  \item {\bf Right, $\con_2 = (\lamp{\labtwo}{\vartwo}{\tmthree})[\ls{\tmfour}_1,\con'_2,\ls{\tmfour}_2]$.}
    Then:
    \[
      \begin{array}{rcl}
      \con_1\of{(\lamp{\labtwo}{\vartwo}{\tmthree})\,[\ls{\tmfour}_1,\con'_2\of{(\lamp{\lab}{\var}{\tm})\ls{\tmtwo}},\ls{\tmfour}_2]}
      & \todist &
      \con_1\of{(\lamp{\labtwo}{\vartwo}{\tmthree})\,[\ls{\tmfour}_1,\con'_2\of{\tm\sub{\var}{\ls{\tmtwo}}},\ls{\tmfour}_2]}
      \\
      & \todist &
      \con_1\of{\tmthree\sub{\vartwo}{[\ls{\tmfour}_1,\con'_2\of{\tm\sub{\var}{\ls{\tmtwo}}},\ls{\tmfour}_2]}}
      \end{array}
    \]
    Then $\redextwo$ has an ancestor $\redextwo_0$, contradicting the hypothesis that $\redex$ creates $\redextwo$.
    So this case is impossible.
  \end{enumerate}
\end{enumerate}
\end{proof}

\begin{lemma}[Basic Stability]
\llem{basic_stability}
Let $\redex \neq \redextwo$.
If $\redexthree_3 \in \redexthree_1/(\redextwo/\redex)$
and $\redexthree_3 \in \redexthree_2/(\redex/\redextwo)$
then there exists a step $\redexthree_0$ such that
$\redexthree_1 \in \redexthree_0/\redex$
and
$\redexthree_2 \in \redexthree_0/\redextwo$.
\end{lemma}
\begin{proof}
We prove an equivalent statement.
Namely,
we prove that if $\redex,\redextwo$ are different coinitial steps such that
$\redex$ creates a step $\redexthree$,
then $\redex/\redextwo$ creates the step $\redexthree/(\redextwo/\redex)$.
It is easy to see that these are equivalent since in $\lambdadist$ there is
no erasure or duplication.

Let $\redex : \conof{(\lamp{\lab}{\var}{\tm})\ls{\tmtwo}} \todist \conof{\subs{\tm}{\var}{\ls{\tmtwo}}}$,
let $\redextwo \neq \redex$ be a step coinitial to $\redex$,
and suppose that $\redex$ creates a step $\redexthree$.
By induction on the context $\con$ we argue that $\redex/\redextwo$ creates $\redexthree/(\redextwo/\redex)$.
\begin{enumerate}
\item {\bf Empty context, $\con = \conbase$.}
  Then $\redex : (\lamp{\lab}{\var}{\tm})\ls{\tmtwo} \todist \tm\sub{\var}{\ls{\tmtwo}}$.
  There are two cases for $\redextwo$, depending on whether it is internal to $\tm$
  or internal to one of the arguments of the list $\ls{\tmtwo}$.
  \begin{enumerate}
  \item {\bf If $\redextwo$ is internal to $\tm$.}
    Then $\tm = \con_1\of{\Sigma}$ where $\Sigma$ is the redex contracted by $\redextwo$.
    Let $\Sigma'$ denote the contractum of $\Sigma$. Then:
    \[
      \xymatrix@R=.5cm{
        (\lamp{\lab}{\var}{\con_1\of{\Sigma}})\,\ls{\tmtwo}
        \ar[r]^-{\redex}
        \ar[d]_-{\redextwo}
      &
        \con_1\of{\Sigma}\sub{\var}{\ls{\tmtwo}}
        \ar[d]_-{\redextwo/\redex}
        \ar[r]^-{\redexthree}
      &
      \\
        (\lamp{\lab}{\var}{\con_1\of{\Sigma'}})\,\ls{\tmtwo}
        \ar[r]^-{\redex/\redextwo}
      &
        \con_1\of{\Sigma'}\sub{\var}{\ls{\tmtwo}}
        \ar[r]^-{\redexthree/(\redextwo/\redex)}
      &
      }
    \]
    By Creation~(\rlem{creation}),
    $\redexthree$ must be created by \CreationCase{III},
    since there are no applications surrounding the subterm contracted by $\redex$.
    This means that $\con_1\of{\Sigma} = \con_2\of{\var^{\typ}\,\ls{\tmthree}}$
    where, moreover, $\ls{\tmtwo}$ may be split as $[\ls{\tmtwo}_1,\lamp{\labtwo}{\vartwo}{\tmfour},\ls{\tmtwo}_2]$
    in such a way that $\labtwo$ is the external label of $\typ$.
    We consider three subcases, depending on whether the contexts
    $\con_1$ and $\con_2$ are disjoint,
    $\con_1$ is a prefix of $\con_2$,
    or $\con_2$ is a prefix of $\con_1$.
    \begin{enumerate}
    \item {\bf If $\con_1$ and $\con_2$ are disjoint.}
      Then there is a two-hole context $\conhat$ such that
      $
        \conhat\of{\conbase,\var^{\typ}\ls{\tmthree}} = \con_1
      $ and $
        \conhat\of{\Sigma,\conbase} = \con_2
      $.
      Given any term, context, or list of terms $\anon$ let us write $\anon^*$ to denote $\anon\sub{\var}{\ls{\tmtwo}}$.
      Then $\redex/\redextwo$ creates $\redexthree/(\redextwo/\redex)$:
      \[
        \xymatrix@R=.5cm{
          (\lamp{\lab}{\var}{\conhat\of{\Sigma,\var^{\typ}\ls{\tmthree}}})\,\ls{\tmtwo}
          \ar[r]^-{\redex}
          \ar[d]_-{\redextwo}
        &
          \conhat^*\of{\Sigma^*,(\lamp{\labtwo}{\vartwo}{\tmfour})\ls{\tmthree}^*}
          \ar[d]_-{\redextwo/\redex}
          \ar[r]^-{\redexthree}
        &
          \conhat^*\of{\Sigma^*,\tmfour\sub{\vartwo}{\ls{\tmthree}}}
        \\
          (\lamp{\lab}{\var}{\conhat\of{\Sigma',\var^{\typ}\ls{\tmthree}}})\,\ls{\tmtwo}
          \ar[r]^-{\redex/\redextwo}
        &
          \conhat^*\of{\Sigma'^*,(\lamp{\labtwo}{\vartwo}{\tmfour})\ls{\tmthree}^*}
          \ar[r]^-{\redexthree/(\redextwo/\redex)}
        &
          \conhat^*\of{\Sigma'^*,\tmfour\sub{\vartwo}{\ls{\tmthree}^*}}
        }
      \]
    \item {\bf If $\con_1$ is a prefix of $\con_2$.}
      Then $\con_2 = \con_1\of{\con'_2}$ which means that $\Sigma = \con'_2\of{\var^{\typ}\ls{\tmthree}}$.
      Recall that $\Sigma$ is a redex, so let us write
      $\Sigma = (\lamp{\labthree}{\varthree}{\tmfive})\ls{\tmsix}$.
      We consider two further subcases, depending on whether the hole of $\con'_2$
      lies to the left or to the right of the application (observe that it cannot be at the root
      since $\lamp{\labthree}{\varthree}{\tmfive}$ is not a variable).
      \begin{enumerate}
      \item {\bf If the hole of $\con'_2$ lies to the left of $(\lamp{\labthree}{\varthree}{\tmfive})\ls{\tmsix}$.}
        More precisely, we have that
        $\con'_2 = (\lamp{\labthree}{\varthree}{\con''_2})\ls{\tmsix}$ and $\tmfive = \con''_2\of{\var^{\typ}\ls{\tmthree}}$.
        Given any term, context, or list of terms $\anon$ let us write $\anon^*$ to denote $\anon\sub{\var}{\ls{\tmtwo}}$.
        Then $\redex/\redextwo$ creates $\redexthree/(\redextwo/\redex)$:
        \[
        {\footnotesize
          \xymatrix@C=.5cm{
            (\lamp{\lab}{\var}{\con_1\of{  (\lamp{\labthree}{\varthree}{\con''_2\of{\var^{\typ}\ls{\tmthree}} })\ls{\tmsix} }})\,\ls{\tmtwo}
            \ar[r]^-{\redex}
            \ar[d]_-{\redextwo}
          &
            \con_1^*\of{  (\lamp{\labthree}{\varthree}{{\con''_2}^*\of{(\lamp{\labtwo}{\vartwo}{\tmfour})\ls{\tmthree}^*} })\ls{\tmsix}^* }
            \ar[d]_-{\redextwo/\redex}
            \ar[r]^-{\redexthree}
          &
            \con_1^*\of{  (\lamp{\labthree}{\varthree}{{\con''_2}^*\of{ \tmfour\sub{\vartwo}{\ls{\tmthree}^*} } })\ls{\tmsix}^* }
          \\
            (\lamp{\lab}{\var}{\con_1\of{  \con''_2\of{\var^{\typ}\ls{\tmthree}} \sub{\varthree}{\ls{\tmsix}} }})\,\ls{\tmtwo}
            \ar[r]^-{\redex/\redextwo}
          &
            %\con_1\of{\Sigma'}\sub{\var}{\ls{\tmtwo}}
            \con_1^*\of{  {\con''_2}^*\of{(\lamp{\labtwo}{\vartwo}{\tmfour})\ls{\tmthree}^*} \sub{\varthree}{\ls{\tmsix}^*} }
            \ar[r]^-{\redexthree/(\redextwo/\redex)}
          &
            \con_1^*\of{  {\con''_2}^*\of{ \tmfour\sub{\vartwo}{\ls{\tmthree}^*} } \sub{\varthree}{\ls{\tmsix}^*} }
          }
          }
        \]
      \item {\bf If the hole of $\con'_2$ lies to the right of $(\lamp{\labthree}{\varthree}{\tmfive})\ls{\tmsix}$.}
        More precisely, we have that $\con'_2 = (\lamp{\labthree}{\varthree}{\tmfive})[\ls{\tmsix}_1,\con''_2,\ls{\tmsix}_2]$
        and $\ls{\tmsix} = [\ls{\tmsix}_1,\con''_2\of{\var^\typ\,\ls{\tmthree}},\ls{\tmsix}_2]$.
        Given any term, context, or list of terms $\anon$ let us write $\anon^*$ to denote $\anon\sub{\var}{\ls{\tmtwo}}$.
        Moreover, since parameters are in 1--1 correspondence with arguments,
        there is exactly one free occurrence of $\varthree$ in $\tmfive$ whose type coincides with
        the type of $\con''_2\of{\var^\typ\,\ls{\tmthree}}$.
        Let $\tmfive = \con_3\of{\varthree^{\typtwo}}$ where the hole of $\con_3$ marks the position
        of such occurrence.
        Then:
        \[
          \con_3\of{\varthree^{\typtwo}}\sub{\varthree}{\ls{\tmsix}^*}
          = \con_3\sub{\varthree}{\ls{\tmsix}^*}\of{(\con''_2\of{\var^\typ\,\ls{\tmthree}})^*}
          = \con_3\sub{\varthree}{\ls{\tmsix}^*}\of{{\con''_2}^*\of{(\lamp{\labtwo}{\vartwo}{\tmfour})\,\ls{\tmthree}^*})}
        \]
        Moreover, if we write $\con'_3$ for the context $\con_3\sub{\varthree}{\ls{\tmsix}}$:
        \[
          (\con'_3)^*
          = (\con_3\sub{\varthree}{\ls{\tmsix}})^*
          = \con_3^*\sub{\varthree}{\ls{\tmsix}^*}
        \]
        by the Substitution Lemma for contexts~(\rlem{substitution_lemma_alt}).
        Then:
        \[
          \xymatrix@R=.5cm{
            (\lamp{\lab}{\var}{\con_1\of{ (\lamp{\labthree}{\varthree}{\con_3\of{\varthree^\typtwo}})\ls{\tmsix} }})\,\ls{\tmtwo}
            \ar[r]^-{\redex}
            \ar[d]_-{\redextwo}
          &
            \con_1^*\of{ (\lamp{\labthree}{\varthree}{\con_3^*\of{\varthree^\typtwo}})\ls{\tmsix}^* }
            \ar[d]_-{\redextwo/\redex}
            \ar[r]^-{\redexthree}
          &
          \\
            (\lamp{\lab}{\var}{\con_1\of{ \con'_3\of{ \con''_2\of{\var^\typ\,\ls{\tmthree}} } }})\,\ls{\tmtwo}
            \ar[r]^-{\redex/\redextwo}
          &
            \con_1^*\of{ (\con'_3)^*\of{ {\con''_2}^*\of{(\lamp{\labtwo}{\vartwo}{\tmfour})\,\ls{\tmthree}^*} } }
            \ar[r]^-{\redexthree/(\redextwo/\redex)}
          &
          }
        \]
        The labels of the steps $\redexthree$ and $\redexthree/(\redextwo/\redex)$ are both $\labtwo$,
        which means that $\redex/\redextwo$ creates $\redexthree/(\redextwo/\redex)$.
      \end{enumerate}
    \item {\bf If $\con_2$ is a prefix of $\con_1$.}
      Then $\con_1 = \con_2\of{\con'_1}$
      which means $\var^{\typ}\,\ls{\tmthree} = \con'_1\of{\Sigma}$.
      Since $\Sigma$ is a redex, $\con'_1$ must be of the form $\var^\typ\,[\ls{\tmthree}_1,\con''_1,\ls{\tmthree}_2]$
      and $\ls{\tmthree} = [\ls{\tmthree}_1,\con''_1\of{\Sigma},\ls{\tmthree}_2]$.
      Given any term, context, or list of terms $\anon$ let us write $\anon^*$ to denote $\anon\sub{\var}{\ls{\tmtwo}}$.
      Then: $\redex/\redextwo$ creates $\redexthree/(\redextwo/\redex)$:
      {\small
      \[
        \xymatrix@R=.5cm{
          (\lamp{\lab}{\var}{\con_2\of{\var^\typ[\ls{\tmthree}_1,\con''_1\of{\Sigma},\ls{\tmthree}_2]}})\ls{\tmtwo}
          \ar[r]^-{\redex}
          \ar[d]_-{\redextwo}
        &
          \con_2^*\of{(\lamp{\labtwo}{\vartwo}{\tmfour})[\ls{\tmthree}_1^*,{\con''_1}^*\of{\Sigma^*},\ls{\tmthree}_2^*]}
          \ar[d]_-{\redextwo/\redex}
          \ar[r]^-{\redexthree}
        &
          \con_2^*\of{ \tmfour\sub{\vartwo}{[\ls{\tmthree}_1^*,{\con''_1}^*\of{\Sigma^*},\ls{\tmthree}_2^*]} }
        \\
          (\lamp{\lab}{\var}{\con_2\of{\var^\typ[\ls{\tmthree}_1,\con''_1\of{\Sigma'},\ls{\tmthree}_2]}})\ls{\tmtwo}
          \ar[r]^-{\redex/\redextwo}
        &
          \con_2^*\of{(\lamp{\labtwo}{\vartwo}{\tmfour})[\ls{\tmthree}_1^*,{\con''_1}^*\of{\Sigma'^*},\ls{\tmthree}_2^*]}
          \ar[r]^-{\redexthree/(\redextwo/\redex)}
        &
          \con_2^*\of{ \tmfour\sub{\vartwo}{[\ls{\tmthree}_1^*,{\con''_1}^*\of{\Sigma'^*},\ls{\tmthree}_2^*]} }
        }
      \]}
    \end{enumerate}
  \item {\bf If $\redextwo$ is internal to $\ls{\tmtwo}$.}
    Then $\ls{\tmtwo} = [\ls{\tmtwo}_1,\con_1\of{\Sigma},\ls{\tmtwo}_2]$
    where $\Sigma$ is the redex contracted by $\redextwo$.
    Let $\Sigma'$ denote the contractum of $\Sigma$. Then:
    \[
      \xymatrix@R=.5cm{
        (\lamp{\lab}{\var}{\tm})\,[\ls{\tmtwo}_1,\con_1\of{\Sigma},\ls{\tmtwo}_2]
        \ar[r]^-{\redex}
        \ar[d]_-{\redextwo}
      &
        \tm\sub{\var}{[\ls{\tmtwo}_1,\con_1\of{\Sigma},\ls{\tmtwo}_2]}
        \ar[d]_-{\redextwo/\redex}
        \ar[r]^-{\redexthree}
      &
      \\
        (\lamp{\lab}{\var}{\tm})\,[\ls{\tmtwo}_1,\con_1\of{\Sigma'},\ls{\tmtwo}_2]
        \ar[r]^-{\redex/\redextwo}
      &
        \tm\sub{\var}{[\ls{\tmtwo}_1,\con_1\of{\Sigma'},\ls{\tmtwo}_2]}
        \ar[r]^-{\redexthree/(\redextwo/\redex)}
      &
      }
    \]
    By Creation~(\rlem{creation}), $\redexthree$ must be created by
    \CreationCase{III},
    since there are no applications surrounding the subterm contracted by $\redex$.
    This means that $\tm = \con_2\of{\var^\typ\,\ls{\tmthree}}$,
    and there is a term in the list $[\ls{\tmtwo}_1,\con_1\of{\Sigma},\ls{\tmtwo}_2]$ 
    of the form $\lamp{\labtwo}{\vartwo}{\tmfour}$
    such that $\labtwo$ is also the external label of the type $\typ$.
    There are two subcases, depending on whether such term is $\con_1\of{\Sigma}$
    or a different term in the list $[\ls{\tmtwo}_1,\con_1\of{\Sigma},\ls{\tmtwo}_2]$.
    \begin{enumerate}
    \item {\bf If $\lamp{\labtwo}{\vartwo}{\tmfour} = \con_1\of{\Sigma}$.}
      Note that $\con_1$ cannot be empty, since this would imply that $\lamp{\labtwo}{\vartwo}{\tmfour} = \Sigma$;
      however $\Sigma$ is a redex, and in particular an application, so this cannot be the case.
      Hence the context $\con_1$ must be non-empty, \ie $\con_1 = \lamp{\labtwo}{\vartwo}{\con'_1}$ with $\tmfour = \con'_1\of{\Sigma}$.
      Given any term, context, or list of terms $\anon$ let us write $\anon^*$ to denote $\anon\sub{\var}{[\ls{\tmtwo}_1,\ls{\tmtwo}_2]}$.
      Then $\redex/\redextwo$ creates $\redexthree/(\redextwo/\redex)$:
      \[
        \xymatrix@R=.5cm{
          (\lamp{\lab}{\var}{\con_2\of{\var^\typ\,\ls{\tmthree}}})[\ls{\tmtwo}_1,\lamp{\labtwo}{\vartwo}{\con'_1\of{\Sigma}},\ls{\tmtwo}_2]
          \ar[r]^-{\redex}
          \ar[d]_-{\redextwo}
        &
          \con_2^*\of{(\lamp{\labtwo}{\vartwo}{\con'_1\of{\Sigma}})\,\ls{\tmthree}^*}
          \ar[d]_-{\redextwo/\redex}
          \ar[r]^-{\redexthree}
        &
        \\
          (\lamp{\lab}{\var}{\con_2\of{\var^\typ\,\ls{\tmthree}}})[\ls{\tmtwo}_1,\lamp{\labtwo}{\vartwo}{\con'_1\of{\Sigma'}},\ls{\tmtwo}_2]
          \ar[r]^-{\redex/\redextwo}
        &
          \con_2^*\of{(\lamp{\labtwo}{\vartwo}{\con'_1\of{\Sigma'}})\,\ls{\tmthree}^*}
          \ar[r]^-{\redexthree/(\redextwo/\redex)}
        &
        }
      \]
    \item {\bf If $\lamp{\labtwo}{\vartwo}{\tmfour} \neq \con_1\of{\Sigma}$.}
      Then $\lamp{\labtwo}{\vartwo}{\tmfour}$
      is either one of the terms in the list $\ls{\tmtwo}_1$
      or one of the terms in the list $\ls{\tmtwo}_2$.
      Given any term, context,
      or list of terms $\anon$ let $\anon^*$ denote $\anon\sub{\var}{[\ls{\tmtwo}_1,\con_1\of{\Sigma}\ls{\tmtwo}_2]}$
      and let $\anon^\dagger$ denote $\anon\sub{\var}{[\ls{\tmtwo}_1,\con_1\of{\Sigma'}\ls{\tmtwo}_2]}$.
      Then $\redex/\redextwo$ creates $\redexthree/(\redextwo/\redex)$:
      \[
        \xymatrix@R=.5cm{
          (\lamp{\lab}{\var}{\con_2\of{\var^\typ\,\ls{\tmthree}}})\,[\ls{\tmtwo}_1,\con_1\of{\Sigma},\ls{\tmtwo}_2]
          \ar[r]^-{\redex}
          \ar[d]_-{\redextwo}
        &
          \con_2^*\of{(\lamp{\labtwo}{\vartwo}{\tmfour})\,\ls{\tmthree}^*}
          \ar[d]_-{\redextwo/\redex}
          \ar[r]^-{\redexthree}
        &
          \con_2^*\of{  \tmfour\sub{\vartwo}{\ls{\tmthree}^*} }
        \\
          (\lamp{\lab}{\var}{\con_2\of{\var^\typ\,\ls{\tmthree}}})\,[\ls{\tmtwo}_1,\con_1\of{\Sigma'},\ls{\tmtwo}_2]
          \ar[r]^-{\redex/\redextwo}
        &
          \con_2^\dagger\of{(\lamp{\labtwo}{\vartwo}{\tmfour})\,\ls{\tmthree}^\dagger}
          \ar[r]^-{\redexthree/(\redextwo/\redex)}
        &
          \con_2^\dagger\of{  \tmfour\sub{\vartwo}{\ls{\tmthree}^\dagger} }
        }
      \]
    \end{enumerate}
  \end{enumerate}
\item {\bf Under an abstraction, $\con = \lamp{\lab}{\vartwo}{\con'}$.}
      \label{stability:case_under_abstraction}
  Straightforward by \ih.
\item {\bf Left of an application, $\con = \con'\ls{\tmthree}$.}
  Let $\Delta$ be the redex contracted by $\redex$, and let $\Delta'$ denote its contractum.
  The step $\redex$ is of the form $\con'\of{\Delta}\ls{\tmthree} \todist \con'\of{\Delta'}\ls{\tmthree}$.
  We consider three subcases, according to Creation~(\rlem{creation}), depending on whether
  $\redexthree$ is created by \CreationCase{I}, \CreationCase{II}, or \CreationCase{III}.
  \begin{enumerate}
  \item \CreationCase{I}
    Then $\con'$ is empty and $\Delta$ has the following particular shape:
    $\Delta = (\lamp{\lab}{\var}{\var^{\typ}})[\lamp{\labtwo}{\vartwo}{\tmfour}]$.
    The step $\redextwo$ can be either internal to the subterm $\tmfour$,
    or internal to one of the subterms in the list $\ls{\tmthree}$.
    If $\redextwo$ is internal to $\tmfour$, let $\tmfour'$ denote the term that results
    from $\tmfour$ after the step $\redextwo$.
    Then:
    \[
      \xymatrix@R=.5cm{
        (\lamp{\lab}{\var}{\var^{\typ}})[\lamp{\labtwo}{\vartwo}{\tmfour}]\,\ls{\tmthree}
        \ar[r]^-{\redex}
        \ar[d]_-{\redextwo}
      &
        (\lamp{\labtwo}{\vartwo}{\tmfour})\,\ls{\tmthree}
        \ar[d]_-{\redextwo/\redex}
        \ar[rr]^-{\redexthree}
      &&
        \tmfour\sub{\vartwo}{\ls{\tmthree}}
      \\
        (\lamp{\lab}{\var}{\var^{\typ}})[\lamp{\labtwo}{\vartwo}{\tmfour'}]\,\ls{\tmthree}
        \ar[r]^-{\redex/\redextwo}
      &
        (\lamp{\labtwo}{\vartwo}{\tmfour'})\,\ls{\tmthree}
        \ar[rr]^-{\redexthree/(\redextwo/\redex)}
      &&
        \tmfour'\sub{\vartwo}{\ls{\tmthree}}
      }
    \]
    Note that $\redex/\redextwo$ creates $\redexthree/(\redextwo/\redex)$, as required.
    If, on the other hand, $\redextwo$ is internal to $\ls{\tmthree}$, the situation is similar.
  \item \CreationCase{II}
    Then $\con'$ is empty and $\Delta$ has the following particular shape:
    $\Delta = (\lamp{\lab}{\var}{\lamp{\labtwo}{\vartwo}{\tmfour}})\,\ls{\tmtwo}$.
    The step $\redextwo$ can be either internal to $\tmfour$, internal to $\ls{\tmtwo}$,
    or internal to $\ls{\tmthree}$.
    If $\redextwo$ is internal to $\tmfour$, let $\tmfour'$ denote the term that results
    from $\tmfour$ after the step $\redextwo$. Then:
    \[
      \xymatrix@R=.5cm{
        (\lamp{\lab}{\var}{\lamp{\labtwo}{\vartwo}{\tmfour}})\,\ls{\tmtwo}\,\ls{\tmthree}
        \ar[r]^-{\redex}
        \ar[d]_-{\redextwo}
      &
        (\lamp{\labtwo}{\vartwo}{\tmfour\sub{\var}{\ls{\tmtwo}}})\,\ls{\tmthree}
        \ar[d]_-{\redextwo/\redex}
        \ar[r]^-{\redexthree}
      &
        \tmfour\sub{\var}{\ls{\tmtwo}}\sub{\vartwo}{\ls{\tmthree}}
      \\
        (\lamp{\lab}{\var}{\lamp{\labtwo}{\vartwo}{\tmfour'}})\,\ls{\tmtwo}\,\ls{\tmthree}
        \ar[r]^-{\redex/\redextwo}
      &
        (\lamp{\labtwo}{\vartwo}{\tmfour'\sub{\var}{\ls{\tmtwo}}})\,\ls{\tmthree}
        \ar[r]^-{\redexthree/(\redextwo/\redex)}
      &
        \tmfour'\sub{\var}{\ls{\tmtwo}}\sub{\vartwo}{\ls{\tmthree}}
      }
    \]
    Note that $\redex/\redextwo$ creates $\redexthree/(\redextwo/\redex)$, as required.
    If, on the other hand, $\redextwo$ is internal to $\ls{\tmtwo}$ or $\ls{\tmthree}$, the situation is similar.
  \item \CreationCase{III}
    Recall that the step $\redex$ is of the form
    $\redex : \con'\of{\Delta}\,\ls{\tmthree} \todist \con'\of{\Delta'}\,\ls{\tmthree}$.
    Since the step $\redexthree$ is created by \CreationCase{III},
    it does not take place at the root of $\con'[\Delta']\,\ls{\tmthree}$,
    but rather it is internal to $\con'[\Delta']$.
    We consider three subcases, depending on whether the step $\redextwo$ takes place
    at the root, to the left of the application or to the right of the application.
    \begin{enumerate}
    \item {\bf If $\redextwo$ takes place at the root.}
      Then $\con'\of{\Delta}$ is an abstraction,
      so $\con'$ cannot be empty, \ie we have that $\con' = \lamp{\labtwo}{\vartwo}{\con''}$.
      Moreover, since $\redexthree$ is created by \CreationCase{III},
      we have that $\Delta$ has the following specific shape:
      $\Delta = (\lamp{\lab}{\var}{\con_2\of{\var^\typ\,\ls{\tmfour}}})\,\ls{\tmtwo}$
      where $\ls{\tmtwo} = [\ls{\tmtwo}_1,\lamp{\labthree}{\varthree}{\tmfive},\ls{\tmtwo}_2]$
      and such that $\labthree$ is the external label of $\typ$.

      Given any term, context, or list of terms $\anon$ let $\anon^*$ denote $\anon\sub{\var}{\ls{\tmtwo}}$
      and let $\anon^\dagger$ denote $\anon\sub{\vartwo}{\ls{\tmthree}}$.
      Note also that by the Substitution lemma~(\rlem{substitution_lemma_alt}) we have that
      $\anon^\dagger\sub{\var}{\ls{\tmtwo}^\dagger} = {\anon^*}^\dagger$.
      Then $\redex/\redextwo$ creates $\redexthree/(\redextwo/\redex)$:
      {\small
      \[
        \xymatrix@R=.5cm{
          (\lamp{\labtwo}{\vartwo}{\con''\of{ (\lamp{\lab}{\var}{\con_2\of{\var^\typ\,\ls{\tmfour}}})\,\ls{\tmtwo} }})\,\ls{\tmthree}
          \ar[r]^-{\redex}
          \ar[d]_-{\redextwo}
        &
          (\lamp{\labtwo}{\vartwo}{\con''\of{ \con_2^*\of{(\lamp{\labthree}{\varthree}{\tmfive})\,\ls{\tmfour}^*}}})\,\ls{\tmthree}
          \ar[d]_-{\redextwo/\redex}
          \ar[r]^-{\redexthree}
        &
          (\lamp{\labtwo}{\vartwo}{\con''\of{ \con_2^*\of{  \tmfive\sub{\varthree}{\ls{\tmfour}^*} }}})\,\ls{\tmthree}
        \\
          {\con''}^\dagger\of{(\lamp{\lab}{\var}{\con_2^\dagger\of{\var^\typ\,\ls{\tmfour}^\dagger}})\,\ls{\tmtwo}^\dagger}
          \ar[r]^-{\redex/\redextwo}
        &
          {\con''}^\dagger\of{ {\con_2^*}^\dagger\of{(\lamp{\labthree}{\varthree}{\tmfive^\dagger})\,{\ls{\tmfour}^*}{}^\dagger} }
          \ar[r]^-{\redexthree/(\redextwo/\redex)}
        &
          {\con''}^\dagger\of{ {\con_2^*}^\dagger\of{  \tmfive^\dagger\sub{\varthree}{{\ls{\tmfour}^*}{}^\dagger} } }
        }
      \]}
    \item {\bf If $\redextwo$ is internal to $\con'\of{\Delta}$.}
      Then it is straightforward by \ih.
    \item {\bf If $\redextwo$ is internal to $\ls{\tmthree}$.}
          \label{stability:case_left_application:III:internal_left}
      Then:
      \[
        \xymatrix@R=.5cm{
          \con'\of{\Delta}\,\ls{\tmthree}
          \ar[r]^-{\redex}
          \ar[d]_-{\redextwo}
        &
          \con'\of{\Delta'}\,\ls{\tmthree}
          \ar[d]_-{\redextwo/\redex}
          \ar[rr]^-{\redexthree}
        &&
          \tmfive\,\ls{\tmthree}
        \\
          \con'\of{\Delta}\,\ls{\tmfour}
          \ar[r]^-{\redex/\redextwo}
        &
          \con'\of{\Delta'}\,\ls{\tmfour}
          \ar[rr]^-{\redexthree/(\redextwo/\redex)}
        &&
          \tmfive\,\ls{\tmfour}
        }
      \]
      It is immediate to note in this case that $\redex/\redextwo$ creates $\redexthree/(\redextwo/\redex)$,
      since the two-step sequences $\redex\redexthree$ and $(\redex/\redextwo)(\redexthree/(\redextwo/\redex))$
      are both isomorphic to
      $\con'\of{\Delta} \todist \con'\of{\Delta'} \todist \tmfive$, only going below different contexts.
    \end{enumerate}
  \end{enumerate}
\item {\bf Right of an application, $\con = \tmthree[\ls{\tmfour}_1,\con',\ls{\tmfour}_2]$.}
  Let $\Delta$ be the redex contracted by $\redex$, and let $\Delta'$ denote its contractum.
  The step $\redex$ is of the form
  $\tmthree[\ls{\tmfour}_1,\con'\of{\Delta},\ls{\tmfour}_2] \todist \tmthree[\ls{\tmfour}_1,\con'\of{\Delta'},\ls{\tmfour}_2]$.
  We consider three subcases, depending on whether the step $\redextwo$ takes place at the root,
  to the left of the application, or to the right of the application.
  \begin{enumerate}
  \item {\bf If $\redextwo$ takes place at the root.}
    Then $\tmthree$ is an abstraction $\tmthree = \lamp{\labtwo}{\vartwo}{\tmthree'}$.
    Moreover, since parameters are in 1--1 correspondence with arguments,
    there is a free occurrence of $\vartwo$ in $\tmthree'$ whose type is also the type of
    the argument $\con'\of{\Delta}$. Let us write $\tmthree'$ as $\tmthree' = \con_1\of{\vartwo^\typ}$,
    where the hole of $\con_1$ marks the position of such occurrence.
    Given any term, context, or list of terms $\anon$ let $\anon^*$ denote $\anon\sub{\var}{\ls{\tmfour}_1,\ls{\tmfour}_2}$.
    Then:
    \[
      \xymatrix@R=.5cm{
        (\lamp{\labtwo}{\vartwo}{\con_1\of{\vartwo^\typ}})[\ls{\tmfour}_1,\con'\of{\Delta},\ls{\tmfour}_2]
        \ar[r]^-{\redex}
        \ar[d]_-{\redextwo}
      &
        (\lamp{\labtwo}{\vartwo}{\con_1\of{\vartwo^\typ}})[\ls{\tmfour}_1,\con'\of{\Delta'},\ls{\tmfour}_2]
        \ar[d]_-{\redextwo/\redex}
        \ar[r]^-{\redexthree}
      &
        (\lamp{\labtwo}{\vartwo}{\con_1\of{\vartwo^\typ}})[\ls{\tmfour}_1,\tmfive,\ls{\tmfour}_2]
      \\
        \con_1^*\of{ \con'\of{\Delta} }
        \ar[r]^-{\redex/\redextwo}
      &
        \con_1^*\of{ \con'\of{\Delta'} }
        \ar[r]^-{\redexthree/(\redextwo/\redex)}
      &
        \con_1^*\of{ \tmfive }
      }
    \]
    So $\redex/\redextwo$ creates $\redexthree/(\redextwo/\redex)$,
    since the two-step sequences $\redex\redexthree$ and $(\redex/\redextwo)(\redexthree/(\redextwo/\redex))$
    are both isomorphic to $\con'\of{\Delta} \todist \con'\of{\Delta'} \todist \tmfive$,
    only going below different contexts.
  \item {\bf If $\redextwo$ is internal to the left of the application.}
    Then:
    \[
      \xymatrix@R=.5cm{
        \tmthree[\ls{\tmfour}_1,\con'\of{\Delta},\ls{\tmfour}_2]
        \ar[r]^-{\redex}
        \ar[d]_-{\redextwo}
      &
        \tmthree[\ls{\tmfour}_1,\con'\of{\Delta'},\ls{\tmfour}_2]
        \ar[d]_-{\redextwo/\redex}
        \ar[r]^-{\redexthree}
      &
        \tmthree[\ls{\tmfour}_1,\tmfive,\ls{\tmfour}_2]
      \\
        \tmthree'[\ls{\tmfour}_1,\con'\of{\Delta},\ls{\tmfour}_2]
        \ar[r]^-{\redex/\redextwo}
      &
        \tmthree'[\ls{\tmfour}_1,\con'\of{\Delta'},\ls{\tmfour}_2]
        \ar[r]^-{\redexthree/(\redextwo/\redex)}
      &
        \tmthree'[\ls{\tmfour}_1,\tmfive,\ls{\tmfour}_2]
      }
    \]
    Then since $\redex\redexthree$ and $(\redex/\redextwo)(\redexthree/(\redextwo/\redex))$
    are isomorphic, it is immediate to conclude,
    similarly as in item~\ref{stability:case_left_application:III:internal_left} of this lemma.
  \item {\bf If $\redextwo$ is internal to the right of the application.}
    We consider two further subcases, depending on whether $\redextwo$ is
    internal to the argument $\con'\of{\Delta}$, or otherwise.
    \begin{enumerate}
    \item {\bf If $\redextwo$ is internal to the argument $\con'\of{\Delta}$.}
      Then it is straightforward by \ih.
    \item {\bf If $\redextwo$ is not internal to the argument $\con'\of{\Delta}$.}
      Then it is either internal to $\ls{\tmfour}_1$ or to $\ls{\tmfour}_2$.
      If it is internal to $\ls{\tmfour}_1$, then:
      \[
        \xymatrix@R=.5cm{
          \tmthree[\ls{\tmfour}_1,\con'\of{\Delta},\ls{\tmfour}_2]
          \ar[r]^-{\redex}
          \ar[d]_-{\redextwo}
        &
          \tmthree[\ls{\tmfour}_1,\con'\of{\Delta'},\ls{\tmfour}_2]
          \ar[d]_-{\redextwo/\redex}
          \ar[r]^-{\redexthree}
        &
          \tmthree[\ls{\tmfour}_1,\tmfive,\ls{\tmfour}_2]
        \\
          \tmthree[\ls{\tmfour}'_1,\con'\of{\Delta},\ls{\tmfour}_2]
          \ar[r]^-{\redex/\redextwo}
        &
          \tmthree[\ls{\tmfour}'_1,\con'\of{\Delta'},\ls{\tmfour}_2]
          \ar[r]^-{\redexthree/(\redextwo/\redex)}
        &
          \tmthree[\ls{\tmfour}'_1,\tmfive,\ls{\tmfour}_2]
        }
      \]
      Then since $\redex\redexthree$ and $(\redex/\redextwo)(\redexthree/(\redextwo/\redex))$
      are isomorphic, it is immediate to conclude,
      similarly as in item~\ref{stability:case_left_application:III:internal_left} of this lemma.
    \end{enumerate}
  \end{enumerate}
\end{enumerate}
\end{proof}

\begin{lemma}[Full Stability]
\llem{full_stability}
Let $\redseq,\redseqtwo$ be coinitial derivations
such that $\names(\redseq) \cap \names(\redseqtwo) = \emptyset$.
Let $\redexthree_1,\redexthree_2,\redexthree_3$ be steps such that
$\redexthree_3 = \redexthree_1/(\redseqtwo/\redseq) = \redexthree_2/(\redseq/\redseqtwo)$.
Then there is a step $\redexthree_0$ such that $\redexthree_1 = \redexthree_0/\redseq$
and $\redexthree_2 = \redexthree_0/\redseqtwo$.
\end{lemma}
\begin{proof}
The proof of Full Stability (\rlem{full_stability})
is straightforward using {\em Basic Stability} (\rlem{basic_stability}).
Consider the diagram
formed by $\redseq = \redex_1\hdots\redex_n$
and $\redseqtwo = \redextwo_1\hdots\redextwo_m$.
By Permutation (\rprop{strong_permutation})
the diagram can be tiled using squares, as in $\lambdadist$
there is no erasure or duplication. Note that for all $i \neq j$ we have
that $\redex_i \neq \redextwo_j$,
since $\names(\redseq)$ and $\names(\redseqtwo)$ are disjoint.
It suffices to apply Basic Stability $n \times m$ times, once for each tile.
\end{proof}

\subsection{Proof of \rprop{lambdadist_lattice} -- $\ulbDerivDist{\tm}$ is a lattice}
\label{appendix_lambdadist_lattice}

Let $\cls{\redseq}$ denote the permutation equivalence class of
a derivation $\redseq$,
and let $\ulbDerivDist{\tm} = \set{\cls{\redseq} \ST \src(\redseq) = \tm}$
be the set of $\todist$-derivations
starting on $\tm$, modulo permutation equivalence.
We already know that $(\ulbDerivDist{\tm},\permle)$ is an upper
semilattice, where the order is given by
$\cls{\redseq} \permle \cls{\redseqtwo} \iffdef \redseq \permle \redseqtwo$,
the bottom is $\bot = \cls{\emptyDerivation}$
and the join is
$\cls{\redseq} \sqcup \cls{\redseqtwo} = \cls{\redseq \sqcup \redseqtwo}$.

Note that $\ulbDerivDist{\tm}$ has a top element:
given that $\lambdadist$ is strongly normalizing~(\rprop{strong_normalization}),
consider a derivation $\redseq : \tm \rtodist^* \tmtwo$ to normal form.
Then for any other derivation $\redseqtwo : \tm \rtodist^* \tmthree$
we have that $\redseqtwo/\redseq = \emptyDerivation$, so it suffices to take
$\top = \cls{\redseq}$.

To see that $\ulbDerivDist{\tm}$ is a lattice for all $\tm \in \termsdist$,
we are left to show that any two coinitial derivations $\redseq,\redseqtwo$
have a {\em meet} $\redseq \sqcap \redseqtwo$
with respect to the prefix order $\permle$,
and define $\cls{\redseq} \sqcap \cls{\redseqtwo} \eqdef \cls{\redseq \sqcap \redseqtwo}$.
To this purpose, we prove various auxiliary results.

\begin{definition}
A step $\redex$ \defn{belongs to} a coinitial derivation $\redseq$,
written $\redex \in \redseq$,
if and only if
some residual of $\redex$ is contracted along $\redseq$.
More precisely,
$\redex \in \redseq$
if there exist $\redseq_1,\redex',\redseq_2$ such that
$\redseq = \redseq_1\redex\redseq_2$ and $\redex' \in \redex/\redseq_1$,
We write $\redex \not\in \redseq$
if it is not the case that $\redex \in \redseq$.
\end{definition}

\begin{lemma}[Characterization of belonging]
\llem{characterization_of_belonging}
Let $\redex$ be a step and $\redseq$ a coinitial derivation
in $\lambdadist$.
Then the following are equivalent:
\begin{enumerate}
\item $\redex \in \redseq$,
\item $\redex \permle \redseq$,
\item $\name(\redex) \in \names(\redseq)$.
\end{enumerate}
{\em Note:} the hypothesis that $\redex$ and $\redseq$ are coinitial is crucial.
In particular, $(1)$ and $(2)$ by definition only hold when $\redex$ and $\redseq$ are coinitial,
while $(3)$ might hold even if $\redex$ and $\redseq$ are not coinitial.
\end{lemma}
\begin{proof}
$(1 \Rightarrow 2)$
  Let $\redseq = \redseq_1\redextwo\redseqtwo_2$ where $\redextwo$ is a residual of $\redex$.
  Suppose moreover, without loss of generality,
  that $\redseq_1$ is minimal, \ie that $\redex \not\in \redseq_1$.
  Since in $\lambdadist$ there is no duplication or erasure,
  $\redex/\redseq_1$ is a singleton, so $\redex/\redseq_1 = \redextwo$.
  This means that $\redex/\redseq_1\redextwo\redseq_2 = \emptyset$,
  so indeed $\redex \permle \redseq_1\redextwo\redseq_2$.\\
$(2 \Rightarrow 3)$
  By induction on $\redseq$.
  If $\redseq$ is empty, the implication is vacuously true,
  so let $\redseq = \redextwo\redseqtwo$
  and consider two subcases, depending on whether $\redex = \redextwo$.
  If $\redex = \redextwo$, then indeed the first step of $\redseq = \redex\redseqtwo$
  has the same label as $\redex$.
  On the other hand, if $\redex \neq \redextwo$, 
  since in $\lambdadist$ there is no duplication or erasure,
  we have that $\redex/\redextwo = \redex'$,
  where $\name(\redex) = \name(\redex')$.
  Note that $\redex \permle \redextwo\redseqtwo$ 
  so $\redex/\redextwo = \redex' \permle \redseqtwo$.
  By applying the \ih we obtain that there must be a step in $\redseqtwo$
  whose label is $\name(\redex) = \name(\redex')$, and we are done.\\
$(3 \Rightarrow 1)$
  By induction on $\redseq$.
  If $\redseq$ is empty, the implication is vacuously true,
  so let $\redseq = \redextwo\redseqtwo$
  and consider two subcases, depending on whether $\name(\redex) = \name(\redextwo)$.
  If $\name(\redex) = \name(\redextwo)$ then $\redex$ and $\redextwo$ must be the same step,
  since terms are correct, which means that labels decorating lambdas are pairwise distinct.
  Hence $\redex \in \redseq = \redex\redseqtwo$.
  On the other hand, if $\name(\redex) \neq \name(\redextwo)$, then
  $\redex \neq \redextwo$
  so since in $\lambdadist$ there is no duplication or erasure
  we have that $\redex/\redextwo = \redex'$,
  where $\name(\redex) = \name(\redex')$.
  By hypothesis, there is a step in the derivation $\redseq = \redextwo\redseqtwo$ whose label
  is $\name(\redex)$, and it is not $\redextwo$, so there must be at least one step in the
  derivation $\redseqtwo$ whose label is $\name(\redex) = \name(\redex')$.
  By \ih $\redex' \in \redseqtwo$
  and then, since $\redex'$ is a residual of $\redex$,
  we conclude that $\redex \in \redextwo\redseqtwo$,
  as required.
\end{proof}

It is immediate to note that,
when composing derivations $\redseq, \redseqtwo$,
the set of labels of $\redseq\redseqtwo$
results from the union of the labels of $\redseq$ and $\redseqtwo$:
\begin{remark}
\lremark{names_concatenation_union}
$\names(\redseq\redseqtwo) = \names(\redseq) \cup \names(\redseqtwo)$
\end{remark}
Indeed, a stronger results holds. Namely, the union is disjoint:

\begin{lemma}
\llem{names_concatenation_disjoint_union}
$\names(\redseq\redseqtwo) = \names(\redseq) \uplus \names(\redseqtwo)$
\end{lemma}
\begin{proof}
Let $\redseq : \tm \rtodist \tmtwo$.
By induction on $\redseq$ we argue that the set of labels decorating the lambdas in $\tmtwo$
is disjoint from $\names(\redseq)$.
The base case is immediate, so let $\redseq = \redexthree\redseqthree$.
The step $\redexthree$ is of the form:
$
  \redexthree :
  \tm =
  \conof{(\lamp{\lab}{\var}{\tmthree})\ls{\tmfour}}
  \todist
  \conof{\subs{\tmthree}{\var}{\ls{\tmfour}}}
$.
Hence the target of $\redexthree$ has no lambdas decorated with the label $\lab$.
Moreover, the derivation $\redseqthree$ is of the form:
$
  \redseqthree :
  \conof{\subs{\tmthree}{\var}{\ls{\tmfour}}}
  \rtodist
  \tmtwo
$
and by \ih the set of labels decorating the lambdas in $\tmtwo$
is disjoint from $\names(\redseqthree)$.
As a consequence,
the set of labels decorating the lambdas in $\tmtwo$
is disjoint from both $\names(\redseqthree)$
and $\set{\lab}$.
This completes the proof.
\end{proof}
As an easy consequence:

\begin{corollary}
\lcoro{length_of_derivation_is_number_of_distinct_names}
If $\lengthof{\redseq}$ denotes the length of $\redseq$,
then
$\lengthof{\redseq} = \#(\names(\redseq))$.
\end{corollary}

\begin{lemma}
\llem{names_after_projection_along_a_step}
Let $\redseq$ and $\redseqtwo$ be coinitial.
Then
$\names(\redseq/\redseqtwo) = \names(\redseq) \setminus \names(\redseqtwo)$
\end{lemma}
\begin{proof}
First we claim that $\names(\redseq/\redex) = \names(\redseq) \setminus \set{\name(\redex)}$.
By induction on $\redseq$.
The base case is immediate, so let $\redseq = \redextwo\redseqtwo$
and consider two subcases, depending on whether $\redex = \redextwo$.
If $\redex = \redextwo$,
then
$
  \names(\redseq/\redex)
  = \names(\redex\redseqtwo/\redex)
  = \names(\redseqtwo)
  = \names(\redex\redseqtwo) \setminus \set{\name(\redex)}
$.
On the other hand if $\redex \neq \redextwo$,
then since
in $\lambdadist$ there is no duplication or erasure
we have that $\redex/\redextwo = \redex'$
where $\name(\redex) = \name(\redex')$
and, similarly,
$\redextwo/\redex = \redextwo'$,
where $\name(\redextwo) = \name(\redextwo')$. Hence:
\[
  \begin{array}{rcll}
  \names(\redseq/\redex)
  & = & \names(\redextwo\redseqtwo/\redex) \\
  & = & \names((\redextwo/\redex)(\redseqtwo/(\redex/\redextwo))) \\
  & = & \names(\redextwo/\redex) \cup \names(\redseqtwo/(\redex/\redextwo)) \\
  & = & \names(\redextwo') \cup \names(\redseqtwo/\redex') \\
  & = & \names(\redextwo') \cup (\names(\redseqtwo) \setminus \set{\name(\redex')} & \text{ by \ih} \\
  & = & \names(\redextwo) \cup (\names(\redseqtwo) \setminus \set{\name(\redex)}) \\
  & = & (\names(\redextwo) \cup \names(\redseqtwo)) \setminus \set{\name(\redex)} & \text{since $\name(\redex) \neq \name(\redextwo)$} \\
  & = & \names(\redextwo\redseqtwo) \setminus \set{\name(\redex)} \\
  & = & \names(\redseq) \setminus \set{\name(\redex)} \\
  \end{array}
\]
which completes the claim.
To see that $\names(\redseq/\redseqtwo) = \names(\redseq) \setminus \names(\redseqtwo)$
for an arbitrary derivation $\redseqtwo$,
proceed by induction on $\redseqtwo$.
If $\redseqtwo$ is empty it is trivial, so
consider the case in which $\redseqtwo = \redex\redseqthree$.
\[
  \begin{array}{rcll}
  \names(\redseq/\redex\redseqthree) & = & \names((\redseq/\redex)/\redseqthree) \\
                                     & = & \names(\redseq/\redex) \setminus \names(\redseqthree) & \text{ by \ih} \\
                                     & = & (\names(\redseq) \setminus \set{\name(\redex)}) \setminus \names(\redseqthree) & \text{ by the previous claim} \\
                                     & = & \names(\redseq) \setminus (\set{\name(\redex)} \cup \names(\redseqthree)) \\
                                     & = & \names(\redseq) \setminus \names(\redex\redseqthree) \\
  \end{array}
\]
\end{proof}

\begin{proposition}[Prefixes as subsets]
\lprop{prefixes_as_subsets}
Let $\redseq,\redseqtwo$ be coinitial derivations in the distributive lambda-calculus.
Then
$\redseq \permle \redseqtwo$ if and only if $\names(\redseq) \subseteq \names(\redseqtwo)$.
\end{proposition}
\begin{proof}
By induction on $\redseq$.
The base case is immediate since $\emptyDerivation \permle \redseqtwo$
and $\emptyset \subseteq \names(\redseqtwo)$ both hold.
So let $\redseq = \redexthree\redseqthree$. First note that
the following equivalence holds:
\begin{equation}
\leqn{prefixes_as_subsets:1}
    \redexthree\redseqthree \permle \redseqtwo
    \HS\iff\HS
    \redexthree \permle \redseqtwo \ \land\ \redseqthree \permle \redseqtwo/\redexthree
\end{equation}
Indeed:
\begin{itemize}
\item[] $(\Rightarrow)$
    Suppose that $\redexthree\redseqthree \permle \redseqtwo$.
    Then, on one hand, $\redexthree \permle \redexthree\redseqthree \permle \redseqtwo$.
    On the other hand, projection is monotonic, so
    $\redseqthree = \redexthree\redseqthree/\redexthree \permle \redseqtwo/\redexthree$.
\item[] $(\Leftarrow)$
    Since $\redseqthree \permle \redseqtwo/\redexthree$
    we have that $\redexthree\redseqthree \permle \redexthree(\redseqtwo/\redexthree) \permeq \redseqtwo(\redexthree/\redseqtwo) = \redseqtwo$
    since $\redexthree/\redseqtwo = \emptyDerivation$.
\end{itemize}
So we have that:
\[
  \begin{array}{rcll}
    \redexthree\redseqthree \permle \redseqtwo
    & \iff &
    \redexthree \permle \redseqtwo \ \land\ \redseqthree \permle \redseqtwo/\redexthree
      & \text{ by \reqn{prefixes_as_subsets:1}} \\
    & \iff &
    \name(\redexthree) \in \names(\redseqtwo) \ \land\ \redseqthree \permle \redseqtwo/\redexthree
      & \text{ by \rlem{characterization_of_belonging}} \\
    & \iff &
    \name(\redexthree) \in \names(\redseqtwo) \ \land\ \names(\redseqthree) \subseteq \names(\redseqtwo/\redexthree)
      & \text{ by \ih} \\
    & \iff &
    \name(\redexthree) \in \names(\redseqtwo) \ \land\ \names(\redseqthree) \subseteq \names(\redseqtwo) \setminus \set{\name(\redexthree)}
      & \text{ by \rlem{names_after_projection_along_a_step}} \\
    & \iff &
    \names(\redexthree\redseqthree) \subseteq \names(\redseqtwo)
  \end{array}
\]
To justify the very last equivalence, the $(\Rightarrow)$ direction is immediate.
For the $(\Leftarrow)$ direction,
the difficulty is ensuring that
$\names(\redseqthree) \subseteq \names(\redseqtwo) \setminus \set{\name(\redexthree)}$
from the fact that
$\names(\redexthree\redseqthree) \subseteq \names(\redseqtwo)$.
To see this it suffices to observe that by \rlem{names_concatenation_disjoint_union},
$\names(\redexthree\redseqthree)$
is the {\em disjoint} union
$\names(\redexthree) \uplus \names(\redseqthree)$,
which means that
$\name(\redexthree) \not\in \names(\redseqthree)$.
\end{proof}
As an easy consequence:

\begin{corollary}[Permutation equivalence in terms of labels]
\lcoro{permutation_equivalence_in_terms_of_names}
Let $\redseq,\redseqtwo$ be coinitial derivations in $\lambdadist$.
Then
$\redseq \permeq \redseqtwo$ if and only if $\names(\redseq) = \names(\redseqtwo)$.
\end{corollary}

\begin{lemma}[Projections are decreasing]
\llem{projections_are_decreasing}
Let $\redex \in \redseq$. Then $\lengthof{\redseq} = 1 + \lengthof{\redseq/\redex}$.
\end{lemma}
\begin{proof}
Observe that $\redex \permle \redseq$ by \rlem{characterization_of_belonging}.
So $\redseq \permeq \redex(\redseq/\redex)$, which gives us that:
\[
  \begin{array}{rcll}
  \lengthof{\redseq} & = & \#\names(\redseq) & \text{ by \rcoro{length_of_derivation_is_number_of_distinct_names}} \\
                     & = & \#\names(\redex(\redseq/\redex)) & \text{ by \rcoro{permutation_equivalence_in_terms_of_names}, since $\redseq \permeq \redex(\redseq/\redex)$} \\
                     & = & \#(\names(\redex) \uplus \names(\redseq/\redex) & \text{ by \rlem{names_concatenation_disjoint_union}} \\
                     & = & 1 + \#\names(\redseq/\redex) \\
                     & = & 1 + \lengthof{\redseq/\redex} & \text{ by \rcoro{length_of_derivation_is_number_of_distinct_names}} \\
  \end{array}
\]
\end{proof}

The proof of \rprop{lambdadist_lattice} proceeds as follows.
If $\redseq$ and $\redseqtwo$ are derivations, we say that a step $\redex$
is a \defn{common} (to $\redseq$ and $\redseqtwo$) whenever $\redex \in \redseq$ and $\redex \in \redseqtwo$.
Define $\redseq \sqcap \redseqtwo$ as follows,
by induction on the length of $\redseq$:
\[
  \redseq \sqcap \redseqtwo \eqdef
    \begin{cases}
    \emptyDerivation & \text{if there are no common steps to $\redseq$ and $\redseqtwo$} \\
    \redex((\redseq/\redex) \sqcap (\redseqtwo/\redex)) & \text{if the step $\redex$ is common to $\redseq$ and $\redseqtwo$ } \\
    \end{cases}
\]
In the second case of the definition, there might be more than one $\redex$
common to $\redseq$ and $\redseqtwo$.
Any such step may be chosen and we make no further assumptions.
To see that this recursive construction is well-defined,
note that the length of $\redseq/\redex$ is lesser than the length of $\redseq$
by the fact that projections are decreasing (\rlem{projections_are_decreasing}).
To conclude the construction, we show that $\redseq \sqcap \redseqtwo$ is an infimum, \ie a greatest lower bound:
\begin{enumerate}
\item {\bf Lower bound.}
  Let us show that $\redseq \sqcap \redseqtwo \permle \redseq$
  by induction on the length of $\redseq$.
  There are two subcases, depending on whether
  there is a step common to $\redseq$ and $\redseqtwo$.

  If there is no common step,
  then $\redseq \sqcap \redseqtwo = \emptyDerivation$
  trivially verifies $\redseq \sqcap \redseqtwo \permle \redseq$.

  On the other hand, if there is a common step, we have by definition that
  $\redseq \sqcap \redseqtwo = \redex((\redseq/\redex) \sqcap (\redseqtwo/\redex))$
  where $\redex$ is common to $\redseq$ and $\redseqtwo$.
  Recall that projections are decreasing (\rlem{projections_are_decreasing})
  so $\lengthof{\redseq} > \lengthof{\redseq/\redex}$.
  This allows us to apply the \ih and conclude:
  \[
  \begin{array}{rcll}
    \redseq \sqcap \redseqtwo & =       & \redex((\redseq/\redex) \sqcap (\redseqtwo/\redex)) & \text{ by definition} \\
                              & \permle & \redex(\redseq/\redex)   & \text{ since by \ih $(\redseq/\redex) \sqcap (\redseqtwo/\redex) \permle \redseq/\redex$} \\
                              & \permeq & \redseq(\redex/\redseq)  \\
                              & =       & \redseq                  & \text{ since $\redex \permle \redseq$ by \rlem{characterization_of_belonging}.} \\
  \end{array}
  \]
  Showing that $\redseq \sqcap \redseqtwo \permle \redseqtwo$ is symmetric,
  by induction on the length of $\redseqtwo$.
\item {\bf Greatest lower bound.}
  Let $\redseqthree$ be a lower bound for $\set{\redseq, \redseqtwo}$,
  \ie $\redseqthree \permle \redseq$ and $\redseqthree \permle \redseqtwo$,
  and let us show that $\redseqthree \permle \redseq \sqcap \redseqtwo$.
  We proceed by induction on the length of $\redseq$.
  There are two subcases, depending on whether
  there is a step common to $\redseq$ and $\redseqtwo$.

  If there is no common step,
  we claim that $\redseqthree$ must be empty.
  Otherwise we would have that $\redseqthree = \redexthree\redseqthree' \permle \redseq$
  so in particular $\redexthree \permle \redseq$ and $\redexthree \in \redseq$ by \rlem{characterization_of_belonging}.
  Similarly, $\redexthree \in \redseqtwo$ so
  $\redexthree$ is a step common to $\redseq$ and $\redseqtwo$,
  which is a contradiction.
  We obtain that $\redseqthree$ is empty,
  so trivially $\redseqthree = \emptyDerivation \permle \redseq \sqcap \redseqtwo$.

  On the other hand, if there is a common step, we have by definition that
  $\redseq \sqcap \redseqtwo = \redex((\redseq/\redex) \sqcap (\redseqtwo/\redex))$
  where $\redex$ is common to $\redseq$ and $\redseqtwo$.
  Moreover, since $\redseqthree \permle \redseq$ and $\redseqthree \permle \redseqtwo$,
  by projecting along $\redex$ we know that
  $\redseqthree/\redex \permle \redseq/\redex$ and $\redseqthree/\redex \permle \redseqtwo/\redex$.
  So:
  \[
    \begin{array}{rcll}
      \redseqthree & \permle & \redseqthree(\redex/\redseqthree) \\
                   & \permeq & \redex(\redseqthree/\redex) \\
                   & \permle & \redex((\redseq/\redex) \sqcap (\redseqtwo/\redex)) & \text{ since by \ih $\redseqthree/\redex \permle (\redseq/\redex) \sqcap (\redseqtwo/\redex)$} \\
                   & =       & \redseq \sqcap \redseqtwo & \text{ by definition}
    \end{array}
  \]
\end{enumerate}
Finally, observe that
the meet of $\set{\redseq,\redseqtwo}$
is unique modulo permutation equivalence,
since
if $\redseqthree$ and $\redseqthree'$ both
have the universal property of being a greatest lower bound,
then $\redseqthree \permle \redseqthree'$
and $\redseqthree' \permle \redseqthree$.

\subsection{Proof of \rprop{labels_morphism} -- The function $\names$ is a morphism of lattices}
\label{appendix_labels_morphism}

We need the following auxiliary result:

\begin{lemma}[Disjoint derivations]
\llem{properties_of_disjoint_derivations}
Let $\redseq$ and $\redseqtwo$ be coinitial derivations in $\lambdadist$.
The following are equivalent:
\begin{enumerate}
\item $\names(\redseq) \cap \names(\redseqtwo) = \emptyset$.
\item $\redseq \sqcap \redseqtwo = \emptyDerivation$.
\item There is no step common to $\redseq$ and $\redseqtwo$.
\end{enumerate}
In this case we say that $\redseq$ and $\redseqtwo$ are \defn{disjoint}.
\end{lemma}
\begin{proof}
The implication $(1 \implies 2)$ is immediate since if we suppose that $\redseq \sqcap \redseqtwo$ is non-empty
then the first step of $\redseq \sqcap \redseqtwo$ is a step $\redexthree$
such that $\redexthree \in \redseq$ and $\redexthree \in \redseqtwo$.
By \rlem{characterization_of_belonging},
this means that $\name(\redexthree) \in \names(\redseq) \cap \names(\redseqtwo)$,
contradicting the fact that $\name(\redexthree)$ and $\names(\redseq)$ are disjoint.

The implication
$(2 \implies 3)$ is immediate by the definition of $\redseq \sqcap \redseqtwo$ (defined in \rprop{lambdadist_lattice}).

Let us check that the implication $(3 \implies 1)$ holds.
By the contrapositive, suppose that $\names(\redseq)$ and $\names(\redseqtwo)$
are not disjoint, and let us show that there is a step common to $\redseq$ and $\redseqtwo$.
Since $\names(\redseq) \cap \names(\redseqtwo) \neq \emptyset$,
we know that the derivation $\redseq$ can be written as $\redseq = \redseq_1 \redex \redseq_2$
where $\name(\redex) \in \names(\redseqtwo)$.
Without loss of generality we may suppose
that $\redex$ is the first step in $\redseq$ with that property,
\ie that $\names(\redseq_1) \cap \names(\redseqtwo) = \emptyset$.
Moreover, let us write $\redseqtwo$ as $\redseqtwo = \redseqtwo_1 \redextwo \redseqtwo_2$ 
where $\name(\redex) = \name(\redextwo)$.

Observe that the label of $\redex$ does not appear anywhere along the sequence of steps $\redseq_1$,
\ie that $\name(\redex) \not\in \names(\redseq_1)$, as a consequence of the
fact that no labels are ever repeated in any sequence of steps (\rlem{names_concatenation_disjoint_union}).
This implies that $\name(\redextwo) \not\in \names(\redseq_1/\redseqtwo_1)$.
Indeed:
\[
  \name(\redextwo)
  = \name(\redex)
  \not\in \names(\redseq_1)
  \supseteq \names(\redseq_1) \setminus \names(\redseqtwo_1)
  =^{(\text{\rlem{names_after_projection_along_a_step}})} \names(\redseq_1/\redseqtwo_1)
\]
This means that $\redextwo$ is not erased by the derivation $\redseq_1/\redseqtwo_1$.
More precisely, $\redextwo/(\redseq_1/\redseqtwo_1)$ is a singleton.

Symmetrically, $\redex/(\redseqtwo_1/\redseq_1)$ is a singleton.
Moreover, $\name(\redextwo/(\redseq_1/\redseqtwo_1)) = \name(\redextwo) = \name(\redex) = \name(\redex/(\redseqtwo_1/\redseq_1))$ 
so we have that $\redextwo/(\redseq_1/\redseqtwo_1) = \redex/(\redseqtwo_1/\redseq_1)$.
The situation is the following, where $\names(\redseq_1) \cap \names(\redseqtwo_1) = \emptyset$:
\[
  \xymatrix@R=.5cm@C=.5cm{
    &
    &
    & \ar@{->>}[ld]_{\redseq_1} \ar@{->>}[rd]^{\redseqtwo_1} &
  \\
    &
    \ar@{->>}[l]_{\redseq_2}
    &
    \ar[l]_{\redex}
    \ar@{->>}[rd]_{\redseqtwo_1/\redseq_1} & & \ar@{->>}[ld]^{\redseq_1/\redseqtwo_1} \ar[r]^{\redextwo}
    &
    \ar@{->>}[r]^{\redseqtwo_2}
    &
  \\
    &
    &
    &
    \ar[d]_{\redex/(\redseqtwo_1/\redseq_1) = \redextwo/(\redseq_1/\redseqtwo_1)}
    &
  \\
    &
    &&&
  }
\]
By Full stability~(\rlem{full_stability}) this means that there exists a step $\redexthree$
such that $\redexthree/\redseq_1 = \redex$ and $\redexthree/\redseqtwo_1 = \redextwo$.
Then $\redexthree \in \redseq_1\redex\redseq_2 = \redseq$
and also $\redexthree \in \redseqtwo_1\redextwo\redseqtwo_2 = \redseqtwo$
so $\redexthree$ is common to $\redseq$ and $\redseqtwo$,
by which we conclude.
\end{proof}

Let $X$ be the set of labels of the steps involved in some derivation to normal form.
More precisely, let $\redseq_0 : \tm \todist^* \tmtwo$ be a derivation to normal form, which exists
by virtue of \resultname{Strong Normalization}~(\rprop{strong_normalization}),
and let $X = \names(\redseq_0)$.
Then we claim that:
\[
  \begin{array}{rrcl}
  \names : & \ulbDerivDist{\tm} & \to     & \mathcal{P}(X) \\
           & \cls{\redseq}      & \mapsto & \names(\redseq)
  \end{array}
\]
is a monomorphism of lattices.
Note that $\names$ is well-defined
over permutation equivalence classes
since if $\redseq \permeq \redseqtwo$ then
$\names(\redseq) = \names(\redseqtwo)$
by \rcoro{permutation_equivalence_in_terms_of_names}.
We are to show that $\names$ is monotonic, that it preserves bottom, joins, top, and meets,
and finally that it is a monomorphism:
\begin{enumerate}
\item {\bf Monotonic, \ie
  $\classof{\redseq} \permle \classof{\redseqtwo}$
  implies $\names(\redseq) \subseteq \names(\redseqtwo)$.}
  By \rprop{prefixes_as_subsets}.
\item {\bf Preserves bottom.}
  Indeed, $\names(\cls{\emptyDerivation}) = \emptyset$.
\item {\bf Preserves joins,
  \ie
  $\names(\classof{\redseq} \sqcup \classof{\redseqtwo}) =
   \names(\redseq) \cup \names(\redseqtwo)$.}
  Indeed:
  \[
    \begin{array}{rcll}
    \names(\redseq \sqcup \redseqtwo) & = & \names(\redseq(\redseqtwo/\redseq)) & \text{ by definition of $\sqcup$} \\
                                      & = & \names(\redseq) \cup \names(\redseqtwo/\redseq) \\
                                      & = & \names(\redseq) \cup (\names(\redseqtwo) \setminus \names(\redseq)) & \text{ by \rlem{names_after_projection_along_a_step}} \\
                                      & = & \names(\redseq) \cup \names(\redseqtwo) \\
    \end{array}
  \]
\item {\bf Preserves top.}
  Recall that $\top = \cls{\redseq_0}$ where $\redseq_0 : \tm \rtodist^* \tmtwo$
  is the derivation to normal form.
  Note moreover that all derivations to normal form
  are permutation equivalent in an orthogonal axiomatic rewrite system, so the choice
  does not matter.
  To conclude, observe that $\names(\redseq_0) = X$ by definition of $X$,
  and $X$ is indeed the top element of $\mathcal{P}(X)$.
\item {\bf Preserves meets,
      \ie $\names(\classof{\redseq} \sqcap \classof{\redseqtwo}) =
           \names(\redseq) \cap \names(\redseqtwo)$.}
  We show the two inclusions:
  \begin{enumerate}
  \item $(\subseteq)$
    It suffices to check that
    $\names(\redseq \sqcap \redseqtwo) \subseteq \names(\redseq)$
    (the inclusion
    $\names(\redseq \sqcap \redseqtwo) \subseteq \names(\redseqtwo)$
    is symmetric).
    By induction on the length of $\redseq \sqcap \redseqtwo$.
    If $\redseq \sqcap \redseqtwo$ is empty it is immediate.
    If it is non-empty,
    $\redseq \sqcap \redseqtwo =
     \redexthree(\redseq/\redexthree \sqcap \redseqtwo/\redexthree)$,
    where $\redexthree$ is a step common to $\redseq$ and $\redseqtwo$.
    Then since $\redexthree$ is common to $\redseq$ and $\redseqtwo$,
    we have that $\name(\redexthree) \in \names(\redseq)$.
    Moreover, by \ih $\names(\redseq/\redexthree \sqcap \redseqtwo/\redexthree) \subseteq \names(\redseq/\redexthree)$.
    So:
    \[
      \begin{array}{rcll}
      \names(\redseq \cap \redseqtwo) & = & \set{\name(\redexthree)} \cup \names(\redseq/\redexthree \sqcap \redseqtwo/\redexthree) \\
                                      & \subseteq & \names(\redseq) \cup \names(\redseq/\redexthree) & \text{by \ih} \\
                                      & = & \names(\redseq) \cup (\names(\redseq) \setminus \set{\name(\redexthree)} & \text{by \rlem{names_after_projection_along_a_step}} \\
                                      & = & \names(\redseq)
      \end{array}
    \]
  \item $(\supseteq)$
    To show  $\names(\redseq) \cap \names(\redseqtwo) \subseteq \names(\redseq \sqcap \redseqtwo)$,
    we first prove the following claim:\bigskip

    {\bf Claim.} $\names(\redseq/(\redseq \sqcap \redseqtwo)) \cap \names(\redseqtwo/(\redseq \sqcap \redseqtwo) = \emptyset$.\\
    {\em Proof of the claim.}
    By \rlem{properties_of_disjoint_derivations} it suffices to show that
    $(\redseq/(\redseq \sqcap \redseqtwo)) \sqcap (\redseqtwo/(\redseq \sqcap \redseqtwo)) = \emptyDerivation$.
    By contradiction, suppose that there is a step $\redexthree$ common to
    the derivations
    $\redseq/(\redseq \sqcap \redseqtwo)$ and $\redseqtwo/(\redseq \sqcap \redseqtwo)$.
    Then the derivation $(\redseq \sqcap \redseqtwo)\redexthree$ is a lower bound for $\set{\redseq,\redseqtwo}$,
    \ie
    $(\redseq \sqcap \redseqtwo)\redexthree \permle \redseq$
    and
    $(\redseq \sqcap \redseqtwo)\redexthree \permle \redseqtwo$.
    Since $\redseq \sqcap \redseqtwo$ is the greatest lower bound for $\set{\redseq,\redseqtwo}$,
    we have that $(\redseq \sqcap \redseqtwo)\redexthree \permle \redseq \sqcap \redseqtwo$.
    But this implies that $\redexthree \permle \emptyDerivation$, which is a contradiction.
    {\em End of claim.}
    \medskip

    Note that $\redseq \sqcap \redseqtwo \permle \redseq$,
    so we have that $\redseq \permeq (\redseq \sqcap \redseqtwo)(\redseq/(\redseq \sqcap \redseqtwo))$,
    and this in turn implies that
    $\names(\redseq) = \names((\redseq \sqcap \redseqtwo)(\redseq/(\redseq \sqcap \redseqtwo)))$
    by \rcoro{permutation_equivalence_in_terms_of_names}.
    Symmetrically,
    $\names(\redseqtwo) = \names((\redseq \sqcap \redseqtwo)(\redseqtwo/(\redseq \sqcap \redseqtwo)))$.
    Then:
    \[
      \begin{array}{rcll}
        &&  \names(\redseq) \cap \names(\redseqtwo) \\
      & = & \names((\redseq \sqcap \redseqtwo)(\redseq/(\redseq \sqcap \redseqtwo))) \cap \names((\redseq \sqcap \redseqtwo)(\redseqtwo/(\redseq \sqcap \redseqtwo))) \\
      & = & \left(\names(\redseq \sqcap \redseqtwo) \cup \names(\redseq/(\redseq \sqcap \redseqtwo))\right) \cap \left(\names(\redseq \sqcap \redseqtwo) \cup \names(\redseqtwo/(\redseq \sqcap \redseqtwo))\right) \\
         && \text{ by \rremark{names_concatenation_union}} \\
      & = &
            \names(\redseq \sqcap \redseqtwo) \cup 
            \left( \names(\redseq/(\redseq \sqcap \redseqtwo)) \cap \names(\redseqtwo/(\redseq \sqcap \redseqtwo)) \right) \\
      && \text{ since $(A \cup B) \cap (A \cup C) = A \cup (B \cap C)$ for arbitrary sets $A,B,C$} \\
      & = &
            \names(\redseq \sqcap \redseqtwo) \\
      && \text{ since $\left( \names(\redseq/(\redseq \sqcap \redseqtwo)) \cap \names(\redseqtwo/(\redseq \sqcap \redseqtwo)) \right) = \emptyset$ by the previous claim} \\
      \end{array}
    \]
  \end{enumerate}

\item {\bf Monomorphism.}
  It suffices to show that $\names$ is injective.
  Indeed, 
  suppose that $\names(\classof{\redseq}) = \names(\classof{\redseqtwo})$.
  By \rcoro{permutation_equivalence_in_terms_of_names}
  we have that $\redseq \permeq \redseqtwo$,
  so
  $\classof{\redseq} = \classof{\redseqtwo}$.
\end{enumerate}

\subsection{Proof of \rprop{simulation} -- Simulation}
\label{appendix_simulation}

To prove \rprop{simulation} we need an auxiliary lemma.

\begin{lemma}[Refinement of a substitution: composition]
\llem{refinement_substitution}
If $\tm' \refines \tm$ and $\tmtwo'_i \refines \tmtwo$ for all $1 \leq i \leq n$,
then $\subs{\tm'}{\var}{[\tmtwo_1', \hdots, \tmtwo_n']} \refines \subs{\tm}{\var}{\tmtwo}$.
\end{lemma}
\begin{proof}
Straightforward by induction on $\tm$.
\end{proof}
We prove items~1 and~2 of \rprop{simulation} separately:

\subsubsection{Item~1 of \rprop{simulation}: simulation of $\tobeta$ with $\todist$}

Let $\tm' \refines \tm$ and let $\tm = \conof{(\lam{\var}{\tmfive}) \tmsix} \tobeta \conof{\subs{\tmfive}{\var}{\tmsix}} = \tmtwo$.
We prove that there exists a term $\tmtwo'$ such that
$\tm' \todist^* \tmtwo'$ and $\tmtwo' \refines \tmtwo$
by induction on $\con$.
\begin{enumerate}
\item {\bf Empty context, $\con = \conbase$.}
  Then $\tm = (\lam{\var}{\tmfive}) \tmsix \to {\subs{\tmfive}{\var}{\tmsix}} = \tmtwo$
  and $\tm'$ is of the form $(\lamp{\lab}{\var}{\tmfive'}) [\tmsix_1', \hdots, \tmsix_n']$.
  We conclude by taking $\tmtwo' = \subs{\tmfive'}{\var}{[\tmsix_1', \hdots, \tmsix_n']}$, using \rlem{refinement_substitution}.
\item {\bf Under an abstraction, $\con = \lam{\vartwo}{\con'}$.}
  Immediate by \ih.
\item {\bf Left of an application, $\con = \con' \tmthree$.}
  Immediate by \ih.
\item {\bf Right of an application, $\con = \tmthree \con'$.}
Then $\tm = \tmthree\,\contwoof{(\lam{\var}{\tmfive}) \tmsix}
      \tobeta \tmthree\,\contwoof{\subs{\tmfive}{\var}{\tmsix}}
      = \tmtwo$
and $\tm' = \tmthree' [\tmthree'_1, \hdots, \tmthree'_n]$,
with $\tmthree' \refines \tmthree$
and $\tmthree'_i \refines \contwoof{(\lam{\var}{\tmfive}) \tmsix}$
for all $i=1..n$.
Applying the \ih on $\tmthree'_i$ for each $i=1..n$, we have:
\[
\xymatrix@R=.5cm{
 \contwoof{(\lam{\var}{\tmfive}) \tmsix} \ar@{}[d]|*=0[@]{\rtimes} \ar[r]^{\beta}
    & \contwoof{\subs{\tmfive}{\var}{\tmsix}} \ar@{}[d]|*=0[@]{\rtimes} \\
 \tmthree'_i \ar@{->>}[r]^{\dist} & \tmthree''_i \\
}
\]
So we may conclude as follows:
\[
\xymatrix@R=.5cm{
 \tmthree\,\contwoof{(\lam{\var}{\tmfive}) \tmsix}
 \ar@{}[d]|*=0[@]{\rtimes} \ar[r]^{\beta}
  & \tmthree\,\contwoof{\subs{\tmfive}{\var}{\tmsix}}
  \ar@{}[d]|*=0[@]{\rtimes} \\
 \tmthree' [\tmthree'_1, \hdots, \tmthree'_n] \ar@{->>}[r]^{\dist}
    & \tmthree' [\tmthree''_1, \hdots, \tmthree''_n] \\
}
\]
\end{enumerate}

\subsubsection{Item~2 of \rprop{simulation}: simulation of $\todist$ with $\tobeta$}

  Let $\tm' \refines \tm$ and $\tm' \todist \tmtwo'$.
  We prove that there exist terms $\tmtwo$ and $\tmtwo''$
  such that
  $\tm \tobeta \tmtwo$,
  $\tmtwo' \todist^* \tmtwo''$,
  and $\tmtwo'' \refines \tmtwo$
  by induction on $\tm'$.
  Moreover, $\tmtwo' \todist^* \tmtwo''$ is a multistep.
  \begin{enumerate}
  \item {\bf Variable $\tm' = \var$.} Impossible.
  \item {\bf Abstraction, $\tm' = \lamp{\lab}{\vartwo}{\tmfour'}$}
    Immediate by \ih.
  \item {\bf Application, $\tm' = \tmthree' [\tmfour'_1,\hdots,\tmfour'_n]$}
    Then $\tm = \tmthree \tmfour$
    where $\tmthree' \refines \tmthree$
    and $\tmfour'_i \refines \tmfour$ for all $i=1..n$.
    There are three subcases,
    depending on whether the step $\tm' \todist \tmtwo'$
    takes place at the root,
    inside $\tmthree'$,
    or inside $\tmfour'_i$ for some $i=1..n$,
    \begin{enumerate}
    \item {\bf Reduction at the root.}
      In this case $\tmthree' = \lamp{\lab}{\var}{\tmtwo'}$.
      Then by \rlem{refinement_substitution}:
      \[
      \xymatrix@R=.5cm@C=.25cm{
      (\lamp{\lab}{\var}{\tmtwo'}) [\tmfour'_1, \hdots, \tmfour'_n] \ar[d]^{\dist} & \refines & (\lam{\var}{\tmtwo}) \tmfour
      \ar@{->}[d]^{\beta} \\
      \subs{\tmtwo'}{\var}{[\tmfour'_1, \hdots, \tmfour'_n]} & \refines & \subs{\tmtwo}{\var}{\tmfour} \\
      }
      \]
      \item {\bf Reduction in $\tmthree'$.}
        Immediate by \ih.
      \item {\bf Reduction in $\tmfour'_i$ for some $i=1..n$.}
        We have that $\tmfour'_i \todist \tmfour''_{i}$, so by \ih
        there exist $\tmfive'_i$ and $\tmfive$ such that:
      \[
      \xymatrix@R=.5cm@C=.25cm{
      \tmfour'_i \ar[d]^{\dist} & \refines & \tmfour \ar[dd]^{\beta} \\
      \tmfour''_{i} \ar@{->>}[d]^{\dist} & & \\
      \tmfive'_{i} & \refines & \tmfive \\
      }
      \]
      For every $j=1..n$, $j \neq i$,
      we have that $\tmfour'_j \refines \tmfour$ and
          $\tmfour \tobeta \tmfive$.
      By item~1 of Simulation~(\rprop{simulation}),
      which we have just proved, we have the following diagram:
      \[
      \xymatrix@R=.5cm@C=.25cm{
      \tmfour'_j \ar@{->>}[d]^{\dist} & \refines & \tmfour \ar[d]^{\beta}\\
      \tmfive'_{j} & \refines & \tmfive \\
      }
      \]
      So:
      \[
      \xymatrix@R=.5cm@C=.25cm{
      \tmthree'[\tmfour'_1,\hdots,\tmfour'_{j-1},\tmfour'_j,\tmfour'_{j+1},\hdots,\tmfour'_n] \ar[d]^{\dist} & \refines & \tmthree\,\tmfour \ar[dd]^{\beta} \\
      \tmthree'[\tmfour'_1,\hdots,\tmfour'_{j-1},\tmfour''_j,\tmfour'_{j+1},\hdots,\tmfour'_n] \ar@{->>}[d]^{\dist} \\
      \tmthree'[\tmfive'_1,\hdots,\tmfive'_{j-1},\tmfive'_j,\tmfive'_{j+1},\hdots,\tmfive'_n] & \refines & \tmthree\,\tmfive \\
      }
      \]
    \end{enumerate}
  \end{enumerate}

\subsection{Proof of \rprop{refinements_and_head_normal_forms} -- Typability characterizes head normalization}
\label{appendix_refinements_and_head_normal_forms}

Before proving \rprop{refinements_and_head_normal_forms}, we need a few auxiliary results.

\begin{definition}[Head normal forms]
A $\lambda$-term $\tm$ is a \defn{head normal form} if it is of the form
$\tm = \lam{\var_1}{\hdots\lam{\var_n}{\vartwo\,\tm_1\,\hdots\,\tm_m}}$
where $n \geq 0$ and $m \geq 0$.
Note that $\vartwo$ might be or not be among the variables $\var_1,\hdots,\var_n$.
Similarly, a correct term $\tm \in \termsdist$ of $\lambdadist$ a head normal form
if it is of the form
$\tm = \lamp{\lab_1}{\var_1}{\hdots\lamp{\lab_n}{\var_n}{\vartwo^{\typ}\,\ls{\tm}_1\,\hdots\,\ls{\tm}_m}}$
where $n \geq 0$ and $m \geq 0$.
\end{definition}

\begin{lemma}
\llem{hnf_has_refinement}
If $\tm$ is a head normal form,
there exists $\tm' \in \termsdist$ such that $\tm' \refines \tm$.
Moreover $\tm'$ can be taken to be of the form
$\lamp{\lab_1}{\var_1}{\hdots\lamp{\lab_n}{\var_n}{\vartwo^{[\,] \tolab{\lab'_1} \hdots [\,] \tolab{\lab'_m} \alpha^{\lab''}} [\,] \hdots [\,]}}$.
\end{lemma}
\begin{proof}
Let $\tm = \lam{\var_1}{\hdots\lam{\var_n}{\vartwo\tm_1\hdots\tm_m}}$.
Let $\set{\lab_1,\hdots,\lab_n,\lab'_1,\hdots,\lab'_m,\lab''}$
be a set of $n+m+1$ pairwise distinct labels and let $\alpha$ be a base type.
Take $\tm' := \lamp{\lab_1}{\var_1}{\hdots\lamp{\lab_n}{\var_n}{\vartwo^{[\,] \tolab{\lab'_1} \hdots [\,] \tolab{\lab'_m} \alpha^{\lab''}} \underbrace{[\,] \hdots [\,]}_{\text{$m$ times}}}}$.
Observe that $\tm'$ is correct and $\tm' \refines \tm$.
\end{proof}

\begin{lemma}[$\todist$-normal forms refine head normal forms]
\llem{normal_forms_refine_hnfs}
Let $\tm' \in \termsdist$ be a $\todist$-normal form
and $\tm' \refines \tm$. Then $\tm$ is a head normal form.
\end{lemma}
\begin{proof}
Observe, by induction on $\tm'$
that if $\tm' \in \termsdist$ is a $\todist$-normal form, it must be
of the form
$
  \lamp{\lab_1}{\var_1}{\hdots\lamp{\lab_n}{\var_n} \vartwo^\typ \ls{\tmtwo}_1 \hdots \ls{\tmtwo}_m}
$,
as it cannot have a subterm of the form $(\lamp{\lab'}{\varthree}{\tmthree})\ls{\tmfour}$.
Then $\lamp{\lab_1}{\var_1}{\hdots\lamp{\lab_n}{\var_n} \vartwo^\typ \ls{\tmtwo}_1 \hdots \ls{\tmtwo}_m} \refines \tm$.
So $\tm$ is of the form $\lam{\var_1}{\hdots\lam{\var_n}{\vartwo \tmtwo_1 \hdots \tmtwo_n}}$,
that is, $\tm$ is a head normal form.
\end{proof}

The following lemma is an adaptation of Subject~Expansion in~\cite{bucciarelli2017non}.
\begin{lemma}[Subject Expansion]
\llem{subject_expansion}
If $\tctx \vdash \conof{\subs{\tm}{\var}{\ls{\tmtwo}}} : \typ$ is derivable,
then $\tctx \vdash \conof{(\lamp{\lab}{\var}{\tm})\ls{\tmtwo}} : \typ$ is derivable.
\end{lemma}
\begin{proof}
The proof proceeds by induction on $\con$,
and the base case by induction on $\tm$, 
similar to the proof that substitution preserves typing in
the proof of Subject~Reduction (\rsec{substitution_preserves_typing}).
\end{proof}

Correctness is not necessarily preserved by $\todist$-expansion.
We need a stronger invariant, \defn{strong sequentiality},
that will be shown to be preserved by expansion
under appropiate conditions:

\begin{definition}[Subterms and free subterms]
The set of \defn{subterms} $\subterms{\tm}$ of a term $\tm$ is formally defined as follows:
\[
  \begin{array}{rcl}
    \subterms{\var^\typ} & \eqdef & \set{\var^\typ} \\
    \subterms{\lamp{\lab}{\var}{\tm}} & \eqdef & \set{\lamp{\lab}{\var}{\tm}} \cup \subterms{\tm} \\
    \subterms{\tm[\tmtwo_i]_{i=1}^n} & \eqdef & \set{\tm[\tmtwo_i]_{i=1}^n} \cup \subterms{\tm} \cup \cup_{i=1}^{n} \subterms{\tmtwo_i}
  \end{array}
\]
The set of \defn{free subterms} $\fsubterms{\tm}$ of a term $\tm$ is defined similarly,
except for the abstraction case, which requires that the subterm in question
do not include occurrences of bound variables:
\[
  \begin{array}{rcl}
    \fsubterms{\var^\typ} & \eqdef & \set{\var^\typ} \\
    \fsubterms{\lamp{\lab}{\var}{\tm}} & \eqdef & \set{\lamp{\lab}{\var}{\tm}} \cup \set{\tmthree \in \fsubterms{\tm} \ST \var \not\in \fv{\tmthree}} \\
    \fsubterms{\tm[\tmtwo_i]_{i=1}^n} & \eqdef & \set{\tm[\tmtwo_i]_{i=1}^n} \cup \fsubterms{\tm} \cup \cup_{i=1}^{n} \fsubterms{\tmtwo_i}
  \end{array}
\]
\end{definition}

\begin{definition}[Strong sequentiality]
A term $\tm$ is \defn{strongly sequential}
if it is correct and, moreover,
for every subterm $\tmtwo \in \subterms{\tm}$
and any two free subterms $\tmtwo_1,\tmtwo_2 \in \fsubterms{\tmtwo}$ lying at disjoint positions of $\tmtwo$,
the types of $\tmtwo_1$ and $\tmtwo_2$ have different external labels.
\end{definition}

\begin{example}
The following examples illustrate the notion of strong sequentiality:
\begin{enumerate}
\item The term $\tm = (\lamp{1}{\var}{\vartwo^{\alpha^2}})[\,]$ is strongly sequential.
      Note that $\tm$ and $\vartwo^{\alpha^2}$ have the same type, namely $\alpha^2$,
      but they do not occur at {\em disjoint} positions.
\item The term $\tm = (\lamp{1}{\var}{\var^{\alpha^2}})[\vartwo^{\alpha^2}]$ is strongly sequential.
      Note that $\var^{\alpha^2}$ and $\vartwo^{\alpha^2}$ both have type $\alpha^2$,
      but they are not simultaneously free subterms of the same subterm of $\tm$.
\item The term $\tm = \lamp{1}{\vartwo}{ \var^{[\alpha^2] \tolab{3} [\alpha^2] \tolab{4} \beta^5} [\vartwo^{\alpha^2}] [\varthree^{\alpha^2}]}$
      is not strongly sequential,
      since $\vartwo^{\alpha^2}$ and $\varthree^{\alpha^2}$
      have the same type and they are both free subterms of
      $\var^{[\alpha^2] \tolab{3} [\alpha^2] \tolab{4} \beta^5} [\vartwo^{\alpha^2}] [\varthree^{\alpha^2}] \in \subterms{\tm}$.
\end{enumerate}
\end{example}

\begin{lemma}[Refinement of a substitution: decomposition]
\llem{refinement_substitution_inverse}
If $\tmthree' \refines \subs{\tm}{\var}{\tmtwo}$
and $\tmthree'$ is strongly sequential,
then $\tmthree'$ is of the form $\subs{\tm'}{\var}{[\tmtwo'_i]_{i=1}^n}$.
Moreover, given a fresh label~$\lab$,
the term $(\lamp{\lab}{\var}{\tm'})[\tmtwo'_i]_{i=1}^{n}$ is strongly sequential,
$\tm' \refines \tm$ and $\tmtwo'_i \refines \tmtwo$ for all $i=1..n$.
\end{lemma}
\begin{proof}
By induction on $\tm$.
\begin{enumerate}
\item {\bf Variable (same), $\tm = \var$.}
  Then $\tmthree' \refines \tmtwo$.
  Let $\typ$ be the type of $\tmthree'$.
  Taking $\tm' := \var^\typ \refines \var$
  we have that $(\lamp{\lab}{\var}{\var^\typ})[\tmthree']$ is strongly sequential.
  Regarding strong sequentiality,
  observe that $\var^\typ$ and $\tmthree'$ have the same type, but they are not
  simultaneously the free subterms of any subterm of $(\lamp{\lab}{\var}{\var^\typ})[\tmthree']$.
\item {\bf Variable (different), $\tm = \vartwo \neq \var$.}
  Then $\tmthree' \refines \vartwo$, so $\tmthree'$ is of the form $\vartwo^\typ$.
  Taking $\tm' := \vartwo^\typ$ we have that
  $(\lamp{\lab}{\var}{\vartwo^\typ})[\,]$ is strongly sequential.
\item {\bf Abstraction, $\tm = \lam{\vartwo}{\tmfour}$.}
  Then $\tmthree' \refines \lam{\vartwo}{\subs{\tmfour}{\var}{\tmtwo}}$ 
  so $\tmthree'$ is of the form $\lamp{\lab'}{\vartwo}{\tmthree''}$
  where $\tmthree'' \refines \subs{\tmfour}{\var}{\tmtwo}$.
  By \ih, $\tmthree''$ is of the form $\subs{\tmfour'}{\var}{[\tmtwo'_i]_{i=1}^{n}}$
  where $(\lamp{\lab}{\var}{\tmfour'})[\tmtwo'_i]_{i=1}^{n}$ is strongly sequential,
  $\tmfour' \refines \tmfour$ and $\tmtwo'_i \refines \tmtwo$ for all $i=1..n$.
  Taking $\tm' := \lamp{\lab'}{\vartwo}{\tmfour'}$,
  we have that
  $\tm' = \lamp{\lab'}{\vartwo}{\tmfour'} \refines \lam{\vartwo}{\tmfour} = \tm$.
  Moreover, the term
  $(\lamp{\lab}{\var}{\tm'})[\tmtwo_i]_{i=1}^{n}
   = (\lamp{\lab}{\var}{\lamp{\labtwo}{\vartwo}{\tmfour'}})[\tmtwo_i]_{i=1}^{n}$
  is strongly sequential.
  Typability is a consequence of Subject Expansion~(\rlem{subject_expansion}).
  The remaining properties are:
  \begin{enumerate}
  \item {\em Uniquely labeled lambdas.}
    The multiset of labels decorating the lambdas of
    $(\lamp{\lab}{\var}{\lamp{\labtwo}{\vartwo}{\tmfour'}})[\tmtwo_i]_{i=1}^{n}$
    is given by
    $\Lambda((\lamp{\lab}{\var}{\lamp{\labtwo}{\vartwo}{\tmfour'}})[\tmtwo_i]_{i=1}^{n})
     = [\lab,\labtwo] + \Lambda(\tmfour') +_{i=1}^{n} \Lambda(\tmtwo_i)$.
    It suffices to check that this multiset has no repeats.
    The label $\lab$ is assumed to be fresh, so it occurs only once.
    By \ih, $\tmthree'' = \subs{\tmfour'}{\var}{[\tmtwo'_i]_{i=1}^{n}}$,
    so using \rlem{labels_over_lambdas_substitution} we have
    $\Lambda(\tmthree') =
     \Lambda(\lamp{\lab'}{\vartwo}{\tmthree'}) =
     [\lab'] + \Lambda(\tmthree'') =
     [\lab'] + \Lambda(\subs{\tmfour'}{\var}{[\tmtwo'_i]_{i=1}^{n}}) =
     [\lab'] + \Lambda(\tmfour') +_{i=1}^n \Lambda(\tmtwo'_i)$.
    Moreover, the term $\tmthree' = \lamp{\lab'}{\vartwo}{\tmthree''}$ is correct,
    so this multiset has no repeats.
  \item {\em Sequential contexts.}
    Let $\tmsix$ be a subterm of $(\lamp{\lab}{\var}{\lamp{\labtwo}{\vartwo}{\tmfour'}})[\tmtwo'_i]_{i=1}^{n}$.
    If $\tmsix$ is a subterm of $\tmfour'$ or a subterm of $\tmtwo'_i$ for some $i=1..n$ we conclude by \ih
    since $(\lamp{\lab}{\var}{\tmfour'})[\tmtwo'_i]_{i=1}^{n}$ is known to be correct.
    Moreover, if $\tctx \oplus \var : \mtyp \oplus \vartwo : \mtyptwo$ is the typing context for $\tmfour'$,
    the typing contexts of
    $\lamp{\labtwo}{\vartwo}{\tmfour'}$ and $\lamp{\lab}{\var}{\lamp{\labtwo}{\vartwo}{\tmfour'}}$
    are respectively $\tctx \oplus \var : \mtyp$ and $\tctx$, which are also sequential.
    Finally, if $\tctxtwo_i$ is the typing context for $\tmtwo'_i$, for each $i=1..n$,
    the typing context for $(\lamp{\lab}{\var}{\tmfour'})[\tmtwo'_i]_{i=1}^{n}$ 
    is of the form $\tctx +_{i=1}^n \tctxtwo_i + \vartwo : \mtyptwo$, and it is sequential by \ih.
    Hence the typing context for the whole term is $\tctx +_{i=1}^n \tctxtwo_i$, and it is sequential.
  \item {\em Sequential types.}
    \label{refinement_substitution_inverse__case_lambda_sequential_types}
    Let $\tmsix$ be a subterm $(\lamp{\lab}{\var}{\lamp{\labtwo}{\vartwo}{\tmfour'}})[\tmtwo'_i]_{i=1}^{n}$,
    and let $\mtypthree \tolab{\lab''} \typthree$ be a type that occurs in the typing context
    or the type of $\tmsix$.
    As in the previous case, we have that
    $\tctx \oplus \var : [\typtwo_i]_{i=1}^n \oplus \vartwo : \mtyptwo \vdash \tmfour' : \typ$ is derivable
    and $\tctxtwo_i \vdash \tmtwo'_i : \typtwo_i$ is derivable for all $i=1..n$.
    Moreover, they are correct by \ih, so if $\tmsix$ is a subterm of $\tmfour'$ or a subterm of some $\tmtwo'_i$,
    we are done.
    There are three cases left for $\tmsix$:
    \begin{enumerate}
    \item {\bf Case $\tmsix = \lamp{\labtwo}{\vartwo}{\tmfour'}$.}
      The typing context is $\tctx \oplus \var : [\typtwo_i]_{i=1}^n$ and the type $\mtyptwo \tolab{\labtwo} \typ$.
    \item {\bf Case $\tmsix = \lamp{\lab}{\var}{\lamp{\labtwo}{\vartwo}{\tmfour'}}$.}
      The typing context is $\tctx$ and the type $[\typtwo_i]_{i=1}^n \tolab{\lab} \mtyptwo \tolab{\labtwo} \typ$.
    \item {\bf Case $\tmsix = (\lamp{\lab}{\var}{\lamp{\labtwo}{\vartwo}{\tmfour'}})[\tmtwo_i]_{i=1}^{n}$.}
      The typing context is $\tctx$ and the type $\mtyptwo \tolab{\labtwo} \typ$.
    \end{enumerate}
    In all three cases, if $\mtypthree \tolab{\lab''} \typthree$ occurs in the typing context
    or the type of $\tmsix$, then $\mtypthree$ can be shown to be sequential using the \ih.
  \item {\em Strong sequentiality.}
    \label{refinement_substitution_inverse__case_lambda_strong_sequentiality}
    Let $\tmsix \in \subterms{(\lamp{\lab}{\var}{\lamp{\labtwo}{\vartwo}{\tmfour'}})[\tmtwo_i]_{i=1}^{n}}$
    be a subterm,
    and let $\tmsix_1,\tmsix_2 \in \fsubterms{\tmsix}$ be free subterms lying at disjoint positions of $\tmsix$.
    We argue that the types of $\tmsix_1$ and $\tmsix_2$ have different external labels.
    Consider the following five possibilities for $\tmsix_1$:
    \begin{enumerate}
    \item {\bf Case $\tmsix_1 = (\lamp{\lab}{\var}{\lamp{\labtwo}{\vartwo}{\tmfour'}})[\tmtwo_i]_{i=1}^{n}$.}
      Impossible since $\tmsix_2$ must be at a disjoint position.
    \item {\bf Case $\tmsix_1 = \lamp{\lab}{\var}{\lamp{\labtwo}{\vartwo}{\tmfour'}}$.}
      Then the external label of the type of $\tmsix_1$ is the label $\lab$, which is fresh,
      so it cannot coincide with the type of any other subterm.
    \item {\bf Case $\tmsix_1 = \lamp{\labtwo}{\vartwo}{\tmfour'}$.}
      Then $\tmsix_2$ must be a subterm of $\tmtwo_i$ for some $i=1..n$.
      Note that $\tmsix$ must be the whole term,
      and $n > 0$, so there is at least one free occurrence of $\var$ in $\lamp{\labtwo}{\vartwo}{\tmfour'}$.
      This means that $\tmsix_1 \not\in \fsubterms{\tmsix}$, so this case is impossible.
    \item {\bf Case $\tmsix_1$ is a subterm of $\tmfour'$.}
      If $\tmsix_2$ is also a subterm of $\tmfour'$, we conclude since by \ih
      $(\lamp{\lab}{\var}{\tmfour'})[\tmtwo'_i]_{i=1}^{n}$ is strongly sequential.
      Otherwise, $\tmsix_2$ is a subterm of $\tmtwo_i$ for some $i=1..n$,
      and we also conclude by \ih.
    \item {\bf Case $\tmsix_1$ is a subterm of $\tmtwo_i$ for some $i=1..n$.}
      If $\tmsix_2$ is a subterm of $\tmtwo_j$ for some $j=1..n$,
      we conclude since by \ih
      $(\lamp{\lab}{\var}{\tmfour'})[\tmtwo'_i]_{i=1}^{n}$ is strongly sequential.
      If $\tmsix_2$ is any other subterm, note that the symmetric case has already been considered
      in one of the previous cases.
    \end{enumerate}
  \end{enumerate}
\item {\bf Application, $\tm = \tmfour\tmfive$.}
  Then $\tmthree' \refines \subs{(\tmfour\tmfive)}{\var}{\tmtwo}$,
  so it is of the form $\tmthree'_0[\tmthree'_j]_{j=1}^m$
  where $\tmthree'_0 \refines \subs{\tmfour}{\var}{\tmtwo}$
  and $\tmthree'_j \refines \subs{\tmfive}{\var}{\tmtwo}$ for all $j=1..m$.
  By \ih we have that $\tmthree'_0$ is of the form
  $\subs{\tmfour'}{\var}{\ls{\tmtwo}_0}$
  and for all $j=1..m$ the term $\tmthree'_j$ is of the form
  $\subs{\tmfive'_j}{\var}{\ls{\tmtwo}_j}$,
  where $\tmfour' \refines \tmfour$
  and $\tmfive'_j \refines \tmfive$ for all $j=1..m$.
  Moreover,
  the length of the list $\ls{\tmtwo}_j$ is $k_j$
  and
  $\ls{\tmtwo}_j = [\tmtwo^{(j)}_i]_{i=1}^{k_j}$ for all $j=0..m$,
  and we have that $\tmtwo^{(j)}_i \refines \tmtwo$ for all $j=0..m,i=1..k_j$.
  By \ih we also know that
  $(\lamp{\lab}{\var}{\tmfour'})\ls{\tmtwo}_0$ is strongly sequential
  and $(\lamp{\lab}{\var}{\tmfive'_j})\ls{\tmtwo}_j$ is strongly sequential for all $j=1..m$.

  Let $\tm' := \tmfour'[\tmfive'_j]_{j=1}^{m}$,
  let $n := \sum_{j=0}^m k_j$,
  and let $[\tmtwo'_1,\hdots,\tmtwo'_n] := \sum_{j=0}^m \ls{\tmtwo}_j$.
  Note that $\tm' = \tmfour'[\tmfive'_j]_{j=1}^{m} \refines \tmfour\tmfive = \tm$
  and $\tmtwo'_i \refines \tmtwo$ for all $i=1..n$.

  Moreover, we have to check that $\tmthree' = \subs{\tm'}{\var}{[\tmtwo'_1,\hdots,\tmtwo'_n]}$.
  To prove this,
  note that $\subs{\tm'}{\var}{[\tmtwo'_1,\hdots,\tmtwo'_n]} = \subs{(\tmfour'[\tmfive'_j]_{j=1}^{m})}{\var}{[\tmtwo'_1,\hdots,\tmtwo'_n]}$.
  Suppose that the multiset $\tmlabel{[\tmtwo'_1,\hdots,\tmtwo'_n]}$ were sequential.
  Then the list of terms $[\tmtwo'_1,\hdots,\tmtwo'_n]$ would be partitioned
  as $(\ls{\tmthree}_0,\hdots,\ls{\tmthree}_m)$
  where $\ls{\tmthree}_j$ is a permutation of $\ls{\tmtwo}_j$ for all $j=0..m$,
  and we would have indeed:
  \[
    \begin{array}{rcll}
      \subs{\tm'}{\var}{[\tmtwo'_1,\hdots,\tmtwo'_n]}
    & = &
      \subs{(\tmfour'[\tmfive'_j]_{j=1}^{m})}{\var}{[\tmtwo'_1,\hdots,\tmtwo'_n]} \\
    & = & 
      \subs{\tmfour'}{\var}{\ls{\tmthree}_0}[\subs{\tmfive'_j}{\var}{\ls{\tmthree}_j}]_{j=1}^{m})
      \hfill\text{ by $\tmlabel{[\tmtwo'_1,\hdots,\tmtwo'_n]}$ sequential} \\
    & = & 
      \subs{\tmfour'}{\var}{\ls{\tmtwo}_0}[\subs{\tmfive'_j}{\var}{\ls{\tmtwo}_j}]_{j=1}^{m})
   \\
    & = & 
      \subs{\tmfour'}{\var}{\ls{\tmtwo}_0}[\subs{\tmfive'_j}{\var}{\ls{\tmtwo}_j}]_{j=1}^{m}) \hfill\text{by \rlem{substitution_permutation}} \\
    & = & 
      \tmthree'_0[\tmthree'_j]_{j=1}^{m} = \tmthree' \hfill\text{by \ih}
    \end{array}
  \]
  To see that $\tmlabel{[\tmtwo'_1,\hdots,\tmtwo'_n]}$ is sequential,
  note that for every $i \neq j$,
  the terms $\tmtwo'_i$ and $\tmtwo'_j$ are free subterms of $\tmthree'$
  and they lie at disjoint positions of $\tmthree'$.
  Since $\tmthree'$ is strongly sequential, the types of $\tmtwo'_i$ and $\tmtwo'_j$
  have different external labels. Hence $\tmlabel{[\tmtwo'_1,\hdots,\tmtwo'_n]}$ is sequential.

  To conclude, we are left to check that $(\lamp{\lab}{\var}{\tm'})[\tmtwo'_1,\hdots,\tmtwo'_n]$ is strongly sequential:
  \begin{enumerate}
  \item {\em Uniquely labeled lambdas.}
    The multiset of labels decorating the lambdas of $(\lamp{\lab}{\var}{\tm'})[\tmtwo'_1,\hdots,\tmtwo'_n]$
    is given by
    $\Lambda((\lamp{\lab}{\var}{\tm'})[\tmtwo'_1,\hdots,\tmtwo'_n]) =
     [\lab] + \Lambda(\tm') +_{i=1}^{n} \Lambda(\tmtwo'_i)$.
    It suffices to check that this multiset has no repeats.
    The label $\lab$ is assumed to be fresh, so it occurs only once.
    We have already argued that
    $\tmthree' = \subs{\tm'}{\var}{[\tmtwo'_1,\hdots,\tmtwo'_n]}$,
    so using \rlem{labels_over_lambdas_substitution} we have
    $\Lambda(\tmthree') = \Lambda(\tm') +_{i=1}^n \Lambda(\tmtwo'_i)$.
    Moreover $\tmthree'$ is correct, so this multiset has no repeats.
  \item {\em Sequential contexts.}
    Suppose that
    $\tctx_0 \oplus \var : \mtyp_0 \vdash \tmfour' : [\typthree_j]_{j=1}^m \tolab{\labtwo} \typ$
    is derivable,
    $\tctx_j \oplus \var : \mtyp_j \vdash \tmfive'_j : \typthree_j$
    is derivable for all $j=1..m$,
    and $\tctxtwo_i \vdash \tmtwo'_i : \typtwo_i$ is derivable for all $i=1..n$.
    Note that $\sum_{j=0}^{m} \mtyp_j = [\typtwo_i]_{i=1}^{n}$.

    Let $\tmsix$ be a subterm of
    $(\lamp{\lab}{\var}{\tmfour'[\tmfive'_j]_{j=1}^{m}})[\tmtwo'_1,\hdots,\tmtwo'_n]$.
    Consider four cases for $\tmsix$:

    \begin{enumerate}
    \item {\bf Case $\tmsix = (\lamp{\lab}{\var}{\tmfour'[\tmfive'_j]_{j=1}^{m}})[\tmtwo'_1,\hdots,\tmtwo'_n]$.}
      The typing context is $\sum_{j=0}^{m} \tctx_j + \sum_{i=1}^n \tctxtwo_i$.
      By Subject~Expansion~(\rlem{subject_expansion})
      the typing context of $\subs{(\tmfour'[\tmfive'_j]_{j=1}^{m})}{\var}{[\tmtwo'_1,\hdots,\tmtwo'_n]} = \tmthree'$
      is also $\sum_{j=0}^{m} \tctx_j + \sum_{i=1}^n \tctxtwo_i$
      and $\tmthree'$ is correct by hypothesis.
      So $\sum_{j=0}^{m} \tctx_j + \sum_{i=1}^n \tctxtwo_i$ is sequential.
    \item {\bf Case $\tmsix = \lamp{\lab}{\var}{\tmfour'[\tmfive'_j]_{j=1}^{m}}$.}
      The typing context is
      $\sum_{j=0}^{m} \tctx_j$,
      which is sequential because $\sum_{j=0}^{m} \tctx_j + \sum_{i=1}^n \tctxtwo_i$ is sequential.
    \item {\bf Case $\tmsix = \tmfour'[\tmfive'_j]_{j=1}^{m}$.}
      The typing context is
      $\sum_{j=0}^{m} \tctx_j \oplus \var : [\typtwo_i]_{i=1}^n$,
      which is sequential because $\sum_{j=0}^{m} \tctx_j$ is sequential
      and, moreover, $[\typtwo_i]_{i=1}^n = \tmlabel{[\typtwo'_1,\hdots,\typtwo'_n]}$
      which we have already shown to be sequential.
    \item {\bf Otherwise.}
      Then $\tmsix$ is a subterm of $\tmfour'$, a subterm of $\tmfive'_j$ for some $j=1..m$,
      or a subterm of some $\tmtwo'_i$ for some $i=1..n$.
      Then we conclude since by \ih $(\lamp{\lab}{\var}{\tmfour'})\ls{\tmtwo}_0$
      and all the $(\lamp{\lab}{\var}{\tmfive'_j})\ls{\tmtwo}_j$
      are strongly sequential.
    \end{enumerate}
  \item {\em Sequential types.}
    Let $\tmsix$ be a subterm of $(\lamp{\lab}{\var}{\tmfour'[\tmfive'_j]_{j=1}^{m}})[\tmtwo'_1,\hdots,\tmtwo'_n]$.
    We claim that if $\mtyptwo \tolab{\lab''} \typfour$ occurs in the context or in the type of $\tmsix$,
    then $\mtyptwo$ is sequential.
    The proof is similar as for subcase \label{refinement_substitution_inverse__case_lambda_sequential_types}.
  \item {\em Strong sequentiality.}
    Let $\tmsix \in \subterms{(\lamp{\lab}{\var}{\tmfour'[\tmfive'_j]_{j=1}^{m}})[\tmtwo'_1,\hdots,\tmtwo'_n]}$
    be a subterm,
    and let $\tmsix_1,\tmsix_2 \in \fsubterms{\tmsix}$ be free subterms lying at disjoint positions of $\tmsix$.
    We claim that the types of $\tmsix_1$ and $\tmsix_2$ have different external labels.
    The proof is similar as for subcase \label{refinement_substitution_inverse__case_lambda_strong_sequentiality}.
  \end{enumerate}
\end{enumerate}
\end{proof}

\begin{lemma}[Backwards Simulation]
\llem{backwards_simulation}
Let $\tm,\tmtwo \in \terms$ be $\lambda$-terms and
let $\tmtwo' \in \termsdist$ be a strongly sequential term such that
$\tm \tobeta \tmtwo$ and $\tmtwo' \refines \tmtwo$.
Then there exists a strongly sequential term $\tm' \in \termsdist$ such that:
\[
\xymatrix@R=.25cm{
 \tm \ar@{}[d]|*=0[@]{\rtimes} \ar[r]^{\beta} & \tmtwo \ar@{}[d]|*=0[@]{\rtimes} \\
 \tm' \ar@{.>>}[r]^{\dist} & \tmtwo' \\
}
\]
\end{lemma}
\begin{proof}
Let $\tm = \conof{(\lam{\var}{\tmthree})\tmfour} \tobeta \conof{\subs{\tmthree}{\var}{\tmfour}} = \tmtwo$.
The proof proceeds by induction on $\con$.
\begin{enumerate}
\item {\bf Empty, $\con = \conbase$.}
  By \rlem{refinement_substitution_inverse} we have that $\tmtwo'$
  is of the form $\subs{\tmthree'}{\var}{[\tmfour'_1,\hdots,\tmfour'_n]}$
  where $\tmthree' \refines \tmthree$ and $\tmfour'_i \refines \tmfour$ for all $i=1..n$.
  Moreover, taking $\lab$ to be a fresh label, $(\lamp{\lab}{\var}{\tmthree'})[\tmfour'_1,\hdots,\tmfour'_n]$ 
  is strongly sequential and
  $(\lamp{\lab}{\var}{\tmthree'})[\tmfour'_1,\hdots,\tmfour'_n] \refines (\lam{\var}{\tmthree})\tmfour$.
\item {\bf Under an abstraction, $\con = \lam{\var}{\con'}$.}
  Straightforward by \ih.
\item {\bf Left of an application, $\con = \con'\,\tmthree$}
  Straightforward by \ih.
\item {\bf Right of an application, $\con = \tmthree\,\con'$}
  Then $\tm = \tmthree\,\tmfour \tobeta \tmthree\,\tmfive = \tmtwo$
  where $\tmfour \tobeta \tmfive$ and $\tmtwo' \refines \tmtwo$.
  Then $\tmtwo'$ is of the form $\tmthree'[\tmfive'_1,\hdots,\tmfive'_n]$ where $\tmfive'_i \refines \tmfive$ for all $i=1..n$.
  By \ih, for all $i=1..n$ we have that there exist $\tmfour'_1,\hdots,\tmfour'_n$ such that:
  \[
  \xymatrix@R=.25cm{
   \tmfour \ar@{}[d]|*=0[@]{\rtimes} \ar[r]^{\beta} & \tmfive \ar@{}[d]|*=0[@]{\rtimes} \\
   \tmfour'_i \ar@{.>>}[r]^{\dist} & \tmfive'_i \\
  }
  \HS\HS\text{So we have:}
  \xymatrix@R=.25cm{
   \tmthree\,\tmfour \ar@{}[d]|*=0[@]{\rtimes} \ar[r]^{\beta} & \tmthree\tmfive \ar@{}[d]|*=0[@]{\rtimes} \\
   \tmthree'[\tmfour'_i]_{i=1}^{n} \ar@{.>>}[r]^{\dist} & \tmthree'[\tmfive'_i]_{i=1}^{n}  \\
  }
  \]
  Moreover, $\tmthree'[\tmfour'_i]_{i=1}^{n}$ is strongly sequential, which can be concluded
  from the facts that $\tmthree'[\tmfive'_i]_{i=1}^{n}$ is strongly sequential by hypothesis,
  $\tmfour'_i$ is strongly sequential for all $i=1..n$ by \ih,
  and $\tmfour'_i$ and $\tmfive'_i$ have the same types by Subject~Expansion~(\rlem{subject_expansion}).
\end{enumerate}
\end{proof}

To prove the equivalence $(1 \iff 2)$ of \rprop{refinements_and_head_normal_forms}, we claim that the following
three statements are equivalent:
\begin{enumerate}
\item There exists $\tm' \in \termsdist$ such that $\tm' \refines \tm$.
\item[1'.]
      There exists $\tm' \in \termsdist$ such that $\tm' \refines \tm$
      and $\tm' \todist^* \lamp{\lab_1}{\var_1}{\hdots\lamp{\lab_n}{\var_n}{\vartwo^\typ[]\hdots[]}}$.
\item There exists a head normal form $\tmtwo$ such that $\tm \tobeta^* \tmtwo$.
\end{enumerate}
Let us prove the chain of implications $1 \implies 2 \implies 1' \implies 1$:
\begin{itemize}
\item $(1 \implies 2)$
  Let $\tm' \refines \tm$.
  By Strong Normalization~(\rprop{strong_normalization}), reduce $\tm' \todist^* \tmtwo'$
  to normal form.
  We claim that there exists a term $\tmtwo$ such that $\tm \tobeta^* \tmtwo$ and $\tmtwo' \refines \tmtwo$.
  Observe that, since $\lambdadist$ is strongly normalizing~(\rprop{strong_normalization})
  and finitely branching, K\"onig's lemma ensures that there is a bound for the length of $\todist$-derivations
  going out from a term $\tm' \in \termsdist$.
  (Alternatively, according to the proof of Strong Normalization in~\rprop{strong_normalization},
  the bound may be explicitly taken to be the number of lambdas in $\tm'$).
  Call this bound the \defn{weight} of $\tm'$.

  We proceed by induction on the weight of $\tm'$.
  If the derivation is empty, we are done by taking $\tmtwo := \tm$.
  If the derivation is non-empty, it is of the form $\tm' \todist \tmthree' \todist^* \tmtwo'$.
  By Simulation~(\rprop{simulation}) there exist terms $\tmthree$ and $\tmthree''$
  such that $\tmthree' \todist^* \tmthree'' \refines \tmthree$
  and $\tm \tobeta \tmthree$. Since the $\lambdadist$-calculus is
  confluent~(\rprop{strong_permutation}) and $\tmtwo'$ is a normal form,
  we have that $\tmthree'' \rtodist \tmtwo'$.
  Note that the weight of $\tm'$ is strictly larger than the weight of $\tmthree''$,
  so by \ih there exists $\tmtwo$ such that $\tmthree \tobeta^* \tmtwo$
  and $\tmtwo' \refines \tmtwo$:
  \[
    \xymatrix@R=.5cm{
      \tm
        \ar@{}[d]|*=0[@]{\rtimes} 
        \ar[rr]^{\beta}
    &
    &
      \tmthree
        \ar@{}[d]|*=0[@]{\rtimes} 
        \ar@{->>}[r]^{\beta}
    &
      \tmtwo
        \ar@{}[d]|*=0[@]{\rtimes} 
    \\
      \tm'
        \ar[r]^{\dist}
    &
      \tmthree'
        \ar@{->>}[r]^{\dist}
        \ar@/_.75cm/@{->>}[rr]^{\dist}
    &
      \tmthree''
        \ar@{->>}[r]^{\dist}
    &
      \tmtwo'
    }
  \]
  Finally, since $\tmtwo'$ is a $\todist$-normal form and $\tmtwo' \refines \tmtwo$,
  \rlem{normal_forms_refine_hnfs}
  ensures that $\tmtwo$ is a head normal form, as required.
\item $(1' \implies 1)$ Obvious.
\item $(2 \implies 1')$
  Let $\tm \tobeta^* \tmtwo$ be a derivation to head normal form.
  We claim that there exists $\tm' \in \termsdist$ such that $\tm'$ is strongly sequential,
  and the normal form of $\tm'$ is of the form $\lamp{\lab_1}{\var_1}{\hdots\lamp{\lab_n}{\var_n}{\vartwo^\typ[]\hdots[]}}$.
  By induction on the length of the derivation $\tm \tobeta^* \tmtwo$.
  If the derivation is empty, $\tm = \tmtwo$ is a head normal form
  and we conclude by \rlem{hnf_has_refinement}, observing that the constructed term $\tm' \refines \tm$
  is strongly sequential.
  If the derivation is non-empty, conclude using the \ih and Backwards Simulation~(\rlem{backwards_simulation}).
\end{itemize}

\subsection{Proof of \rprop{compatibility_of_simulation_residuals_and_permutation_equivalence} -- Compatibility}
\label{appendix_compatibility_of_simulation_residuals_and_permutation_equivalence}

We need a few auxiliary results.

\begin{lemma}[Simulation residuals and composition]
\llem{simulation_residuals_composition}
If $\redseq,\redseqtwo$ are composable derivations and $\tm' \refines \src(\redseq)$,
then:
\begin{enumerate}
\item $\tm'/\redseq\redseqtwo = (\tm'/\redseq)/\redseqtwo$
\item $\redseq\redseqtwo/\tm' = (\redseq/\tm')(\redseqtwo/(\tm'/\redseq))$
\end{enumerate}
\end{lemma}
\begin{proof}
Straightforward by induction on $\redseq$.
\end{proof}
An $n$-hole context $\con'$ is a term with $n$ ocurrences of a hole $\conbase$.
We extend refinement for contexts by setting $\conbase \refines \conbase$.
We write $\conof{\tm_1,\hdots,\tm_n}$ for the capturing substitution of the
$i$-th occurrence of $\conbase$ in $\con$ by $\tm_i$ for all $i=1..n$.
If $\redex : \tm \tobeta \tmtwo$ is a step, we write $\conof{\redex}$
for the step $\conof{\redex} : \conof{\tm} \tobeta \conof{\tmtwo}$.
Note that in general a context $\con$ may be refined by an $n$-hole context,
for example a $0$-hole context, $\var^{[] \tolab{1} \alpha^2}[\,] \refines \var(\vartwo\,\conbase)$,
or a $2$-hole context, $\var^{[\alpha^1,\beta^2] \tolab{3} \gamma^4}[\conbase,\conbase] \refines \var\conbase$.

\begin{lemma}[Simulation of contexts]
\llem{simulation_of_context}
Let $\contwo \refines \con$,
let $\tm'_1, \hdots, \tm'_n \in \termsdist$ and $\tm \in \terms$
such that $\tm'_i \refines \tm'$ for each $i \in \set{1, \hdots, n}$,
and $\contwoof{\tm'_1, \hdots, \tm'_n} \refines \conof{\tm}$.
Moreover, let $\redex : \tm \tobeta \tmtwo$.
Then $\contwoof{\tm'_1, \hdots, \tm'_n} / \conof{\redex} = \contwoof{\tm_1 / \redex, \hdots, \tm_n / \redex}$,
and $\names(\conof{\redex} / \contwoof{\tm'_1, \hdots, \tm'_n}) = \sum_{i=1}^n \names(\redex / \tm'_i)$.
\end{lemma}
\begin{proof}
Straightforward by induction on $\con$.
\end{proof}

In the proof of the following lemma~(\rlem{cube_lemma_for_simulation_residuals}),
sometimes we will use the previous lemma~(\rlem{simulation_of_context}) implicitly.

\begin{lemma}[Basic Cube Lemma for simulation residuals]
\llem{cube_lemma_for_simulation_residuals}
Let $\redex : \tm \tobeta \tmtwo$ and $\redextwo : \tm \tobeta \tmthree$ be coinitial steps,
and let $\tm' \in \termsdist$ be a correct term such that $\tm' \refines \tm$.
Then the following equality between sets of coinitial steps holds:
$
  (\redex/\tm')/(\redextwo/\tm') = (\redex/\redextwo)/(\tm'/\redextwo)
$.
\end{lemma}
\begin{proof}
If $\redex = \redextwo$ then it is easy to see that the proposition holds, so we can assume that $\redex \neq \redextwo$.
  Also, note that it is enough to see that $\names((\redex/\tm')/(\redextwo/\tm')) = \names(((\redex/\redextwo)/(\tm'/\redextwo))$, as we will do that in some cases.
We proceed by induction on $\tm$.
\begin{enumerate}
  \item {\bf Variable $\tm = \var$.} Impossible.
  \item {\bf Abstraction, $\tm = \lam{\var}{\tmthree}$.}
    Immediate by \ih.
  \item {\bf Application, $\tm = \tmfive \tmsix$.}
    Three cases, depending on the position of $\redex$.
    \begin{enumerate}
      \item {\bf If $\redex$ is at the root.}
        Then $\tm = (\lam{\var}{\tmfour}) \tmsix$, and $\redex : (\lam{\var}{\tmfour}) \tmsix \to \subs{\tmfour}{\var}{\tmsix}$.
        Then $\redextwo$ may be inside $\tmfour$ or inside $\tmsix$. In any case,
        the situation is:
        \[
         \xymatrix@R=.25cm{
           (\lam{\var}{\tmfour}) \tmsix
             \ar@{->}[r]^{\redex} \ar@{->}[d]^{\redextwo}
               & \subs{\tmfour}{\var}{\tmsix} \\
           (\lam{\var}{\tmfour^\circ}) \tmsix^\circ
         }
         \xymatrix@R=.25cm{
           (\lamp{\lab}{\var}{\tmfour'}) \ls{\tmsix}
             \ar@{->}[r]^-{\redex / \tm'} \ar@{->>}[d]^{\redextwo / \tm'}
             & \subs{\tmfour'}{\var}{\ls{\tmsix}} \\
           (\lamp{\lab}{\var}{\tmfour'^\circ}) \ls{\tmsix}^\circ
         }
         \]
         Note that $(\redex / \tm') / (\redextwo / \tm')$ has only one element $\redex_1$,
         namely the step contracting the lambda labeled with $\lab$.
         Note that $\redex / \redextwo$ also happens to have only one element,
         $\redex_2 : (\lam{\var}{\tmfour^\circ}) \tmsix^\circ \tobeta \subs{\tmfour^\circ}{\var}{\tmsix^\circ}$.
         It is easy to see that $\redex_2/((\lamp{\lab}{\var}{\tmfour'}) \ls{\tmsix}) = \set{\redex_1}$, as required.
      \item {\bf If $\redex$ is in $\tmfive$.}
        We consider three subcases, depending on the position of $\redextwo$.
        \begin{enumerate}
          \item {\bf If $\redextwo$ is at the root.}
            Then $\tmfive = \lam{x}{\conof{(\lam{\vartwo}{\tmthree}) \tmthreevariant}}$ and
            (for appropiate $\con^\circ,\tmthree^\circ,\tmthreevariant^\circ$):
            \[
            \xymatrix@R=.25cm@C=1cm{
              (\lam{\var}{\conof{(\lam{\vartwo}{\tmthree}) \tmthreevariant}}) \tmsix
                \ar@{->}[r]^{\redex} \ar@{->}[d]^{\redextwo} %\ar@{}[dd]|*=0[@]{\rtimes}
                  & (\lam{\var}{\conof{\subs{\tmthree}{\vartwo}{\tmthreevariant}}}) \tmsix & \\
              \con^{\circ}\of{(\lam{\vartwo}{\tmthree^{\circ}}) \tmthreevariant^{\circ}}
                  \ar@{->}[r]^{\redex / \redextwo}
                  & \con^{\circ}\of{\subs{\tmthree^{\circ}}{\vartwo}{\tmthreevariant^{\circ}}} \\
            }
            \]
            \[
            \xymatrix@R=.25cm@C=1cm{
              (\lamp{\lab}{\var}{\con'\of{(\lamp{\lab_1}{\vartwo}{\tmthree_1}) \tmthreevariant_1, \hdots, (\lamp{\lab_n}{\vartwo}{\tmthree_n}) \tmthreevariant_n}}) \ls{\tmsix}
                \ar@{->>}[r]^{\redex / \tm'} \ar@{->}[d]^{\redextwo / \tm'}
                  & (\lamp{\lab}{\var}{\con'\of{\subs{\tmthree_1}{\vartwo}{\tmthreevariant_1}, \hdots, \subs{\tmthree_n}{\vartwo}{\tmthreevariant_n}}}) \ls{\tmsix} & \\
                \con'^{\circ}\of{(\lamp{\lab_1}{\vartwo}{\tmthree^{\circ}_1}) \tmthreevariant^{\circ}_1, \hdots, (\lamp{\lab_n}{\vartwo}{\tmthree^{\circ}_n}) \tmthreevariant^{\circ}_n}
                  \ar@{->>}[r]^{(\redex / \tm') / (\redextwo / \tm')}
                  & \con'^{\circ}\of{\subs{\tmthree^{\circ}_1}{\vartwo}{\tmthreevariant^{\circ}_1}, \hdots, \subs{\tmthree^{\circ}_n}{\vartwo}{\tmthreevariant^{\circ}_n}} \\
              }
            \]
            Note that $\names((\redex / \tm') / (\redextwo / \tm')) = \set{\lab_1, \hdots, \lab_n}$.
            Similarly, $\names((\redex / \redextwo) / (\tm' / \redextwo)) = \set{\lab_1, \hdots, \lab_n}$,
            as required.
          \item {\bf If $\redextwo$ is in $\tmfive$.}
            Straightforward by \ih.
          \item {\bf If $\redextwo$ is in $\tmsix$.}
            Straightforward since the steps are disjoint.
        \end{enumerate}
      \item {\bf If $\redex$ is in $\tmsix$.}
        We consider three subcases, depending on the position of $\redextwo$.
        \begin{enumerate}
          \item {\bf If $\redextwo$ is at the root.} Then $\tmfive = \lam{\var}{\tmtwo}$
            and $\tm = (\lam{\var}{\tmtwo}) \conof{(\lam{\vartwo}{\tmthree}) \tmthreevariant}$.
            \[
            \xymatrix@R=.25cm@C=1cm{
              (\lam{\var}{\tmtwo}) \conof{(\lam{\vartwo}{\tmthree}) \tmthreevariant}
                \ar@{->}[r]^{\redex} \ar@{->}[d]^{\redextwo}
                  & (\lam{\var}{\tmtwo}) \conof{\subs{\tmthree}{\vartwo}{\tmthreevariant}} & \\
                \subs{\tmtwo}{\var}{\conof{(\lam{\vartwo}{\tmthree}) \tmthreevariant}}
                  \ar@{->>}[r]^{\redex / \redextwo}
                    & \subs{\tmtwo}{\var}{\conof{\subs{\tmthree}{\vartwo}{\tmthreevariant}}} \\
            }
            \]
            \[
            \xymatrix@R=.25cm@C=2cm{
                \ar@{->>}[r]^{\redex / \tm'} \ar@{->}[d]^{\redextwo / \tm'}
              (\lamp{\lab}{\var}{\tmtwo'}) [\con_i\of{(\lamp{\lab_{i,j}}{\vartwo}{\tmthree_{i,j}} \tmthreevariant_{i,j})}_{j=1}^{m_i}]_{i=1}^n
                  & (\lamp{\lab}{\var}{\tmtwo'}) [\con_i\of{\subs{\tmthree_{i,j}}{\vartwo}{\tmthreevariant_{i,j}}}_{j=1}^{m_i}]_{i=1}^n & \\
               \subs{\tmtwo'}{\var}{[\con_i\of{(\lamp{\lab_{i,j}}{\vartwo}{\tmthree_{i,j}} \tmthreevariant_{i,j})}_{j=1}^{m_i}]_{i=1}^n} \ar@{->}[r]^{(\redex / \tm') / (\redextwo / \tm')}
                  & \subs{\tmtwo'}{\var}{[\con_i\of{\subs{\tmthree_{i,j}}{\vartwo}{\tmthreevariant_{i,j}}}_{j=1}^{m_i}]_{i=1}^n}
              }
            \]
            Note that $\redex / \redextwo$ has as many elements as there are free occurrences
            of $\var$ in $\tmtwo$. In particular, $\redextwo$ may erase or multiply $\redex$.
            In turn $\tmtwo'$ has a number of free occurrences of $\var$,
            more precisely $n$ free occurrences of $\var$,
            \ie the cardinality of the argument of the application.
            By \rlem{simulation_of_context}, each step in $\redex / \redextwo$,
            when projected onto $\tm' / \redextwo$ yields a set of labels $\set{\lab_{i,1}, \hdots, \lab_{i, m_i}}$.
            So
            $\names((\redex / \redextwo) / (\tm' / \redextwo)) = \sum_{i=1}^n \set{\lab_{i,1}, \hdots, \lab_{i, m_i}} = \names((\redex / \tm') / (\redextwo / \tm'))$
            as required.
          \item {\bf If $\redextwo$ is in $\tmfive$.}
            Straightforward since the steps are disjoint.
          \item {\bf If $\redextwo$ is in $\tmsix$.}
            Straightforward by \ih.
        \end{enumerate}
    \end{enumerate}
\end{enumerate}
\end{proof}

\begin{lemma}[Simulation residual of a development]
\llem{simulation_residual_of_a_development}
Let $\redexset$ be a set of coinitial steps in the $\lambda$-calculus,
let $\redseq$ be a complete development of $\redexset$,
and let $\tm' \in \termsdist$ be a correct term such that $\tm' \refines \src(\redseq)$.
Then $\redseq/\tm'$ is a complete development of $\redexset/\tm'$.
\end{lemma}
\begin{proof}
Recall that developments are finite and the $\lambda$-calculus is finitely branching so,
by K\"onig's lemma, the length of a development of a set $\redexset$ is bounded.
This bound is called the {\em depth} of $\redexset$.
We proceed by induction on the depth of $\redexset$.
\begin{enumerate}
  \item {\bf Base case.}
    Immediate.
  \item {\bf Induction.}
    Then $\redexset$ must be non-empty. Let $\redexset = \set{\redex_1, \hdots, \redex_{n},\redextwo}$ for some $n \geq 0$,
    and let $\redseq$ be a complete development of $\redexset$.
    Without loss of generality, we may assume that $\redseq = \redextwo\redseq'$,
    where $\redseq'$ is a complete development of $\redexset' = \redexset/\redextwo$.
    Note that the depth of $\redexset'$ is strictly lesser than the depth of $\redexset$.
    By \ih, $\redseq' / (\tm' / \redextwo)$ is a complete development of
    $\redexset'/(\tm'/\redextwo) = \bigcup_{i=1}^n \frac{\redex_i / \redextwo}{\tm' / \redextwo}$.
    By the Basic Cube Lemma~(\rlem{cube_lemma_for_simulation_residuals}) this set is equal to
    $\bigcup_{i=1}^n \frac{\redex_i / \tm'}{\redextwo / \tm'}$.
    Then $(\redextwo / \tm') (\redseq' / (\tm' / \redextwo))$ is a complete development of
    $\redexset / \tm' = (\redex_1 / \tm') \cup \hdots \cup (\redex_{n} / \tm') \cup (\redextwo/\tm')$.
    To conclude, note that
    $\redseq / \tm' =
        \redextwo\redseq'/\tm' =
        (\redextwo/\tm')(\redseq'/(\tm' / \redextwo))$ so we are done.
\end{enumerate}
\end{proof}

\begin{lemma}[Compatibility for developments]
\llem{compatibility_for_developments}
Let $\redexset$ be a set of coinitial steps,
and let $\redseq$ and $\redseqtwo$ be complete developments of $\redexset$,
and let $\tm' \in \termsdist$ be a correct term such that $\tm' \refines \src(\redseq)$.
Then $\redseq/\tm' \equiv \redseqtwo/\tm'$.
\end{lemma}
\begin{proof}
This is an immediate consequence of \rlem{simulation_residual_of_a_development},
since $\redseq/\tm'$ and $\redseqtwo/\tm'$ are both complete developments of $\redexset/\tm'$,
hence permutation equivalent.
\end{proof}

The proof of \rprop{compatibility_of_simulation_residuals_and_permutation_equivalence}
proceeds as follows. Let $\redseq \permeq \redseqtwo$ be permutation equivalent derivations in the $\lambda$-calculus,
let $\tm' \in \termsdist$ be a correct term such that $\tm' \refines \src(\redseq)$,
and let us show that:
\begin{enumerate}
\item $\tm'/\redseq = \tm'/\redseqtwo$
\item $\redseq/\tm' \permeq \redseqtwo/\tm'$
\end{enumerate}
Recall that, in an orthogonal axiomatic rewrite system,
permutation equivalence may be defined as
the reflexive and transitive closure of the
permutation axiom
$\redseqthree_1 \redex \redseqtwo \redseqthree_2 \permeq \redseqthree_1 \redextwo \redseq \redseqthree_2$,
where
$\redseqthree_1$ and $\redseqthree_2$ are arbitrary derivations,
$\redseqtwo$ is a complete development of $\redextwo/\redex$,
and $\redseq$ is a complete development of $\redex/\redextwo$.
For this definition of permutation equivalence, see for example~\cite[Def.~2.17]{thesismellies},
and .
We prove each item separately:
\begin{enumerate}
\item
  For the first item, we proceed by induction on the derivation that $\redseq \permeq \redseqtwo$.
  Reflexivity and transitivity are trivial, so we concentrate on the permutation axiom itself.
  Let $\redseqthree_1 \redex \redseqtwo \redseqthree_2 \permeq \redseqthree_1 \redextwo \redseq \redseqthree_2$,
  where
  $\redseqthree_1$ and $\redseqthree_2$ are arbitrary derivations,
  $\redseqtwo$ is a complete development of $\redextwo/\redex$,
  and $\redseq$ is a complete development of $\redex/\redextwo$,
  and let us show that $\tm'/\redseqthree_1 \redex \redseqtwo \redseqthree_2 = \tm'/\redseqthree_1 \redextwo \redseq \redseqthree_2$.
  By \rlem{simulation_residuals_composition} we have that:
  \[
    \tm'/\redseqthree_1 \redex \redseqtwo \redseqthree_2 = ((\tm'/\redseqthree_1)/\redex \redseqtwo)/\redseqthree_2
    \HS
    \text{ and }
    \HS
    \tm'/\redseqthree_1 \redextwo \redseq \redseqthree_2 = ((\tm'/\redseqthree_1)/\redextwo \redseq)/ \redseqthree_2
  \]
  so, without loss of generality,
  it suffices to show that for an arbitrary term $\tmtwo' \in \termsdist$,
  we have $\tmtwo'/\redex\redseqtwo = \tmtwo'/\redextwo\redseq$.
  By definition of simulation residual, the derivations below have the indicated sources and targets:
  \[
    \begin{array}{rcl}
    \redex\redseqtwo/\tmtwo' & : & \tmtwo' \to \tmtwo'/\redex\redseqtwo \\
    \redextwo\redseq/\tmtwo' & : & \tmtwo' \to \tmtwo'/\redextwo\redseq \\
    \end{array}
  \]
  Moreover:
  \[
    \begin{array}{rcll}
      \redex\redseqtwo/\tmtwo'
    & = & (\redex/\tmtwo')(\redseqtwo/(\tmtwo'/\redex)) \\
    & \permeq & (\redex/\tmtwo')((\redextwo/\redex)/(\tmtwo'/\redex))     & \text{by \rlem{compatibility_for_developments}} \\
    & \permeq & (\redex/\tmtwo')((\redextwo/\tmtwo')/(\redex/\tmtwo'))    & \text{by the basic cube lemma~(\rlem{cube_lemma_for_simulation_residuals})} \\
    & \permeq & (\redextwo/\tmtwo')((\redex/\tmtwo')/(\redextwo/\tmtwo')) & \text{since $A(B/A) \permeq B(A/B)$ holds in general} \\
    & \permeq & (\redextwo/\tmtwo')((\redex/\redextwo)/(\tmtwo'/\redextwo)) & \text{by the basic cube lemma~(\rlem{cube_lemma_for_simulation_residuals})} \\
    & = & (\redextwo/\tmtwo')(\redseq/(\tmtwo'/\redextwo)) & \text{by \rlem{compatibility_for_developments}} \\
    & = & \redextwo\redseq/\tmtwo'
    \end{array}
  \]
  So $\redex\redseqtwo/\tmtwo'$ and $\redextwo\redseq/\tmtwo'$ are permutation equivalent.
  In particular, they have the same target, so
  $\tmtwo'/\redex\redseqtwo = \tmtwo'/\redextwo\redseq$ as required.
\item
  The proof of the second item is also by induction on the derivation that $\redseq \permeq \redseqtwo$.
  Let $\redseqthree_1 \redex \redseqtwo \redseqthree_2 \permeq \redseqthree_1 \redextwo \redseq \redseqthree_2$,
  where
  $\redseqthree_1$ and $\redseqthree_2$ are arbitrary derivations,
  $\redseqtwo$ is a complete development of $\redextwo/\redex$,
  and $\redseq$ is a complete development of $\redex/\redextwo$,
  and let us show that $\redseqthree_1 \redex \redseqtwo \redseqthree_2/\tm' = \redseqthree_1 \redextwo \redseq \redseqthree_2/\tm'$.
  By \rlem{simulation_residuals_composition} we have that:
  \[
    \redseqthree_1 \redex \redseqtwo \redseqthree_2/\tm' = (\redseqthree_1/\tm') (\redex \redseqtwo/\tmtwo') (\redseqthree_2/\tmthree')
    \HS
    \text{ and }
    \HS
    \redseqthree_1 \redextwo \redseq \redseqthree_2 / \tm'= (\redseqthree_1 / \tm') (\redextwo \redseq/\tmtwo') (\redseqthree_2/\tmthree'')
  \]
  where $\tmtwo' = \tm'/\redseqthree_1$,
        $\tmthree' = \tmtwo'/\redseqthree_1\redex\redseqtwo$,
        and
        $\tmthree'' = \tmtwo'/\redseqthree_1\redextwo\redseq$.
  Similarly as before, we can prove that $\redex\redseqtwo/\tmtwo' \permeq \redextwo\redseq/\tmtwo'$.
  By item 1 of this proposition, we have $\tmthree' = \tmthree''$,
  so $\redseqthree_2/\tmthree' \permeq \redseqthree_2/\tmthree''$,
  which finishes the proof.
\end{enumerate}

\subsection{Proof of \rlem{generalized_cube_lemma} -- Cube Lemma for simulation residuals}
\label{appendix_generalized_cube_lemma}

Let $\redseq : \tm \rtobeta \tmtwo$ and $\redseqtwo : \tm \rtobeta \tmthree$ be coinitial derivations,
and let $\tm' \in \termsdist$ be a correct term such that $\tm' \refines \tm$.
The statement of \rlem{generalized_cube_lemma} claims that the following equivalence holds:
\[
  (\redseq/\tm')/(\redseqtwo/\tm') \permeq (\redseq/\redseqtwo)/(\tm'/\redseqtwo)
\]
Before proving \rlem{generalized_cube_lemma},
we prove three auxiliary lemmas, all of which are particular cases of the main result.

\begin{lemma}
\llem{generalized_cube_lemma__claim1}
$(\redex / \tm') / (\redseqtwo / \tm') \permeq (\redex / \redseqtwo) / (\tm' / \redseqtwo)$.
\end{lemma}
\begin{proof}
By induction on $\redseqtwo$:
\begin{enumerate}
  \item {\bf Empty, $\redseqtwo = \emptyDerivation$.}
    Immediate since $\redex / \tm' = \redex / \tm'$.
  \item {\bf Non-empty, $\redseqtwo = \redextwo \redseqtwo'$.}
    Then:
    \[
      \begin{array}{rcll}
        (\redex/\tm')/(\redextwo\redseqtwo'/\tm')
      & = &
        (\redex/\tm')/((\redextwo/\tm')(\redseqtwo'/(\tm'/\redextwo)))
        & \text{ by definition} \\
      & = &
        ((\redex/\tm')/(\redextwo/\tm'))/(\redseqtwo'/(\tm'/\redextwo))
        & \text{ by definition} \\
      & \permeq &
        ((\redex/\redextwo)/(\tm'/\redextwo))/(\redseqtwo'/(\tm'/\redextwo))
        & \text{ by \rlem{cube_lemma_for_simulation_residuals}} \\
      & \permeq &
        ((\redex/\redextwo)/\redseqtwo')/((\tm'/\redextwo)/\redseqtwo')
        & \text{ by \ih} \\
      & = &
        (\redex/\redextwo\redseqtwo')/(\tm'/\redextwo\redseqtwo')
        & \text{ by definition}
      \end{array}
    \]
\end{enumerate}
\end{proof}

\begin{lemma}
\llem{generalized_cube_lemma__claim2}
Let $\redexset$ be a set of coinitial steps in the $\lambda$-calculus.
Then $(\redexset/\tm')/(\redextwo/\tm')$ and $(\redexset/\redextwo)/(\tm'/\redextwo)$
are equal as sets.
\end{lemma}
\begin{proof}
Using \rlem{cube_lemma_for_simulation_residuals}:
\[
    (\redexset/\tm')/(\redextwo/\tm')
  =
    \bigcup_{\redex \in \redexset} (\redex/\tm')/(\redextwo/\tm')
  =
    \bigcup_{\redex \in \redexset} (\redex/\redextwo)/(\tm'/\redextwo)
  =
    (\redexset/\redextwo)/(\tm'/\redextwo)
\]
\end{proof}

\begin{lemma}
\llem{generalized_cube_lemma__claim3}
Let $\redexset$ be a set of coinitial steps in the $\lambda$-calculus,
and let $\redexset$ also stand for some (canonical) complete development of $\redexset$.
Then $(\redseq/\tm')/(\redexset/\tm') \permeq (\redseq/\redexset)/(\tm'/\redexset)$.
\end{lemma}
\begin{proof}
By induction on $\redseq$:
\begin{enumerate}
  \item {\bf Empty, $\redseq = \emptyDerivation$.}
    Immediate since $\emptyDerivation = \emptyDerivation$.
  \item {\bf Non-empty, $\redseq = \redex\redseq'$.}
    Then:
    \[
      \begin{array}{rcll}
        (\redex\redseq' / \tm') / (\redexset / \tm') 
      & = &
        (\redex/\tm') (\redseq' / (\tm'/\redex)) / (\redexset / \tm')
      & \text{ by definition}
      \\
      & = &
        ((\redex/\tm') / (\redexset / \tm')) ((\redseq' / (\tm'/\redex)) / ((\redexset / \tm') / (\redex/\tm')) )
      & \text{ by definition}
      \\
      & \permeq &
        ((\redex/\redexset) / (\tm' / \redexset)) ((\redseq' / (\tm'/\redex)) / ((\redexset / \tm') / (\redex/\tm')) )
        & \text{ by \rlem{generalized_cube_lemma__claim1}} \\
      & = &
        ((\redex/\redexset) / (\tm' / \redexset)) ((\redseq' / (\tm'/\redex)) / ((\redexset / \redex) / (\tm'/\redex)) )
        & \text{ by \rlem{generalized_cube_lemma__claim2}} \\
      & = &
        ((\redex/\redexset) / (\tm' / \redexset)) ((\redseq' / (\redexset / \redex)) / ((\tm'/\redex) / (\redexset / \redex)))
        & \text{ by \ih} \\
      & = &
        ((\redex/\redexset) / (\tm' / \redexset)) ((\redseq' / (\redexset / \redex)) / ((\tm'/\redexset) / (\redex / \redexset)))
        & \text{ by \rprop{compatibility_of_simulation_residuals_and_permutation_equivalence} ($\star$)} \\
      & = &
        (\redex/\redexset)(\redseq' / (\redexset / \redex)) / (\tm' / \redexset)
        & \text{ by definition} \\
      & = &
        (\redex\redseq'/\redexset) / (\tm' / \redexset)
        & \text{ by definition} \\
      \end{array}
    \]
For the equality marked with ($\star$), observe that
$(\tm'/\redex) / (\redexset / \redex) = \tm'/(\redex \sqcup \redexset)$
and
$(\tm'/\redexset) / (\redex / \redexset) = \tm'/(\redexset \sqcup \redex)$
by definition.
Moreover $\redex \sqcup \redexset \permeq \redexset \sqcup \redex$,
so by Compatibility~(\rprop{compatibility_of_simulation_residuals_and_permutation_equivalence}),
$(\tm'/\redex) / (\redexset / \redex) = (\tm'/\redexset) / (\redex / \redexset)$.
\end{enumerate}
\end{proof}

Now the proof of \rlem{generalized_cube_lemma} proceeds by induction on $\redseq$:
\begin{enumerate}
  \item {\bf Empty $\redseq = \emptyDerivation$.} Immediate since $\emptyDerivation = \emptyDerivation$.
  \item {\bf Non-empty, $\redseq = \redex \redseq'$.} Then:
    \[
      \begin{array}{rcll}
        (\redex \redseq'/\redseqtwo)/(\tm'/\redseqtwo)
      & = &
        (\redex/\redseqtwo)(\redseq'/(\redseqtwo/\redex))/(\tm'/\redseqtwo)
        & \text{ by definition} \\
      & = &
        ((\redex/\redseqtwo)/(\tm'/\redseqtwo))((\redseq'/(\redseqtwo/\redex))/((\tm'/\redseqtwo)/(\redex/\redseqtwo)))
        & \text{ by definition} \\
      & \permeq &
        ((\redex/\tm')/(\redseqtwo/\tm'))((\redseq'/(\redseqtwo/\redex))/((\tm'/\redseqtwo)/(\redex/\redseqtwo)))
        & \text{ by \rlem{generalized_cube_lemma__claim1}} \\
      & = &
        ((\redex/\tm')/(\redseqtwo/\tm'))((\redseq'/(\redseqtwo/\redex))/((\tm'/\redex)/(\redseqtwo/\redex)))
        & \text{ by \rprop{compatibility_of_simulation_residuals_and_permutation_equivalence} ($\star$)} \\
      & = &
        ((\redex/\tm')/(\redseqtwo/\tm'))
        ((\redseq'/(\tm'/\redex))/((\redseqtwo/\redex)/(\tm'/\redex)))
        & \text{ by \ih} \\
      & = &
        ((\redex/\tm')/(\redseqtwo/\tm'))
        ((\redseq'/(\tm'/\redex))/((\redseqtwo/\tm')/(\redex/\tm')))
        & \text{ by \rlem{generalized_cube_lemma__claim3}} \\
      & = &
        (\redex/\tm')(\redseq'/(\tm'/\redex))/(\redseqtwo/\tm')
        & \text{ by definition} \\
      & = &
        (\redex\redseq'/\tm')/(\redseqtwo/\tm')
        & \text{ by definition} \\
      \end{array}
    \]
\end{enumerate}
For the equality marked with ($\star$), observe that
$(\tm'/\redseqtwo) / (\redex / \redseqtwo) = \tm'/(\redseqtwo \sqcup \redex)$
and
$(\tm'/\redex) / (\redseqtwo / \redex) = \tm'/(\redex \sqcup \redseqtwo)$
by definition.
Moreover $\redseqtwo \sqcup \redex \permeq \redex \sqcup \redseqtwo$,
so by Compatibility~(\rprop{compatibility_of_simulation_residuals_and_permutation_equivalence}),
$(\tm'/\redseqtwo) / (\redex / \redseqtwo) = (\tm'/\redex) / (\redseqtwo / \redex)$.

\subsection{Proof of \rprop{properties_of_garbage} -- Properties of garbage}
\label{appendix_properties_of_garbage}

\begin{enumerate}
\item Let $\redseqtwo \permle \redseq$.
      Then $\redseqtwo\redseqthree \permeq \redseq$ for some $\redseqthree$,
      so $\redseqtwo/\tm' \permle (\redseqtwo/\tm')(\redseqthree/(\tm'/\redseqtwo)) = \redseqtwo\redseqthree/\tm' \permeq \redseq/\tm'$
      by \resultname{Compatibility}~(\rprop{compatibility_of_simulation_residuals_and_permutation_equivalence}).
\item Note that $\redseq\redseqtwo/\tm' = (\redseq/\tm')(\redseqtwo/(\tm'/\redseq))$.
      So $\redseq\redseqtwo/\tm'$ is empty if and only if
      $\redseq/\tm'$ and $\redseqtwo/(\tm'/\redseq)$ are empty.
\item Suppose that $\redseq/\tm' = \emptyDerivation$.
      Then $(\redseq/\redseqtwo)/(\tm'/\redseqtwo) = (\redseq/\tm')/(\redseqtwo/\tm')$
      by the \resultname{Cube Lemma}~(\rlem{generalized_cube_lemma}).
\item By the \resultname{Cube Lemma}~(\rlem{generalized_cube_lemma}):
      $(\redseq \sqcup \redseqtwo)/\tm' = \redseq(\redseqtwo/\redseq)/\tm' =
      (\redseq/\tm')((\redseqtwo/\redseq)/(\tm'/\redseq) \permeq
      (\redseq/\tm')((\redseqtwo/\tm')/(\redseq/\tm')) =
      (\redseq/\tm') \sqcup (\redseqtwo/\tm')$.
      So $(\redseq \sqcup \redseqtwo)/\tm'$ is empty if and only if
      $\redseq/\tm'$ and $\redseqtwo/\tm'$ are empty.
\end{enumerate}

\subsection{Proof of \rlem{sieving_well_defined} -- Sieving is well-defined}
\label{appendix_sieving_well_defined}

Note that $\tm'/\redex_0 \refines \redseq/\redex_0$ by definition of simulation residual,
so the recursive call can be made.

To see that recursion is well-founded, consider the measure
given by $M(\redseq,\tm') = \#\names(\redseq/\tm')$.
Let $\redex_0$ be a coarse step for $(\redseq,\tm')$.
Observe that $\redex_0 \permle \redseq$
so $\redex_0/\tm' \permle \redseq/\tm'$ by~\rcoro{algebraic_simulation}.
Then $\names(\redex_0/\tm') \subseteq \names(\redseq/\tm')$ by \rprop{prefixes_as_subsets}.
So we have:
\[
  \begin{array}{rcll}
         \#\names(\redseq/\tm')
   & > & \#(\names(\redseq/\tm') \setminus \names(\redex_0/\tm')) & \text{since $\redex_0/\tm' \neq \emptyset$ and $\names(\redex_0/\tm') \subseteq \names(\redseq/\tm')$} \\
   & = & \#\names((\redseq/\tm') / (\redex_0/\tm')) & \text{by \rlem{names_after_projection_along_a_step}} \\
   & = & \#\names((\redseq/\redex_0) / (\tm'/\redex_0))
  \end{array}
\]
as required.

\subsection{Proof of \rprop{properties_of_sieving} -- Properties of sieving}
\label{appendix_properties_of_sieving}

To prove \rprop{properties_of_sieving} we first prove various auxiliary results.

\begin{lemma}[The sieve is a prefix]
\llem{sieve_is_prefix}
Let $\redseq : \tm \tobeta^* \tmtwo$ and $\tm' \refines \tm$.
Then $\redseq \sieve \tm' \permle \redseq$.
\end{lemma}
\begin{proof}
By induction on the length of $\redseq \sieve \tm'$.
If there are no coarse steps for $(\redseq,\tm')$,
then trivially $\redseq \sieve \tm' = \emptyDerivation \permle \redseq$.
If there is a coarse step for $(\redseq,\tm')$,
let $\redex_0$ be the leftmost such step. Then:
  \[
    \begin{array}{rcll}
      \redseq \sieve \tm'
      & = & \redex_0((\redseq/\redex_0) \sieve (\tm'/\redex_0)) \\
      & \permle & \redex_0(\redseq/\redex_0) & \text{by \ih} \\
      & \equiv  & \redseq(\redex_0/\redseq)  & \text{since $A(B/A) \permeq B(A/B)$ in general} \\
      & =       & \redseq                    & \text{since $\redex_0 \permle \redseq$ as $\redex_0$ is coarse for $(\redseq,\tm')$} \\
    \end{array}
  \]
\end{proof}

We also need the following technical lemma:
\begin{lemma}[Refinement of a context]
\llem{refinement_context}
The following are equivalent:
  \begin{enumerate}
  \item $\tm' \refines \conof{\tmtwo}$,
  \item $\tm'$ is of the form $\con'\of{\tmtwo'_{1},\hdots,\tmtwo'_n}$,
        where $\con'$ is an $n$-hole context such that $\con' \refines \con$ and
        $\tmtwo'_i \refines \tmtwo$ for all $1 \leq i \leq n$.
        Note that $n$ might be $0$, in which case $\con'$ is a term.
  \end{enumerate}
Moreover, in the implication $(1 \implies 2)$,
the decomposition is unique,
\ie the context $\con'$,
the number of holes $n \geq 0$,
and the terms $\tmtwo'_1,\hdots,\tmtwo'_n$
are the unique possible such objects.
\end{lemma}
\begin{proof}
Straightforward by induction on $\con$.
\end{proof}

\begin{lemma}[Garbage only interacts with garbage]
\llem{garbage_only_creates_garbage}
\llem{garbage_only_duplicates_garbage}
The following hold:
\begin{enumerate}
\item {\bf Garbage only creates garbage.}
  Let $\redex$ and $\redextwo$ be composable steps in the $\lambda$-calculus,
  and let $\tm' \refines \src(\redex)$.
  If $\redex$ creates $\redextwo$ and $\redex$ is $\tm'$-garbage,
  then $\redextwo$ is $(\tm'/\redex)$-garbage.
\item {\bf Garbage only duplicates garbage.}
  Let $\redex$ and $\redextwo$ be coinitial steps in the $\lambda$-calculus
  and let $\tm' \refines \src(\redex)$.
  If $\redex$ duplicates $\redextwo$, \ie $\#(\redextwo/\redex) > 1$,
  and $\redex$ is $\tm'$-garbage,
  then $\redextwo$ is $(\tm'/\redex)$-garbage.
\end{enumerate}
\end{lemma}
\begin{proof}
We prove each item separately:
\begin{enumerate}
\item
  According to L\'evy~\cite{Tesis:Levy:1978},
  there are three creation cases in the $\lambda$-calculus.
  We consider the three possibilities for $\redex$ creating $\redextwo$:
  \begin{enumerate}
  \item
    {\bf Case I,
      $
         \conof{(\lam{\var}{\var})(\lam{\vartwo}{\tmtwo})\tmthree}
       \toabeta{\redex}
         \conof{(\lam{\vartwo}{\tmtwo})\tmthree}
       \toabeta{\redextwo}
         \conof{\subs{\tmtwo}{\vartwo}{\tmthree}}
      $.}
    Then by \rlem{refinement_context},
    the term $\tm'$ is of the form
    $\con'\of{\Delta_1,\hdots,\Delta_n}$
    where $\con'$ is an $n$-hole context such that $\con' \refines \conof{\conbase\,\tmthree}$
    and $\Delta_i \refines (\lam{\var}{\var})(\lam{\vartwo}{\tmtwo})$ for all $1 \leq i \leq n$.
    Since $\redex$ is garbage, we know that actually $n = 0$.
    So $\tm' \refines \conof{\conbase\,\tmthree}$ and $\redex/\tm' : \tm' \rtodist \tm' = \tm'/\redex$ in zero steps.
    Hence $\tm' \refines \conof{(\lam{\vartwo}{\tmtwo})\tmthree}$,
    so by \rlem{refinement_context},
    the term $\tm'$ can be written in a unique way as
    $\con''\of{\Sigma_1,\hdots,\Sigma_m}$, where $\con''$ is an $m$-hole context
    such that $\con'' \refines \conof{\conbase\,\tmthree}$ and $\Sigma_i \refines \lam{\vartwo}{\tmtwo}$ for all $1 \leq i \leq m$.
    Since the decomposition is unique and $\tm' \refines \conof{\conbase\,\tmthree}$,
    we conclude that $m = 0$.
    Hence $\redextwo$ is $(\tm'/\redex)$-garbage.
  \item
    {\bf Case II,
      $
         \conof{(\lam{\var}{\lam{\vartwo}{\tmtwo}})\,\tmthree\,\tmfour}
       \toabeta{\redex}
         \conof{(\lam{\vartwo}{\subs{\tmtwo}{\var}{\tmthree}})\,\tmfour}
       \toabeta{\redextwo}
         \conof{\subs{\subs{\tmtwo}{\var}{\tmthree}}{\vartwo}{\tmfour}}
      $.}
    Then by \rlem{refinement_context},
    the term $\tm'$ is of the form
    $\con'\of{\Delta_1,\hdots,\Delta_n}$
    where $\con'$ is an $n$-hole context such that $\con' \refines \conof{\conbase\,\tmfour}$
    and $\Delta_i \refines (\lam{\var}{\lam{\vartwo}{\tmtwo}})\,\tmthree$ for all $1 \leq i \leq n$.
    Since $\redex$ is garbage, we know that actually $n = 0$.
    So $\tm' \refines \conof{\conbase\,\tmfour}$ and $\redex/\tm' : \tm' \rtodist \tm' = \tm'/\redex$
    in zero steps. Hence $\tm' \refines \conof{(\lam{\vartwo}{\subs{\tmtwo}{\var}{\tmthree}})\,\tmfour}$,
    so by \rlem{refinement_context}, the term $\tm'$ can be written in a unique way as
    $\con''\of{\Sigma_1,\hdots,\Sigma_n}$, where $\con''$ is an $m$-hole context such that
    $\con'' \refines \conof{\conbase\,\tmfour}$ and $\Sigma_i \refines \lam{\vartwo}{\subs{\tmtwo}{\var}{\tmthree}}$
    for all $1 \leq i \leq m$.
    Since the decomposition is unique and $\tm' \refines \conof{\conbase\,\tmfour}$,
    we conclude that $m = 0$.
    Hence $\redextwo$ is $(\tm'/\redex)$-garbage.
  \item {\bf Case III, $
         \con_1\of{(\lam{\var}{\con_2\of{\var\,\tmtwo}})\,(\lam{\vartwo}{\tmthree})}
       \toabeta{\redex}
         \con_1\of{\hat{\con}_2\of{(\lam{\vartwo}{\tmthree})\,\hat{\tmtwo}}}
       \toabeta{\redextwo}
         \con_1\of{\hat{\con}_2\of{\subs{\tmthree}{\vartwo}{\hat{\tmtwo}}}}
      $,
      where
        $\hat{\con}_2 = \subs{\con_2}{\var}{\lam{\vartwo}{\tmthree}}$
      and
        $\hat{\tm} = \subs{\tm}{\var}{\lam{\vartwo}{\tmthree}}$.}
    Then by \rlem{refinement_context},
    the term $\tm'$ is of the form $\con'\of{\Delta_1,\hdots,\Delta_n}$
    where $\con'$ is an $n$-hole context such that $\con' \refines \con_1$
    and $\Delta_i \refines (\lam{\var}{\con_2\of{\var\,\tmtwo}})\,(\lam{\vartwo}{\tmthree})$ for all $1 \leq i \leq n$.
    Since $\redex$ is garbage, we know that actually $n = 0$.
    So $\tm' \refines \con_1$ and $\redex/\tm' : \tm' \rtodist \tm' = \tm'/\redex$ in zero steps.
    Hence $\tm' \refines \con_1\of{\hat{\con}_2\of{(\lam{\vartwo}{\tmthree})\,\hat{\tmtwo}}}$,
    so by \rlem{refinement_context}, the term $\tm'$ can be written in a unique way as
    $\con''\of{\Sigma_1,\hdots,\Sigma_m}$, where $\con'' \refines \con_1$ and
    $\Sigma_i \refines \hat{\con}_2\of{(\lam{\vartwo}{\tmthree})\,\hat{\tmtwo}}$ for all $1 \leq i \leq m$.
    Since the decomposition is unique and $\tm' \refines \con_1$, we conclude that $m = 0$.
    Hence $\redextwo$ is $(\tm'/\redex)$-garbage.
  \end{enumerate}
\item
  Since $\redex$ duplicates $\redextwo$,
  the redex contracted by $\redextwo$ 
  lies inside the argument of $\redex$, that is,
  the source term is of the form $\con_1\of{(\lam{\var}{\tm})\con_2\of{(\lam{\vartwo}{\tmtwo})\tmthree}}$
  where the pattern of $\redex$ is $(\lam{\var}{\tm})\con_2\of{(\lam{\vartwo}{\tmtwo})\tmthree}$,
  and the pattern of $\redextwo$ is $(\lam{\vartwo}{\tmtwo})\tmthree$.
  By \rlem{refinement_context},
  the term $\tm'$ is of the form $\con'\of{\Delta_1,\hdots,\Delta_n}$
  where $\con'$ is an $n$-hole context
  such that $\con' \refines \con$ and $\Delta_i \refines (\lam{\var}{\tm})\con_2\of{(\lam{\vartwo}{\tmtwo})\tmthree}$ for all $1 \leq i \leq n$.
  Since $\redex$ is garbage, we know that $n = 0$.
  By \rlem{refinement_context},
  the term $\tm'$ can be written as $\con''\of{\Sigma_1,\hdots,\Sigma_m}$
  where $\con'' \refines \con_1\of{(\lam{\var}{\tm})\con_2}$
  and $\Sigma_i \refines (\lam{\vartwo}{\tmtwo})\tmthree$ for all $1 \leq i \leq m$.
  Note that $\tm' \refines \con_1$ so $\tm' \refines \con_1\of{(\lam{\var}{\tm})\con_2}$,
  as can be checked by induction on $\con_1$.
  Since the decomposition is unique, this means that $m = 0$,
  and thus $\redextwo$ is garbage.
\end{enumerate}
\end{proof}

\begin{proposition}[Characterization of garbage]
\lprop{characterization_of_garbage}
Let $\redseq : \tm \rtobeta \tmtwo$ and $\tm' \refines \tm$.
The following are equivalent:
\begin{enumerate}
\item $\redseq \sieve \tm' = \emptyDerivation$.
\item There are no coarse steps for $(\redseq,\tm')$.
\item The derivation $\redseq$ is $\tm'$-garbage.
\end{enumerate}
\end{proposition}
\begin{proof}
It is immediate to check that items 1 and 2 are equivalent, by definition of sieving,
so let us prove $2 \implies 3$ and $3 \implies 2$:
\begin{itemize}
\item $(2 \implies 3)$
  We prove the contrapositive, namely that if $\redseq$ is not garbage,
  there is a coarse step for $(\redseq,\tm')$.
  Suppose that $\redseq$ is not garbage, \ie $\redseq/\tm' \neq \emptyDerivation$.
  Then by \rprop{properties_of_garbage}
  we have that $\redseq$ can be written as $\redseq = \redseq_1 \redex \redseq_2$ where
  all the steps in $\redseq_1$ are garbage and $\redex$ is not garbage.
  By the fact that garbage only creates garbage~(\rlem{garbage_only_creates_garbage})
  the step $\redex$ has an ancestor $\redex_0$, \ie $\redex \in \redex_0/\redseq_1$.
  Moreover, since garbage only duplicates garbage~(\rlem{garbage_only_duplicates_garbage})
  we have that $\redex_0/\redseq_1 = \redex$.
  Given that $\redex$ is not garbage, we have that:
  \[
    \begin{array}{rcll}
          (\redex_0/\tm')/(\redseq_1/\tm')
    & = & (\redex_0/\redseq_1)/(\tm'/\redseq_1) & \text{ by \rlem{generalized_cube_lemma}} \\
    & = & \redex/(\tm'/\redseq_1) \\
    & \neq & \emptyset \\
    \end{array}
  \]
  Since $(\redex_0/\tm')/(\redseq_1/\tm') \neq \emptyset$, in particular,
  $\redex_0/\tm' \neq \emptyset$, which means that $\redex_0$ is not garbage.
  Moreover, $\redex_0 \permle \redseq_1\redex\redseq_2 = \redseq$.
  So $\redex_0$ is coarse for $(\redseq,\tm')$.
\item $(3 \implies 2)$
  Let $\redseq$ be garbage, suppose that there is a coarse step $\redex$ for $(\redseq,\tm')$,
  and let us derive a contradiction.
  Since $\redex$ is coarse for $(\redseq,\tm')$,
  we have that $\redex \permle \redseq$,
  so $\redex/\tm' \permle \redseq/\tm'$ by \rcoro{simulation_residuals_and_prefixes}.
  But $\redseq/\tm'$ is empty because $\redseq$ is $\tm'$-garbage,
  that is, $\redex/\tm' \permle \redseqtwo/\tm' = \emptyDerivation$,
  which means that $\redex$ is also $\tm'$-garbage.
  This contradicts the fact that $\redex$ is coarse for $(\redseq,\tm')$.
\end{itemize}
\end{proof}

\begin{lemma}[The leftmost coarse step has at most one residual]
\llem{leftmost_coarse_step_has_at_most_one_residual}
Let $\redex_0$ be the leftmost coarse step for $(\redseq,\tm')$,
and let $\redseqtwo \permle \redseq$.
Then $\#(\redex_0/\redseqtwo) \leq 1$.
%:%Moreover, if $\redex_0 \permle \redseqtwo$ then $\redex_0 \in \redseqtwo$.
%:%Recall that $\redex \in \redseq$ means that $\redseq$ can be written as $\redseq_1\redextwo\redseqtwo_2$
%:%such that $\redextwo$ is a residual of $\redex$, \ie $\redextwo \in \redex/\redseq_1$.
\end{lemma}
\begin{proof}
By induction on the length of $\redseqtwo$. The base case is immediate.
For the inductive step, let $\redseqtwo = \redextwo \redseqthree \permle \redseq$.
Then in particular $\redextwo \permle \redseq$.
We consider two cases, depending on whether $\redex_0 = \redextwo$.
\begin{enumerate}
\item {\bf Equal, $\redex_0 = \redextwo$.}
  Then $\redex_0/\redseqtwo = \redex_0/\redex_0\redseqthree = \emptyset$ and we are done.
  %:%Moreover, $\redex_0 \in \redex_0\redseqthree$.
\item {\bf Non-equal, $\redex_0 \neq \redextwo$.}
  First we argue that $\redex_0/\redextwo$ has exactly one residual $\redex_1 \in \redex_0/\redextwo$.
  To see this, we consider two further cases, depending on whether
  $\redextwo$ is $\tm'$-garbage or not:
  \begin{enumerate}
  \item {\bf If $\redextwo$ is not $\tm'$-garbage.}
    Then $\redextwo \permle \redseq$ and $\redextwo/\tm' = \emptyset$, so $\redextwo$ is coarse for $(\redseq,\tm')$.
    Since $\redex_0$ is the leftmost coarse step, this means that $\redex$ is to the left of $\redextwo$.
    So $\redex_0$ has exactly one residual $\redex_1 \in \redex_0/\redextwo$.
  \item {\bf If $\redextwo$ is $\tm'$-garbage.}
    Let us write the term $\tm$ as $\tm = \conof{(\lam{\var}{\tmtwo})\tmthree}$, where
    $(\lam{\var}{\tmtwo})\tmthree$ is the pattern of the redex $\redextwo$.
    By \rlem{refinement_context}
    the term $\tm'$ is of the form $\tm' = \con'[\Delta_1,\hdots,\Delta_n]$,
    where $\con'$ is a many-hole context such that $\con' \refines \con$
    and $\Delta_i \refines (\lam{\var}{\tmtwo})\tmthree$ for all $1 \leq i \leq n$.
    We know that $\redextwo$ is garbage, so $n = 0$ and $\tm' = \con'$ is actually
    a $0$-hole context (\ie a term).
    On the other hand, $\redex_0$ is coarse for $(\redseq,\tm')$, so in particular
    it is not $\tm'$-garbage.
    This means that the pattern of the redex $\redex_0$ cannot occur inside the argument $\tmthree$
    of the redex $\redextwo$.
    So $\redextwo$ does not erase or duplicate $\redex$,
    \ie $\redex_0$ has exactly one residual $\redex_1 \in \redex_0/\redextwo$.
  \end{enumerate}
  Now we have that $\redex_0/\redextwo\redseqthree = \redex_1/\redseqthree$.
  We are left to show that $\#(\redex_0/\redextwo\redseqthree) \leq 1$.
  Let us show that we may apply the \ih on $\redex_1$.
  More precisely, observe that $\redseqthree \permle \redseq/\redextwo$ since $\redextwo\redseqthree \permle \redseq$.
  To apply the \ih it suffices to show that $\redex_1$ is coarse for $(\redseq/\redextwo,\tm'/\redextwo)$.
  Indeed, we may check the two conditions
  for the definition that $\redex_1$ is coarse for $(\redseq/\redextwo,\tm'/\redextwo)$.
  \begin{enumerate}
  \item Firstly, $\redex_1 = \redex_0/\redextwo \permle \redseq/\redextwo$ holds, as a consequence of the fact that $\redex_0 \permle \redseq$.
  \item Secondly, we may check that $\redex_1$ is not $(\tm'/\redextwo)$-garbage.
        To see this, \ie that $\redex_1/(\tm'/\redextwo)$ is non-empty,
        we check that $\names(\redex_1/(\tm'/\redextwo))$ is non-empty.
        \[
          \begin{array}{rcll}
          \names(\redex_1/(\tm'/\redextwo))
          & = & \names((\redex_0/\redextwo)/(\tm'/\redextwo)) & \text{ by definition of $\redex_1$} \\
          & = & \names((\redex_0/\tm')/(\redextwo/\tm')) & \text{ by \rlem{generalized_cube_lemma}} \\
          & = & \names(\redex_0/\tm') \setminus \names(\redextwo/\tm') & \text{ by \rlem{names_after_projection_along_a_step}} \\
          & = & \names(\redex_0/\tm')
          \end{array}
        \]
        For the last step, note that $\names(\redex_0/\tm')$ and $\names(\redextwo/\tm')$ are disjoint,
        since $\redex_0 \neq \redextwo$.
  \end{enumerate}
  Applying the \ih, we obtain that $\#(\redex_0/\redextwo\redseqthree) = \#(\redex_1/\redseqthree) \leq 1$.  
\end{enumerate}
\end{proof}

\begin{lemma}
\llem{stable_non_garbage}
Let $\redex$ be a step, let $\redseq$ a coinitial derivation,
and let $\tm' \refines \src(\redex)$.
Suppose that $\redex$ is not $\tm'$-garbage,
and that $\redex / \redseq_1$ is a singleton for every prefix $\redseq_1 \permle \redseq$.
Then $\redex/\redseq$ is not $(\tm'/\redseq)$-garbage.
\end{lemma}
\begin{proof}
By induction on $\redseq$.
The base case, when $\redseq = \emptyDerivation$, is immediate
since we know that $\redex$ is not garbage.
For the inductive step, suppose that $\redseq = \redextwo\redseqtwo$.
We know that $\redex/\redextwo$ is a singleton,
so let $\redex_1 = \redex/\redextwo$.
Note that
$
  \redex_1/(\tm'/\redextwo)
  = (\redex/\redextwo)/(\tm'/\redextwo)
  = (\redex/\tm')/(\redextwo/\tm')
$ by \rlem{generalized_cube_lemma}.
We know that $\redex/\tm'$ is non-empty, because $\redex$ is not garbage.
Moreover, $\names(\redex/\tm')$ and $\names(\redextwo/\tm')$ are disjoint
since $\redex \neq \redextwo$.
So $\#\names((\redex/\tm')/(\redextwo/\tm')) = \names(\redex/\tm')$ by \rlem{names_after_projection_along_a_step}.
This means that the set $\redex_1/(\tm'/\redextwo)$ is non-empty,
so $\redex_1$ is not $(\tm'/\redextwo)$-garbage .
By \ih we obtain that $\redex_1/\redseqtwo$ is not $((\tm'/\redextwo)/\redseqtwo)$-garbage,
which means that
$(\redex/\redextwo\redseqtwo)/(\tm'/\redextwo\redseqtwo) \neq \emptyset$,
\ie that $\redex/\redextwo\redseqtwo$ is not garbage, as required.
To be able to apply the \ih, observe that if $\redseqtwo_1$ is a prefix of $\redseqtwo$,
then $\redextwo\redseqtwo_1$ is a prefix of $\redseq$,
so the fact that $\redex$ has a single residual after $\redextwo\redseqtwo_1$
implies that the step $\redex_1 = \redex/\redextwo$ has a single residual after $\redseqtwo_1$.
\end{proof}

\begin{proposition}[Characterization of garbage-free derivations]
\lprop{characterization_of_garbage_free_derivations}
Let $\redseq : \tm \rtobeta \tmtwo$ and $\tm' \refines \tm$.
The following are equivalent:
\begin{enumerate}
\item $\redseq$ is $\tm'$-garbage-free.
\item $\redseq \permeq \redseq \sieve \tm'$.
\item $\redseq \permeq \redseqtwo \sieve \tm'$ for some derivation $\redseqtwo$.
\end{enumerate}
\end{proposition}
\begin{proof}
Let us prove $1 \implies 2 \implies 3 \implies 1$:
\begin{itemize}
\item $(1 \implies 2)$
  Suppose that $\redseq$ is $\tm'$-garbage-free,
  and let us show that $\redseq \permeq \redseq \sieve \tm'$
  by induction on the length of $\redseq \sieve \tm'$.

  If there are no coarse steps for $(\redseq,\tm')$,
  By \rprop{characterization_of_garbage}, any derivation with no coarse steps is garbage.
  So $\redseq$ is $\tm'$-garbage.
  Since $\redseq$ is garbage-free, this means that $\redseq = \emptyDerivation$.
  Hence $\redseq = \emptyDerivation = \redseq \sieve \tm'$, as required.

  If there exists a coarse step for $(\redseq,\tm')$, let $\redex_0$ be the leftmost such step.
  Note that $\redseq \permeq \redex_0(\redseq/\redex_0)$ since $\redex_0 \permle \redseq$.
  Moreover, we claim that $\redseq/\redex_0$ is $(\tm'/\redex_0)$-garbage-free.
  Let $\redseqtwo \permle \redseq/\redex_0$ such that $(\redseq/\redex_0)/\redseqtwo$ is garbage
  with respect to the term $(\tm'/\redex_0)/\redseqtwo = \tm'/\redex_0\redseqtwo$,
  and let us show that $(\redseq/\redex_0)/\redseqtwo$ is empty.
  Note that:
  \[
    \begin{array}{rcll}
    \redex_0\redseqtwo
    & \permle & \redex_0(\redseq/\redex_0) & \text{since $\redseqtwo \permle \redseq/\redex_0$} \\
    & \permeq & \redseq                    & \text{as already noted} \\
    \end{array}
  \]
  Moreover, we know that the derivation $\redseq/\redex_0\redseqtwo = (\redseq/\redex_0)/\redseqtwo$ is
  $(\tm'/\redex_0\redseqtwo)$-garbage.
  So, given that $\redseq$ is $\tm'$-garbage-free, we conclude that
  $\redseq/\redex_0\redseqtwo = \emptyDerivation$, that is $(\redseq/\redex_0)/\redseqtwo = \emptyDerivation$,
  which completes the proof of the claim that $\redseq/\redex_0$ is $(\tm'/\redex_0)$-garbage-free.
  We conclude as follows:
  \[
    \begin{array}{rcll}
      \redseq & \permeq & \redex_0(\redseq/\redex_0)                          & \text{as already noted} \\
              & \permeq & \redex_0((\redseq/\redex_0) \sieve (\tm'/\redex_0)) & \text{by \ih since $\redseq/\redex_0$ is $(\tm'/\redex_0)$-garbage-free} \\
              & \permeq & \redseq \sieve \tm'                                 & \text{by definition of sieving} \\
    \end{array}
  \]
\item $(2 \implies 3)$ Obvious, taking $\redseqtwo := \redseq$.
\item $(3 \implies 1)$
  Let $\redseq \permeq \redseqtwo \sieve \tm'$.
  Let us show that $\redseq$ is garbage-free by induction on the length of $\redseqtwo \sieve \tm'$.

  If there are no coarse steps for $(\redseqtwo,\tm')$,
  then $\redseq \permeq \redseqtwo \sieve \tm' = \emptyDerivation$, which means that $\redseq = \emptyDerivation$.
  Observe that the empty derivation is trivially garbage-free.

  If there exists a coarse step for $(\redseqtwo,\tm')$,
  let $\redex_0$ be the leftmost such step.
  Then $\redseq \permeq \redseqtwo \sieve \tm' = \redex_0((\redseqtwo/\redex_0) \sieve (\tm'/\redex_0))$.
  To see that $\redseq$ is $\tm'$-garbage-free,
  let $\redseqthree \permle \redseq$ such that $\redseq/\redseqthree$ is garbage,
  and let us show that $\redseq/\redseqthree$ is empty. 
  We know that $\redseq/\redseqthree$ is of the following form (modulo permutation equivalence):
  \[
     \redseq/\redseqthree \permeq
     \frac{\redex_0((\redseqtwo/\redex_0) \sieve (\tm'/\redex_0))}{\redseqthree}
     =
     \left(\frac{\redex_0}{\redseqthree}\right) \left(\frac{(\redseqtwo/\redex_0) \sieve (\tm'/\redex_0)}{\redseqthree/\redex_0}\right)
  \]
  That is, we know that the following derivation is $(\tm'/\redseqthree)$-garbage,
  and it suffices to show that it is empty:
  \[
    \left(\frac{\redex_0}{\redseqthree}\right) \left(\frac{(\redseqtwo/\redex_0) \sieve (\tm'/\redex_0)}{\redseqthree/\redex_0}\right)
  \]
  Recall that, in general, $AB$ is garbage if and only if $A$ and $B$ are garbage (\rprop{properties_of_garbage}).
  Similarly, $AB$ is empty if and only if $A$ and $B$ are empty.
  So it suffices to prove the two following implications:
  \[
    \begin{array}{clll}
    {\bf (A)} &
    \text{ If } & \redex_0/\redseqthree \text{ is garbage,} & \text{ then it is empty.} \\
    {\bf (B)} &
    \text{ If } & ((\redseqtwo/\redex_0) \sieve (\tm'/\redex_0))/(\redseqthree/\redex_0) \text{ is garbage,} & \text{ then it is empty.}
    \end{array}
  \]
  Let us check that each implication holds:
  \begin{itemize}
  \item {\bf (A)}
    Suppose that $\redex_0/\redseqthree$ is $(\tm'/\redseqthree)$-garbage,
    and let us show that $\redex_0/\redseqthree$ is empty.
    Knowing that the derivation $\redex_0/\redseqthree$ is garbage means that
    $(\redex_0/\redseqthree)/(\tm'/\redseqthree) = \emptyset$.
    Since $\redex_0$ is the leftmost step coarse for $(\redseqtwo,\tm')$,
    by \rlem{leftmost_coarse_step_has_at_most_one_residual}
    we have that $\#(\redex_0/\redseqthree) \leq 1$.
    If $\redex_0/\redseqthree$ is empty we are done, since this is what we wanted to prove.

    The remaining possibility is that $\redex_0/\redseqthree$ be a singleton.
    We argue that this case is impossible.
    Note that for every prefix $\redseqthree_1 \permle \redseqthree$,
    the set $\redex_0/\redseqthree_1$ is also a singleton,
    since otherwise it would be empty, as a consequence of \rlem{leftmost_coarse_step_has_at_most_one_residual}.
    So we may apply \rlem{stable_non_garbage} and conclude that,
    since $\redex_0$ is not $\tm'$-garbage
    then $\redex_0/\redseqthree$ is not $(\tm'/\redseqthree)$-garbage.
    This contradicts the hypothesis.
  \item {\bf (B)}
    Suppose that $((\redseqtwo/\redex_0) \sieve (\tm'/\redex_0))/(\redseqthree/\redex_0)$ is garbage
    with respect to the term $(\tm'/\redex_0)/(\redseqthree/\redex_0)$, and let us show that it is empty.
    Since $(\redseqtwo/\redex_0) \sieve (\tm'/\redex_0)$
    is a shorter derivation than $\redseqtwo \sieve \tm'$,
    we may apply the \ih we obtain that
    $(\redseqtwo/\redex_0) \sieve (\tm'/\redex_0)$ is $(\tm'/\redex_0)$-garbage-free.
    Moreover, the following holds:
    \[\redseqthree/\redex_0 \permle \redseq/\redex_0 \permeq (\redseqtwo/\redex_0) \sieve (\tm'/\redex_0)\]

    So, by definition of $(\redseqtwo/\redex_0) \sieve (\tm'/\redex_0)$
    being garbage-free,
    the fact that
    the derivation
    $((\redseqtwo/\redex_0) \sieve (\tm'/\redex_0))/(\redseqthree/\redex_0)$ is garbage
    implies that it is empty, as required.
  \end{itemize}
\end{itemize}
\end{proof}
 
\begin{lemma}[The projection after a sieve is garbage]
\llem{projection_after_sieving_is_garbage}
Let $\redseq : \tm \rtobeta \tmtwo$ and $\tm' \refines \tm$.
Then $\redseq/(\redseq \sieve \tm')$ is $(\tm'/(\redseq \sieve \tm'))$-garbage.
\end{lemma}
\begin{proof}
By induction on the length of $\redseq \sieve \tm'$.

If there are no coarse steps for $(\redseq,\tm')$,
then $\redseq \sieve \tm' = \emptyDerivation$.
Moreover, by \rprop{characterization_of_garbage}, a derivation with no
coarse steps is garbage. So $\redseq/(\redseq \sieve \tm') = \redseq$
is garbage, as required.

If there exists a coarse step for $(\redseq,\tm')$,
let $\redex_0$ be the leftmost such step,
and let $\redseqtwo = \redseq/\redex_0$ and $\tmtwo' = \tm'/\redex_0$.
Then:
\[
  \begin{array}{rcll}
  \redseq/(\redseq \sieve \tm')
  & = & \redseq/(\redex_0((\redseq/\redex_0) \sieve (\tm'/\redex_0)) & \text{ by definition} \\
  & = & (\redseq/\redex_0)/((\redseq/\redex_0) \sieve (\tm'/\redex_0)) \\
  & = & \redseqtwo/(\redseqtwo \sieve \tmtwo') \\
  \end{array}
\]
By \ih, $\redseqtwo/(\redseqtwo \sieve \tmtwo')$ is $(\tmtwo'/(\redseqtwo \sieve \tmtwo'))$-garbage.
That is:
\[
  \frac{
    \redseqtwo/(\redseqtwo \sieve \tmtwo')
  }{
    \tmtwo'/(\redseqtwo \sieve \tmtwo')
  } = \emptyDerivation
\]
Unfolding the definitions of $\redseqtwo$ and $\tmtwo'$ we have
that:
\[
  \frac{
    (\redseq/\redex_0)/((\redseq/\redex_0) \sieve (\tm'/\redex_0))
  }{
    (\tm'/\redex_0)/((\redseq/\redex_0) \sieve (\tm'/\redex_0))
  } = \emptyDerivation
\]
Equivalently:
\[
  \frac{
    \redseq/\redex_0((\redseq/\redex_0) \sieve (\tm'/\redex_0))
  }{
    \tm'/\redex_0((\redseq/\redex_0) \sieve (\tm'/\redex_0))
  } = \emptyDerivation
\]
Finally, by definition of sieve,
\[
  \frac{
    \redseq/(\redseq \sieve \tm')
  }{
    \tm'/(\redseq \sieve \tm')
  } = \emptyDerivation
\]
which means that $\redseq/(\redseq \sieve \tm')$ is $(\tm'/(\redseq \sieve \tm'))$-garbage,
as required.
\end{proof}

The proof of \rprop{properties_of_sieving} is a consequence of the preceding results, namely:
\begin{enumerate}
\item {\bf Proof that $\redseq \sieve \tm'$ is $\tm'$-garbage-free and $\redseq \sieve \tm' \permle \redseq$.}\\
  Note that $\redseq \sieve \tm'$ is $\tm'$-garbage-free by \rprop{characterization_of_garbage_free_derivations}.
  Moreover, $\redseq \sieve \tm' \permle \redseq$ by~\rlem{sieve_is_prefix}.
\item {\bf Proof that $\redseq/(\redseq \sieve \tm')$ is $(\tm'/(\redseq \sieve \tm'))$-garbage.}\\
  This is precisely~\rlem{projection_after_sieving_is_garbage}.
\item {\bf Proof that $\redseq$ is $\tm'$-garbage if and only if $\redseq \sieve \tm' = \emptyDerivation$.}\\
  An immediate consequence of~\rprop{characterization_of_garbage}.
\item {\bf Proof that $\redseq$ is $\tm'$-garbage-free if and only if $\redseq \sieve \tm' \permeq \redseq$.}\\
  An immediate consequence of~\rprop{characterization_of_garbage_free_derivations}.
\end{enumerate}

\subsection{Proof of \rprop{semilattices_of_garbage_free_and_garbage_derivations} -- Semilattices of garbage-free and garbage derivations}
\label{appendix_semilattices_of_garbage_free_and_garbage_derivations}

To prove \rprop{semilattices_of_garbage_free_and_garbage_derivations} we need a few
auxiliary lemmas:

\begin{lemma}
\llem{sieving_is_compatible_with_permutation_equivalence}
Let $\redseq \permeq \redseqtwo$. Then $\redseq \sieve \tm' = \redseqtwo \sieve \tm'$.
\end{lemma}
\begin{proof}
Observe that, given two permutation equivalent derivations $\redseq$ and $\redseqtwo$,
a step $\redex$ is coarse for $(\redseq,\tm')$ if and only if $\redex$ is coarse for $(\redseqtwo,\tm')$,
since
$
  (\redex \permle \redseq) \iff (\redex/\redseq = \emptyset) \iff (\redex/\redseqtwo = \emptyset) \iff (\redex \permle \redseqtwo)
$.
Using this observation, the proof is straightforward by induction on the length of $\redseq \sieve \tm'$.
\end{proof}

\begin{lemma}[Sieving trailing garbage]
\llem{sieving_trailing_garbage}
Let $\redseq$ and $\redseqtwo$ be composable derivations, and let $\tm' \refines \src(\redseq)$.
If $\redseqtwo$ is $(\tm'/\redseq)$-garbage, then $\redseq\redseqtwo \sieve \tm' = \redseq \sieve \tm'$.
\end{lemma}
\begin{proof}
By induction on the length of $\redseq \sieve \tm'$.
If there are no coarse steps for $(\redseq,\tm')$,
then by \rprop{characterization_of_garbage},
the derivation $\redseq$ is $\tm'$-garbage,
so $\redseq\redseqtwo$ is $\tm'$-garbage by \rprop{properties_of_garbage}.
Resorting to \rprop{characterization_of_garbage}
we obtain that $\redseq\redseqtwo \sieve \tm' = \emptyDerivation = \redseq \sieve \tm'$,
as required.

If there exists a coarse step for $(\redseq,\tm')$,
let $\redex_0$ be the leftmost such step.
Then since $\redex_0 \permle \redseq$ also $\redex_0 \permle \redseq\redseqtwo$,
so $\redex_0$ is a coarse step for $(\redseq\redseqtwo,\tm')$.
In particular, since there exists at least one coarse step for $(\redseq\redseqtwo,\tm')$,
let $\redextwo_0$ be the leftmost such step.
We argue that $\redex_0 = \redextwo_0$.
We consider two cases, depending on whether $\redextwo_0$ is coarse for $(\redseq,\tm')$:
\begin{enumerate}
\item {\bf If $\redextwo_0$ is coarse for $(\redseq,\tm')$.}
  Then $\redex_0$ and $\redextwo_0$ are both simultaneously
  coarse for $(\redseq,\tm')$ and for $(\redseq\redseqtwo,\tm')$.
  Note that $\redex_0$ cannot be to the left of $\redextwo_0$,
  since then $\redextwo_0$ would not be the leftmost coarse step for $(\redseq\redseqtwo,\tm')$.
  Symmetrically, $\redextwo_0$ cannot be to the left of $\redex_0$,
  since then $\redex_0$ would not be the leftmost coarse step for $(\redseq,\tm')$.
  Hence $\redex_0 = \redextwo_0$ as claimed.
\item {\bf If $\redextwo_0$ is not coarse for $(\redseq,\tm')$.}
  We argue that this case is impossible.
  Note that $\redextwo_0 \permle \redseq\redseqtwo$ but it is not the case that $\redextwo_0 \permle \redseq$,
  so $\redextwo_0/\redseq$ is not empty.
  Note also that $\redextwo_0/\redseq \permle \redseqtwo$,
  so by \rcoro{simulation_residuals_and_prefixes}
  we have that $(\redextwo_0/\redseq)/(\tm'/\redseq) \permle \redseqtwo/(\tm'/\redseq)$.
  Moreover, $\redseqtwo/(\tm'/\redseq) = \emptyDerivation$ is empty because $\redseqtwo$ is $(\tm'/\redseq)$-garbage.
  This means that $(\redextwo_0/\redseq)/(\tm'/\redseq) = \emptyDerivation$.
  Then we have the following chain of equalities:
  \[
    \begin{array}{rcll}
      \emptyset & = & \names((\redextwo_0/\redseq)/(\tm'/\redseq)) \\
                & = & \names((\redextwo_0/\tm')/(\redseq/\tm'))               & \text{by \rcoro{permutation_equivalence_in_terms_of_names} and \rlem{generalized_cube_lemma}} \\
                & = & \names(\redextwo_0/\tm') \setminus \names(\redseq/\tm') & \text{by \rlem{names_after_projection_along_a_step}} \\
                & = & \names(\redextwo_0/\tm')
    \end{array}
  \]
  To justify the last step, start by noting that no residual of $\redextwo_0$ is contracted along the derivation $\redseq$.
  Indeed, if $\redextwo_0$ had a residual, then $\redseq = \redseq_1\redextwo_1\redseq_2$ where $\redextwo_1 \in \redextwo_0/\redseq$.
  But recall that $\redextwo_0$ is the leftmost coarse step for $(\redseq\redseqtwo,\tm')$ and
  $\redseq_1 \permle \redseq\redseqtwo$, so it has at most one residual (\rlem{leftmost_coarse_step_has_at_most_one_residual}).
  This means that $\redextwo_0/\redseq_1 = \redextwo_1$, so $\redextwo_0/\redseq = \emptyset$, which is a contradiction.
  Given that no residuals of $\redextwo_0$ are contracted along the derivation $\redseq$,
  we have that the sets $\names(\redextwo_0/\tm')$ and $\names(\redseq/\tm')$ are disjoint
  which justifies the last step.
  
  According to the chain of equalities above, we have that $\names(\redextwo_0/\tm') = \emptyset$.
  This means that $\redextwo_0$ is $\tm'$-garbage.
  This in turn contradicts the fact that $\redextwo_0$ is coarse for $(\redseq\redseqtwo,\tm')$,
  confirming that this case is impossible.
\end{enumerate}
We have just claimed that $\redex_0 = \redextwo_0$. Then we conclude as follows:
\[
  \begin{array}{rcll}
  \redseq\redseqtwo \sieve \tm'
  & = & \redex_0((\redseq\redseqtwo/\redex_0) \sieve (\tm'/\redex_0)) \\
  & = & \redex_0(((\redseq/\redex_0)(\redseqtwo/(\redex_0/\redseq))) \sieve (\tm'/\redex_0)) \\
  & = & \redex_0((\redseq/\redex_0) \sieve (\tm'/\redex_0)) \HS\text{by \ih, as $\redseqtwo/(\redex_0/\redseq)$ is garbage by \rprop{properties_of_garbage}} \\
  & = & \redseq \sieve \tm' \\
  \end{array}
\]
\end{proof}

Let us check the two items of \rprop{semilattices_of_garbage_free_and_garbage_derivations}:
\begin{enumerate}
\item
  {\bf The set $\ulbFree{\tm'}{\tm}$ forms a finite lattice.}
  Let us check all the conditions:
  \begin{enumerate}
  \item {\bf Partial order.}
    First let us show that $\leqF$ is a partial order.
    \begin{enumerate}
    \item {\bf Reflexivity.}
      It is immediate that $\cls{\redseq} \leqF \cls{\redseq}$ holds since $\redseq/\redseq = \emptyDerivation$ is garbage.
    \item {\bf Antisymmetry.}
      Let $\cls{\redseq} \leqF \cls{\redseqtwo} \leqF \cls{\redseq}$.
      This means that $\redseq/\redseqtwo$ and $\redseqtwo/\redseq$ are garbage.
      Then:
      \[
        \begin{array}{rcll}
        \redseq
        & \permeq & \redseq \sieve \tm' & \text{ since $\redseq$ is garbage-free, by \rprop{characterization_of_garbage_free_derivations}} \\
        & \permeq & \redseq(\redseqtwo/\redseq) \sieve \tm' & \text{ since $\redseqtwo/\redseq$ is garbage, by \rlem{sieving_trailing_garbage}} \\
        & \permeq & \redseqtwo(\redseq/\redseqtwo) \sieve \tm' & \text{ since $A(B/A) \permeq B(A/B)$ in general, using \rlem{sieving_is_compatible_with_permutation_equivalence}} \\
        & \permeq & \redseqtwo \sieve \tm' & \text{ since $\redseq/\redseqtwo$ is garbage, by \rlem{sieving_trailing_garbage}} \\
        & \permeq & \redseqtwo & \text{ since $\redseqtwo$ is garbage-free, by \rprop{characterization_of_garbage_free_derivations}} \\
        \end{array}
      \]
      Since $\redseq \permeq \redseqtwo$ we conclude that $\cls{\redseq} = \cls{\redseqtwo}$,
      as required.
    \item {\bf Transitivity.}
      Let $\cls{\redseq} \leqF \cls{\redseqtwo} \leqF \cls{\redseqthree}$
      and let us show that $\cls{\redseq} \leqF \cls{\redseqthree}$.
      Note that $\redseq/\redseqtwo$ and $\redseqtwo/\redseqthree$ are garbage.
      By the fact that the projection of garbage is garbage~(\rprop{properties_of_garbage})
      the derivation $(\redseq/\redseqtwo)/(\redseqtwo/\redseqthree)$ is garbage.
      The composition of garbage is also garbage~(\rprop{properties_of_garbage}),
      so $(\redseqtwo/\redseqthree)((\redseq/\redseqtwo)/(\redseqtwo/\redseqthree))$ is garbage.
      In general the following holds:
      \[
        \begin{array}{rcll}
        \redseq/\redseqthree
        & \permle & (\redseq/\redseqthree)((\redseqtwo/\redseqthree)/(\redseq/\redseqthree))    & \text{ since $A \permle AB$ in general} \\
        & \permeq & (\redseqtwo/\redseqthree)((\redseq/\redseqthree)/(\redseqtwo/\redseqthree)) & \text{ since $A(B/A) \permeq B(A/B)$ in general} \\
        & \permeq & (\redseqtwo/\redseqthree)((\redseq/\redseqtwo)/(\redseqthree/\redseqtwo))   & \text{ since $A(B/A) \permeq B(A/B)$ in general} \\
        \end{array}
      \]
      So since any prefix of a garbage derivation is garbage~(\rprop{properties_of_garbage})
      we conclude that $\redseq/\redseqthree$ is garbage.
      This means that $[\redseq] \leqF [\redseqthree]$, as required.
    \end{enumerate}
  \item {\bf Finite.}
    Let us check that $\ulbFree{\tm'}{\tm}$ is a finite lattice, \ie that it has a finite number of elements.
    Recall that $\ulbFree{\tm'}{\tm}$ is the set of equivalence classes $\cls{\redseq}$ where $\redseq$
    is $\tm'$-garbage-free.

    On one hand, recall that the $\lambdadist$-calculus is finitely branching and strongly normalizing~(\rprop{strong_normalization}).
    So by K\"onig's lemma the set of $\todist$-derivations starting from $\tm'$
    is bounded in length. More precisely, there is a constant $M$ such that
    for any $\tmtwo' \in \termsdist$ and any derivation $\redseq : \tm' \todist^* \tmtwo'$
    we have that the length $|\redseq|$ is bounded by $M$.

    On the other hand, let $F = \set{\redseq \ST \redseq \text{ is $\tm'$-garbage-free and } \redseq \sieve \tm' = \redseq}$.
    Consider the mapping $\varphi : \ulbFree{\tm'}{\tm} \to F$
    given by $\cls{\redseq} \mapsto \redseq \sieve \tm'$, and note that:
    \begin{itemize}
    \item
      $\varphi$ does not depend on the choice of representative,
      that is, if $\cls{\redseq} = \cls{\redseqtwo}$ then $\varphi(\cls{\redseq}) = \redseq \sieve \tm' = \redseqtwo \sieve \tm' = \varphi(\cls{\redseqtwo})$.
      This is a consequence of~\rlem{sieving_is_compatible_with_permutation_equivalence}.
    \item
      $\varphi$ is a well-defined function, that is $\varphi(\cls{\redseq}) \in F$.
      This is because $\varphi(\cls{\redseq}) = \redseq \sieve \tm'$ is $\tm'$-garbage-free by \rprop{properties_of_sieving}.
      Moreover, we have that $\varphi(\cls{\redseq}) \sieve \tm' = (\cls{\redseq} \sieve \tm') \sieve \tm' = \cls{\redseq} \sieve \tm'= \varphi(\cls{\redseq})$,
      which is also a consequence of~\rlem{sieving_is_compatible_with_permutation_equivalence},
      and the fact that $\cls{\redseq} \sieve \tm' \permeq \redseq$~(\rprop{properties_of_sieving}).
    \item
      $\varphi$ is injective. Indeed, suppose that $\varphi(\cls{\redseq}) = \varphi(\cls{\redseqtwo})$,
      that is, $\redseq \sieve \tm' = \redseqtwo \sieve \tm'$.
      Then $\redseq \permeq \redseq \sieve \tm' = \redseqtwo \sieve \tm' \permeq \redseqtwo$
      by \rprop{properties_of_sieving} since $\redseq$ and $\redseqtwo$ are $\tm'$-garbage-free.
      Hence $\cls{\redseq} = \cls{\redseqtwo}$.
    \end{itemize}
    Since $\varphi : \ulbFree{\tm'}{\tm} \to F$ is injective, it suffices to show that $F$ is finite
    to conclude that $\ulbFree{\tm'}{\tm}$ is finite.
    Recall that the $\lambda$-calculus is finitely branching, so the number of derivations of a certain length
    is finite. To show that $F$ is finite, it suffices to show that the length of derivations in $F$ is
    bounded by some constant.
    Let $\redseq$ be a derivation in $F$. We have that $\redseq = \redseq \sieve \tm'$.
    By construction of the sieve, none of the steps of $\redseq \sieve \tm'$ are garbage.
    That is $\redseq = \redseq \sieve \tm' = \redex_1\hdots\redex_n$ where for all $i$ we have that
    $\redex_i/(\tm'/\redex_1\hdots\redex_{i-1}) \neq \emptyset$.
    So we have that the length of $\redseq/\tm'$ is greater than the length of $\redseq$
    for all $\redseq \in F$:
    \[
      |\redseq/\tm'| = \sum_{i=1}^{n} |\underbrace{\redex/(\tm'/\redex_1\hdots\redex_{i-1})}_{\neq \emptyset}| \geq n = |\redseq|
    \]
    As a consequence given any $\tobeta$-derivation $\redseq \in F$ we have that $|\redseq| \leq |\redseq/\tm'| \leq M$.
    This concludes the proof that $\ulbFree{\tm'}{\tm}$ is finite.
  \item {\bf Bottom element.}
    As the bottom element take $\bot_{(\ulbFree{\tm'}{\tm})} := \cls{\emptyDerivation}$.
    Observe that this is well-defined since $\emptyDerivation$ is $\tm'$-garbage-free.
    Moreover, let us show that $\bot_{(\ulbFree{\tm'}{\tm})}$ is the least element.
    Let $[\redseq]$ be an arbitrary element of $\ulbFree{\tm'}{\tm}$
    and let us check that $\bot_{(\ulbFree{\tm'}{\tm})} \leqF [\redseq]$.
    This is immediate since,
    by definition, $\bot_{(\ulbFree{\tm'}{\tm})} \leqF [\redseq]$
    if and only if $\emptyDerivation/\redseq$ is garbage.
    But $\emptyDerivation/\redseq = \emptyDerivation$ is trivially garbage.
  \item {\bf Join.}
    Let $[\redseq],[\redseqtwo]$ be arbitrary elements of $\ulbFree{\tm'}{\tm}$,
    and let us check that $[\redseq] \lorF [\redseqtwo]$ is the join.
    First observe that $[\redseq] \lorF [\redseqtwo]$ is well-defined
    \ie that $(\redseq \sqcup \redseqtwo) \sieve \tm'$ is $\tm'$-garbage-free,
    which is an immediate consequence of \rprop{characterization_of_garbage_free_derivations}.
    Moreover, it is indeed the least upper bound of $\set{[\redseq],[\redseqtwo]}$:
    \begin{enumerate}
    \item {\bf Upper bound.}
      Let us show that $[\redseq] \leqF [\redseq] \lorF [\redseqtwo]$; the proof for $[\redseqtwo]$ is symmetrical.
      We must show that $\redseq/((\redseq \sqcup \redseqtwo) \sieve \tm')$ is garbage.
      Note that $\redseq \permle \redseq \sqcup \redseqtwo$,
      so in particular $\redseq/((\redseq \sqcup \redseqtwo) \sieve \tm') \permle (\redseq \sqcup \redseqtwo)/((\redseq \sqcup \redseqtwo) \sieve \tm')$.
      Given that any prefix of a garbage derivation is garbage (\rprop{properties_of_garbage}),
      it suffices to show that $(\redseq \sqcup \redseqtwo)/((\redseq \sqcup \redseqtwo) \sieve \tm')$ is garbage.
      This is a straightforward consequence of the fact that projecting after a sieve is garbage (\rlem{projection_after_sieving_is_garbage}).
    \item {\bf Least upper bound.}
      Let $[\redseq], [\redseqtwo] \leqF [\redseqthree]$,
      and let us show that $[\redseq] \lorF [\redseqtwo] \leqF [\redseqthree]$.
      We know that $\redseq/\redseqthree$ and $\redseqtwo/\redseqthree$ are garbage,
      and we are to show that $((\redseq \sqcup \redseqtwo) \sieve \tm')/\redseqthree$
      is garbage.
      Note that $(\redseq \sqcup \redseqtwo) \sieve \tm' \permle \redseq \sqcup \redseqtwo$
      as a consequence of the fact that the sieve is a prefix (\rlem{sieve_is_prefix}).
      So in particular $((\redseq \sqcup \redseqtwo) \sieve \tm')/\redseqthree \permle (\redseq \sqcup \redseqtwo)/\redseqthree$.
      Given that any prefix of a garbage derivation is garbage (\rprop{properties_of_garbage}),
      it suffices to show that $(\redseq \sqcup \redseqtwo)/\redseqthree$ is garbage.
      But $(\redseq \sqcup \redseqtwo)/\redseqthree \permeq \redseq/\redseqthree \sqcup \redseqtwo/\redseqthree$
      so we conclude by the fact that the join of garbage is garbage (\rprop{properties_of_garbage}).
    \end{enumerate}
  \item {\bf Top element.}
    Since $\ulbFree{\tm'}{\tm}$ is finite,
    its elements are $\set{\cls{\redseqthree_1}, \hdots, \cls{\redseqthree_n}}$.
    It suffices to take $\top := \cls{\redseqthree_1} \lorF \hdots \lorF \cls{\redseqthree_n}$
    as the top element.
  \item {\bf Meet.}
    Let $\cls{\redseq},\cls{\redseqtwo} \in \ulbFree{\tm'}{\tm}$.
    Note that the set $L = \set{\cls{\redseqthree} \in \ulbFree{\tm'}{\tm} \ST \cls{\redseqthree} \leqF \cls{\redseq} \text{ and } \cls{\redseqthree} \leqF \cls{\redseqtwo}}$
    is finite because, as we have already proved, the set $\ulbFree{\tm'}{\tm}$ is itself finite.
    Then $L = \set{\cls{\redseqthree_1}, \hdots, \cls{\redseqthree_n}}$.
    Take $\cls{\redseq} \landF \cls{\redseqtwo} := \cls{\redseqthree_1} \lorF \hdots \lorF \cls{\redseqthree_n}$.
    It is straightforward to show that $\cls{\redseq} \landF \cls{\redseqtwo}$
    thus defined is the greatest lower bound for $\set{\cls{\redseq},\cls{\redseqtwo}}$.
  \end{enumerate}

\item
  {\bf The set $\ulbGarbage{\tm'}{\tm}$ forms an upper semilattice.}
  The semilattice structure of $\ulbGarbage{\tm'}{\tm}$ is inherited from $\ulbDerivLam{\tm}$.
  More precisely, $\cls{\redseq} \permle \cls{\redseqtwo}$ is declared to hold if $\redseq \permle \redseqtwo$,
  the bottom element is $\cls{\emptyDerivation}$, and the join is
  $\cls{\redseq} \sqcup \cls{\redseqtwo} = \cls{\redseq \sqcup \redseqtwo}$.
  Let us check that this is an upper semilattice:
  \begin{enumerate}
  \item {\bf Partial order.}
    The relation $(\permle)$ is a partial order on $\ulbGarbage{\tm'}{\tm}$ because it is already a partial order in $\ulbDerivLam{\tm}$.
  \item {\bf Bottom element.}
    It suffices to note that the empty derivation $\emptyDerivation : \tm \tobeta^* \tm$
    is $\tm'$-garbage, so $\cls{\emptyDerivation} \in \ulbGarbage{\tm'}{\tm}$
    is the least element.
  \item {\bf Join.}
    By \rprop{properties_of_garbage} we know that if $\redseq$ and $\redseqtwo$
    are $\tm'$-garbage then $\redseq \sqcup \redseqtwo$ is $\tm'$-garbage,
    so $\cls{\redseq} \sqcup \cls{\redseqtwo}$ is indeed the least upper bound
    of $\set{\cls{\redseq}, \cls{\redseqtwo}}$.
  \end{enumerate}
\end{enumerate}

\subsection{Proof of \rthm{factorization_ulb_derivations} -- Factorization}
\label{appendix_factorization_ulb_derivations}

We need a few auxiliary lemmas:

\begin{lemma}[Garbage-free/garbage decomposition]
\llem{garbage_free_garbage_decomposition}
Let $\redseq : \tm \rtobeta \tmtwo$ and $\tm' \refines \tm$.
Then there exist $\redseq_1,\redseq_2$ such that:
\begin{enumerate}
\item $\redseq \permeq \redseq_1\redseq_2$
\item $\redseq_1$ is $\tm'$-garbage-free,
\item $\redseq_2$ is $\tm'$-garbage.
\end{enumerate}
Moreover, $\redseq_1$ and $\redseq_2$ are unique modulo permutation equivalence,
and we have that
$\redseq_1 \permeq \redseq \sieve \tm'$ and $\redseq_2 \permeq \redseq/(\redseq \sieve \tm')$.
\end{lemma}
\begin{proof}
Let us prove that said decomposition exists and that it is unique:
\begin{itemize}
\item {\bf Existence.}
  Take $\redseq_1 := \redseq \sieve \tm'$ and $\redseq_2 := \redseq/(\redseq \sieve \tm')$.
  Then we may check the three conditions in the statement:
  \begin{enumerate}
  \item
    Recall that $\redseq \sieve \tm' \permle \redseq$ holds by \rlem{sieve_is_prefix}, so:
    \[
      \begin{array}{rcll}
      \redseq_1\redseq_2
      & = & (\redseq \sieve \tm')(\redseq/(\redseq \sieve \tm')) \\
      & \permeq & \redseq((\redseq \sieve \tm')/\redseq) & \text{ since $A(B/A) \permeq B(A/B)$ holds in general} \\
      & \permeq & \redseq                                & \text{ since $\redseq \sieve \tm' \permle \redseq$}
      \end{array}
    \]
    as required.
  \item
    The derivation $\redseq_1 = \redseq \sieve \tm'$
    is $\tm'$-garbage-free.
    This is as an immediate consequence of \rprop{characterization_of_garbage_free_derivations},
    namely the result of sieving is always garbage-free.
  \item
    The derivation $\redseq_2 = \redseq/(\redseq \sieve \tm')$ is garbage.
    This is an immediate consequence of \rlem{projection_after_sieving_is_garbage},
    namely projection after a sieve is always garbage.
  \end{enumerate}
\item {\bf Uniqueness, modulo permutation equivalence.}
  Let $\redseq \permeq \redseqtwo_1,\redseqtwo_2$
  where $\redseqtwo_1$ is $\tm'$-garbage-free,
  and $\redseqtwo_2$ is $\tm'$-garbage.
  Then we argue that $\redseqtwo_1 \permeq \redseq_1$ and $\redseqtwo_2 \permeq \redseq_2$.

  Since $\redseq_1\redseq_2 \permeq \redseq \permeq \redseqtwo_1\redseqtwo_2$
  and sieving is compatible with permutation equivalence (\rlem{sieving_is_compatible_with_permutation_equivalence})
  we have that $\redseq_1\redseq_2 \sieve \tm' \permeq \redseqtwo_1\redseqtwo_2 \sieve \tm'$.
  By \rlem{sieving_trailing_garbage}, we know that trailing garbage does not affect sieving,
  hence $\redseq_1 \sieve \tm' \permeq \redseqtwo_1 \sieve \tm'$.
  Moreover, $\redseq_1$ and $\redseqtwo_1$ are garbage-free, which
  by \rprop{characterization_of_garbage_free_derivations} means
  that $\redseq_1 \permeq \redseqtwo_1$.
  This means that $\redseq_1$ is unique, modulo permutation equivalence.
  Finally, since $\redseq_1\redseq_2 \permeq \redseqtwo_1\redseqtwo_2$ and $\redseq_1 \permeq \redseqtwo_1$,
  we have that $\redseq_1\redseq_2/\redseq_1 \permeq \redseqtwo_1\redseqtwo_2/\redseqtwo_1$,
  that is $\redseq_2 \permeq \redseqtwo_2$.
  This means that $\redseq_2$ is unique, modulo permutation equivalence.
\end{itemize}
\end{proof}

\begin{lemma}
\llem{lorf_ignores_garbage}
Let $\tm' \refines \src(\redseq) = \src(\redseqtwo)$.
Then $\cls{(\redseq \sqcup \redseqtwo) \sieve \tm'} = \cls{\redseq \sieve \tm'} \lorF \cls{\redseqtwo \sieve \tm'}$.
\end{lemma}
\begin{proof}
Let:
\[
  \begin{array}{rcl}
    \alpha & := & \redseq/(\redseq \sieve \tm') \\
    \beta  & := & \redseqtwo/(\redseqtwo \sieve \tm') \\
    \gamma & := & (\alpha/((\redseqtwo \sieve \tm')/(\redseq \sieve \tm'))) \sqcup (\beta/((\redseq \sieve \tm')/(\redseqtwo \sieve \tm'))) \\
  \end{array}
\]
Note that $\alpha$ and $\beta$ are garbage by \rlem{projection_after_sieving_is_garbage}
and hence $\gamma$ is also garbage, as a consequence of
the facts that the join of garbage is garbage
and that the projection of garbage is garbage~(\rprop{properties_of_garbage}).
Remark that, in general, $AB \sqcup CD \permeq (A \sqcup C)(B/(C/A) \sqcup D/(A/C))$.
Then the statement of this lemma is a consequence of the following chain of equalities:
\[
  \begin{array}{rcll}
  \cls{(\redseq \sqcup \redseqtwo) \sieve \tm'}
  & = &
  \cls{((\redseq \sieve \tm')\alpha \sqcup (\redseqtwo \sieve \tm')\beta) \sieve \tm'}
  & \text{by \rlem{garbage_free_garbage_decomposition}}
  \\
  & = &
  \cls{((\redseq \sieve \tm') \sqcup (\redseqtwo \sieve \tm'))\gamma \sieve \tm'}
  & \text{by the previous remark}
  \\
  & = &
  \cls{((\redseq \sieve \tm') \sqcup (\redseqtwo \sieve \tm')) \sieve \tm'}
  & \text{by~\rlem{sieving_trailing_garbage}}
  \\
  & = & \cls{\redseq \sieve \tm'} \lorF \cls{\redseqtwo \sieve \tm'}
  & \text{by definition of $\lorF$}
  \end{array}
\]
\end{proof}

To prove \rthm{factorization_ulb_derivations},
let us check that $\grothy{\ulbF}{\ulbG}$ is well-defined,
that it is an upper semilattice,
and finally that $\ulbDerivLam{\tm} \simeq \grothy{\ulbF}{\ulbG}$
are isomorphic as upper semilattices.

\begin{enumerate}
\item {\bf The Grothendieck construction $\grothy{\ulbF}{\ulbG}$ is well-defined.}
  This amounts to checking that $\ulbG$ is indeed a lax 2-functor:
  \begin{enumerate}
  \item {\bf The mapping $\ulbG$ is well-defined on objects.}
    Note that $\ulbG(\cls{\redseq}) = \ulbGarbage{\tm'/\redseq}{\tgt(\redseq)}$, which is a poset.
    Moreover, the choice of the representative $\redseq$ of the equivalence class $\cls{\redseq}$
    does not matter,
    since if $\redseq$ and $\redseqtwo$ are permutation equivalent derivations,
    then $\tm'/\redseq = \tm'/\redseqtwo$ (by \rprop{compatibility_of_simulation_residuals_and_permutation_equivalence})
    and $\tgt(\redseq) = \tgt(\redseqtwo)$.
  \item {\bf The mapping $\ulbG$ is well-defined on morphisms.}
    Let us check that given $\cls{\redseq},\cls{\redseqtwo} \in \ulbF$
    such that $\cls{\redseq} \leqF \cls{\redseqtwo}$
    then
    $
      \ulbG(\ptF{\cls{\redseqtwo}}{\cls{\redseq}}) : \ulbG(\cls{\redseq}) \to \ulbG(\cls{\redseqtwo})
    $ is a morphism of posets, \ie a monotonic function.

    First, we can see that it is well-defined,
    since if $\cls{\alpha} \in \ulbG([\redseq])$
    then the image $\ulbG(\ptF{\cls{\redseqtwo}}{\cls{\redseq})(\cls{\alpha}}) = \cls{\redseq\alpha/\redseqtwo}$
    is an element of $\ulbG([\redseqtwo])$,
    since
    $
      \redseq\alpha/\redseqtwo = (\redseq/\redseqtwo)(\alpha/(\redseqtwo/\redseq))
    $
    is garbage, as it is the composition of garbage derivations~(\rprop{properties_of_garbage}):
    the derivation $\redseq/\redseqtwo$ is garbage since $\cls{\redseq} \leqF \cls{\redseqtwo}$ (by definition),
    and the derivation $\alpha/(\redseqtwo/\redseq)$ is garbage since $\alpha$ is garbage~(\rprop{properties_of_garbage}).
    Moreover, the choice of representative does not matter,
    since if $\redseq_1 \permeq \redseq_2$ and $\redseqtwo_1 \permeq \redseqtwo_2$ and $\alpha_1 \permeq \alpha_2$
    then $\redseq_1\alpha_1/\redseqtwo_1 \permeq \redseq_2\alpha_2/\redseqtwo_2$.

    We are left to verify that $\ulbG(\ptF{\cls{\redseqtwo}}{\cls{\redseq}})$ is monotonic.
    Let $\cls{\alpha}, \cls{\beta} \in \ulbG([\redseq])$
    such that $\cls{\alpha} \permle \cls{\beta}$, and let us show that
    $\ulbG(\ptF{\cls{\redseqtwo}}{\cls{\redseq}})(\cls{\alpha}) \permle \ulbG(\ptF{[\redseqtwo]}{[\redseq]})(\cls{\beta})$.
    Indeed, $\alpha \permle \beta$, so
    $
      \redseq\alpha/\redseqtwo
      = (\redseq/\redseqtwo)(\alpha/(\redseqtwo/\redseq))
      \permle (\redseq/\redseqtwo)(\beta/(\redseqtwo/\redseq))
      = \redseq\beta/\redseqtwo
    $.
  \item {\bf Identity.}
    Let $\cls{\redseq} \in \ulbF$.
    Let us check that $\ulbG(\ptF{\cls{\redseq}}{\cls{\redseq}}) = \identity_{\ulbG(\cls{\redseq})}$
    is the identity morphism.
    Indeed, if $\cls{\alpha} \in \ulbG(\cls{\redseq})$,
    then $\ulbG(\ptF{\cls{\redseq}}{\cls{\redseq}})(\cls{\alpha}) = \cls{\redseq\alpha/\redseq} = \cls{\alpha}$.
  \item {\bf Composition.}
    Let $\cls{\redseq}, \cls{\redseqtwo}, \cls{\redseqthree} \in \ulbF$
    such that $\cls{\redseq} \leqF \cls{\redseqtwo} \leqF \cls{\redseqthree}$.
    Let us check that there is a 2-cell
    $\ulbG((\ptF{\cls{\redseqthree}}{\cls{\redseqtwo}}) \circ (\ptF{\cls{\redseqtwo}}{\cls{\redseq}})) \permle
     \ulbG(\ptF{\cls{\redseqthree}}{\cls{\redseqtwo}}) \circ \ulbG(\ptF{\cls{\redseqtwo}}{\cls{\redseq}})$.
    Note that
    $(\ptF{\cls{\redseqthree}}{\cls{\redseqtwo}}) \circ (\ptF{\cls{\redseqtwo}}{\cls{\redseq}}) : \cls{\redseq} \leqF \cls{\redseqthree}$
    is a morphism in the upper semilattice $\ulbF$ seen as a category.
    Moreover, since it is a semilattice, there a unique morphism $\cls{\redseq} \leqF \cls{\redseqthree}$,
    namely $\ptF{\cls{\redseqthree}}{\cls{\redseq}}$, so we have that
    $
      (\ptF{\cls{\redseqthree}}{\cls{\redseqtwo}}) \circ (\ptF{\cls{\redseqtwo}}{\cls{\redseq}}) = \ptF{\cls{\redseqthree}}{\cls{\redseq}}
    $.
    Now if $\cls{\alpha} \in \ulbG{\cls{\redseq}}$, then:
    \[
      \begin{array}{rcll}
      \ulbG((\ptF{\cls{\redseqthree}}{\cls{\redseqtwo}}) \circ (\ptF{\cls{\redseqtwo}}{\cls{\redseq}}))(\cls{\alpha})
      & = & \ulbG(\ptF{\cls{\redseqthree}}{\cls{\redseq}})(\cls{\alpha}) \\
      & = & \redseq\alpha/\redseqthree \\
      & \permle & (\redseq\alpha/\redseqthree)((\redseqtwo/\redseq\alpha)/(\redseqthree/\redseq\alpha)) \\
      & = & \redseq\alpha(\redseqtwo/\redseq\alpha)/\redseqthree \\
      & \permeq & \redseqtwo(\redseq\alpha/\redseqtwo)/\redseqthree \HS\text{ since $A(B/A) \permeq B(A/B)$} \\
      & = & \ulbG(\ptF{\cls{\redseqthree}}{\cls{\redseqtwo}})(\redseq\alpha/\redseqtwo) \\
      & = & (\ulbG(\ptF{\cls{\redseqthree}}{\cls{\redseqtwo}}) \circ \ulbG(\ptF{\cls{\redseqtwo}}{\cls{\redseq}}))(\cls{\alpha})
      \end{array}
    \]
  \end{enumerate}
  so
  $\ulbG((\ptF{\cls{\redseqthree}}{\cls{\redseqtwo}}) \circ (\ptF{\cls{\redseqtwo}}{\cls{\redseq}}))
  \permle
  \ulbG(\ptF{\cls{\redseqthree}}{\cls{\redseqtwo}}) \circ \ulbG(\ptF{\cls{\redseqtwo}}{\cls{\redseq}})$
  as required.

\item {\bf The Grothendieck construction $\grothy{\ulbF}{\ulbG}$ is an upper semilattice.}
  \begin{enumerate}
  \item {\bf Partial order.}
    Recall that $\grothy{\ulbF}{\ulbG}$ is always poset
    with the order given by $(a,b) \leq (a',b')$ if and only if
    $a \leqF a'$ and $\ulbG(\ptF{a'}{a})(b) \permle b'$.
  \item {\bf Bottom element.}
    We argue that $(\bot_{\ulbF},\bot_{\ulbG(\bot_\ulbF)})$ is the bottom
    element. Let $(\cls{\redseq},\cls{\redseqtwo}) \in \grothy{\ulbF}{\ulbG}$.
    Then clearly $\bot_{\ulbF} \leqF \cls{\redseq}$. Moreover,
    $
     \ulbG(\ptF{[\redseq]}{[\bot_{\ulbF}]})(\bot_\ulbG) =
     [\emptyDerivation/\redseq] = [\emptyDerivation] \permle \cls{\redseqtwo}
    $.
  \item {\bf Join.}
    Let us show that $(a, b) \lor (a',b') = (a \lorF a', \ulbG(\ptF{(a \lorF a')}{a})(b) \sqcup \ulbG(\ptF{(a \lorF a')}{a'})(b'))$.
    is the least upper bound of $\set{(a, b), (a',b')}$.
    \begin{enumerate}
    \item {\bf Upper bound.}
      Let us show that $(a, b) \leq (a, b) \lor (a',b')$.
      Recall that $(a, b) \lor (a', b') = (a \lorF a, \ulbG(\ptF{(a \lor a')}{a})(b) \sqcup \ulbG(\ptF{(a \lorF a')}{a'})(b'))$.
      First, we have $a \leqF a \lorF a'$.
      Moreover, $\ulbG(\ptF{(a \lorF a')}{a})(b) \permle \ulbG(\ptF{(a \lor a')}{a})(b) \sqcup \ulbG(\ptF{(a \lorF a')}{a'})(b')$, as required.
    \item {\bf Least upper bound.}
      Let $(a,b) = (\cls{\redseq},\cls{\redseqtwo})$
      and $(a',b') = (\cls{\redseq'},\cls{\redseqtwo'})$.
      Moreover, let $(\cls{\redseq''},\cls{\redseqtwo''})$ be an upper bound,
      \ie an element
      such that $(\cls{\redseq},\cls{\redseqtwo}) \leq (\cls{\redseq''},\cls{\redseqtwo''})$
      and $(\cls{\redseq'},\cls{\redseqtwo'}) \leq (\cls{\redseq''},\cls{\redseqtwo''})$.
      Let us show that
      $(\cls{\redseq},\cls{\redseqtwo}) \lor (\cls{\redseq'},\cls{\redseqtwo'}) \leq (\cls{\redseq''},\cls{\redseqtwo''})$.
      First note that $\cls{\redseq} \leqF \cls{\redseq''}$ and $\cls{\redseq'} \leqF \cls{\redseq''}$
      so $\cls{\redseq} \lorF \cls{\redseq'} \leqF \cls{\redseq''}$.\\
      Moreover, we know that:
      \[
        \cls{\redseq\redseqtwo/\redseq''} = \ulbG(\pt{\cls{\redseq''}}{\cls{\redseq}})(\cls{\redseqtwo}) \permle \cls{\redseqtwo''}
        \ \text{ and }\ %
        \cls{\redseq'\redseqtwo'/\redseq''} = \ulbG(\pt{\cls{\redseq''}}{\cls{\redseq'}})(\cls{\redseqtwo'}) \permle \cls{\redseqtwo''}
      \]
      Let $\alpha := (\redseq \sqcup \redseq') \sieve \tm'$.
      First we claim that $\alpha \permle \redseq\redseqtwo \sqcup \redseq'\redseqtwo'$.
      Indeed, $\alpha = (\redseq \sqcup \redseq') \sieve \tm' \permle \redseq \sqcup \redseq'$
      by \rlem{sieve_is_prefix}, and it is easy to check that
      $\redseq \sqcup \redseq' \permle \redseq\redseqtwo \sqcup \redseq'\redseqtwo'$.
      What we have to check is the following inequality:
      \[
        \ulbG(\pt{\cls{\redseq''}}{\cls{\alpha}})(
          \ulbG(\pt{\cls{\alpha}}{\cls{\redseq}})(\cls{\redseqtwo})
          \sqcup
          \ulbG(\pt{\cls{\alpha}}{\cls{\redseq'}})(\cls{\redseqtwo'})
        )
        \permle
        \cls{\redseqtwo''}
      \]
      Indeed:
      \[
        \begin{array}{rcll}
        &&
        \ulbG(\pt{\cls{\redseq''}}{\cls{\alpha}})(
          \ulbG(\pt{\cls{\alpha}}{\cls{\redseq}})(\cls{\redseqtwo})
          \sqcup
          \ulbG(\pt{\cls{\alpha}}{\cls{\redseq'}})(\cls{\redseqtwo'})
        ) \\
        & = &
        \cls{\alpha((\redseq\redseqtwo/\alpha) \sqcup (\redseq'\redseqtwo'/\alpha))/\redseq''}
        \\
        & = &
        \cls{\alpha((\redseq\redseqtwo \sqcup \redseq'\redseqtwo')/\alpha)/\redseq''}
        &\hspace{-3cm} \text{since $A/C \sqcup B/C \permeq (A \sqcup B)/C$}
        \\
        & = &
        \cls{(\redseq\redseqtwo \sqcup \redseq'\redseqtwo')(\alpha/(\redseq\redseqtwo \sqcup \redseq'\redseqtwo'))/\redseq''}
        &\hspace{-3cm} \text{since $A(B/A) \permeq B(A/B)$}
        \\
        & = &
        \cls{(\redseq\redseqtwo \sqcup \redseq'\redseqtwo')/\redseq''}
        &\hspace{-3cm} \text{since $\alpha \permle \redseq\redseqtwo \sqcup \redseq'\redseqtwo'$}
        \\
        & = &
        \cls{\redseq\redseqtwo/\redseq'' \sqcup \redseq'\redseqtwo'/\redseq''}
        &\hspace{-3cm} \text{since $A/C \sqcup B/C \permeq (A \sqcup B)/C$}
        \\
        & \permle &
        \cls{\redseqtwo''}
        &\hspace{-3cm} \text{since $\cls{\redseq\redseqtwo/\redseq''} \permle \cls{\redseqtwo''}$ and $\cls{\redseq'\redseqtwo'/\redseq''} \permle \cls{\redseqtwo''}$}
        \end{array}
      \]
    \end{enumerate}
  \end{enumerate}

\item {\bf There is an isomorphism $\ulbDerivLam{\tm} \simeq \grothy{\ulbF}{\ulbG}$ of upper semilattices.}
  As stated, the isomorphism is given by:
  \[
  \begin{array}{rclcl}
    \varphi & : & \ulbDerivLam{\tm}        & \to     & \grothy{\ulbF}{\ulbG} \\
               &&    \cls{\redseq} & \mapsto & (\cls{\redseq \sieve \tm'}, \cls{\redseq/(\redseq \sieve \tm')}) \\
  \\
    \psi    & : & \grothy{\ulbF}{\ulbG} & \to & \ulbDerivLam{\tm} \\
               &&  (\cls{\redseq},\cls{\redseqtwo}) & \mapsto & [\redseq\redseqtwo] \\
  \end{array}
  \]
  Note that $\varphi$ and $\psi$ are well-defined mappings, since their value does not
  depend on the choice of representative, due, in particular, to the fact that sieving
  is compatible with permutation equivalence~(\rlem{sieving_is_compatible_with_permutation_equivalence}).
  Let us check that $\varphi$ and $\psi$ are morphisms of upper semilatices,
  and that they are mutual inverses:
  \begin{enumerate}
  \item {\bf $\varphi$ is a morphism of upper semilattices.}
    Let us check the three conditions:
    \begin{enumerate}
    \item {\bf Monotonic.}
      Let $\cls{\redseq} \permle \cls{\redseqtwo}$ in $\ulbDerivLam{\tm}$,
      and let us show that the following inequality holds:
      \[
        \varphi(\cls{\redseq}) =
        (\cls{\redseq \sieve \tm'}, \cls{\redseq/(\redseq \sieve \tm')})
        \leq
        (\cls{\redseqtwo \sieve \tm'}, \cls{\redseqtwo/(\redseqtwo \sieve \tm')})
        = \varphi(\cls{\redseqtwo})
      \]
      We check the two conditions (by definition of $\grothy{\ulbF}{\ulbG}$):
      \begin{enumerate}
      \item On the first hand,
            $\cls{\redseq \sieve \tm'} \leqF \cls{\redseqtwo \sieve \tm'}$
            since
            \[
              \begin{array}{rcll}
              (\redseq \sieve \tm')/(\redseqtwo \sieve \tm')
              & \permle &
              \redseq/(\redseqtwo \sieve \tm') & \text{since $\redseq \sieve \tm' \permle \redseq$ by \rlem{sieve_is_prefix}} \\
              & \permle & \redseqtwo/(\redseqtwo \sieve \tm') & \text{since $\redseq \permle \redseqtwo$ by hypothesis} \\
              \end{array}
            \]
            Note that this is garbage by \rlem{projection_after_sieving_is_garbage}.
            So by \rprop{properties_of_garbage}, $(\redseq \sieve \tm')/(\redseqtwo \sieve \tm')$ is also garbage,
            as required.
      \item On the other hand, let us show that
            $\ulbG(\ptF{\cls{\redseqtwo \sieve \tm'}}{\cls{\redseq \sieve \tm'}})(\cls{\redseq/(\redseq \sieve \tm')})
            \permle
            \redseqtwo/(\redseqtwo \sieve \tm')$.
            In fact:
            \[
              \begin{array}{rcll}
                \ulbG(\ptF{\cls{\redseqtwo \sieve \tm'}}{\cls{\redseq \sieve \tm'}})(\cls{\redseq/(\redseq \sieve \tm')})
              & = &
                \cls{ (\redseq \sieve \tm')(\redseq/(\redseq \sieve \tm')) / (\redseqtwo \sieve \tm') } & \text{ by definition} \\
              & = &
                \cls{ \redseq / (\redseqtwo \sieve \tm') } & \text{ by \rlem{garbage_free_garbage_decomposition}} \\
              & \permle &
                \cls{ \redseqtwo / (\redseqtwo \sieve \tm') } & \text{ since $\redseq \permle \redseqtwo$} \\
              \end{array}
            \]
             
      \end{enumerate}
    \item {\bf Preserves bottom.}
      By definition:
      $\varphi(\bot_{\ulbDerivLam{\tm}})
       = (\cls{\emptyDerivation \sieve \tm'}, \cls{\emptyDerivation/(\emptyDerivation \sieve \tm')})
       = (\cls{\emptyDerivation},\cls{\emptyDerivation})
       = (\bot_{\ulbF},\bot_{\ulbG(\bot_\ulbF)})$.
    \item {\bf Preserves joins.}
      Let $\cls{\redseq},\cls{\redseqtwo} \in \ulbDerivLam{\tm}$, and let us show that
      $\varphi(\cls{\redseq} \sqcup \cls{\redseqtwo}) = \varphi(\cls{\redseqtwo}) \lor \varphi(\cls{\redseqtwo})$.
      Indeed, note that:
      \[
        \varphi(\cls{\redseq} \sqcup \cls{\redseqtwo}) = (\alpha,\beta)
      \]
      where
      \[
        \begin{array}{rcl}
        \alpha & = & \cls{(\redseq \sqcup \redseqtwo) \sieve \tm'} \\
        \beta  & = & \cls{(\redseq \sqcup \redseqtwo) / ((\redseq \sqcup \redseqtwo) \sieve \tm')} \\
        \end{array}
      \]
      and
      \[
        \varphi(\cls{\redseq}) \lor \varphi(\cls{\redseqtwo})
        = (\cls{\redseq \sieve \tm'}, \cls{\redseq / (\redseq \sieve \tm')}) \lor
          (\cls{\redseqtwo \sieve \tm'}, \cls{\redseqtwo / (\redseqtwo \sieve \tm')}) \\
        = (\alpha', \beta') \\
      \]
      where
      \[
        \begin{array}{rcl}
        \alpha' & = & \cls{\redseq \sieve \tm'} \lorF \cls{\redseqtwo \sieve \tm'} \\
        \beta'  & = &
               \ulbG(\ptF{\alpha}{\cls{\redseq \sieve \tm'}})(\cls{\redseq / (\redseq \sieve \tm')}) \sqcup \ulbG(\ptF{\alpha}{\cls{\redseqtwo \sieve \tm'}})(\cls{\redseqtwo / (\redseqtwo \sieve \tm')})
        \end{array}
      \]
      It suffices to show that $\alpha = \alpha'$ and $\beta = \beta'$.
      Let us show each separately:
      \begin{enumerate}
      \item {\bf Proof of $\alpha = \alpha'$.}
        The equality
        $
          \alpha
          = \cls{(\redseq \sqcup \redseqtwo) \sieve \tm'}
          = \cls{\redseq \sieve \tm'} \lorF \cls{\redseqtwo \sieve \tm'}
          = \alpha'
        $
        is an immediate consequence of \rlem{lorf_ignores_garbage}.
      \item {\bf Proof of $\beta = \beta'$.}
        Note that:
        \[
          \begin{array}{rcll}
          \beta'
                & = & \ulbG(\ptF{\alpha'}{\cls{\redseq \sieve \tm'}})(\cls{\redseq / (\redseq \sieve \tm')}) \sqcup
                      \ulbG(\ptF{\alpha'}{\cls{\redseqtwo \sieve \tm'}})(\cls{\redseqtwo / (\redseqtwo \sieve \tm')}) \\
                & = &
                     \cls{
                       (\redseq \sieve \tm')(\redseq / (\redseq \sieve \tm'))/\alpha'
                       \sqcup
                       (\redseqtwo \sieve \tm')(\redseqtwo / (\redseqtwo \sieve \tm'))/\alpha'
                     } \\
                & = &
                     \cls{\redseq/\alpha' \sqcup \redseqtwo/\alpha'}
                     &\hspace{-5cm}\text{by \rlem{garbage_free_garbage_decomposition}} \\
                & = &
                     \cls{(\redseq \sqcup \redseqtwo)/\alpha'}
                     &\hspace{-5cm}\text{since $A/C \sqcup B/C \permeq (A \sqcup B)/C$} \\
                & = &
                     \cls{(\redseq \sqcup \redseqtwo)/((\redseq \sqcup \redseqtwo) \sieve \tm')}
                     &\hspace{-5cm}\text{since $\alpha' = \alpha = (\redseq \sqcup \redseqtwo) \sieve \tm'$} \\
                & = & \beta \\
          \end{array}
        \]
        as required. 
      \end{enumerate}
    \end{enumerate}
  \item {\bf $\psi$ is a morphism of upper semilattices.}
    Let us check the three conditions:
    \begin{enumerate}
    \item {\bf Monotonic.}
      Let $(\cls{\redseq_1},\cls{\redseqtwo_1}) \leq (\cls{\redseq_2},\cls{\redseqtwo_2})$ in $\grothy{\ulbF}{\ulbG}$
      and let us show that
      $\psi(\cls{\redseq_1},\cls{\redseqtwo_1}) \permle \psi(\cls{\redseq_2},\cls{\redseqtwo_2})$
      in $\ulbDerivLam{\tm}$.
      Indeed, we know that $\ulbG(\ptF{\cls{\redseq_2}}{\cls{\redseq_1}})(\cls{\redseqtwo_1}) \permle \cls{\redseqtwo_2}$,
      that is to say $\redseq_1\redseqtwo_1/\redseq_2 \permle \redseqtwo_2$.
      Then:
      \[
        \redseq_1\redseqtwo_1/\redseq_2\redseqtwo_2
        = (\redseq_1\redseqtwo_1/\redseq_2)/\redseqtwo_2
        = \emptyDerivation
      \]
      which means that $\redseq_1\redseqtwo_1 \permle \redseq_2\redseqtwo_2$.
      This immediately implies that
      $\psi(\cls{\redseq_1},\cls{\redseqtwo_1}) \permle \psi(\cls{\redseq_2},\cls{\redseqtwo_2})$.
    \item {\bf Preserves bottom.}
      Recall that the bottom element $\bot_{(\grothy{\ulbF}{\ulbG})}$
      is defined as $(\bot_\ulbF, \bot_{\ulbG(\bot_\ulbF)})$, that is
      $(\cls{\emptyDerivation}, \cls{\emptyDerivation})$.
      Therefore $\psi(\bot_{(\grothy{\ulbF}{\ulbG})}) = \cls{\emptyDerivation} = \bot_{\ulbDerivLam{\tm}}$.
    \item {\bf Preserves joins.}
      Let $(\cls{\redseq_1},\cls{\redseqtwo_1})$ and $(\cls{\redseq_2},\cls{\redseqtwo_2})$
      be elements of $\grothy{\ulbF}{\ulbG}$, and let us show that
      $\psi((\cls{\redseq_1},\cls{\redseqtwo_1}) \lor (\cls{\redseq_2},\cls{\redseqtwo_2})) =
       \psi(\cls{\redseq_1},\cls{\redseqtwo_1}) \sqcup \psi(\cls{\redseq_2},\cls{\redseqtwo_2}))$.

      Let:
      \[
      \begin{array}{rcl}
        \alpha & := & (\redseq_1 \sqcup \redseq_2) \sieve \tm'
      \end{array}
      \]

      First we claim that $\alpha \permle \redseq_1\redseqtwo_1 \sqcup \redseq_2\redseqtwo_2$.
      This is because by \rlem{sieve_is_prefix}
      we know that $\alpha = (\redseq_1 \sqcup \redseq_2) \sieve \tm' \permle \redseq_1 \sqcup \redseq_2$.
      Moreover, it is easy to check that
      $\redseq_1 \sqcup \redseq_2 \permle \redseq_1\redseqtwo_1 \sqcup \redseq_2\redseqtwo_2$.
      Using this fact we have:
      \[
      \begin{array}{rcll}
        \psi((\cls{\redseq_1},\cls{\redseqtwo_1}) \lor (\cls{\redseq_2},\cls{\redseqtwo_2}))
      & = &
        \psi(\cls{\alpha},\ulbG(\ptF{\cls{\alpha}}{\cls{\redseq_1}})(\cls{\redseqtwo_1}) \sqcup \ulbG(\ptF{\cls{\alpha}}{\cls{\redseq_2}})(\cls{\redseqtwo_2}))
      \\
      & = &
        \psi(\cls{\alpha}, \cls{(\redseq_1\redseqtwo_1/\alpha) \sqcup (\redseq_2\redseqtwo_2/\alpha)})
      \\
      & = &
        \psi(\cls{\alpha}, \cls{(\redseq_1\redseqtwo_1 \sqcup \redseq_2\redseqtwo_2)/\alpha})
      \\&&\text{ since $A/C \sqcup B/C \permle (A \sqcup B)/C$} \\
      & = &
        \cls{\alpha((\redseq_1\redseqtwo_1 \sqcup \redseq_2\redseqtwo_2)/\alpha)} \\
      & = &
        \cls{(\redseq_1\redseqtwo_1 \sqcup \redseq_2\redseqtwo_2)(\alpha/(\redseq_1\redseqtwo_1 \sqcup \redseq_2\redseqtwo_2))} \\
      & = &
        \cls{\redseq_1\redseqtwo_1 \sqcup \redseq_2\redseqtwo_2}
        \\&&\text{ since
                     $\alpha \permle \redseq_1\redseqtwo_1 \sqcup \redseq_2\redseqtwo_2$,
                   so $\alpha/(\redseq_1\redseqtwo_1 \sqcup \redseq_2\redseqtwo_2) = \emptyDerivation$} \\
      & = &
        \psi(\cls{\redseq_1},\cls{\redseqtwo_1}) \sqcup \psi(\cls{\redseq_2},\cls{\redseqtwo_2}))
      \end{array}
      \]
      as required.
    \end{enumerate}
  \item {\bf Left inverse: $\psi \circ \varphi = \id$.}
    Let $\cls{\redseq} \in \ulbDerivLam{\tm}$.
    Then by \rlem{garbage_free_garbage_decomposition}:
    \[
      \psi(\varphi(\cls{\redseq}))
      = \psi(\cls{\redseq \sieve \tm'}, \cls{\redseq/(\redseq \sieve \tm')})
      = \cls{(\redseq \sieve \tm')(\redseq/(\redseq \sieve \tm'))}
      = \cls{\redseq}
    \]
  \item {\bf Right inverse: $\varphi \circ \psi = \id$.}
    Let $(\cls{\redseq},\cls{\redseqtwo}) \in \grothy{\ulbF}{\ulbG}$.
    Note that
    $\redseq$ is $\tm'$-garbage-free
    and $\redseqtwo$ is $\tm'$-garbage,
    so by~\rlem{sieving_trailing_garbage}
    and \rprop{characterization_of_garbage_free_derivations}
    we know that $\redseq\redseqtwo \sieve \tm' = \redseq \sieve \tm' \permeq \redseq$.
    Hence:
    \[
      \varphi(\psi(\cls{\redseq},\cls{\redseqtwo}))
      = \varphi(\cls{\redseq\redseqtwo})
      = (\cls{\redseq\redseqtwo \sieve \tm'}, \cls{\redseq\redseqtwo/(\redseq\redseqtwo \sieve \tm')})
      = (\cls{\redseq}, \cls{\redseq\redseqtwo/\redseq})
      = (\cls{\redseq}, \cls{\redseqtwo})
    \]
  \end{enumerate}
\end{enumerate}

\end{document}